%% file: thesis2012_main_hepth.tex
\documentclass[12pt]{book}
\usepackage{amsfonts,amsmath,amssymb,times}
\usepackage[numbers,sort&compress]{natbib}
\usepackage{amsthm,amscd,bm,framed,color,fancyhdr}
\usepackage[normalem]{ulem}
\usepackage{graphicx,url,setspace,relsize,numcompress}
\makeatletter

\@addtoreset{equation}{section}
\makeatother
\def\~{{\sim}}

\def\lb{\linebreak}
\newcommand{\user}[2]{{#2}}
\providecommand\phantomsection{}

\theoremstyle{definition}
\newtheorem{defn}{Definition}[chapter]
\newtheorem{pr}[defn]{Proposition}
\newtheorem{remark}[defn]{Remark}
\newtheorem{theorem}[defn]{Theorem}
\newtheorem{lemma}[defn]{Lemma}
\newtheorem{ex}[defn]{Example}

\let\Oldcdots\cdots
\renewcommand{\cdots}{{\mathsmaller{\Oldcdots}}}
\let\Oldwedge\wedge
\renewcommand{\wedge}{{\mathsmaller{\Oldwedge}}}
\let\Oldcirc\circ
\renewcommand{\circ}{{\mathsmaller{\mathsmaller{\Oldcirc}}\,}}

\begin{document}
\begin{titlepage}
\begin{center}
{\large PALACK\'Y UNIVERSITY OLOMOUC} 
\end{center}
\begin{center}
{FACULTY OF SCIENCE} 
\end{center}
\begin{center}
{DEPARTMENT OF ALGEBRA AND GEOMETRY} 
\end{center}
\begin{center}
\Huge{PARAMETER INVARIANT 
LAGRANGIAN FORMULATION OF KAWAGUCHI GEOMETRY}
\end{center}
\addvspace{15mm}
\begin{center}
{Ph.D. Thesis}
\end{center}
\addvspace{10mm}
\begin{center}
{\Large ERICO TANAKA}
\end{center}
\addvspace{28mm}
\begin{center}
\resizebox{30mm}{!}{\includegraphics{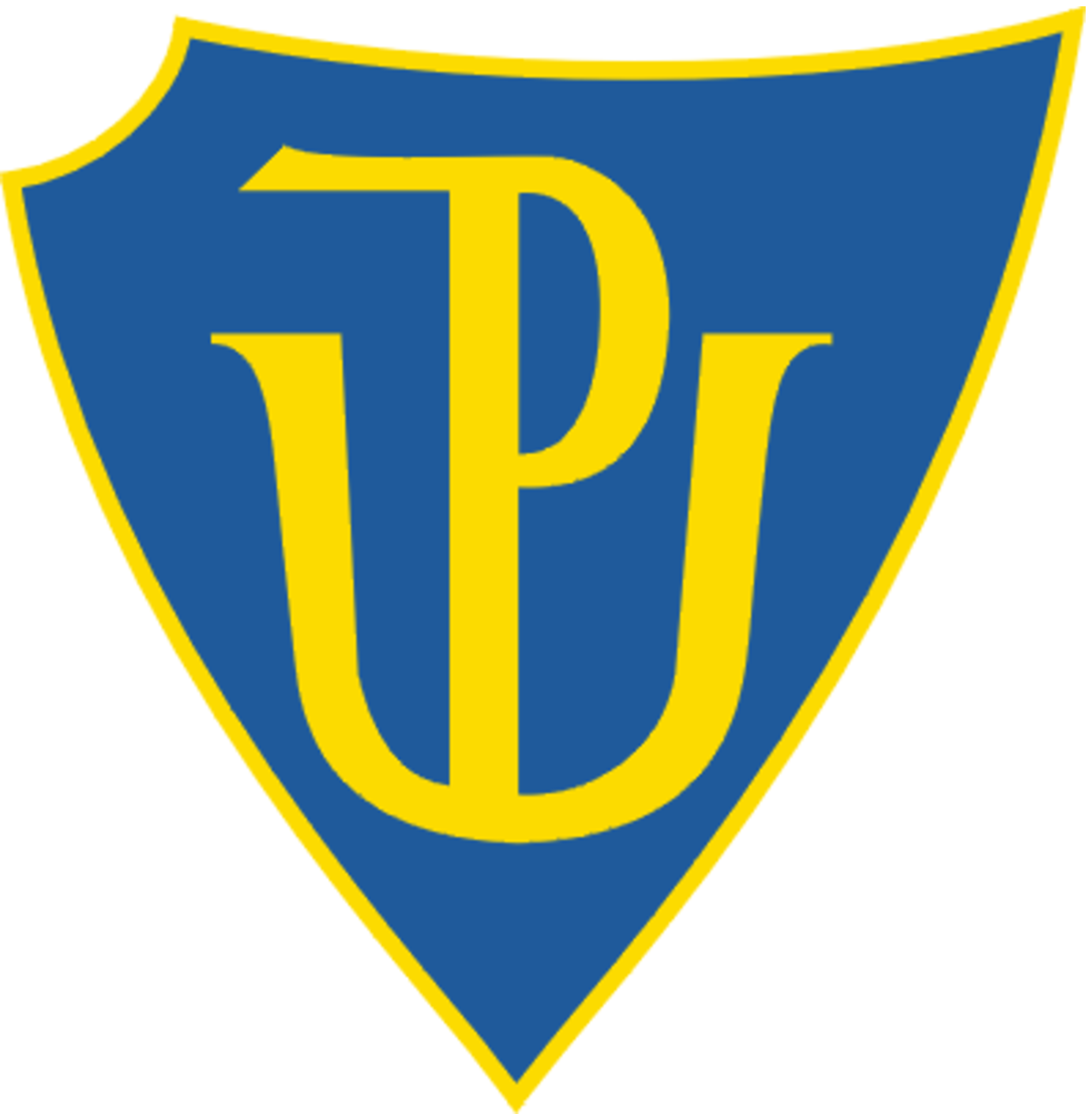}}
\end{center}
\addvspace{13mm} 
\begin{center}
{Program: P1102 Mathematics,} 
\end{center}
\begin{center}
{Global Analysis and Mathematical Physics} 
\end{center}
\addvspace{2mm}
\begin{center}
{\large Supervisor: PROF. RNDR. DEMETER KRUPKA, DRSC.}
\end{center}
\addvspace{2mm}
\begin{center}
{\large OLOMOUC 2013}
\end{center}
\addvspace{30cm}
\date{}
\end{titlepage}

\begin{center}
I declare that this dissertation is my own work. \\
All sources used for the text is either acknowledged or listed in the references.\\
\end{center}
\vspace{3cm}
\begin{flushleft}
At Olomouc 25. 6. 2013, \\
\end{flushleft}
\begin{flushright}
Erico Tanaka  \\
\vspace{2cm} 
Revised version 15.10.2013
\end{flushright}


\allowdisplaybreaks
\include{thesis2012_chap1}
\include{thesis2012_chap2}

\include{thesis2012_chap3}

\include{thesis2012_chap4}

\include{thesis2012_chap5}

\include{thesis2012_chap6}

\include{thesis2012_chap_app}
\clearpage
\addcontentsline{toc}{chapter}{\bibname}
\phantomsection

\end{document}

%% file: thesis2012_chap1.tex
\section*{Preface}
\phantomsection
\addcontentsline{toc}{section}{Preface}

This thesis was submitted to Palack\'y university, Faculty of Science, 
in partial fulfillment of the requirements for the degree of Ph.D. in Mathematics. 
Related publications are~\cite{OT2,Ta2,Ta3,OTY3}. 

The work is devoted to the constructions of Lagrangian formulation, 
which has reparameterisation invariant property. 
In the basic courses of analytical mechanics, the time is taken as the parameter, 
and the motion of a particle is described by a trajectory in a $n$-dimensional configuration space. 
Therefore, position and velocity (or momentum) of the particle is expressed as a function of time, 
and if the equations of motion which determines this trajectory could be derived from a Lagrangian, 
this Lagrangian was given as a function of these variables; namely time, position and velocity. 
Such view point matches our intuition well, 
and indeed covers wide range of physical phenomena we experience in our everyday life. 
However, since our concept on time and space changed drastically after the emerging of the theory of relativity, 
it became a major movement in theoretical and mathematical physics to reconstruct the existing theory 
in such a way that does not distinguish time as a special coordinate among the others of the same spacetime. 
The theory which treats time and space equally is said to have the property of {\it covariance}. 
In case of mechanics, the motion of a particle will be realised as a trajectory on a $(n+1)$-dimensional manifold $M$, and each point on the trajectory is the position of the particle in $M$; {\it i.e}, its local coordinate expression is given by the coordinate functions on $M$, without any preferrence to a specific one as time. In the case of relativity theory, such case was considered with the aid of Riemannian geometry. 
In this thesis, we will try to consider such situation with Finsler and Kawaguchi geometry, 
which is a generalisation of Riemannian geometry. 
While Finsler is viable for first order (velocity) Lagrangian mechanics, 
Kawaguchi is considered for higher order case. 
These geometries have further possibilities to express more complicated physical theories such as 
irreversible systems or hysterisis, etc. 
However, these problems are out of the scope of this thesis, and would be left for future research.  

There are further issues similar to the case of mechanics, in the case of field theory. 
Modern theoritical physics especially in particle physics has sought a theory, which does not depend on the choice of $k$-dimensional spacetime $M$. Now our target is shifted from position and velocity to field configurations and its derivative with respect to the spacetime. If we are to consider the ``covariant'' version of such theory in the analogy of mechanics, we should now consider a spacetime $M$ lying in a greater manifold; the total space, which is also made from the fields. 
This way of thinking is not possible with Riemannian geometry, or Finsler geometry, 
and we have to use the Kawaguchi geometry, which is still not well-established. 
In this thesis, we will treat the case only where spacetime lying in the total space is diffeomorphic to 
closed $k$-rectangle of $\mathbb{R}^k$. For the case of higher order, we will also restrict ourselves to where the total space is 
$\mathbb{R}^n$, $n \ge k$. Such approach essentialy means, 
that we do not distinguish between the fields and spacetime. 
We will call such property as {\it generalised covariance}. 

We will provide the mathematical foundations for constructing such theories, 
and introduce the Lagrange formulation in the above context. 
Some elementary examples are given to make comparison to the conventional theory. 
Though the theory is far from complete, we hope that further research will supplement the imperfection, 
and our foundations will become useful in constructing more precise and viable theory for 
both mathematics and physics. 
\\ \\

I would like to thank Prof. Demeter Krupka for his generous and enduring help for this work, 
also for giving me the chance to study at Palack\'y university. 
I also thank Prof. Lorenzo Fatibene and 
Prof. Mauro Francaviglia for both scientific discussions and support during my stay at Torino University. 
I also found the discussion with Prof. David Saunders, Prof. Takayoshi Ootsuka, invaluable and inspiring. 
Thanks also goes a lot to Prof. Zby\v{n}ek Urban and Ph.D. students at Torino University, 
who encouraged me a lot and gave me impressive questions. 
This thesis would not have finished without the help of Prof. Jan Brajer\v{c}ik and Prof. Milan Demko. 
Also I thank them for the hospitality they gave me during my stay at 
Pre\v{s}ov University. 

During my Ph.D., I was financially supported by several grants and facilities, 
Czech Science Foundation (grant 201/09/0981), 
Palack\'y University (PrF-2010-008, PrF-2011-022, PrF-2012-017), 
Czech Ministry of Education, Youth, and Sports (grant MEB 041005), 
JSPS Institutional Program for Young Researcher Overseas Visits, 
Slovensk\'a akademick\'a informa\v{c}n\'a agent\'ura (National Scholarship Programme of the Slovak Republic 
for the Support of Mobility of Researchers), 
Yukawa Institute Computer Facility, I espcially thank Prof. Olga Rossi and Prof. Jan Kuhr for the support. 

Finally, I both thank and apologise to my parents who took so much pain for supporting me abroad. 
\\ \\ \\
Jun. 2013     Olomouc, 
\begin{flushright}
Erico Tanaka
\end{flushright}

\newpage

\section*{Basic symbol list (unless otherwise stated):}
\phantomsection
\addcontentsline{toc}{section}{Basic symbol list}

$\mathbb{R}$ : real numbers \\
${\mathbb{R}^n}$: real $n$-dimensional Cartesian vector space with its natural topology \\
$V$: vector space over $\mathbb{R}$\\
${V^*}$: dual vector space of $V$\\
$M$: ${C^\infty }$-differentiable manifold \\
$TM$: tangent bundle of $M$ \\
${\Lambda ^k}TM$: all $k$-multivectors over $M$ or $k$-multivector bundle of $M$ \\
${C^\infty }(M)$: module of ${C^\infty }$-functions \\
${\mathfrak{X}}(M)$: ${C^\infty }(M)$-module of all vector fields over $M$ \\
${{\mathfrak{X}}^k}(M)$ : ${C^\infty }(M)$-module of all $k$-multivector fields over $M$ \\
${{\mathfrak{X}}^{\wedge k}}(M)$ : ${C^\infty }(M)$-module of all decomposable $k$-multivector fields over $M$ \\
${\tilde {\mathfrak{X}}^{\wedge k}}(M)$ : ${C^\infty }(M)$-module of all locally decomposable $k$-multivector fields over $M$ \\
${\Omega ^k}(M)$: ${C^\infty }(M)$-module of all differential $k$-forms over $M$ \\
${J^r}Y$ : $r$-th jet-prolongation of fibred manifold $Y$ \\
$\mathrm{pr}_k$ : Cartesian product projection onto the $k$-th set of product \\
$_n{C_k}$: Combination $\displaystyle{_n{C_k}: = \frac{{n!}}{{(n - k)!k!}}}$

\tableofcontents
\newpage
\chapter{Introduction} \label{chap_1}

	In this thesis, we will discuss the parameterisation invariant theory of Lagrangian formulation 
in terms of Finsler and Kawaguchi geometry. 
By setting up a clear and simple mathematical construction, 
we hope the theory to be viable for considering and extending the basic theories of physics, 
such as mechanics and field theory. 
The Finsler geometry is the foundation we will use for the first order mechanics, 
while Kawaguchi geometry is considered for the higher order mechanics and field theories. 
However, we would like to emphasise that in this thesis, 
these geometries are not the direct objects of our research, 
in contrary, we will only use their basic properties to build the structures we need for Lagrangian formulation. 
For instance, no fundamental tensor or connection will be discussed. 
Also it must be noted that our definition of Finsler geometry is much looser than those introduced in 
standard textbooks~\cite{ChernChenLam, Chern2}, 
for the aim to make it more applicable to the problems of physics. 
The only crucial condition we require for the Finsler function is the homogeneity condition. 
We also consider the Hilbert form as a fundamental structure rather than the Finsler function, 
which we will take as our Lagrangian.  

Kawaguchi geometry, which is the generalisation of Finsler geometry, is still in its developing state, 
and there are no standard definitions written in modern mathematical language. 
There are two directions of generalisation of Finsler geometry, 
one to higher order and another to multi-dimensional parameter space. 
In this thesis, we propose a new definition of Kawaguchi geometry, 
especially for the second order $1$-dimensional parameter space, and first order $k$-dimensional parameter space, 
using multivector bundle and a global differential form, which we call as Kawaguchi form. 
The Kawaguchi form is constructed in a way such that it satisfies similar properties as the Hilbert form in the case of Finsler geometry. We will take this Kawaguchi form as our Lagrangian. 
We will also consider the structure of second order $k$-dimensional parameter space, but only locally. 

Using these structures, we will consider the Lagrange formulation, 
and obtain the Euler-Lagrange equations that are reparameterisation invariant. 
Examples on simple case as Newton dynamics and De Broglie field is presented, 
and the results are compared with the standard formulation, which is parameter dependent. 

The reason that we expect Finsler and Kawaguchi geometry (in the above context) and the Lagrange formulation considered on these setting to be important in constructing viable theories in physics is because of the parameterisation invariance and its extendability compared to, for example, Riemannian geometry. 
Ootsuka, Tanaka and Yahagi proposed concrete examples of such application in~\cite{Oo1,OT2,OTY3}. 
This thesis will provide the mathematical background for these discussions. 
Especially, we intended to prepare a foundation that can provide a classical field theory a geometrically natural extension, which unifies the spacetime and field, in language of physics. 
Mathematically, this means we will consider the spacetime as submanifold embedded in a higher dimensional space, 
without any fibration over the parameter space. 
In this thesis, we will only consider the case where spacetime is diffeomorphic to a closed $k$-dimensional rectangle in $\mathbb{R}^k$. 

The parameter invariant theories of calculus of variations are also considered in different mathematical settings, 
notably by Grigore, D. Krupka, M. Krupka, Saunders, Urban, 
in terms of Jets and Grassmannian fibrations~\cite{Grigore1, Krupka-MKrupka, Saunders1,Urban1}. 

The structure of this thesis will be as the following. 
In the following chapter, we will begin with setting up the basic structures and definitions used in the theory. 
Basics definitions such as bundles, multivector bundles, induced charts, multi-tangent maps, 
integration on a submanifold are given. 
In Chapter 3, we will introduce the Finsler geometry and its basic properties.  
Some historical concepts are described briefly. Curves, arc segment, 
its parameterisation and length are described. 
We will also discuss the reason why the standard definition of Finsler geometry is too strict for the application to physics. 
The relation between Hilbert form and Cartan form is also presented here. 
In Chapter 4, we will introduce the Kawaguchi geometry and its basic properties.
We propose a global definition of Kawaguchi geometry, such that for the higher order case, 
the length of an arc segment will be invariant with respect to reparameterisation. 
For the multi dimensional case, we will also introduce $k$-curve, $k$-patches and its $k$-area. 
Similarly as in the previous case, 
we propose a definition of Kawaguchi geometry, such that the $k$-area of a 
$k$-patch remain unchanged by reparameterisation. We especially construct the second order 
$1$-dimensional parameter case, first order $k$-dimensional parameter case globally, 
and finally the second order $k$-dimensional parameter case locally 
(namely, the total space is $\mathbb{R}^{n}$, with $n \ge k$). 

In Chapter 5 we finally discuss on Lagrange formulation, 
using the structures introduced in the previous chapters. 
First for the Finsler case, then Kawaguchi case for second order $1$-dimensional parameter case, 
and first order $k$-dimensional case. 
The obtained results are compared for concrete example such as Newton dynamics and De Broglie field. 
We will summarise our results in Chapter 6. 
\\ \\
About the references: 

For the basic structures as manifolds, coordinate charts, tangent vectors, vector fields, we referred to the text by D. Krupka, "Advanced Analysis on Manifolds (to be published)", Y. Matsumoto, "Foundation of Manifolds('½—l'Ì'ÌŠî'b)", B. O'Neill, "Semi-Riemannian Geometry With Applications to Relativity". 
The book "Metrical differential geometry(Œv—Ê"÷•ªŠô‰½Šw)" by M. Matsumoto is one of the basic references for Finsler and Kawaguchi geometry, which unfortunately is not translated to English. 
Other references will be stated when appeared.
\\ \\
About the notations: 

Unless otherwise stated, the double occurrence of indices in the formula means summation, following the standard convention of Einstein. 
The symmetrisation of indices is denoted by round parenthesis, for example,
\begin{align}
{A_{(i}}{B_{j)}}: = \frac{1}{{2!}}\left( {{A_i}{B_j} + {A_j}{B_i}} \right)
\end{align}
The anti-symmetrisation of indices is denoted by square parenthesis, for example,
\begin{align}
{A_{[i}}{B_{j]}}: = \frac{1}{{2!}}\left( {{A_i}{B_j} - {A_j}{B_i}} \right)
\end{align}
The bases of $k$-multivector at a point $p \in M$ is often expressed $\displaystyle{{\left( {\frac{\partial }{{\partial {x^{{\mu _1}}}}} \wedge  \cdots  \wedge \frac{\partial }{{\partial {x^{{\mu _k}}}}}} \right)_p}}$, 
which abbreviates $\displaystyle{{\left( {\frac{\partial }{{\partial {x^{{\mu _1}}}}}} \right)_p} \wedge  \cdots  \wedge {\left( {\frac{\partial }{{\partial {x^{{\mu _k}}}}}} \right)_p}}$.
Throughout this thesis, we will consider a real smooth manifold, which has Hausdorff, second-countable and connected topology.
\\ \\
About the conventions:
 
The differential forms (cotangent vectors) are related to the tensor product 
by the following convention, 
\begin{align}
\alpha  \wedge \beta  := \alpha  \otimes \beta  - \beta  \otimes \alpha. 
\end{align} 
$\alpha ,\beta $ are 1-forms(cotangent vectors). 
The $k$-covector in the form 
$\displaystyle{{\alpha _1} \wedge {\alpha _2} \wedge  \cdots  \wedge {\alpha _k} 
\in {\Lambda ^k}{V^*}}$,  maps $k$ vectors ${X_1},{X_2}, \cdots ,{X_k} \in V$ to a number by \\
$\displaystyle{{\alpha _1} \wedge {\alpha _2} \wedge  \cdots  \wedge {\alpha _k}({X_1},{X_2}, \cdots ,{X_k}): = 
\det ({\alpha _i}({X_j}))}$. 
The general $k$-covector maps $k$ vectors ${X_1},{X_2}, \cdots ,{X_k} \in V$ 
to a number by its linear extension. 
In coordinate basis 
\begin{align}
&\alpha  = \frac{1}{{k!}}{\alpha _{{i_1} \cdots {i_k}}}d{x^{{i_1}}} 
\wedge  \cdots  \wedge d{x^{{i_k}}},  \quad 
{X_i} = {X_i}^j\frac{\partial }{{\partial {x^j}}}, \\ 
&\alpha ({X_1},{X_2}, \cdots ,{X_k}): = {\alpha _{{i_i} \cdots {i_k}}}{X_1}^{{j_1}} 
\cdots X_k^{{j_k}}\delta _{{i_1}}^{[{j_1}} \cdots \delta _{{i_k}}^{{j_k}]}
={\alpha _{{i_i} \cdots {i_k}}}{X_1}^{{i_1}} \cdots X_k^{{i_k}}.
\end{align}

%% file: thesis2012_chap2.tex
\chapter{Basic structures} \label{chap_2}
Our aim here is to set up the space where Finsler and Kawaguchi structure will be endowed, 
and the calculus of variation could be carried out. 
We begin by defining standard vector bundle structures and 
the products of them, and introduce the $k$-multivector bundle which is a 
$k$-fold antisymmetric tensor product bundle. 
This bundle will be the stage for considering the theory of calculus of variation 
for first order $k$ parameter space, i.e., first order field theory, 
where the parameter space corresponds to the spacetime. 
For constructing the stage for second order field theory, we will extend this concept to second order 
$k$-antisymmetric tensor product space by standard manipulations on vector bundles. 
We will begin by introducing the fibred manifold and fibred coordinates, 
trivialisation and concept of bundles, then additionally the properties of a vector space to consider 
vector bundles and its product spaces. Then we will further consider the sub-bundles of them, 
by certain bundle isomorphism. The obtained final bundle will be the stage for our parameterisation 
invariant theory of calculus of variation, which we apply to the field theory. 
This chapter mostly refers to the textbook by D. J. Saunders, 
"The Geometry of Jet Bundles", with slight changes in notations. 

\section{Bundles}  \label{sec_bundles}

\begin{defn}Fibred manifold \\
A {\it fibred manifold} is a triple $(E,\pi ,M)$ where $E$ and $M$ are manifolds, and 
$\pi :E \to M$ is a surjective submersion. $E$ is called the {\it total space}, $M$ the {\it base space}, 
and $\pi $ is a {\it projection}. 
The subset ${\pi ^{ - 1}}(p) \subset E$ over each point $p \in M$ is called a {\it fibre}, 
and is usually denoted by ${E_p}$. 
\end{defn}
We occasionally use the projection $\pi $ to denote the total fibred manifold, instead of writing the triple.

\begin{defn}   Adapted chart of a fibred manifold \\
Let $(E,\pi ,M)$ be a fibred manifold such that $\dim M = n,\dim E = n + m$. 
The {\it adapted chart} of an open set $V \subset E$ is a chart 
$(V,\psi ),\psi  = ({y^\nu }),\nu  = 1, \cdots ,n + m$, such that for any points 
$a,b \in V$, $\pi (a) = \pi (b) = p$, $p \in M$, 
then $p{r_1}(\psi (a)) = p{r_1}(\psi (b))$, where $p{r_1}: = {\mathbb{R}^{n + m}} \to {\mathbb{R}^n}$
. 
\end{defn}

The existence of such adapted chart is guaranteed by the following lemma. 

\begin{lemma}
Let $f:E \to M$ be a ${C^r}(1 \leqslant r \leqslant \infty )$ map, 
and $\dim M = n,\;\dim E = n + m$. 
If $f$ is a submersion at ${q_0} \in E$, that is, if for the neighbourhood $V$ of ${q_0}$, 
the tangent map ${T_q}f : {T_q}E \to {T_{f(q)}}M$ at $\forall q \in V$ is surjective and has constant rank, 
then there exists a chart $(V,\psi ), \psi  = ({y^\nu })$ on $E$ such that, 
for any $a, b \in V$, $f(a) = f(b) = p$, $p \in M$, then $p{r_1}(\psi (a)) = p{r_1}(\psi (b))$, 
where $p{r_1}: = {\mathbb{R}^{n + m}} \to {\mathbb{R}^n}$.
\end{lemma}
For the proof, we refer to~\cite{Krupka1}, \cite{Saunders2}.

The adapted chart $(V,\psi )$ on the total space $E$ induces a chart on the base space $M$ by $\pi $. 
This induced chart can be denoted as $(\pi (V),\varphi )$, where the coordinate function 
$\varphi :\pi (V) \to {\mathbb{R}^n}$ is given by setting 
$\varphi (\pi (a)) = p{r_1}(\psi (a)),a \in V$. 
It is convenient to use the same notation for the first $n$ coordinate functions for both 
$\psi $ and $\varphi $, so that, $(V,\psi )$, $\psi  = ({x^1}, \cdots ,{x^n},{y^1}, \cdots ,{y^m})$ and 
$(\pi (V), \varphi )$, $\varphi  = ({x^1}, \cdots ,{x^n})$.

\begin{defn}   Local trivialisation \\
Let $(E,\pi ,M)$ a fibred manifold and $p \in M$. 
A {\it local trivialisation of $\pi $ around $p$} is a triple 
$({U_p},{F_p},{t_p})$, where ${U_p}$ is a neighbourhood of $p$, 
${F_p}$ is a manifold, and 
\begin{align}
{t_p}:{\pi ^{ - 1}}({U_p}) \to {U_p} \times {F_p}
\end{align}
is a diffeomorphism satisfying the condition 
$p{r_1} \circ {t_p} = \pi {|_{{\pi ^{ - 1}}({U_p})}}$, with 
$p{r_1}$ denoting the Cartesian product projection onto the first set. 
\end{defn}

\begin{defn}   Bundle \\
A fibred manifold which has at least one local trivialisation around each point on 
$M$ is called a {\it bundle}. 
\end{defn}

We occasionally use the projection $\pi$ to denote the bundle itself, instead of writing the triple.

\begin{ex}  Tangent bundle \\
For the case $E = TM$, a tangent space of $M$, 
the triple $(TM, \tau ,M)$ with $\tau $ being a natural projection which sends each tangent vector 
${v_p} \in {T_p}M$ at point $p \in M$ to $p \in M$ becomes a bundle, and is called the tangent bundle. 
\end{ex}

Each chart on $M$ induces a local trivialisation of $TM$. 
The local trivialisation of the tangent bundle around $p \in M$ could be introduced in the following way. 
Let $({U_p},\varphi )$, $\varphi  = ({x^\mu }), \mu  = 1, \cdots ,n$, such that 
$p \in {U_p}$ be a chart on the base space $M$. 
For any element $\xi$ of ${\tau ^{ - 1}}({U_p}) \subset TM$ at $p$ have the coordinate expression, 
$\displaystyle{\xi  = {\xi ^\mu }{\left( {\frac{\partial }{{\partial {x^\mu }}}} \right)_p}}$. 
Then the trivialisation $({U_p}, {F_p}, {t_p})$ induced by the chart $({U_p}, \varphi )$ on the base space 
$M$ is $({U_p}, {\mathbb{R}^n}, {t_p})$, where ${t_p}$ is given by 
${t_p}(\xi ) = (\tau (\xi ), ({\xi ^1}, \cdots , {\xi ^n}))$. 
Let $(U,\varphi)$ be a chart on $M$.
In general, the set ${\pi ^{ - 1}}(U)$, $U \subset M$ of a bundle 
$(E,\pi ,M)$ cannot be covered by a single chart. H
owever, for the case of a tangent bundle this is possible. 
The charts on $M$ induces such specific charts on $TM$ via local trivialisation. 

\begin{defn}   Induced charts of a tangent bundle  \label{TMchart} \\
Let $(TM,\tau ,M)$ be a tangent bundle with $\dim M = n$, 
and $(U,\varphi)$, $\varphi = ({x^\mu })$ a chart on $U \subset M$. 
Since for any $p \in M$ we have the local trivialisation 
$({U_p}, {\mathbb{R}^n}, {t_p})$, ${t_p}({\tau ^{ - 1}}({U_p})) = {U_p} \times {\mathbb{R}^n}$, 
we define the {\it induced chart of a tangent bundle} on ${\tau ^{ - 1}}({U_p}) \subset TM$ for any 
$p \in M$ by $({\tau ^{ - 1}}({U_p}),\psi )$, 
$\psi (\xi ) = (\varphi (\tau (\xi )),({\xi ^1}, \cdots ,{\xi ^n}))$, where 
$\xi  \in {\tau ^{ - 1}}({U_p})$. \\
Especially, we may use the convenient expressions as 
$\psi  = ({x^1},{x^2}, \cdots ,{x^n},{y^1},{y^2}, \cdots ,{y^n})$, where 
${y^\mu }(\xi ) = {\xi ^\mu },\mu  = 1, \cdots ,n$. 
\end{defn}

These induced charts define on $TM$ the structure of ${C^\infty }$-manifold of dimension $2n$.

\begin{defn}   Global trivialisation \\
Let $(E,\pi ,M)$ be a fibred manifold. A {\it global trivialisation of $\pi $} 
is a triple $(M,F,t)$, where $F$ is a manifold, 
and $t:E \to M \times F$ is a diffeomorphism satisfying the condition 
$p{r_1} \circ t = \pi $, with $p{r_1}$ denoting the Cartesian product projection onto the first set. 
$F$ is called a {\it typical fibre of $\pi $}.
\end{defn}

\begin{defn}   Trivial fibre bundle \\
A fibred manifold which has a global trivialisation is called a {\it trivial fibre bundle}. 
\end{defn}

We show in Figure \ref{fig_fibrebundles}, three diagrams of fibred manifoldsF general, bundle, and trivial. 

\begin{figure}
  \centering
  \includegraphics[width=15cm]{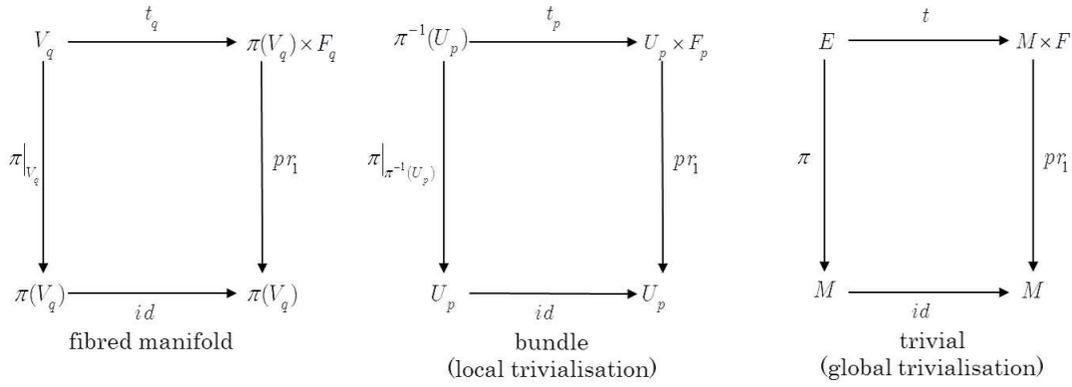}
  \caption{Fibre bundles} \label{fig_fibrebundles}
\end{figure}
 
\begin{defn} Fibred product bundles \\
Let $(E,\pi ,M)$ and $(H,\rho ,M)$ be bundles over the same base manifold $M$. 
The {\it fibred product bundle} is a triple 
$(E{ \times _M} H, \pi { \times _M}\rho ,M)$, where the total space $E{ \times _M}H$ is defined by 
\begin{align}
E \times_M H :=  \{ (p,q) \in E \times H:\pi (p) = \rho (q)\}
\end{align}
and the projection map $\pi { \times _M}\rho $ is defined by 
\begin{align}
(\pi { \times _M}\rho )(p,q) = \pi (p) = \rho (q).
\end{align}
\end{defn}

Bellow we will check that the triple $(E{ \times _M}H,\pi { \times _M}\rho ,M)$ indeed has a bundle structure. 

First, the total space $E{ \times _M}H$ is a submanifold of $E \times H$, 
since $E{ \times _M}H = {(\pi  \times \rho )^{^ - 1}}{\Delta _M}$, where ${\Delta _M}$ is the diagonal set 
\begin{align}
{\Delta _M} = \{ (p,q) \in M \times M \,|\, p,q \in M, \, p = q \}.
\end{align}
Suppose the adapted charts on ${V_E} \subset E$, ${V_H} \subset H$, such that 
$\pi ({V_E}) \cap \rho ({V_H}) \ne \emptyset $ are given by $({V_E}, {\psi _E})$, 
${\psi _E} = ({x^\mu },{y^a})$, $({V_H},{\psi _H}),{\psi _H} = ({x^\mu },{z^A})$, 
then it induces a chart on $U = (\pi ({V_E}) \cap \rho ({V_H})) \subset M$, namely, 
$(U,\varphi)$, $\varphi  = ({x^\mu })$. 
Let the adapted chart on ${V_{E \times H}} \subset E \times H$ be 
$({V_{E \times H}}$, ${\psi _{E \times H}})$, with 
${V_{E \times H}} = {(\pi  \times \rho )^{ - 1}}U$, 
${\psi _{E \times H}}(p,q) = \left( {\varphi (\pi (p)), {y^a}(p), \varphi (\rho (q)),{z^A}(q)} \right)$, 
where 
$(p,q) \in E \times H$. 
Then by considering the coordinate expressions of equations of a submanifold 
$ \varphi(\pi (p)) = \varphi(\rho (q))$, the total space $E{ \times _M}H$ has an adapted chart 
$({V_{E{ \times _M}H}}, {\psi _{E{ \times _M}H}})$, where
\begin{align}
&{V_{E{ \times _M}H}} = (E{ \times _M}H) \cap {V_{E \times H}}, \nonumber  \\
&{\psi _{E{ \times _M}H}}(p,q) = \left( {\varphi (\pi (p)),{y^a}(p),{z^A}(q)} \right),
\end{align}
with $(p,q) \in E \times H,\pi (p) = \rho (q)$. 

These charts define on $E{ \times _M}H$ the structures of a ${C^\infty }$-manifold. 

Next, we will obtain the local trivialisation $({U_r},{F_r},{t_r})$ on every point of $r \in M$ 
in the following way. 

Since $(E,\pi ,M)$ and $(H,\rho ,M)$ are both bundles, 
they have a local trivialisation around each point $r \in M$. 
Denote these as $({U_r},{E_r},{s_r})$, 
$({V_r}, {H_r}, {u_r})$, then ${s_r} : {\pi ^{ - 1}}({U_r}) \to {U_r} \times {E_r}$ ,
${u_r} : {\rho ^{ - 1}}({V_r}) \to {V_r} \times {H_r}$. 
Then for any $r \in M$, we can obtain the local trivialisation of 
$(E{ \times _M}H, \pi { \times _M}\rho, M)$ around $r$ by 
${t_r} : {(\pi { \times _M} \rho )^{ - 1}}({U_r} \cap {V_r}) \to ({U_r} \cap {V_r}) \times {F_r}$, 
where ${F_r} = {E_r} \times {H_r}$, and 
${t_r}(p,q) = (\pi (p), {y^a }(p), {z^A }(q)) = (\rho (q), {y^a }(p), {z^A }(q))$. 
Accordingly, the triple $(E{ \times _M}H, \pi { \times _M} \rho, M)$ becomes a bundle. 

\begin{defn}  Bundle morphism \\
If $(E,\pi ,M)$ and $(H,\rho ,N)$ are bundles, then a {\it bundle morphism from $\pi $ to $\rho $} is a pair 
$(f,\bar f)$ where $f:E \to H$, $\bar f:M \to N$ and $\rho  \circ f = \bar f \circ \pi $. 
The map $\bar f$ is called a {\it projection} of $f$. 
\end{defn}

\begin{ex}  \label{ex_tangentbundlemorphism}
Let $(E,\pi ,M)$ be a bundle, and $(TE,T\pi ,TM)$ its tangent bundle. There is a bundle morphism from 
$T\pi $ to $\pi $, which is a pair of tangent bundle projections $({\tau _E},{\tau _M})$. 
\end{ex}

\begin{ex}  
Let $(E,\pi ,M)$ be a bundle, and $(X,\bar X)$ a bundle morphism to $(TE,T\pi ,TM)$. 
$X$ is called a {\it projectable vector field}, and its projection $\bar X$ satisfies the relation, 
$\bar X \circ \pi  = T\pi  \circ X$. 
\end{ex}

\begin{defn} Pull-back bundle \\
Let $(E,\pi ,M)$ be a bundle, and $f:N \to M$ a smooth map. 
The pull-back of the bundle $\pi $ by $f$ is denoted by 
$({f^*}E,{f^*}\pi ,N)$,
where the total space ${f^*}E$ is defined by, 
\begin{align}
{f^*}E = \{ (u,x) \in E \times N \, | \, \pi (u) = f(x)\},
\end{align}
that is, ${f^*}E = E{ \times _M} N$, and the projection map ${f^*}\pi $ is defined by 
${f^*}\pi (u,x) = x$. 
The bundle ${f^*}\pi $ is called a {\it pull-back bundle} of $\pi $ by $f$.  
\end{defn}

\begin{defn}  Sub-bundle \\
If $(E,\pi ,M)$ is a bundle and $E' \subset E$ is a submanifold such that the triple 
$(E',{\left. \pi  \right|_{E'}},\pi (E'))$ is itself a bundle, then 
${\left. \pi  \right|_{E'}}$ is called a {\it sub-bundle} of $\pi $. 
\end{defn}

\begin{defn}  Vector bundle \\
A {\it vector bundle} is a quintuple $(E,\pi ,M,\sigma ,\mu )$ with the following structures, \\
1. $(E,\pi ,M)$ is a bundle \\
2. Denote the fibre over $p \in M$ as ${E_p}$. Namely, ${E_p} = {\pi ^{ - 1}}(p)$. Then, \\
(a) $\sigma :E{ \times _M}E \to E$ satisfies, for each $p \in M$, $\sigma ({E_p} \times {E_p}) \subset {E_p}$ \\
(b) $\mu :\mathbb{R} \times E \to E$ satisfies, for each $p \in M$, $\mu (\mathbb{R} \times {E_p}) \subset {E_p}$ \\
(c) $({E_p},{\left. \sigma  \right|_{{E_p} \times {E_p}}},{\left. \mu  \right|_{\mathbb{R} \times {E_p}}})$ is a real vector space for each $p \in M$ \\
3. for each $p \in M$, there is a local trivialization $({W_p},{\mathbb{R}^n},{t_p})$ called a {\it linear local trivialisation}, satisfying the condition that, for $q \in {W_p}$, the map $p{r_2} \circ {\left. {{t_p}} \right|_{{E_p}}}:{E_p} \to {\mathbb{R}^n}$ where 
${\left. {{t_p}} \right|_{{E_p}}}:{E_p} \to \{ q\}  \times {\mathbb{R}^n}$ and $p{r_2}:\{ q\}  \times {\mathbb{R}^n} \to {\mathbb{R}^n}$, is a linear isomorphism. 
\end{defn}
Under the linear local trivialization, the maps $\sigma $ and $\mu $ correspond to addition of vectors on ${\mathbb{R}^n}$, and scalar multiplication of vectors on ${\mathbb{R}^n}$, respectively. 

\begin{ex} 
The tangent bundle $(TM,\tau ,M)$ is a vector bundle. The linear local trivialisation around each point $p \in M$ is given by $({U_p},{\mathbb{R}^n},{t_p})$, for each $p \in M$, the fibre ${T_p}M$ has the property of vector space, therefore $a{v_1} + b{v_2} \in {T_p}M,\;a,b \in \mathbb{R},\;{v_1},{v_2} \in {T_p}M$. 
\end{ex}

\begin{defn}  Vector bundle adapted charts   \label{def_localsectionofvb} \\ 
Let $m = \dim M$, $n = \dim E$, and $({\pi ^{ - 1}}(W),\psi ),\psi  = ({u^1}, \cdots ,{u^n})$ the adapted coordinates on $(E,\pi ,M)$ induced by the chart $(W,\varphi ),\varphi  = ({x^1}, \cdots ,{x^m})$ on $W \subset M$, chosen to be linear on $\pi $. Such charts are called {\it vector bundle adapted charts}. 
\end{defn}
In such charts, the elements of $E$ can be expressed by $\xi  = {\xi ^\alpha }{e_\alpha }$, $\alpha  = 1, \cdots ,m$, where the base ${e_\alpha } \in {\Gamma _W}(\pi )$ is a family of local section defined by ${u^\beta }({e_\alpha }(p)) = \delta _\alpha ^\beta $ for all $p \in W$. In this way, linear operations on sections could be defined pointwise. 

\begin{ex} 
The vector bundle adapted chart of the tangent bundle $(TM,\tau ,M)$, 
is the induced chart $({\tau ^{ - 1}}(U),\psi )$,$\psi (\xi ) = ({x^\mu }(\tau (\xi )),{\xi ^\mu })$ 
given in (\ref{TMchart}), and the corresponding local sections are the vector fields 
$\displaystyle{\frac{\partial }{{\partial {x^\mu }}}}$. 
\end{ex}
 
\begin{ex} 
The vector bundle adapted chart of the cotangent bundle $({T^*}M,{\tau ^*},M)$, is the induced chart $({({\tau ^*})^{ - 1}}(U),{\psi ^*})$,${\psi ^*}(\alpha ) = ({x^\mu }({\tau ^*}(\alpha )),{\alpha _\mu })$, where $\alpha  \in {T^*}M$, and the corresponding local sections are the 1-forms $d{x^\mu }$.  
\end{ex}

\begin{defn}  Tensor product \\
Let $(E,\pi ,M)$ and $(F,\rho ,M)$ be vector bundles with fibres ${E_p}$, ${F_p}$ respectively. The {\it tensor product} of $\pi $ and $\rho $ is the vector bundle with fibres ${E_p} \otimes {F_p}$ and is denoted $(E \otimes F,\pi  \otimes \rho ,M)$. 
\end{defn}

\begin{defn} :  Antisymmetric/symmetric tensor product \\
Let $({E_1},{\pi _1},M), \cdots ,({E_k},{\pi _k},M)$ be vector bundles. 
The {\it $k$-fold antisymmetric tensor product} is a vector bundle with completely antisymmetric properties 
denoted by $({E_1} \wedge {E_2} \wedge  \cdots  \wedge {{\rm E}_k}, \lb[3] 
{\pi _1} \wedge {\pi _2} \wedge  \cdots  \wedge {\pi _k},M)$. 
For $k = 2$, every fibre is defined by the antisymmetric product 
\begin{align}
{E_{1,p}} \wedge {E_{2,p}} = {E_{1,p}} \otimes {E_{2,p}} - {E_{2,p}} \otimes {E_{1,p}}
\end{align}
for all $p \in M$. 
Similarly, the {\it $k$-fold symmetric product} is a vector bundle with completely symmetric properties 
denoted by $({E_1} \odot {E_2} \odot  \cdots  \odot {{\rm E}_k},{\pi _1} \odot {\pi _2} \odot  
\cdots  \odot {\pi _k},M)$. For $k = 2$, every fibre is defined by the symmetric product 
\begin{align}
{E_{1,p}} \odot {E_{2,p}} = {E_{1,p}} \otimes {E_{2,p}} + {E_{2,p}} \otimes {E_{1,p}}
\end{align}
for all $p \in M$. 
\end{defn}
Any tensor product of tangent and cotangent bundles which has the same base space $M$ become a tensor product bundle over $M$. 

\subsection{Second order tangent bundle } \label{subsec_second_TM}
Here we will review the definitions of ${T^2}M$ and its basic properties which will be used for 
considering the theory of second order mechanics. 
Higher order theory could be constructed by similar induction, 
and is introduced briefly in the following section. 
These method of construction will be also used in order to 
create the basic structures for the higher order field theory, which later will be introduced in this chapter.

\begin{defn}  Second order tangent bundle ${T^2}M$ over $TM$ \label{def_2ndorder_TMoverTM} \\
Let $(TM, {\tau _M}, M)$ be the tangent bundle with base space $M$ and 
$(TTM, {\tau _{TM}}, TM)$ the tangent bundle with the base space $TM$. 
The tangent map $T{\tau _M}$ of ${\tau _M}$ also induces the bundle $(TTM, T{\tau _M}, TM)$. 
Denote the subset of elements $\xi  \in TTM$ which satisfy the equations of submanifold 
\begin{align}
{\tau _{TM}}(\xi ) = T{\tau _M}(\xi )
\end{align}
as ${T^2}M$, and a map from 
${T^2}M$ to $TM$ by $\tau _M^{2,1}: = {\left. {{\tau _{TM}}} \right|_{{T^2}M}}$. 
The triple $({T^2}M,\tau _M^{2,1},TM)$ becomes a bundle, which is a sub-bundle of 
$(TTM,{\tau _{TM}},TM)$. We will call this a {\it second order tangent bundle over $TM$}. 
\end{defn}

To check $({T^2}M,\tau _M^{2,1},TM)$ is a bundle, we first introduce manifold structure on ${T^2}M$, 
namely we will introduce charts, and prove that these charts define a smooth structure on ${T^2}M$. 
Let $\dim M = n$, and the chart on $U \subset M$ be $(U,\varphi ),\varphi  = ({x^\mu })$. 
The induced chart on $TM$ and $TTM$ is 
$(V,\psi )$, $V = \tau _M^{ - 1}(U)$, $\psi = ({x^\mu }, {\dot x^\mu })$, 
and $({\tilde V^2}, {\tilde \psi ^2})$, ${\tilde V^2} = \tau _{TM}^{ - 1}(V)$, 
${\tilde \psi ^2} = ({x^\mu },{\dot x^\mu },{y^\mu },{\dot y^\mu })$ respectively. 
The elements of ${T^2}M$ have a specific form, namely they satisfy the condition 
${\tau _{TM}}(\xi ) = T{\tau _M}(\xi )$, $\xi  \in TTM$. 
Let the local expression of ${\xi _q} \in {T_q}TM,q \in V \subset TM$ be 
\begin{align}
{\xi _q} = \xi _1^\mu {\left( {\frac{\partial }{{\partial {x^\mu }}}} \right)_q} 
+ \xi _2^\mu {\left( {\frac{\partial }{{\partial {{\dot x}^\mu }}}} \right)_q},
\end{align}
then ${\tau _{TM}}({\xi _q}) = q$ and $\displaystyle{T{\tau _M}({\xi _q}) 
= \xi _1^\mu {\left( {\frac{\partial }{{\partial {x^\mu }}}} \right)_{{\tau _M}(q)}}}$. 
In coordinate expressions, 
\begin{align}
({x^\mu }(q),{\dot x^\mu }(q)) = \left( {{x^\mu }\left( {\xi _1^\nu 
{{\left( {\frac{\partial }{{\partial {x^\nu }}}} \right)}_{{\tau _M}(q)}}} \right),{{\dot x}^\mu }
\left( {\xi _1^\nu {{\left( {\frac{\partial }{{\partial {x^\nu }}}} \right)}_{{\tau _M}(q)}}} \right)} \right) 
= \left( {({x^\mu } \circ {\tau _M})(q),\xi _1^\mu } \right).
\end{align}
Therefore, the elements of ${T^2}U$ have the form 
\begin{align}
{v_q} = {\dot x^\mu }(q){\left( {\frac{\partial }{{\partial {x^\mu }}}} \right)_q} 
+ {v^\mu }{\left( {\frac{\partial }{{\partial {{\dot x}^\mu }}}} \right)_q}, 
\end{align}
for $q \in V$. 
Furthermore, since every elements of ${T^2}U$ in chart expression are in the form, 
\begin{align}
({x^\mu }({v_q}) , {\dot x^\mu }({v_q}) , {y^\mu }({v_q}) , {\dot y^\mu }({v_q})) 
= (({x^\mu } \circ {\tau _M}(q) , {\dot x^\mu }(q) , {\dot x^\mu }(q) , {v^\mu }),
\end{align} 
we can choose a chart 
$ ({V^2} , {\psi ^2})$, ${\psi ^2} = ({x^\mu } , {\dot x^\mu } , {\ddot x^\mu })$, 
${V^2} = {\tilde V^2} \cap {T^2}M, \, \mu  = 1, \dots ,n$ on ${T^2}M$. 

Clearly the fibres of $\tau _M^{2,1}$ are not vector spaces, 
since in general the scalar multiplication of ${v_q}$ does not belong to the same fibre. 

Now, let $({V^2}, {\psi ^2})$, ${\psi ^2} = ({x^\mu }, {y^\mu }, {z^\mu })$ and 
$({\bar V^2}, {\bar \psi ^2})$, ${\bar \psi ^2} = ({\bar x^\mu }, {\bar y^\mu }, {\bar z^\mu })$ 
be two charts on ${T^2}M$, with ${V^2} \cap {\bar V^2} \ne \emptyset$. 
Then express the element 
${\xi _q} \in {V^2} \cap {\bar V^2}$, $q \in TM$, by these charts, 
\begin{align}
{\xi _q} = {y^\mu }(q){\left( {\frac{\partial }{{\partial {x^\mu }}}} \right)_q} 
 + {\xi ^\mu }{\left( {\frac{\partial }{{\partial {y^\mu }}}} \right)_q} 
 = {\bar y^\mu }(q){\left( {\frac{\partial }{{\partial {{\bar x}^\mu }}}} \right)_q} 
 + {\bar \xi ^\mu }{\left( {\frac{\partial }{{\partial {{\bar y}^\mu }}}} \right)_q},
\end{align}
with ${z^\mu }({\xi _q}) = {\xi ^\mu },\,{\bar z^\mu }({\xi _q}) = {\bar \xi ^\mu }$. 
The bases of $TM$ will transform as 
\begin{align}
&\frac{\partial }{{\partial {{\bar x}^\mu }}} 
= \frac{{\partial {x^\nu }}}{{\partial {{\bar x}^\mu }}}\frac{\partial }{{\partial {x^\nu }}} 
+ \frac{{\partial {y^\nu }}}{{\partial {{\bar x}^\mu }}}\frac{\partial }{{\partial {y^\nu }}}, \nonumber \\
&\frac{\partial }{{\partial {{\bar y}^\mu }}} 
= \frac{{\partial {y^\nu }}}{{\partial {{\bar y}^\mu }}}\frac{\partial }{{\partial {y^\nu }}} 
= \frac{{\partial {x^\nu }}}{{\partial {{\bar x}^\mu }}}\frac{\partial }{{\partial {y^\nu }}}, 
\end{align}
and we have the transformation rule for the new coordinate ${z^\mu }$ by 
\begin{align}
{{\bar{z}}^{\mu }}=\frac{\partial {{{\bar{x}}}^{\mu }}}{\partial {{x}^{\nu }}}{{z}^{\nu }}
+\frac{\partial {{{\bar{y}}}^{\mu }}}{\partial {{x}^{\nu }}}{{y}^{\nu }}
=\frac{\partial {{{\bar{x}}}^{\mu }}}{\partial {{x}^{\nu }}}{{z}^{\nu }}
+\frac{{{\partial }^{2}}{{{\bar{x}}}^{\mu }}}{\partial {{x}^{\nu }}\partial {{x}^{\rho }}}{{y}^{\nu }}{{y}^{\rho }}.
\end{align}
These transformations are smooth, and the charts form a smooth atlas on ${T^2}M$. 
The natural inclusion $\iota :{{T}^{2}}M\to TTM$ has a coordinate expression, 
\begin{align}
{{x}^{\mu }}\circ \iota ={{x}^{\mu }},\,{{\dot{x}}^{\mu }}\circ \iota ={{y}^{\mu }}, 
\,{{y}^{\mu }}\circ \iota ={{y}^{\mu }},\,{{\dot{y}}^{\mu }}\circ \iota ={{z}^{\mu }} \label{inclusionmap_2}.  
\end{align}
The {\it l.h.s.} denotes coordinates on $TTM$ while the {\it r.h.s.} denotes coordinates on $T^2 M$. 
This shows that 
${T^2}M$ is a submanifold of $TTM$. 
Now $\tau _M^{2,1}$ is a surjective submersion by definition, so it remains to check the local trivialisation. 
The local trivialisation of $({T^2}M, \tau _M^{2,1}, TM)$ around any point $p \in TM$ is given by 
$({V_p}, {\mathbb{R}^n}, {t_p})$, 
\begin{align}
{t_p}:{(\tau _M^{2,1})^{ - 1}}({V_p}) \to {V_p} \times {\mathbb{R}^n}, \quad p \in {V_p}, 
\end{align}
where ${V_p}$ is an open set of $TM$, which in chart expression for any 
$\xi  \in {(\tau _M^{2,1})^{ - 1}}({V_p})$ is 
\begin{align}
{t_p}(\xi ) = (\tau _M^{2,1}(\xi ),{\ddot x^\mu }(\xi )).
\end{align}
Therefore, $({T^2}M,\tau _M^{2,1},TM)$ is indeed a bundle. 

\begin{defn} Second order tangent bundle \\ 
The triple $({T^2}M,\tau _M^{2,0},M)$ with $\tau _M^{2.0} = {\tau _M} \circ \tau _M^{2,1}$ 
is also a bundle with the trivialisation 
$({U_p},{\mathbb{R}^{2n}},{t_p})$, 
${t_p}:{(\tau _M^{2,0})^{ - 1}}({U_p}) \to {U_p} \times {\mathbb{R}^{2n}}$, around any $p \in M$, 
which in chart expression for any 
$\xi  \in {(\tau _M^{2,0})^{ - 1}}({U_p})$ is 
\begin{align}
{t_p}(\xi ) 
= (\tau _M^{2,0}(\xi ),{\dot x^\mu }(\xi ),{\ddot x^\mu }(\xi )).
\end{align}
We will call this a {\it second order tangent bundle over $M$}, or simply, {\it second order tangent bundle}. 
\end{defn}

In the theory of second order mechanics, 
the dynamical variables are the section of the second order tangent bundle. 

\subsection{Higher order tangent bundle} \label{higher_TM}
Now we will briefly introduce the higher order tangent bundles. 
This concept will be used to define the total derivatives of the higher order (Chapter 5), 
when deriving the reparameterisation invariant Euler-Lagrange equations. 
We will especially give the construction of $({T^3}M,\tau _M^{3,2},{T^2}M)$, 
and the $r$-th order $({{T}^{r}}M,\tau _{M}^{r,r-1},{{T}^{r}}M)$ 
could be obtained iteratively. 
\begin{defn}  Third order tangent bundle ${{T}^{3}}M$ over ${{T}^{2}}M$ \label{def_3rdorder_TMoverT2M} \\
Consider the bundle morphism 
$(T\tau _{M}^{2,1}, \tau _{M}^{2,1})$ from $(T{{T}^{2}}M, {{\tau }_{{{T}^{2}}M}}, {{T}^{2}}M)$ 
to $(TTM, {{\tau }_{TM}}, \linebreak[3] TM)$. 
We define the set ${{T}^{3}}M$ by, 
\begin{align}
{{T}^{3}}M:=\{u\in T{{T}^{2}}M\,|\,\iota \circ \tau_{T^2 M}(u)=T\tau _{M}^{2,1}(u)\},
\end{align}
where $\iota $ is the inclusion map $\iota :{{T}^{2}}M\to TTM$, 
and its coordinate expression given by (\ref{inclusionmap_2}).
Let $({{\tilde{V}}^{2}}, {{\tilde{\psi }}^{2}})$, ${{\tilde{\psi }}^{2}} = 
({{x}^{\mu }},{{y}^{\mu }}, {{\dot{x}}^{\mu }}, {{\dot{y}}^{\mu }})$ be the induced chart on $TTM$ and 
$({{\tilde{V}}^{3}}, {{\tilde{\psi }}^{3}})$, 
${{\tilde{\psi }}^{3}} = ({{x}^{\mu }}, {{y}^{\mu }}, {{z}^{\mu }}, {{\dot{x}}^{\mu }}, {{\dot{y}}^{\mu }}, 
{{\dot{z}}^{\mu }})$ be the induced chart on $T{{T}^{2}}M$. 
The elements of $T{{T}^{2}}M$ have the local expressions 
\begin{align}
{{u}_{q}}=u_{1}^{\mu }{{\left( \frac{\partial }{\partial {{x}^{\mu }}} \right)}_{q}}
+ u_{2}^{\mu }{{\left( \frac{\partial }{\partial {{y}^{\mu }}} \right)}_{q}} 
+ u_{3}^{\mu }{{\left( \frac{\partial }{\partial {{z}^{\mu }}} \right)}_{q}}, \, \, q\in {{T}^{2}}M
\end{align}
and the submanifold equations will give 
\begin{align}
&  {x^\mu }({T_q}\tau _M^{2,1}(u)) = {x^\mu } \circ \iota (q), \hfill \nonumber \\
&  {y^\mu }({T_q}\tau _M^{2,1}(u)) = {y^\mu } \circ \iota (q), \hfill \nonumber \\
&  {{\dot x}^\mu }({T_q}\tau _M^{2,1}(u)) = u_1^\mu  = {y^\mu } \circ \iota (q), \hfill \nonumber \\
&  {{\dot y}^\mu }({T_q}\tau _M^{2,1}(u)) = u_2^\mu  = {z^\mu } \circ \iota (q), \hfill  
\end{align}
where the coordinate functions in the {\it l.h.s.} are on $TTM$, 
while the {\it r.h.s} represents the coordinate functions of ${T^2}M$.  
Then we will have for coordinate functions on $T{{T}^{2}}M$, 
\begin{align}
& {{x}^{\mu }}(u)={{x}^{\mu }}\circ \iota (q), \nonumber \\ 
& {{y}^{\mu }}(u)={{y}^{\mu }}\circ \iota (q), \nonumber \\ 
& {{{\dot{x}}}^{\mu }}(u)=u_{1}^{\mu }={{y}^{\mu }}\circ \iota (q), \nonumber \\ 
& {{{\dot{y}}}^{\mu }}(u)=u_{2}^{\mu }={{z}^{\mu }}\circ \iota (q), 
\end{align}
therefore, the elements of ${{T}^{3}}M$ will have the form 
\begin{align}
{{w}_{q}}={{y}^{\mu }}(q){{\left( \frac{\partial }{\partial {{x}^{\mu }}} \right)}_{q}}
+ {{z}^{\mu }}(q){{\left( \frac{\partial }{\partial {{y}^{\mu }}} \right)}_{q}}
+ {{w}^{\mu }}{{\left( \frac{\partial }{\partial {{z}^{\mu }}} \right)}_{q}},
\end{align}
and we will take the induced chart on 
${{T}^{3}}M$ as $({{V}^{3}}, {{\psi }^{3}})$, ${{\psi }^{3}} = ({{x}^{\mu }}, {{y}^{\mu }}, {{z}^{\mu }}, {{w}^{\mu }})$. 
We can check the set $({{T}^{3}}M, \tau _{M}^{3,2}, {{T}^{2}}M)$ with 
$\tau _{M}^{3,2} := {{\left. {{\tau }_{{{T}^{2}}M}} \right|}_{{{T}^{3}}M}}$ is a bundle in the similar way 
we did for the case of $({{T}^{2}}M, \tau _{M}^{2,1}, TM)$. 
Clearly, $\tau _{M}^{3,2}$ is a sub-bundle of $(T{{T}^{2}}M, {{\tau }_{{{T}^{2}}M}}, {{T}^{2}}M)$. 
We will call the set $({{T}^{3}}M, \tau _{M}^{3,2}, {{T}^{2}}M)$ the 
{\it third order tangent bundle over ${{T}^{2}}M$}. 
\end{defn} 
We can similarly construct the $r$-th order tangent bundle over ${{T}^{r-1}}M$; 
namely $({T^r}M, \lb[3] \tau _M^{r,r - 1},{T^{r - 1}}M)$ by induction. 
The bundle projection is defined by, 
\begin{align}
\tau _M^{r,r - 1}: = {\left. {{\tau _{{T^{r - 1}}M}}} \right|_{{T^r}M}}.  
\end{align}
Consider the bundle morphism $(T\tau _{M}^{r-1,r-2}, \lb[3] \tau _{M}^{r-1,r-2})$ from 
$(T{{T}^{r-1}}M, \lb[3] {{\tau }_{{{T}^{r-1}}M}}, \lb[3] {{T}^{r-1}}M)$ to  
$(T{{T}^{r-2}}M,\lb[3] {{\tau }_{{{T}^{r-2}}M}},\linebreak[3] {{T}^{r-2}}M)$. 
Then we will define the total space $T^r M$ by 
\begin{align}
{T^r}M: = \{ u \in T{T^{r - 1}}M{\mkern 1mu} |{\mkern 1mu} T\tau _M^{r - 1,r - 2}(u) 
= {\iota _{r-1}} \circ {\tau _{{T^{r - 1}}M}}(u)\}, \label{higherbundle_r}
\end{align}
where  ${\iota _{r-1}}:{T^{r-1}}M \to T{T^{r - 2}}M$ is the inclusion map. 
$({T^r}M,\tau _M^{r,r - 1},{T^{r - 1}}M)$ is a sub-bundle of 
$(T{T^{r - 1}}M,{\tau _{{T^{r - 1}}M}},{T^{r - 1}}M)$. 

\subsection{Multivector bundle} \label{subsec_multivectorbundles}
The completely antisymmetric tensor product bundles are the structures we need for further discussions on calculus of variation, especially when the dimension of parameter space is greater than one, namely for the consideration of a field theory. It is related to the concept of multivectors and multivector fields, which is introduced below. 

\begin{defn}   $k$-multivectors \\
A {\it $k$-multivector} with $k \leqslant n$ is an element of exterior algebra over a vector space $V$ denoted by ${\Lambda ^k}(V)$. It is a linear combination of the multivectors of the form ${v_1} \wedge  \cdots  \wedge {v_k}$, ${v_1}, \cdots ,{v_k} \in V$. 
When the $k$-multivector is in the form ${v_1} \wedge  \cdots  \wedge {v_k}$, it is called a {\it decomposable multivector}.  
\end{defn}
This space ${\Lambda ^k}(V)$ has a natural space of the vector space. 

One geometric way to understand the vector field was to see it as a section of a tangent bundle, 
$(TM,\tau ,M)$. 
Similarly, the multivector fields could be understood as a section of the 
$k$-fold antisymmetric tensor product bundle of  $({\Lambda ^k}TM,{\Lambda ^k}{\tau _M},M)$. 
We denote 
\begin{align}
{\Lambda ^k}TM := \bigcup\limits_{p \in M} {{\Lambda ^k}{T_p}M} 
\end{align}
and ${\Lambda ^k}{\tau _M} := {\tau _M} \wedge  \cdots  \wedge {\tau _M}$ ($k$-alternating products). 
Below we will give the definition, and see that this is indeed a vector bundle. 

\begin{defn}   $k$-Multivector bundle  \label{lkTM} \\
The triple $({\Lambda ^k}TM,\pi ,M)$ where $\pi  = {\Lambda ^k}{\tau _M}$ is a projection 
$\pi (v) = p$ for $v \in {\Lambda ^k}TM$, $p \in M$, has a vector bundle structure, and is called the {\it $k$-multivector bundle}. 
\end{defn}

First, we will introduce a smooth structure on ${\Lambda ^k}TM$. 
Let the chart on the base space $M$ be $(U,\varphi ), 
\varphi  = ({x^\mu })$, and $v \in {\pi ^{ - 1}}(p),p \in U$. 
The bases of $k$-multivector constructed from natural bases 
$\displaystyle{\left( {{{\left( {\frac{\partial }{{\partial {x^1}}}} \right)}_p}, \cdots , 
{{\left( {\frac{\partial }{{\partial {x^n}}}} \right)}_p}} \right)}$ on $U$ are in the form 
$\displaystyle{{\left( {\frac{\partial }{{\partial {x^{{i_1}}}}}} \right)_p} \wedge  \cdots  
\wedge {\left( {\frac{\partial }{{\partial {x^{{i_k}}}}}} \right)_p}}$, 
where ${i_1}, \cdots ,{i_k}$ are integers taken from $1, \cdots ,n$ without overlapping, 
and we denote this by the abbreviation 
$\displaystyle{{\left( {\frac{\partial }{{\partial {x^{{i_1}}}}} 
\wedge  \cdots  \wedge \frac{\partial }{{\partial {x^{{i_k}}}}}} \right)_p}}$. 
The coordinate expression of $v$ by these bases is  
\begin{align}
v = \frac{1}{{k!}}{v^{{i_1} \cdots {i_k}}}{\left( {\frac{\partial }{{\partial {x^{{i_1}}}}} \wedge  \cdots  
\wedge \frac{\partial }{{\partial {x^{{i_k}}}}}} \right)_p},
\end{align}
where ${v^{{i_1} \cdots {i_k}}}$ are real numbers with alternating superscripts. 
We may also use the local expression of $v$ using ordered bases, i.e., 
\begin{align}
v = \sum\limits_{{i_1} < {i_2} <  \cdots  < {i_k}} {{v^{{i_1} \cdots {i_k}}}{{
\left( {\frac{\partial }{{\partial {x^{{i_1}}}}} \wedge  \cdots  
\wedge \frac{\partial }{{\partial {x^{{i_k}}}}}} \right)}_p}}.
\end{align}
Define the functions ${y^{{\mu _1} \cdots {\mu _k}}}$ on 
${\pi ^{ - 1}}(U)$ by 
${y^{{\mu _1} \cdots {\mu _k}}}(v) = {v^{{\mu _1} \cdots {\mu _k}}}$, 
then we can obtain the induced chart on 
$V \subset {\Lambda ^k}TM$, by $(V,\psi ),\; 
V = {\pi ^{ - 1}}(U),\;\psi  = ({x^\mu },{y^{{\mu _1} \cdots {\mu _k}}})$. 
When considering an exact value, we assume the superscripts are ordered. 

To see the coordinate transformations, 
let $(V,\psi )$, $V = {\pi ^{ - 1}}(U)$, $\psi  = ({x^\mu },{y^{{\mu _1} \cdots {\mu _k}}})$ and 
$(\bar V,\bar \psi )$, $\bar V = {\pi ^{ - 1}}(\bar U)$, 
$\bar \psi  = ({\bar x^\mu },{\bar y^{{\mu _1} \cdots {\mu _k}}})$ be two charts on 
${\Lambda ^k}TM$, with $V \cap \bar V \ne \emptyset $. 
Then express the element ${v_p} \in V \cap \bar V$, 
$p \in U \cap \bar U \subset M$ by these charts, 
\begin{align}
{v_p} = \frac{1}{{k!}}{y^{{\mu _1} \cdots {\mu _k}}}(p){
\left( {\frac{\partial }{{\partial {x^{{\mu _1}}}}} \wedge  \cdots  
\wedge \frac{\partial }{{\partial {x^{{\mu _k}}}}}} \right)_p} 
= \frac{1}{{k!}}{\bar y^{{\mu _1} \cdots {\mu _k}}}(p){
\left( {\frac{\partial }{{\partial {{\bar x}^{{\mu _1}}}}} \wedge  \cdots  
\wedge \frac{\partial }{{\partial {{\bar x}^{{\mu _k}}}}}} \right)_p}.\nonumber \\
\end{align}
The bases of ${\Lambda ^k}TM$ will transform as 
\begin{align}
\frac{\partial }{{\partial {{\bar x}^{{\mu _1}}}}} \wedge  \cdots  
\wedge \frac{\partial }{{\partial {{\bar x}^{{\mu _k}}}}} 
= \frac{{\partial {x^{{\nu _1}}}}}{{\partial {{\bar x}^{{\mu _1}}}}} \cdots 
\frac{{\partial {x^{{\nu _k}}}}}{{\partial {{\bar x}^{{\mu _k}}}}}
\frac{\partial }{{\partial {x^{{\nu _1}}}}} \wedge  \cdots  \wedge 
\frac{\partial }{{\partial {x^{{\nu _k}}}}}, 
\end{align}
so the transformation equation for the coordinate functions on ${\Lambda ^k}TM$ are 
\begin{align}
&{\bar x^\mu } = {\bar x^\mu }({x^\nu }), \\
&{\bar y^{{\nu _1} \cdots {\nu _k}}} = 
\frac{{\partial {{\bar x}^{{\nu _1}}}}}{{\partial {x^{{\mu _1}}}}} \cdots 
\frac{{\partial {{\bar x}^{{\nu _k}}}}}{{\partial {x^{{\mu _k}}}}}{y^{{\mu _1} \cdots {\mu _k}}}. 
\end{align}
These transformations are smooth. 
Such induced charts define on ${\Lambda ^k}TM$ the structure of 
${C^\infty }$-manifold of dimension $n + {}_n{C_k}$.

Now, $({\Lambda ^k}TM,\pi ,M)$ naturally becomes a bundle by the projection, 
$\pi (v) = p$ for $v \in {\Lambda ^k}{T_p}M,p \in M$, since for any $v \in {\Lambda ^k}TM$, 
there always exist a unique $p \in M$, such that $v \in {\Lambda ^k}{T_p}M$, 
and we assumed $\dim M$ and therefore $\dim {\Lambda ^k}TM$ are both constant. 

This bundle structure can be also constructed by taking a 
{\it $k$-fold alternating product of $(TM,\tau ,M)$}, 
which is denoted by $({\Lambda ^k}TM,{\Lambda ^k}\tau ,M)$, 
and is a sub-bundle of the tensor product bundle $({ \otimes ^k}TM,{ \otimes ^k}\tau ,M)$. 

\begin{defn}   $k$-multivector field \\
{\it$k$-multivector field} is a section of $({\Lambda ^k}TM, \pi ,M)$. 
We denote all sections of $\pi$ by $\Gamma ({\Lambda ^k}TM)$, 
or equivalently ${{\mathfrak{X}}^k}(M)$. 
\end{defn}

\begin{defn}   Local coordinate expression of $k$-multivector field \\
Let $(U,\varphi ),\varphi  = ({x^\mu })$ be a chart on $M$, 
the local expression of $Y \in \Gamma ({\Lambda ^k}TU)$ is 
\begin{align}
Y = \frac{1}{{k!}}{f^{{i_1} \cdots {i_k}}}\frac{\partial }{{\partial {x^{{i_1}}}}} \wedge  \cdots  
\wedge \frac{\partial }{{\partial {x^{{i_k}}}}} 
\end{align}
with ${f^{{i_1} \cdots {i_k}}} \in {C^\infty }(U)$ alternating in all the superscripts. 
\end{defn}

\begin{defn}  Decomposable $k$-multivector field \\
Let $X$ be a $k$-multivector field. $X$ is a {\it decomposable $k$-multivector field} 
iff there exists ${X_1}, \cdots ,{X_k} \in {\mathfrak{X}}(M)$ such that $X = {X_1} \wedge  \cdots  \wedge {X_k}$. 
We denote all decomposable $k$-multivector fields by ${{\mathfrak{X}}^{ \wedge k}}(M)$. 
\end{defn}

\begin{defn}  Locally decomposable $k$-multivector field \\
Let $X$ be a $k$-multivector field. We say that $X$ is {\it decomposable at $p \in M$}, if there exists a neighbourhood ${U_p} \subset M$ of $p$, and ${X_1}, \cdots ,{X_k} \in {\mathfrak{X}}({U_p})$, such that $X = {X_1} \wedge  \cdots  \wedge {X_k}$ on ${U_p}$. 
$X$ is called {\it locally decomposable $k$-multivector field} iff for every $p \in M$ there exists a neighbourhood ${U_p} \subset M$ and ${X_1}, \cdots ,{X_k} \in {\mathfrak{X}}({U_p})$ such that $X = {X_1} \wedge  \cdots  \wedge {X_k}$ on each ${U_p}$.
\end{defn}
We denote all locally decomposable $k$-multivector fields by ${\widetilde {\mathfrak{X}}^{ \wedge k}}(M)$. 
$k$-multivector fields forms a ${C^\infty }(M)$-module, and it is a dual concept to the $k$-form. 
We will show this in the following Lemma \ref{lem_multivec}, and for this purpose, we introduce the theorem by Morita. 

\begin{theorem} Correspondence of alternating map and Forms (Morita)  \label{thm_Morita} \\
Let $M$ be a n-dimensional ${C^\infty }$-manifold, and ${\Omega ^k}(M)$ a module of all 
$k$-forms over $M$. ${\Omega ^k}(M)$ can be naturally identified with the module of ${C^\infty }(M)$ multi-linear, 
alternating mapping ${\mathfrak{X}}(M) \times  \dots  \times {\mathfrak{X}}(M) \to {C^\infty }(M)$. 
\end{theorem}
\begin{proof} 
Suppose we have a map $\omega $ which satisfies the above properties. 
We first prove that the value $\omega ({X_1}, \cdots ,{X_k})(p)$, 
is determined only by the values of ${X_1}, \cdots ,{X_k} \in {\mathfrak{X}}(M)$ at point $p \in M$. 
By multi-linearity, it is sufficient to prove that for certain ${X_i}$, with $1 \leqslant i \leqslant k$, 
$\omega ({X_1}, \cdots ,{X_k})(p) = 0$ for ${X_i}(p) = 0$. 
Take $i = 1$, and local coordinates around $p$ : $(U;{x_1}, \cdots ,{x_n})$. On $U$, 
the vector field ${X_1}$ is expressed by $\displaystyle{{X_1} = {f_j}\frac{\partial }{{\partial {x^j}}}}$, 
with ${f_i}(p) = 0$. 
Consider an open set $V$ which $\bar V \subset U$, 
and a function $h \in {C^\infty }(M)$ such that, 
\begin{align}
h(q) = \left\{ {\begin{array}{*{20}{c}}
  {1 \dots q \in V} \\ 
  {0 \dots q \notin U} 
\end{array}} \right. .
\end{align} 
Let $\displaystyle{{Y_j} = h\frac{\partial }{{\partial {x^j}}}}$, 
${Y_j} \in {\mathfrak{X}}(M)$, 
and ${\tilde f_j}: = h{f_j}$, then, ${\tilde f_j} \in {C^\infty }(M)$. 

By simple modification; ${X_1} = {f_j}{Y_j} + (1 - h){X_1}$, and by the multi-linearity of the map $\omega $, 
\begin{align}
&\omega ({X_1}, \cdots ,{X_k})(p) = \omega ({f_j}{Y_j} + (1 - h){X_1},{X_2}, \cdots ,{X_k})(p)  \nonumber \\
&= {f_j}(p)\omega ({Y_j},{X_2}, \cdots ,{X_k})(p) + (1 - h(p))\omega ({X_1}, \cdots ,{X_k})(p) = 0.  
\end{align}
Having this, we can obtain a $k$-form as the following. 
For every $p \in M$, and tangent vectors ${\tilde X_1}, \cdots ,{\tilde X_k} \in {T_p}M$, 
choose vector fields ${X_i},\;(i = 0, \cdots ,k)$, such that ${X_i}(p) = {\tilde X_i}$. 
Put ${\tilde \omega _p}({\tilde X_1}, \cdots ,{\tilde X_k}) = \omega ({X_1}, \cdots ,{X_k})(p)$. 
By the previous discussion, this does not depend on the choice of vector field, and it is obvious that 
${\tilde \omega _p}$ is ${C^\infty }$. Then, $\tilde \omega $ is the differential form identified with $\omega $. 
\end{proof}

We can prove a similar result regarding multivector fields in the following lemma. 

\begin{lemma} Duality of locally decomposable $k$-vector fields and $k$-form  \label{lem_multivec}\\ 
1. Suppose we have a map $\omega :{{\mathfrak{X}}^k}(M) \to {C^\infty }(M)$ 
which is ${C^\infty }(M)$-linear. Then there exists a unique $k$-form $\Omega $ on $M$ such that, 
on each chart $(U,\varphi ), \varphi  = ({x^\mu })$ of $M$ is related to 
$\omega $ by 
\begin{align}
\omega (X) = \frac{1}{{k!}}{X^{{\mu _1} \cdots {\mu _k}}} \Omega 
\left( {\frac{\partial }{{\partial {x^{{\mu _1}}}}}, \cdots , 
\frac{\partial }{{\partial {x^{{\mu _k}}}}}} \right)
\end{align}
for $X \in {{\mathfrak{X}}^k}(M)$, ${X^{{\mu _1} \cdots {\mu _k}}} \in {C^\infty }(M)$. \\
2. Let $\{(U_\iota,\varphi_\iota )\}$ be the set of charts covering $M$. 
The restriction ${\left. \omega  \right|_{L.D.}}: {\widetilde {\mathfrak{X}}^{ \wedge k}}(M) \to {C^\infty }(M)$ 
given by ${\left. \omega  \right|_{L.D.}}(X) = \omega (X)$ where 
$X \in {\widetilde {\mathfrak{X}}^{ \wedge k}}(M)$, is related to $\Omega$ by 
\begin{align}
\omega (X) = \Omega \left( {{X_1}, \cdots ,{X_k}} \right)
\end{align}
where ${X_1}, \cdots ,{X_k} \in {\mathfrak{X}}(U_\iota)$ are the arbitrary decomposition of $X$ on each $U_\iota$.
\end{lemma}
\begin{proof}
Suppose we have a $k$-form $\Omega $. Let ${X_1}, \cdots ,{X_k}$ be vector fields on $M$, 
and let $\displaystyle{{X_i} = X_i^\mu \frac{\partial }{{\partial {x^\mu }}}}$, 
$i = 1, \dots ,k$ 
be the chart expression of ${X_i}$ in the chart $(U,\varphi ), \varphi  = ({x^\mu })$. 
Then by multi-linearity of a $k$-form, 
\begin{align}
{\Omega ({X_1}, \cdots ,{X_k}) = X_1^{{\mu _1}} \cdots X_k^{{\mu _k}} \Omega 
\left( {\frac{\partial }{{\partial {x^{{\mu _1}}}}}, \cdots ,\frac{\partial }{{\partial {x^{{\mu _k}}}}}} \right)}.
\end{align}
On the other hand, the $k$-alternating product of ${X_1}, \cdots ,{X_k}$ forms a decomposable vector fields, 
$\displaystyle{{X_1} \wedge  \cdots  \wedge {X_k} = X_1^{{\mu _1}} \cdots X_k^{{\mu _k}}
\left( {\frac{\partial }{{\partial {x^{{\mu _1}}}}} \wedge  \cdots  
\wedge \frac{\partial }{{\partial {x^{{\mu _k}}}}}} \right)}$, where $X_1^{{\mu _1}}, \cdots , 
X_k^{{\mu _k}} \in {C^\infty }(M)$. 
Set 
\begin{align}
{\omega \left( {\frac{\partial }{{\partial {x^{{\mu _1}}}}} \wedge  \cdots  
\wedge \frac{\partial }{{\partial {x^{{\mu _k}}}}}} \right) = \Omega 
\left( {\frac{\partial }{{\partial {x^{{\mu _1}}}}}, \cdots , 
\frac{\partial }{{\partial {x^{{\mu _k}}}}}} \right)}.
\end{align} 
Now we extend $\omega$ on ${{\mathfrak{X}}^k}(M)$ by linearity. 
The chart $(U,\varphi ), \varphi  = ({x^\mu })$ induces the bases of multivectors in the form 
$\displaystyle{\frac{\partial }{{\partial {x^{{\mu _1}}}}} \wedge  \cdots  
\wedge \frac{\partial }{{\partial {x^{{\mu _k}}}}}}$, 
and any multivector field $Y \in {{\mathfrak{X}}^k}(M)$ has a coordinate expression 
\begin{align}
Y = \frac{1}{{k!}}{Y^{{\mu _1} \cdots {\mu _k}}}\frac{\partial }{{\partial {x^{{\mu _1}}}}} \wedge  \cdots  
\wedge \frac{\partial }{{\partial {x^{{\mu _k}}}}}
\end{align}
with ${Y^{{\mu _1} \cdots {\mu _k}}} \in {C^\infty }(U)$. 
Set 
\begin{align}
\omega (Y) = \frac{1}{{k!}}{Y^{{\mu _1} \cdots {\mu _k}}}\omega 
\left( {\frac{\partial }{{\partial {x^{{\mu _1}}}}} \wedge  \cdots  
\wedge \frac{\partial }{{\partial {x^{{\mu _k}}}}}} \right) = \frac{1}{{k!}}{Y^{{\mu _1} \cdots {\mu _k}}} \Omega 
\left( {\frac{\partial }{{\partial {x^{{\mu _1}}}}}, \cdots , \frac{\partial }{{\partial {x^{{\mu _k}}}}}} \right). 
\end{align}
We will show that this expression is independent of charts. 

Let $(\bar U,\bar \varphi ),\bar \varphi  = ({\bar x^\mu })$ be another chart on $M$ such that 
$U \cap \bar U \ne \emptyset $. 
On $U \cap \bar U$, we have also 
\begin{align}
\omega (Y) = \frac{1}{{k!}}{\bar Y^{{\mu _1} \cdots {\mu _k}}} \Omega 
\left( {\frac{\partial }{{\partial {{\bar x}^{{\mu _1}}}}}, \cdots , 
\frac{\partial }{{\partial {{\bar x}^{{\mu _k}}}}}} \right).
\end{align}
But since we have 
\begin{align}
Y = \frac{1}{{k!}}{Y^{{\mu _1} \cdots {\mu _k}}}\frac{\partial }{{\partial {x^{{\mu _1}}}}} \wedge  \cdots  
\wedge \frac{\partial }{{\partial {x^{{\mu _k}}}}} 
= \frac{1}{{k!}}{Y^{{\nu _1} \cdots {\nu _k}}} 
\frac{{\partial {{\bar x}^{{\mu _1}}}}}{{\partial {x^{{\nu _1}}}}} \cdots 
\frac{{\partial {{\bar x}^{{\mu _k}}}}}{{\partial {x^{{\nu _k}}}}} 
\frac{\partial }{{\partial {{\bar x}^{{\mu _1}}}}} \wedge  
\cdots  \wedge \frac{\partial }{{\partial {{\bar x}^{{\mu _k}}}}} 
\end{align}
we get 
\begin{align}
{\bar Y^{{\mu _1} \cdots {\mu _k}}} = {Y^{{\nu _1} \cdots {\nu _k}}} 
\frac{{\partial {{\bar x}^{{\mu _1}}}}}{{\partial {x^{{\nu _1}}}}} \cdots 
\frac{{\partial {{\bar x}^{{\mu _k}}}}}{{\partial {x^{{\nu _k}}}}},
\end{align}
and since
\begin{align}
\Omega \left( {\frac{\partial }{{\partial {{\bar x}^{{\mu _1}}}}}, \cdots , 
\frac{\partial }{{\partial {{\bar x}^{{\mu _k}}}}}} \right) 
= \frac{{\partial {x^{{\nu _1}}}}}{{\partial {{\bar x}^{{\mu _1}}}}} \cdots 
\frac{{\partial {x^{{\nu _k}}}}}{{\partial {{\bar x}^{{\mu _k}}}}}\Omega 
\left( {\frac{\partial }{{\partial {x^{{\nu _1}}}}}, \cdots , 
\frac{\partial }{{\partial {x^{{\nu _k}}}}}} \right),
\end{align}
we have 
\begin{align}
\omega (Y) = \frac{1}{{k!}}{Y^{{\mu _1} \cdots {\mu _k}}}\Omega 
\left( {\frac{\partial }{{\partial {x^{{\mu _1}}}}}, \cdots , 
\frac{\partial }{{\partial {x^{{\mu _k}}}}}} \right) = \frac{1}{{k!}}{\bar Y^{{\mu _1} \cdots {\mu _k}}} \Omega 
\left( {\frac{\partial }{{\partial {{\bar x}^{{\mu _1}}}}}, \cdots , 
\frac{\partial }{{\partial {{\bar x}^{{\mu _k}}}}}} \right). 
\end{align}
Therefore, this expression does not depend on the chart, and $\Omega $ is uniquely determined on $M$. 

Conversely, suppose we have a ${C^\infty }(M)$-linear map 
$\omega :{{\mathfrak{X}}^k}(M) \to {C^\infty }(M)$. By the linearity of $\omega $, 
for any multivector field $Y \in {{\mathfrak{X}}^k}(M)$ we have 
\begin{align}
\omega (Y) = \frac{1}{{k!}}{Y^{{\mu _1} \cdots {\mu _k}}}\omega \left( {\frac{\partial }{{\partial {x^{{\mu _1}}}}} \wedge  \cdots  \wedge \frac{\partial }{{\partial {x^{{\mu _k}}}}}} \right)
\end{align}
on each chart $(U,\varphi ),\varphi  = ({x^\mu })$ of $M$. 
Set 
\begin{align}
\Omega \left( \displaystyle{{\frac{\partial }{{\partial {x^{{\mu _1}}}}}, \cdots , 
\frac{\partial }{{\partial {x^{{\mu _k}}}}}}} \right) = \omega 
\left( \displaystyle{{\frac{\partial }{{\partial {x^{{\mu _1}}}}} \wedge  \cdots  \wedge 
\frac{\partial }{{\partial {x^{{\mu _k}}}}}}} \right),
\end{align} 
then since the multivectors are multi-linear and skew symmetric in the vector fields 
$\displaystyle{\frac{\partial }{{\partial {x^\mu }}}}$, by the Theorem \ref{thm_Morita}, 
$\Omega$ is identified as a $k$-form on $U$. 
Then for the general 
\begin{align}
\omega (Y) = \frac{1}{{k!}}{Y^{{\mu _1} \cdots {\mu _k}}} \Omega 
\left( {\frac{\partial }{{\partial {x^{{\mu _1}}}}}, \cdots , \frac{\partial }{{\partial {x^{{\mu _k}}}}}} \right), 
\end{align}
$\displaystyle{\frac{1}{{k!}}{Y^{{\mu _1} \cdots {\mu _k}}}\Omega }$ is also a $k$-form on $U$, 
since ${Y^{{\mu _1} \cdots {\mu _k}}}$ is a function on $U$. 
Globalisation is carried out similarly. 
Thus we have proved the first part of the lemma. 

Now consider the special case where $X \in {\widetilde {\mathfrak{X}}^{ \wedge k}}(M)$ is locally decomposable. 
Without any loss of generality, it is always possible to choose an open covering 
${\{ {U_\iota }\} _{\iota  \in I}}$ of $M$ such that for each ${U_\iota }$, 
$X$ is decomposable, namely $X = {X_{\iota ,1}} \wedge  \cdots  \wedge {X_{\iota ,k}}$, with ${X_{\iota ,1}}, 
\cdots ,{X_{\iota ,k}} \in {\mathfrak{X}}({U_\iota })$. 
$\iota  \in I$ is an index taken from countable index set $I$. 
Let ${\mathcal{A}} = {\{ ({U_\iota },{\varphi _\iota })\} _{\iota  \in I}}$ be an atlas of $M$, 
then on each ${U_\iota }$ we have $\omega (X) = \omega ({X_{\iota ,1}} \wedge  \cdots  \wedge {X_{\iota ,k}})$. 
We will prove that this does not depend on the decomposition of $X$. 
Suppose $X$ has another decomposition, 
$X = {\tilde X_{\iota ,1}} \wedge  \cdots  \wedge {\tilde X_{\iota ,k}}$, ${\tilde X_{\iota ,1}}, \cdots , 
{\tilde X_{\iota ,k}} \in {\mathfrak{X}}({U_\iota })$ on 
${U_\iota }$. Let the coordinates on ${U_\iota }$ be 
$\varphi = ({x^\mu })$, then the local expression of $X$ becomes, 
\begin{align}
X = X_{\iota ,1}^{{\mu _1}} \cdots X_{\iota ,k}^{{\mu _k}} 
{\frac{\partial }{{\partial {x^{{\mu _1}}}}} \wedge  \cdots  \wedge 
\frac{\partial }{{\partial {x^{{\mu _k}}}}} = \tilde X_{\iota ,1}^{{\mu _1}} \cdots 
\tilde X_{\iota ,k}^{{\mu _k}}\frac{\partial }{{\partial {x^{{\mu _1}}}}} \wedge  \cdots  
\wedge \frac{\partial }{{\partial {x^{{\mu _k}}}}}}, 
\end{align}
implying $X_{\iota ,1}^{[{\mu _1}} \cdots 
X_{\iota ,k}^{{\mu _k}]} = \tilde X_{\iota ,1}^{[{\mu _1}} \cdots \tilde X_{\iota ,k}^{{\mu _k}]}$. 
Then by the first part of the lemma, we have 
\begin{align}
&\omega ({X_{\iota ,1}} \wedge  \cdots  \wedge {X_{\iota ,k}})  \nonumber \\
&= X_{\iota ,1}^{{\mu _1}} \cdots 
  X_{\iota ,k}^{{\mu _k}}\Omega \left( {\frac{\partial }{{\partial {x^{{\mu _1}}}}}, \cdots , 
  \frac{\partial }{{\partial {x^{{\mu _k}}}}}} \right) = X_{\iota ,1}^{[{\mu _1}} \cdots 
  X_{\iota ,k}^{{\mu _k}]}\Omega \left( {\frac{\partial }{{\partial {x^{{\mu _1}}}}}, \cdots , 
  \frac{\partial }{{\partial {x^{{\mu _k}}}}}} \right) \nonumber \\
&= \tilde X_{\iota ,1}^{[{\mu _1}} \cdots \tilde X_{\iota ,k}^{{\mu _k}]}\Omega 
   \left( {\frac{\partial }{{\partial {x^{{\mu _1}}}}}, \cdots ,\frac{\partial }{{\partial {x^{{\mu _k}}}}}} \right) 
   = \omega ({{\tilde X}_{\iota ,1}} \wedge  \cdots  \wedge {{\tilde X}_{\iota ,k}}) 
\end{align}
Therefore, this map does not depend on the choice of decomposition, and also implies 
$\omega ({X_{\iota ,1}} \wedge  \cdots  \wedge {X_{\iota ,k}}) 
= \Omega ({X_{\iota ,1}}, \cdots ,{X_{\iota ,k}})$ on each ${U_\iota }$. 
\end{proof}

\begin{defn}  Tangent mapping of multivectors \\
Let $M,N$ be smooth manifolds and $f:M \to N$ a ${C^\infty }$-map. 
Let $(U,\varphi ),\varphi  = ({x^\mu })$ be a chart on $M$ and $(V,\psi ), \psi  = ({y^\nu })$ 
a chart on $N$, such that $f(U) \subset V$. 
We extend the bundle morphism $(Tf,f)$ of tangent bundles 
$Tf:TM \to TN$ and $f:M \to N$ to bundle morphism of multivector bundles, 
$({\Lambda ^k}TM,{\Lambda ^k}{\tau _M},M)$ to $({\Lambda ^k}TN, {\Lambda ^k}{\tau _N},N)$, 
similarly as in the case of contravariant tensor bundles. 
Let $v \in {\Lambda ^k}{T_p}M$, and let the coordinate expression be 
$\displaystyle{v = \frac{1}{{k!}}{v^{{\mu _1} \cdots {\mu _k}}}
{\left( {\frac{\partial }{{\partial {x^{{\mu _1}}}}} \wedge  \cdots  
\wedge \frac{\partial }{{\partial {x^{{\mu _k}}}}}} \right)_p}}$. 
Then we define the image of this multivector by the tangent map by 
\begin{align}
{\Lambda ^k}Tf(v) = \frac{1}{{k!}}{v^{{\mu _1} \cdots {\mu _k}}}
{\left( {Tf\left( {\frac{\partial }{{\partial {x^{{\mu _1}}}}}} \right) \wedge  \cdots  
\wedge Tf\left( {\frac{\partial }{{\partial {x^{{\mu _k}}}}}} \right)} \right)_{f(p)}}.
\end{align}
\end{defn}
To distinguish between the usual tangent mappings, we used the notation ${\Lambda ^k}Tf(v)$ 
to state that the map is acting on a $k$-multivector, and occasionally call them {\it multi-tangent map}.  

\subsection{Second order multivector bundle} \label{subsec_2nd_k_bundle}
In the previous sections, we have prepared the basics of multivectors and second order tangent bundle. 
With these foundations, we can now construct the underlying structures required for the second order field theory, 
which we call the bundle of {\it second order multivectors}, 
and denote by $({({\Lambda ^k}T)^2}M,{\Lambda ^k}\tau _M^{2,1},{\Lambda ^k}TM)$. 
The section of the bundle of second order multivectors is called {\it second order multivector field}, 
which corresponds to the physical fields as we shall see later. 

Before constructing the bundle of second order multivectors, we begin with some basic observations.
\begin{pr}  \label{prop_vectorbundle}
If $(E,\pi ,M)$ is a vector bundle, then $(TE,T\pi ,TM)$ is a vector bundle. 
\end{pr}
\begin{proof}
Let $m = \dim M$, $n = \dim E$, and $({\pi ^{ - 1}}(W),\psi ), \psi  = ({x^1}, \cdots {x^m},{u^1}, \cdots ,{u^n})$ 
a vector bundle adapted chart on $E$, induced by the chart $(W,\varphi ), \varphi  = ({x^1}, \cdots ,{x^m})$ 
on $W \subset M$. 
The element of the fibre ${E_p}$ would then have a chart expression 
$\phi  = {\phi ^i}{({e_i})_p}$, $i = 1, \cdots ,n$, where ${e_i}$ 
are the local sections defined from the chart as in Definition \ref{def_localsectionofvb}. 
Let $p \in U \subset W$, and denote the local trivialisation of $(E,\pi ,M)$ at $p$ by 
$({U_p},{\mathbb{R}^n},{s_p})$. 
The chart on $M$ and induces the local trivialisation on the bundles $(TM,{\tau _M},M)$ and 
$(TE,{\tau _E},E)$, by $({U_p},{\mathbb{R}^m},{t_{M,p}})$ and 
$({s_p}({\pi ^{ - 1}}({U_p})) = {U_p} \times {\mathbb{R}^n},{\mathbb{R}^{m + n}},{t_{E,q}})$, 
$q \in E,\,\;\pi (q) = p$, respectively. 
We take induced charts on $TM$ and $TE$ associated to these trivialisation as 
$({\tau _M}^{ - 1}({U_p}),{\varphi ^1}),{\varphi ^1} = ({x^\mu },{y^\mu })$ and 
$({\pi ^{ - 1}}({\tau _M}^{ - 1}({U_p})),{\psi ^1}),{\psi ^1} = ({x^\mu },{u^i},{\dot x^\mu },{\dot u^i})$. 
Furthermore, these trivialisations gives rise to the local trivialisation $({\hat U_p},{F_p},{\hat t_p})$ on 
$(TE,T\pi ,TM)$, by 
\begin{align}
{\hat t_p}:T{\pi ^{ - 1}}({U_p} \times {\mathbb{R}^m}) \to {U_p} 
\times {\mathbb{R}^m} \times {\mathbb{R}^{2n}},
\end{align}
where we denoted ${U_p} \times {\mathbb{R}^m} = {\hat U_p}$, and ${\mathbb{R}^{2n}} = {F_p}$. 
Here we used the isomorphism of the trivialisation, 
\begin{align}
S:{U_p} \times {\mathbb{R}^n} \times {\mathbb{R}^{m + n}} \to {U_p} 
\times {\mathbb{R}^m} \times {\mathbb{R}^{2n}},
\end{align}
which in local coordinates could be expressed as $S({x^\mu },{u^\alpha },{\dot x^\mu },{\dot u^\alpha }) 
= ({x^\mu },{\dot x^\mu },{u^\alpha },{\dot u^\alpha })$. $S$ is called a {\it swap map}. 
Let us see how this work. Let ${\xi _\phi } \in {T_\phi }E$, $\phi 
= {\phi ^i}{e_i}(p) \in {E_p}$. 
The local coordinate expression of ${\xi _\phi }$ 
is given by $\displaystyle{{\xi _\phi } = {\xi ^\mu }{\left( {\frac{\partial }{{\partial {x^\mu }}}} \right)_\phi } 
+ {\bar \xi ^i}{\left( {\frac{\partial }{{\partial {u^i}}}} \right)_\phi }}$ or equivalently, 
${\psi ^1}({\xi _\phi }) = ({x^\mu }(\phi ),{\phi ^i},{\xi ^\mu },{\bar \xi ^i})$. 
The swap map will take this to 
\begin{align}
S{\psi ^1}({\xi _\phi }) = S({x^\mu }(\phi ),{\phi ^i},{\xi ^\mu },{\bar \xi ^i}) 
= ({x^\mu }(\phi ),{\xi ^\mu },{\phi ^i},{\bar \xi ^i}),
\end{align}
which by the local section ${e_i},\;{\dot e_i}$ of $T\pi $, 
defined by ${u^i}({e_j}(z)) = \delta _j^i$, ${\dot u^i}({\dot e_j}(z)) 
= \delta _j^i$ for all $z \in {\tau _M}^{ - 1}({U_p})$, could be expressed by ${\xi _q} 
= {\phi ^i}{({e_i})_q} + {\bar \xi ^i}{({\dot e_i})_q} \in {(TE)_q}$, with 
$q = ({x^\mu }(p),{\xi ^\mu }) \in {\tau _M}^{ - 1}({U_p})$. It is easy to see that this point $q$ 
is consistent with the projection ${T_\phi }\pi ({\xi _\phi })$. 
The swap map $S({\xi _\phi }) = {\xi _q}$ simply changes the base point of the vector in $TE$. 
We can now see the fibres of $T\pi $ becomes a vector space over each point of 
$z \in {\tau _M}^{ - 1}({U_p}) \subset TM$. Considering for every $p \in M$, 
these process for each trivialisation, we can conclude that $T\pi $ is a vector bundle. 
\end{proof}
Proposition \ref{prop_vectorbundle} tells that the 
$k$-multivector bundle $({\Lambda ^k}TM,{\Lambda ^k}{\tau _M},M)$ 
induces a vector bundle $(T({\Lambda ^k}TM),T{\Lambda ^k}{\tau _M},TM)$ by the tangent map 
$T{\Lambda ^k}{\tau _M}:T({\Lambda ^k}TM) \to TM$. The element of $T({\Lambda ^k}TM)$ 
is a 1-vector at a point on ${\Lambda ^k}TM$, and if we choose a vector bundle adapted chart 
$(V,\psi )$, $V = {({\Lambda ^k}TM)^{ - 1}}({U_p})$, $\psi  = ({x^\mu },{y^{{\mu _1} \cdots {\mu _k}}})$ on 
${\Lambda ^k}TM$, induced by the chart $(U,\varphi ),\varphi  = ({x^1}, \cdots ,{x^m})$ on $M$, 
it has a coordinate expression 
\begin{align}
{\xi _v} = {\xi ^\mu }{\left( {\frac{\partial }{{\partial {x^\mu }}}} \right)_v} 
+ \frac{1}{{k!}}{\xi ^{{\mu_1} \cdots {\mu_k}}}{\left( {\frac{\partial }{{\partial {y^{{\mu _1} 
\cdots {\mu _k}}}}}} \right)_v}, \label{TlkTM}
\end{align}
with ${\xi _v} \in {T_v}({\Lambda ^k}TM)$, $v \in {({\Lambda ^k}TM)_p}$. 
or similarly, 
\begin{align}
{\psi ^1}({\xi _v}) = \left( {{x^\mu }(p),{v^{{\mu _1} \cdots {\mu _k}}},{\xi ^\mu },{\xi ^{{\mu _1} 
\cdots {\mu _k}}}} \right), 
\end{align}
by the induced chart $({V^1},{\psi ^1})$, ${V^1} = {({\tau _{{\Lambda ^k}TM}})^{ - 1}}(V)$, 
${\psi ^1} = ({x^\mu },{y^{{\mu _1} \cdots {\mu _k}}},{\dot x^\mu },{\dot y^{{\mu _1} \cdots {\mu _k}}})$ 
on $T({\Lambda ^k}TM)$. 
The swap map 
\begin{align}
S{\psi ^1}({\xi _v}) = S\left( {{x^\mu }(p),{v^{{\mu _1} 
\cdots {\mu _k}}},{\xi ^\mu },{\xi ^{{\mu _1} \cdots {\mu _k}}}} \right) 
= \left( {{x^\mu }(p),{\xi ^\mu },{v^{{\mu _1} \cdots {\mu _k}}},{\xi ^{{\mu _1} \cdots {\mu _k}}}} \right)
\end{align}
will give the multivector the expression 
\begin{align}
{\xi _q} = \frac{1}{{k!}}{v^{{\mu _1} \cdots {\mu _k}}}{(\frac{\partial }{{\partial {y^{{\mu _1} 
\cdots {\mu _k}}}}})_q} + \frac{1}{{k!}}{\xi ^{{\mu _1} \cdots {\mu _k}}}
{(\frac{\partial }{{\partial {{\dot y}^{{\mu _1} \cdots {\mu _k}}}}})_q} \in {(T{\Lambda ^k}TM)_q},
\end{align}
with $q = ({x^\mu }(p),{\xi ^\mu }) \in {\tau _M}^{ - 1}({U_p})$, 
where the local section 
\begin{align}
{e_{{\mu _1} \cdots {\mu _k}}} 
= \frac{\partial }{{\partial {y^{{\mu _1} \cdots {\mu _k}}}}},\;{\dot e_{{\mu _1} \cdots {\mu _k}}} 
= \frac{\partial }{{\partial {{\dot y}^{{\mu _1} \cdots {\mu _k}}}}}
\end{align} 
of $T{\Lambda ^k}{\tau _M}$ was defined by ${y^{{\mu _1} \cdots {\mu _k}}}({e_{^{{\nu _1} 
\cdots {\nu _k}}}}(z)) = k!\delta _{{\nu _1}}^{[{\mu _1}} \cdots \delta _{{\nu _k}}^{{\mu _k}]}$, 
${\dot y^{{\mu _1} \cdots {\mu _k}}}({\dot e_{^{{\nu _1} \cdots {\nu _k}}}}(z)) 
= k!\delta _{{\nu _1}}^{[{\mu _1}} \cdots \delta _{{\nu _k}}^{{\mu _k}]}$, 
for all $z \in {\Lambda ^k}{\tau _M}^{ - 1}({U_p})$. 
The map $T{\Lambda ^k}{\tau _M}$ 
sends ${\xi _v}$ to ${T_{{\Lambda ^k}{\tau _M}(v)}}M$ by 
\begin{align}
&T{\Lambda ^k}{\tau _M}({\xi _v}) 
= {\left. {\frac{{\partial {x^\mu }{\Lambda ^k}{\tau _M}{\psi ^{ - 1}}}}{{\partial {x^\nu }}}} 
\right|_{\psi (v)}}{\xi ^\nu }{\left( {\frac{\partial }{{\partial {x^\mu }}}} \right)_p} \nonumber \\
&\quad + \frac{1}{{k!}}{\left. {\frac{{\partial {x^\mu }{\Lambda ^k}{\tau _M}{\psi ^{ - 1}}}}{{\partial {y^{{\nu _1} 
\cdots {\nu _k}}}}}} \right|_{\psi (v)}}{\xi ^{{\nu _1} \cdots {\nu _k}}}
{\left( {\frac{\partial }{{\partial {x^\mu }}}} \right)_p} = {\xi ^\mu }
{\left( {\frac{\partial }{{\partial {x^\mu }}}} \right)_p}, 
\end{align}
which is indeed the base point $q$ of ${\xi _q}$, 
and we see that ${\left( {T{\Lambda ^k}{\tau _M}} \right)_q}$ becomes a vector space. 

The $k$-multivector bundle $({\Lambda ^k}TM,{\Lambda ^k}{\tau _M},M)$ 
also induces a second order bundle 
$({\Lambda ^k}T{\Lambda ^k}TM,{\Lambda ^k}{\tau _{{\Lambda ^k}TM}},{\Lambda ^k}TM)$ 
by iteration. We could introduce the induced charts for ${\Lambda ^k}T{\Lambda ^k}TM$ 
similarly as we did in Definition \ref{lkTM}. 
Let $(V,\psi ),\psi  = ({x^\mu },{y^{{\mu _1} \cdots {\mu _k}}})$ 
be a chart on ${\Lambda ^k}TM$, 
then the natural bases of ${\Lambda ^k}T{\Lambda ^k}TM$ on $V$ are 
\begin{align}
\left( {\frac{\partial }{{\partial {x^\mu }}}, \quad
\frac{\partial }{{\partial {y^{{\mu _1} \cdots {\mu _k}}}}}} \right),
\end{align} 
where ${\mu _1}, \cdots ,{\mu _k} = 1, \cdots ,n$ are alternating indices. 
The element of ${\Lambda ^k}T{\Lambda ^k}TM$ at a point $p \in {\Lambda ^k}TM$ will have the form 
\begin{align}
&w = \frac{1}{{k!}}{w^{{\mu _1} \cdots {\mu _k}}}{\left( {\frac{\partial }{{\partial {x^{{\mu _1}}}}} \wedge  
\cdots  \wedge \frac{\partial }{{\partial {x^{{\mu _k}}}}}} \right)_p} 
+ \frac{1}{{(k - 1)!}}{w^{I{\mu _1} \cdots {\mu _{k - 1}}}}{\left( {\frac{\partial }{{\partial {y^I}}} \wedge 
\frac{\partial }{{\partial {x^{{\mu _1}}}}} \wedge  \cdots  \wedge 
\frac{\partial }{{\partial {x^{{\mu _{k - 1}}}}}}} \right)_p} \nonumber \\
&\quad  + \frac{1}{{(k - 2)!2!}}{w^{{I_1}{I_2}{\mu _1} \cdots {\mu _{k - 2}}}}
{\left( {\frac{\partial }{{\partial {y^{{I_1}}}}} \wedge \frac{\partial }{{\partial {y^{{I_2}}}}} \wedge 
\frac{\partial }{{\partial {x^{{\mu _1}}}}} \wedge  
\cdots  \wedge \frac{\partial }{{\partial {x^{{\mu _{k - 2}}}}}}} \right)_p} +  \cdots  \nonumber \\
&\quad  + \frac{1}{{(k - 1)!}}{w^{{I_1}{I_2} \cdots {I_{k - 1}}\mu }}
{\left( {\frac{\partial }{{\partial {y^{{I_1}}}}} \wedge  \cdots  \wedge 
\frac{\partial }{{\partial {y^{{I_{k - 1}}}}}} \wedge 
\frac{\partial }{{\partial {x^\mu }}}} \right)_p} \nonumber \\
&\quad  + \frac{1}{{k!}}{w^{{I_1}{I_2} \cdots {I_k}}}
{\left( {\frac{\partial }{{\partial {y^{{I_1}}}}} \wedge 
\frac{\partial }{{\partial {y^{{I_2}}}}} \wedge  \cdots  
\wedge \frac{\partial }{{\partial {y^{{I_k}}}}}} \right)_p}.  
\end{align}
Here we have introduced a multi-index notation for visibility. 
The upper case latin index ${I_1},{I_2}, \cdots ,{I_k}$ 
is a combination of alternating $k$ indices, 
such that $I = {\mu _1} \cdots {\mu _k}$, and has its own unique label, 
and we suppressed the factorial coefficient which we will understand as already included in the summation over $I$ 
appearing consequently. 
For example, in the case $n=4$ and $k=2$, 
$I = {\mathbf{1}},{\mathbf{2}}, \cdots ,{}_4{C_2} = {\mathbf{6}}$, 
and 
\begin{align}
{\mathbf{1}} := (1\,2), {\mathbf{2}} := (1\,3), {\mathbf{3}} := (1\,4), {\mathbf{4}} := (2\,3), 
{\mathbf{5}} := (2\,4), {\mathbf{6}} := (3\,4), 
\end{align}
we are able to write explicitly
\begin{align}
&\frac{1}{2}{w^{IJ}}\frac{\partial }{{\partial {y^I}}} \wedge \frac{\partial }{{\partial {y^J}}} \nonumber \\
&  = {w^{{\mathbf{12}}}}\frac{\partial }{{\partial {y^{\mathbf{1}}}}} \wedge 
  \frac{\partial }{{\partial {y^{\mathbf{2}}}}} 
  + {w^{{\mathbf{13}}}}\frac{\partial }{{\partial {y^{\mathbf{1}}}}} \wedge 
  \frac{\partial }{{\partial {y^{\mathbf{3}}}}} 
  + {w^{{\mathbf{14}}}}\frac{\partial }{{\partial {y^{\mathbf{1}}}}} \wedge 
  \frac{\partial }{{\partial {y^{\mathbf{4}}}}} 
  + {w^{{\mathbf{15}}}}\frac{\partial }{{\partial {y^{\mathbf{1}}}}} \wedge 
  \frac{\partial }{{\partial {y^{\mathbf{5}}}}} 
  + {w^{{\mathbf{16}}}}\frac{\partial }{{\partial {y^{\mathbf{1}}}}} \wedge 
  \frac{\partial }{{\partial {y^{\mathbf{6}}}}} \nonumber \\
&  + {w^{{\mathbf{23}}}}\frac{\partial }{{\partial {y^{\mathbf{2}}}}} \wedge 
   \frac{\partial }{{\partial {y^{\mathbf{3}}}}} 
   + {w^{{\mathbf{24}}}}\frac{\partial }{{\partial {y^{\mathbf{2}}}}} \wedge 
   \frac{\partial }{{\partial {y^{\mathbf{4}}}}} 
   + {w^{{\mathbf{25}}}}\frac{\partial }{{\partial {y^{\mathbf{2}}}}} \wedge 
   \frac{\partial }{{\partial {y^{\mathbf{5}}}}} 
   + {w^{{\mathbf{26}}}}\frac{\partial }{{\partial {y^{\mathbf{2}}}}} \wedge 
   \frac{\partial }{{\partial {y^{\mathbf{6}}}}} \nonumber \\
&  + {w^{{\mathbf{34}}}}\frac{\partial }{{\partial {y^{\mathbf{3}}}}} \wedge 
   \frac{\partial }{{\partial {y^{\mathbf{4}}}}} 
   + {w^{{\mathbf{35}}}}\frac{\partial }{{\partial {y^{\mathbf{3}}}}} \wedge 
   \frac{\partial }{{\partial {y^{\mathbf{5}}}}} 
   + {w^{{\mathbf{36}}}}\frac{\partial }{{\partial {y^{\mathbf{3}}}}} \wedge 
   \frac{\partial }{{\partial {y^{\mathbf{6}}}}} \nonumber \\
&  + {w^{{\mathbf{45}}}}\frac{\partial }{{\partial {y^{\mathbf{4}}}}} \wedge 
   \frac{\partial }{{\partial {y^{\mathbf{5}}}}} 
   + {w^{{\mathbf{46}}}}\frac{\partial }{{\partial {y^{\mathbf{4}}}}} \wedge 
   \frac{\partial }{{\partial {y^{\mathbf{6}}}}} \nonumber \\
&   + {w^{{\mathbf{56}}}}\frac{\partial }{{\partial {y^{\mathbf{5}}}}} \wedge 
   \frac{\partial }{{\partial {y^{\mathbf{6}}}}} \nonumber \\
&  = {w^{(12)(13)}}\frac{\partial }{{\partial {y^{(12)}}}} \wedge 
   \frac{\partial }{{\partial {y^{(13)}}}} 
   + {w^{(12)(14)}}\frac{\partial }{{\partial {y^{(12)}}}} \wedge 
   \frac{\partial }{{\partial {y^{(14)}}}} 
   + {w^{(12)(23)}}\frac{\partial }{{\partial {y^{(12)}}}} \wedge 
   \frac{\partial }{{\partial {y^{(23)}}}} \nonumber \\
&  + {w^{(12)(24)}}\frac{\partial }{{\partial {y^{(12)}}}} \wedge 
   \frac{\partial }{{\partial {y^{(24)}}}} 
   + {w^{(12)(34)}}\frac{\partial }{{\partial {y^{(12)}}}} \wedge 
   \frac{\partial }{{\partial {y^{(34)}}}} 
   + {w^{(13)(14)}}\frac{\partial }{{\partial {y^{(13)}}}} \wedge 
   \frac{\partial }{{\partial {y^{(14)}}}} \nonumber \\
&  + {w^{(13)(23)}}\frac{\partial }{{\partial {y^{(13)}}}} \wedge 
   \frac{\partial }{{\partial {y^{(23)}}}} 
   + {w^{(13)(24)}}\frac{\partial }{{\partial {y^{(13)}}}} \wedge 
   \frac{\partial }{{\partial {y^{(24)}}}} 
   + {w^{(13)(34)}}\frac{\partial }{{\partial {y^{(13)}}}} \wedge 
   \frac{\partial }{{\partial {y^{(34)}}}} \nonumber  \\
& + {w^{(14)(23)}}\frac{\partial }{{\partial {y^{(14)}}}} \wedge 
   \frac{\partial }{{\partial {y^{(23)}}}} 
   + {w^{(14)(24)}}\frac{\partial }{{\partial {y^{(14)}}}} \wedge 
   \frac{\partial }{{\partial {y^{(24)}}}} 
   + {w^{(14)(34)}}\frac{\partial }{{\partial {y^{(14)}}}} \wedge 
   \frac{\partial }{{\partial {y^{(24)}}}} \nonumber \\
&  + {w^{(23)(24)}}\frac{\partial }{{\partial {y^{(23)}}}} \wedge 
   \frac{\partial }{{\partial {y^{(24)}}}} 
   + {w^{(23)(34)}}\frac{\partial }{{\partial {y^{(23)}}}} \wedge 
   \frac{\partial }{{\partial {y^{(34)}}}} 
   + {w^{(24)(34)}}\frac{\partial }{{\partial {y^{(24)}}}} \wedge 
   \frac{\partial }{{\partial {y^{(34)}}}}. 
\end{align}
Therefore in this notation, coordinate functions are labelled as 
\begin{align}
{y^I}: = {y^{{\mu _1} \cdots {\mu _k}}}, 
\quad {z^{{I_1}{I_2}{\nu _3} \cdots {\nu _k}}} := {z^{(\mu _1^1 \cdots \mu _k^1)(\mu _1^2 \cdots \mu _k^2){\nu _3} 
\cdots {\nu _k}}},
\end{align}
 etc. 
The summation conventions between ordered and non-ordered indices are, 
\begin{align}
&\quad \quad \quad \frac{{\partial K}}{{\partial {y^I}}}{y^I} 
:= \sum\limits_{{\mu _1} < {\mu _2} <  \cdots  < {\mu _k}} 
{\frac{{\partial K}}{{\partial {y^{{\mu _1} \cdots {\mu _k}}}}}{y^{{\mu _1} \cdots {\mu _k}}}} 
= \frac{1}{{k!}}\frac{{\partial K}}{{\partial {y^{{\mu _1} 
\cdots {\mu _k}}}}}{y^{{\mu _1} \cdots {\mu _k}}},   \\
&\frac{1}{{l!(k - l)!}}\frac{{\partial K}}{{\partial {z^{{I_1} \cdots {I_l}{\nu _{l + 1}} 
\cdots {\nu _k}}}}}{z^{{I_1} \cdots {I_l}{\nu _{l + 1}} \cdots {\nu _k}}} 
:= \sum\limits_{{I_1} < {I_2} <  \cdots  < {I_l}} 
{\sum\limits_{{\nu _{l + 1}} < {\nu _{l + 2}} <  \cdots  < {\nu _k}} 
{\frac{{\partial K}}{{\partial {z^{{I_1} \cdots {I_l}{\nu _{l + 1}} 
\cdots {\nu _k}}}}}{z^{{I_1} \cdots {I_l}{\nu _{l + 1}} \cdots {\nu _k}}}} } \nonumber \\ 
&= \sum\limits_{{\text{ordered by }}I{\text{ }}} 
{\sum\limits_{\mu _1^1 < \cdots  < \mu _k^1} { 
\cdots \sum\limits_{\mu _1^l < \cdots  < \mu _k^l} 
{\sum\limits_{{\nu _{l + 1}} < \cdots  < {\nu _k}} 
{\frac{{\partial K}}{{\partial {z^{(\mu _1^1 \cdots \mu _k^1) 
\cdots (\mu _1^l \cdots \mu _k^l){\nu _{l + 1}} 
\cdots {\nu _k}}}}}{z^{(\mu _1^1 \cdots \mu _k^1) 
\cdots (\mu _1^l \cdots \mu _k^l){\nu _{l + 1}} 
\cdots {\nu _k}}}} } } }, \nonumber \\
&\hspace{10cm}1 \leqslant l \leqslant k, 
\end{align}
and the induced chart of the space ${\Lambda ^k}T{\Lambda ^k}TM$ could be introduced as 
$({\tilde W^2},{\tilde \psi ^2})$, 
${\tilde W^2} = {({\Lambda ^k}{\tau _{{\Lambda ^k}TM}})^{ - 1}}(V)$, 
${\tilde \psi ^2} = ({x^\mu },{y^{{\mu _1} \cdots {\mu _k}}},{z^{{\mu _1} \cdots {\mu _k}}}, 
{z^{\operatorname{I} {\mu _1} \cdots {\mu _{k - 1}}}},{z^{{I_1}{I_2}{\mu _1} 
\cdots {\mu _{k - 2}}}}, \cdots ,{z^{{I_1} \cdots {I_k}}}).$ 

Now, consider the extended tangent map 
${\Lambda ^k}T{\Lambda ^k}{\tau _M}:{\Lambda ^k}T{\Lambda ^k}TM \to {\Lambda ^k}TM$. 
This map sends $k$-multivector at $p \in {\Lambda ^k}TM$ to $k$-multivector at 
${\Lambda ^k}{\tau _M}(p) \in M$, 
and induces a vector bundle $({\Lambda ^k}T{\Lambda ^k}TM, {\Lambda ^k}T{\Lambda ^k}{\tau _M}, {\Lambda ^k}TM)$. 
As we did in the case for constructing the bundle $({T^2}M, 
{\left. {{\tau _{TM}}} \right|_{{T^2}M}},TM)$, 
we can use the bundle isomorphism between the two bundles, 
$({\Lambda ^k}T{\Lambda ^k}TM, \lb[2] {\Lambda ^k}{\tau _{{\Lambda ^k}TM}},\lb[3] {\Lambda ^k}TM)$ and 
$({\Lambda ^k}T{\Lambda ^k}TM, \lb[2] {\Lambda ^k}T{\Lambda ^k}{\tau _M},\lb[3] {\Lambda ^k}TM)$, 
to construct the bundle of second order multivectors. Let $w \in {\Lambda ^k}T{\Lambda ^k}TM$ be the 
$k$-multivector at point $p \in {\Lambda ^k}TM$. 
Then by the previous definition, 
\begin{align}
&{\Lambda ^k}T{\Lambda ^k}{\tau _M}(w) = \frac{1}{{k!}}{w^{{\mu _1} 
  \cdots {\mu _k}}}{\left( {T{\Lambda ^k}{\tau _M} 
  \left( {\frac{\partial }{{\partial {x^{{\mu _1}}}}}} \right) \wedge 
  T{\Lambda ^k}{\tau _M}\left( {\frac{\partial }{{\partial {x^{{\mu _2}}}}}} \right) \wedge  
  \cdots  \wedge T{\Lambda ^k}{\tau _M}\left( 
  {\frac{\partial }{{\partial {x^{{\mu _k}}}}}} \right)} \right)_p} \nonumber \\
& + \frac{1}{{k!}}{w^{I{\mu _1} \cdots {\mu _{k - 1}}}}{\left( {T{\Lambda ^k}{\tau _M} 
  \left( {\frac{\partial }{{\partial {y^I}}}} \right) \wedge T{\Lambda ^k}{\tau _M} 
  \left( {\frac{\partial }{{\partial {x^{{\mu _1}}}}}} \right) \wedge  \cdots  \wedge T{\Lambda ^k}{\tau _M} 
  \left( {\frac{\partial }{{\partial {x^{{\mu _{k - 1}}}}}}} \right)} \right)_p} \nonumber \\
& +  \cdots  
  + \frac{1}{{k!}}{w^{{I_1}{I_2} \cdots {I_{k - 1}}\mu }}{\left( {T{\Lambda ^k}{\tau _M} 
  \left( {\frac{\partial }{{\partial {y^{{I_1}}}}}} \right) \wedge  \cdots  \wedge T{\Lambda ^k}{\tau _M} 
  \left( {\frac{\partial }{{\partial {y^{{I_{k - 1}}}}}}} \right) \wedge T{\Lambda ^k}{\tau _M} 
  \left( {\frac{\partial }{{\partial {x^\mu }}}} \right)} \right)_p} \nonumber \\
& + \frac{1}{{k!}}{w^{{I_1}{I_2} \cdots {I_k}}}{\left( {T{\Lambda ^k}{\tau _M}
  \left( {\frac{\partial }{{\partial {y^{{I_1}}}}}} \right) \wedge T{\Lambda ^k}{\tau _M}
  \left( {\frac{\partial }{{\partial {y^{{I_2}}}}}} \right) \wedge  
  \cdots  \wedge T{\Lambda ^k}{\tau _M}\left( {\frac{\partial }{{\partial {y^{{I_k}}}}}} \right)} \right)_p} \nonumber \\
& = \frac{1}{{k!}}{w^{{\mu _1} \cdots {\mu _k}}}{\left( {\frac{\partial }{{\partial {x^{{\mu _1}}}}} \wedge  
\cdots  \wedge \frac{\partial }{{\partial {x^{{\mu _k}}}}}} \right)_{{\Lambda ^k}{\tau _M}(p)}} 
\in {\Lambda ^k}{T_{{\Lambda ^k}{\tau _M}(p)}}M.
\end{align}
On the other hand, ${\Lambda ^k}{\tau _{{\Lambda ^k}TM}}(w) = p$, 
and the equation for the isomorphism is 
\begin{align}
\frac{1}{{k!}}{w^{{\mu _1} 
\cdots {\mu _k}}}{\left( {\frac{\partial }{{\partial {x^{{\mu _1}}}}} \wedge  
\cdots  \wedge \frac{\partial }{{\partial {x^{{\mu _k}}}}}} \right)_{{\Lambda ^k}{\tau _M}(p)}} = p,
\end{align}
and in coordinate expression becomes, 
\begin{align}
{x^\mu }(p) = {x^\mu }({\Lambda ^k}{\tau _M}(p)), \quad 
{y^{{\mu _1} \cdots {\mu _k}}}(p) 
= {w^{{\mu _1} \cdots {\mu _k}}} = {z^{{\mu _1} \cdots {\mu _k}}}(w).
\end{align} 
Therefore, we can take as the chart on ${({\Lambda ^k}T)^2}M$, 
$({W^2},{\psi ^2})$, ${W^2} = {({\left. {{\Lambda ^k}{\tau _{{\Lambda ^k}TM}}} 
\right|_{{{({\Lambda ^k}T)}^2}M}})^{ - 1}}(V)$, 
${\psi ^2} = ({x^\mu },{y^{{\mu _1} \cdots {\mu _k}}},{z^{\operatorname{I} {\mu _1} 
\cdots {\mu _{k - 1}}}},{z^{{I_1}{I_2}{\mu _1} \cdots {\mu _{k - 2}}}} \cdots ,{z^{{I_1} \cdots {I_k}}})$. 
The coordinate ${z^{{\mu _1} \cdots {\mu _k}}}$ 
is not present from the original chart on ${\Lambda ^k}T{\Lambda ^k}TM$ 
because of the above equation of submanifolds. 

\begin{defn}  Second order $k$-multivector bundle ${({\Lambda ^k}T)^2}M$ over 
${\Lambda ^k}TM$ \label{def_2ndorder_kTMoverkTM} \\
Let $({\Lambda ^k}TM,{\Lambda ^k}{\tau _M},M)$ be the 
$k$-multivector bundle with the base space $M$ and $({\Lambda ^k}T{\Lambda ^k}TM, \lb[3]
{\Lambda ^k}{\tau _{{\Lambda ^k}TM}},{\Lambda ^k}TM)$ 
the $k$-multivector bundle with the base space ${\Lambda ^k}TM$. 
Denote the subset of elements $w \in {\Lambda ^k}T{\Lambda ^k}TM$ 
which satisfy the equations of submanifold 
\begin{eqnarray}
{\Lambda ^k}{\tau _{{\Lambda ^k}TM}}(w) = {\Lambda ^k}T{\Lambda ^k}{\tau _M}(w)
\end{eqnarray} 
as ${({\Lambda ^k}T)^2}M$, 
and a map from ${({\Lambda ^k}T)^2}M$ to ${\Lambda ^k}TM$ by 
\begin{eqnarray}
{\Lambda ^k}\tau _M^{2,1}: = {\left. {{\Lambda ^k}{\tau _{{\Lambda ^k}TM}}} 
\right|_{{{({\Lambda ^k}T)}^2}M}}.
\end{eqnarray} 
The triple $({({\Lambda ^k}T)^2}M, {\Lambda ^k}\tau _M^{2,1}, {\Lambda ^k}TM)$ 
becomes a bundle, and we call them {\it second order $k$-multi-vector bundle over ${\Lambda ^k}TM$}. 
\end{defn}
Again this is not a vector bundle, i.e. the fibres of ${\Lambda ^k}\tau _M^{2,1}$ are not vector spaces, 
since in general, the scalar multiplication of ${w_p} \in {({\Lambda ^k}T)^2}M$, $p \in {\Lambda ^k}TM$ 
does not belong to the same fibre.
As in Section \ref{subsec_second_TM}, we can similarly check that 
$({({\Lambda ^k}T)^2}M,{\Lambda ^k}\tau _M^{2,1},{\Lambda ^k}TM)$ is a bundle. 
Namely we first introduce manifold structure on ${({\Lambda ^k}T)^2}M$, 
and then consider the local trivialisation. The charts on ${({\Lambda ^k}T)^2}M$ 
are already introduced, so let us check the coordinate transformations. \\
Let $({W^2},{\psi ^2})$, ${\psi ^2} = ({x^\mu },{y^{{\mu _1} 
\cdots {\mu _k}}},{z^{\operatorname{I} {\mu _1} \cdots {\mu _{k - 1}}}},{z^{{I_1}{I_2}{\mu _1} 
\cdots {\mu _{k - 2}}}} \cdots ,{z^{{I_1} \cdots {I_k}}})$  
and $({\bar W^2},{\bar \psi ^2})$, ${\bar \psi ^2} 
= \left({\bar x^\mu }, {\bar y^{{\mu _1} 
\cdots {\mu _k}}}, {\bar z^{\operatorname{I} {\mu _1} \cdots {\mu _{k - 1}}}}, {\bar z^{{I_1}{I_2}{\mu _1} 
\cdots {\mu _{k - 2}}}}, \cdots ,{\bar z^{{I_1} \cdots {I_k}}}\right)$ 
be two charts on 
${({\Lambda ^k}T)^2}M$, with ${W^2} \cap {\bar W^2} \ne \emptyset $. 
Then express the element 
${w_q} \in {W^2} \cap {\bar W^2}$, $q \in {\Lambda ^k}TM$, by these charts, 
\begin{align}
&{w_q} = \frac{1}{{k!}}{y^{{\mu _1} \cdots {\mu _k}}}(q){\left( {\frac{\partial }{{\partial {x^{{\mu _1}}}}} \wedge  
  \cdots  \wedge \frac{\partial }{{\partial {x^{{\mu _k}}}}}} \right)_q} 
  + \frac{1}{{(k - 1)!}}{w^{{I_1}{\mu _2} \cdots {\mu _k}}}{\left( {\frac{\partial }{{\partial {y^{{I_1}}}}} \wedge 
  \frac{\partial }{{\partial {x^{{\mu _2}}}}} \wedge  \cdots  \wedge 
  \frac{\partial }{{\partial {x^{{\mu _k}}}}}} \right)_q} \nonumber \\
&  \quad  + \frac{1}{{(k - 2)!2!}}{w^{{I_1}{I_2}{\mu _3} \cdots {\mu _k}}}
  {\left( {\frac{\partial }{{\partial {y^{{I_1}}}}} \wedge \frac{\partial }{{\partial {y^{{I_2}}}}} \wedge 
  \frac{\partial }{{\partial {x^{{\mu _3}}}}} \wedge  \cdots  \wedge 
  \frac{\partial }{{\partial {x^{{\mu _k}}}}}} \right)_q} \nonumber \\
& \quad  +  \cdots  + \frac{1}{{(k - l)!l!}}{w^{{I_1} \cdots {I_l}{\mu _{l + 1}} \cdots {\mu _k}}}
  {\left( {\frac{\partial }{{\partial {y^{{I_1}}}}} \wedge  \cdots  \wedge \frac{\partial }{{\partial {y^{{I_l}}}}} \wedge 
  \frac{\partial }{{\partial {x^{{\mu _{l + 1}}}}}} \wedge  \cdots  \wedge 
  \frac{\partial }{{\partial {x^{{\mu _k}}}}}} \right)_q} \nonumber \\
& \quad  +  \cdots  + \frac{1}{{k!}}{w^{{I_1}{I_2} \cdots {I_k}}}
  {\left( {\frac{\partial }{{\partial {y^{{I_1}}}}} \wedge \frac{\partial }{{\partial {y^{{I_2}}}}} \wedge  
  \cdots  \wedge \frac{\partial }{{\partial {y^{{I_k}}}}}} \right)_q} \nonumber \\
& = \frac{1}{{k!}}{{\bar y}^{{\mu _1} \cdots {\mu _k}}}(q)
   {\left( {\frac{\partial }{{\partial {{\bar x}^{{\mu _1}}}}} \wedge  \cdots  \wedge
    \frac{\partial }{{\partial {{\bar x}^{{\mu _k}}}}}} \right)_q} 
    + \frac{1}{{(k - 1)!}}{{\bar w}^{{I_1}{\mu _2} \cdots {\mu _k}}
    }{\left( {\frac{\partial }{{\partial {{\bar y}^{{I_1}}}}} \wedge 
    \frac{\partial }{{\partial {{\bar x}^{{\mu _2}}}}} \wedge  
    \cdots  \wedge \frac{\partial }{{\partial {{\bar x}^{{\mu _k}}}}}} \right)_q} \nonumber \\
& \quad  + \frac{1}{{(k - 2)!2!}}{{\bar w}^{{I_1}{I_2}{\mu _3} 
  \cdots {\mu _k}}}{\left( {\frac{\partial }{{\partial {{\bar y}^{{I_1}}}}} \wedge 
  \frac{\partial }{{\partial {{\bar y}^{{I_2}}}}} \wedge 
  \frac{\partial }{{\partial {{\bar x}^{{\mu _3}}}}} \wedge  
  \cdots \wedge \frac{\partial }{{\partial {{\bar x}^{{\mu _k}}}}}} \right)_q} \nonumber \\
& \quad  +  \cdots  + \frac{1}{{(k - l)!l!}}{{\bar w}^{{I_1} \cdots {I_l}{\mu _{l + 1}} 
  \cdots {\mu _k}}}{\left( {\frac{\partial }{{\partial {{\bar y}^{{I_1}}}}} \wedge  
  \cdots  \wedge \frac{\partial }{{\partial {{\bar y}^{{I_l}}}}} \wedge 
  \frac{\partial }{{\partial {{\bar x}^{{\mu _{l + 1}}}}}} \wedge  
  \cdots  \wedge \frac{\partial }{{\partial {{\bar x}^{{\mu _k}}}}}} \right)_q} \nonumber \\
& \quad  +  \cdots  + \frac{1}{{k!}}{{\bar w}^{{I_1}{I_2} \cdots {I_k}}}
{\left( {\frac{\partial }{{\partial {{\bar y}^{{I_1}}}}} \wedge 
\frac{\partial }{{\partial {{\bar y}^{{I_2}}}}} \wedge  
\cdots  \wedge \frac{\partial }{{\partial {{\bar y}^{{I_k}}}}}} \right)_q},
\end{align}
with 
\begin{align}
&{z^{{I_1}{\mu _2} \cdots {\mu _k}}}({w_q}) = {w^{{I_1}{\mu _2} \cdots {\mu _k}}},{z^{{I_1}{I_2}{\mu _3} 
  \cdots {\mu _k}}}({w_q}) = {w^{{I_1}{I_2}{\mu _3} \cdots {\mu _k}}}, 
  \cdots ,{z^{{I_1}{I_2} \cdots {I_k}}}({w_q}) = {w^{{I_1}{I_2} \cdots {I_k}}}, \nonumber \\
&{{\bar z}^{{I_1}{\mu _2} \cdots {\mu _k}}}({w_q}) = {{\bar w}^{{I_1}{\mu _2} 
  \cdots {\mu _k}}},{{\bar z}^{{I_1}{I_2}{\mu _3} \cdots {\mu _k}}}({w_q}) 
  = {{\bar w}^{{I_1}{I_2}{\mu _3} \cdots {\mu _k}}}, \cdots ,{{\bar z}^{{I_1}{I_2} 
  \cdots {I_k}}}({w_q}) = {w^{{I_1}{I_2} \cdots {I_k}}}. 
\end{align}
Since the base of one vectors transform as
\begin{align}
\frac{\partial }{{\partial {{\bar x}^\mu }}} 
= \frac{{\partial {x^\nu }}}{{\partial {{\bar x}^\mu }}}\frac{\partial }{{\partial {x^\nu }}} 
+ \frac{{\partial {y^I}}}{{\partial {{\bar x}^\mu }}}\frac{\partial }{{\partial {y^I}}}, 
\frac{\partial }{{\partial {{\bar y}^I}}} = \frac{{\partial {y^J}}}{{\partial {{\bar y}^I}}}
\frac{\partial }{{\partial {y^J}}} = k!\frac{{\partial {x^J}}}{{\partial {{\bar x}^I}}}\frac{\partial }{{\partial {y^J}}}, 
\end{align}
the bases of ${\Lambda ^k}T{\Lambda ^k}TM$ will transform as 
\begin{align}
&\frac{\partial }{{\partial {{\bar x}^{{\mu _1}}}}} \wedge  \cdots  \wedge \frac{\partial }{{\partial {{\bar x}^{{\mu _k}}}}} 
= \left( {\frac{{\partial {x^{{\nu _1}}}}}{{\partial {{\bar x}^{{\mu _1}}}}}\frac{\partial }{{\partial {x^{{\nu _1}}}}} 
+ \frac{{\partial {y^{{I_1}}}}}{{\partial {{\bar x}^{{\mu _1}}}}}\frac{\partial }{{\partial {y^{{I_1}}}}}} \right) \wedge  
\cdots  \wedge \left( {\frac{{\partial {x^{{\nu _k}}}}}{{\partial {{\bar x}^{{\mu _k}}}}}
\frac{\partial }{{\partial {x^{{\nu _k}}}}} + \frac{{\partial {y^{{I_k}}}}}{{\partial {{\bar x}^{{\mu _k}}}}} 
\frac{\partial }{{\partial {y^{{I_k}}}}}} \right) \nonumber \\
&  = \frac{{\partial {x^{{\nu _1}}}}}{{\partial {{\bar x}^{{\mu _1}}}}} 
   \cdots \frac{{\partial {x^{{\nu _k}}}}}{{\partial {{\bar x}^{{\mu _k}}}}}
   \frac{\partial }{{\partial {x^{{\nu _1}}}}} \wedge  \cdots  \wedge \frac{\partial }{{\partial {x^{{\nu _k}}}}} 
   + {}_k{C_1}\frac{{\partial {y^{{I_1}}}}}{{\partial {{\bar x}^{{\mu _1}}}}} 
   \frac{{\partial {x^{{\nu _2}}}}}{{\partial {{\bar x}^{{\mu _2}}}}} 
   \cdots \frac{{\partial {x^{{\nu _k}}}}}{{\partial {{\bar x}^{{\mu _k}}}}} 
   \frac{\partial }{{\partial {y^{{I_1}}}}} \wedge \frac{\partial }{{\partial {x^{{\nu _2}}}}} \wedge  
   \cdots  \wedge \frac{\partial }{{\partial {x^{{\nu _k}}}}} \nonumber \\
&   + {}_k{C_2}\frac{{\partial {y^{{I_1}}}}}{{\partial {{\bar x}^{{\mu _1}}}}} 
   \frac{{\partial {y^{{I_2}}}}}{{\partial {{\bar x}^{{\mu _2}}}}} 
   \frac{{\partial {x^{{\nu _3}}}}}{{\partial {{\bar x}^{{\mu _3}}}}} 
   \cdots \frac{{\partial {x^{{\nu _k}}}}}{{\partial {{\bar x}^{{\mu _k}}}}} 
   \frac{\partial }{{\partial {y^{{I_1}}}}} \wedge \frac{\partial }{{\partial {y^{{I_2}}}}} \wedge 
   \frac{\partial }{{\partial {x^{{\nu _3}}}}} \wedge  
   \cdots  \wedge \frac{\partial }{{\partial {x^{{\nu _k}}}}} \nonumber \\
&  +  \cdots \!  + \frac{{\partial {y^{{I_1}}}}}{{\partial {{\bar x}^{{\mu _1}}}}} 
   \frac{{\partial {y^{{I_2}}}}}{{\partial {{\bar x}^{{\mu _2}}}}} 
   \cdots \frac{{\partial {y^{{I_k}}}}}{{\partial {{\bar x}^{{\mu _k}}}}} 
   \frac{\partial }{{\partial {y^{{I_1}}}}} \wedge \frac{\partial }{{\partial {y^{{I_2}}}}} \wedge  
   \cdots  \wedge \frac{\partial }{{\partial {y^{{I_k}}}}}, \nonumber \\  \nonumber \\
&\frac{\partial }{{\partial {{\bar y}^{{I_1}}}}} \wedge \frac{\partial }{{\partial {{\bar x}^{{\mu _2}}}}} \wedge  
\cdots  \wedge \frac{\partial }{{\partial {{\bar x}^{{\mu _k}}}}} \nonumber \\
&= k!\frac{{\partial {x^{{J_1}}}}}{{\partial {{\bar x}^{{I_1}}}}}\frac{\partial }{{\partial {y^{{J_1}}}}} \wedge 
\left( {\frac{{\partial {x^{{\nu _2}}}}}{{\partial {{\bar x}^{{\mu _2}}}}}\frac{\partial }{{\partial {x^{{\nu _2}}}}} 
+ \frac{{\partial {y^{{I_2}}}}}{{\partial {{\bar x}^{{\mu _2}}}}}\frac{\partial }{{\partial {y^{{I_2}}}}}} \right) \wedge  
\cdots  \wedge \left( {\frac{{\partial {x^{{\nu _k}}}}}{{\partial {{\bar x}^{{\mu _k}}}}} 
\frac{\partial }{{\partial {x^{{\nu _k}}}}} 
+ \frac{{\partial {y^{{I_k}}}}}{{\partial {{\bar x}^{{\mu _k}}}}}\frac{\partial }{{\partial {y^{{I_k}}}}}} \right) \nonumber \\
&  = k!\left( {\frac{{\partial {x^{{J_1}}}}}{{\partial {{\bar x}^{{I_1}}}}} 
   \frac{{\partial {x^{{\nu _2}}}}}{{\partial {{\bar x}^{{\mu _2}}}}} 
   \cdots \frac{{\partial {x^{{\nu _k}}}}}{{\partial {{\bar x}^{{\mu _k}}}}}\frac{\partial }{{\partial {y^{{J_1}}}}} \wedge 
   \frac{\partial }{{\partial {x^{{\nu _2}}}}} \wedge  \cdots  \wedge \frac{\partial }{{\partial {x^{{\nu _k}}}}}} \right. \nonumber \\
&   + {}_{k - 1}{C_1}\frac{{\partial {x^{{J_1}}}}}{{\partial {{\bar x}^{{I_1}}}}} 
   \frac{{\partial {y^{{J_2}}}}}{{\partial {{\bar x}^{{\mu _2}}}}} 
   \frac{{\partial {x^{{\nu _3}}}}}{{\partial {{\bar x}^{{\mu _3}}}}} 
   \cdots \frac{{\partial {x^{{\nu _k}}}}}{{\partial {{\bar x}^{{\mu _k}}}}} 
   \frac{\partial }{{\partial {y^{{J_1}}}}} \wedge \frac{\partial }{{\partial {y^{{J_2}}}}} \wedge 
   \frac{\partial }{{\partial {x^{{\nu _3}}}}} \wedge  
   \cdots  \wedge \frac{\partial }{{\partial {x^{{\nu _k}}}}} \nonumber \\
& \left. { +  \cdots  + \frac{{\partial {x^{{J_1}}}}}{{\partial {{\bar x}^{{I_1}}}}} 
  \frac{{\partial {y^{{J_2}}}}}{{\partial {{\bar x}^{{\mu _2}}}}} 
  \cdots \frac{{\partial {y^{{I_k}}}}}{{\partial {{\bar x}^{{\mu _k}}}}}\frac{\partial }{{\partial {y^{{I_1}}}}} \wedge 
  \frac{\partial }{{\partial {y^{{I_2}}}}} \wedge  \cdots  \wedge \frac{\partial }{{\partial {y^{{I_k}}}}}} \right), \nonumber \\ \nonumber \\ 
& \vdots \nonumber \\  \nonumber \\
& \frac{\partial }{{\partial {{\bar y}^{{I_1}}}}} \wedge  
  \cdots  \wedge \frac{\partial }{{\partial {{\bar y}^{{I _k}}}}} = {(k!)^k} 
  \frac{{\partial {x^{{J_1}}}}}{{\partial {{\bar x}^{{I_1}}}}}\frac{\partial }{{\partial {y^{{J_1}}}}} \wedge 
  \cdots  \wedge \frac{{\partial {x^{{J_k}}}}}{{\partial {{\bar x}^{{I_k}}}}} 
  \frac{\partial }{{\partial {y^{{J_k}}}}} \nonumber \\
&  = {(k!)^k}\frac{{\partial {x^{{J_1}}}}}{{\partial {{\bar x}^{{I_1}}}}} 
   \cdots \frac{{\partial {x^{{J_k}}}}}{{\partial {{\bar x}^{{I_k}}}}} 
   \frac{\partial }{{\partial {y^{{J_1}}}}} \wedge  \cdots  \wedge \frac{\partial }{{\partial {y^{{J_k}}}}},  
\end{align}
and the coordinate transformations of $w \in \Lambda^k T \Lambda^k T M$ with base point on $\Lambda^k TM$ will be
\begin{align}
&w = \frac{1}{{k!}}{{\bar y}^{{\mu _1} \cdots {\mu _k}}} 
  \frac{{\partial {x^{{\nu _1}}}}}{{\partial {{\bar x}^{{\mu _1}}}}} 
  \cdots \frac{{\partial {x^{{\nu _k}}}}}{{\partial {{\bar x}^{{\mu _k}}}}} 
  \frac{\partial }{{\partial {x^{{\nu _1}}}}} \wedge  \cdots  \wedge 
  \frac{\partial }{{\partial {x^{{\nu _k}}}}} \nonumber \\
&  + \! \left( {\frac{1}{{k!}}{{\bar y}^{{\mu _1} \cdots {\mu _k}}}{}_k{C_1} 
   \frac{{\partial {y^{{J_1}}}}}{{\partial {{\bar x}^{{\mu _1}}}}} 
   \frac{{\partial {x^{{\nu _2}}}}}{{\partial {{\bar x}^{{\mu _2}}}}} 
   \cdots \frac{{\partial {x^{{\nu _k}}}}}{{\partial {{\bar x}^{{\mu _k}}}}} 
   + \frac{{k!}}{{(k - 1)!}}{{\bar w}^{{I_1}{\mu _2} \cdots {\mu _k}}}
   \frac{{\partial {x^{{J_1}}}}}{{\partial {{\bar x}^{{I_1}}}}} 
   \frac{{\partial {x^{{\nu _2}}}}}{{\partial {{\bar x}^{{\mu _2}}}}} 
   \cdots \frac{{\partial {x^{{\nu _k}}}}}{{\partial {{\bar x}^{{\mu _k}}}}}} \!  \right) \nonumber \\
&\hspace{8cm}   \times \frac{\partial }{{\partial {y^{{J_1}}}}} \wedge \frac{\partial }{{\partial {x^{{\nu _2}}}}} \wedge  
   \cdots  \wedge \frac{\partial }{{\partial {x^{{\nu _k}}}}} \nonumber \\
&  + \left( {\frac{1}{{k!}}{{\bar y}^{{\mu _1} \cdots {\mu _k}}}{}_k{C_2} 
   \frac{{\partial {y^{{J_1}}}}}{{\partial {{\bar x}^{{\mu _1}}}}} 
   \frac{{\partial {y^{{J_2}}}}}{{\partial {{\bar x}^{{\mu _2}}}}} 
   \frac{{\partial {x^{{\nu _3}}}}}{{\partial {{\bar x}^{{\mu _3}}}}} 
   \cdots \frac{{\partial {x^{{\nu _k}}}}}{{\partial {{\bar x}^{{\mu _k}}}}} 
   + \frac{{k!}}{{(k - 1)!}}{{\bar w}^{{I_1}{\mu _2} \cdots {\mu _k}}}{}_{k - 1}{C_1} 
   \frac{{\partial {x^{{J_1}}}}}{{\partial {{\bar x}^{{I_1}}}}}
   \frac{{\partial {y^{{J_2}}}}}{{\partial {{\bar x}^{{\mu _2}}}}}
   \frac{{\partial {x^{{\nu _3}}}}}{{\partial {{\bar x}^{{\mu _3}}}}} 
   \cdots \frac{{\partial {x^{{\nu _k}}}}}{{\partial {{\bar x}^{{\mu _k}}}}}} \right. \nonumber \\
& \left. { + \frac{{{{(k!)}^2}}}{{(k - 2)!2!}}{{\bar w}^{{I_1}{I_2}{\mu _3} 
  \cdots {\mu _k}}}{}_{k - 2}{C_0}\frac{{\partial {x^{{J_1}}}}}{{\partial {{\bar x}^{{I_1}}}}}
  \frac{{\partial {x^{{J_2}}}}}{{\partial {{\bar x}^{{I_2}}}}}
  \frac{{\partial {x^{{\nu _3}}}}}{{\partial {{\bar x}^{{\mu _3}}}}} 
  \cdots \frac{{\partial {x^{{\nu _k}}}}}{{\partial {{\bar x}^{{\mu _k}}}}}} \right) 
  \frac{\partial }{{\partial {y^{{J_1}}}}} \wedge \frac{\partial }{{\partial {y^{{J_2}}}}} \wedge 
  \frac{\partial }{{\partial {x^{{\nu _3}}}}} \wedge  \cdots  \wedge \frac{\partial }{{\partial {x^{{\nu _k}}}}} \nonumber \\
&  +  \cdots \! + \! \left( \frac{1}{{k!}}{{\bar y}^{{\mu _1} 
   \cdots {\mu _k}}}{}_k{C_l}\frac{{\partial {y^{{J_1}}}}}{{\partial {{\bar x}^{{\mu _1}}}}} 
   \cdots \frac{{\partial {y^{{J_l}}}}}{{\partial {{\bar x}^{{\mu _l}}}}}
   \frac{{\partial {x^{{\nu _{l + 1}}}}}}{{\partial {{\bar x}^{{\mu _{l + 1}}}}}} 
   \cdots \frac{{\partial {x^{{\nu _k}}}}}{{\partial {{\bar x}^{{\mu _k}}}}} \right. \nonumber \\
&  + \frac{{k!}}{{(k - 1)!}}{{\bar w}^{{I_1}{\mu _2} \cdots {\mu _k}}}{}_{k - 1}{C_{l - 1}} 
   \frac{{\partial {x^{{J_1}}}}}{{\partial {{\bar x}^{{I_1}}}}}
   \frac{{\partial {y^{{J_2}}}}}{{\partial {{\bar x}^{{\mu _2}}}}} 
   \cdots \frac{{\partial {y^{{J_l}}}}}{{\partial {{\bar x}^{{\mu _l}}}}}
   \frac{{\partial {x^{{\nu _{l + 1}}}}}}{{\partial {{\bar x}^{{\mu _{l + 1}}}}}} 
   \cdots \frac{{\partial {x^{{\nu _k}}}}}{{\partial {{\bar x}^{{\mu _k}}}}}  \nonumber \\
&  + \frac{{{{(k!)}^2}}}{{(k - 2)!2!}}{{\bar w}^{{I_1}{I_2}{\mu _3} 
   \cdots {\mu _k}}}{}_{k - 2}{C_{l - 2}}\frac{{\partial {x^{{J_1}}}}}{{\partial {{\bar x}^{{I_1}}}}}
   \frac{{\partial {x^{{J_2}}}}}{{\partial {{\bar x}^{{I_2}}}}} 
   \frac{{\partial {y^{{J_3}}}}}{{\partial {{\bar x}^{{\mu _3}}}}} 
   \cdots \frac{{\partial {y^{{J_l}}}}}{{\partial {{\bar x}^{{\mu _l}}}}} 
   \frac{{\partial {x^{{\nu _{l + 1}}}}}}{{\partial {{\bar x}^{{\mu _{l + 1}}}}}} 
   \cdots \frac{{\partial {x^{{\nu _k}}}}}{{\partial {{\bar x}^{{\mu _k}}}}} \nonumber \\
& \left. { +  \cdots \! + \frac{{{{(k!)}^l}}}{{(k - l)!l!}}{{\bar w}^{{I_1} 
  \cdots {I_l}{\mu _{l + 1}} \cdots {\mu _k}}}{}_{k - l}{C_0} 
  \frac{{\partial {x^{{J_1}}}}}{{\partial {{\bar x}^{{I_1}}}}} 
  \cdots \frac{{\partial {x^{{J_l}}}}}{{\partial {{\bar x}^{{I_l}}}}}
  \frac{{\partial {x^{{\nu _{l + 1}}}}}}{{\partial {{\bar x}^{{\mu _{l + 1}}}}}} 
  \cdots \frac{{\partial {x^{{\nu _k}}}}}{{\partial {{\bar x}^{{\mu _k}}}}}} \right) \nonumber \\
&\hspace{8cm}  \times \frac{\partial }{{\partial {y^{{J_1}}}}} \wedge  \cdots  \wedge 
  \frac{\partial }{{\partial {y^{{J_l}}}}} \wedge \frac{\partial }{{\partial {x^{{\nu _{l + 1}}}}}} \wedge  
  \cdots  \wedge \frac{\partial }{{\partial {x^{{\nu _k}}}}} \nonumber \\
&   +  \cdots \! +\! \left( {\frac{1}{{k!}}{{\bar y}^{{\mu _1} \cdots {\mu _k}}}
   \frac{{\partial {y^{{J_1}}}}}{{\partial {{\bar x}^{{\mu _1}}}}}
   \frac{{\partial {y^{{J_2}}}}}{{\partial {{\bar x}^{{\mu _2}}}}} 
   \cdots \frac{{\partial {y^{{J_k}}}}}{{\partial {{\bar x}^{{\mu _k}}}}} 
   + \frac{{k!}}{{(k - 1)!}}{{\bar w}^{{I_1}{\mu _2} \cdots {\mu _k}}}
   \frac{{\partial {x^{{J_1}}}}}{{\partial {{\bar x}^{{I_1}}}}}
   \frac{{\partial {y^{{J_2}}}}}{{\partial {{\bar x}^{{\mu _2}}}}} 
   \cdots \frac{{\partial {y^{{I_k}}}}}{{\partial {{\bar x}^{{\mu _k}}}}}} \right. \nonumber \\
&   +  \cdots \! + \frac{{{{(k!)}^l}}}{{(k - l)!l!}}{{\bar w}^{{I_1} 
  \cdots {I_l}{\mu _{l + 1}} \cdots {\mu _k}}}\frac{{\partial {x^{{J_1}}}}}{{\partial {{\bar x}^{{I_1}}}}} 
  \cdots \frac{{\partial {x^{{J_l}}}}}{{\partial {{\bar x}^{{I_l}}}}}
  \frac{{\partial {y^{{J_{l + 1}}}}}}{{\partial {{\bar x}^{{\mu _{l + 1}}}}}} 
  \cdots \frac{{\partial {y^{{I_k}}}}}{{\partial {{\bar x}^{{\mu _k}}}}} \nonumber \\
& \hspace{2cm} \left. +  \cdots \! + \! \frac{{{{(k!)}^k}}}{{k!}}{{\bar w}^{{I_1}{I_2} 
  \cdots {I_k}}}\frac{{\partial {x^{{J_1}}}}}{{\partial {{\bar x}^{{I_1}}}}} 
  \cdots \frac{{\partial {x^{{J_k}}}}}{{\partial {{\bar x}^{{I_k}}}}} \! \right)  
 \times \frac{\partial }{{\partial {y^{{J_1}}}}} \wedge  
\cdots  \wedge \frac{\partial }{{\partial {y^{{J_k}}}}}.   
\end{align}
To reduce the space, we made an abbreviation,  
\begin{align}
\frac{\partial {{x}^{{{J}_{1}}}}}{\partial {{{\bar{x}}}^{{{I}_{1}}}}}
:=\frac{\partial {{x}^{[{{j}_{1}}}}}{\partial {{{\bar{x}}}^{{{i}_{1}}}}}
\cdots \frac{\partial {{x}^{{{j}_{k}}]}}}{\partial {{{\bar{x}}}^{{{i}_{k}}}}}.
\end{align}
Now we will get the transformation rule for the coordinates as, 
\begin{align}
&{\bar y^{{\nu _1} \cdots {\nu _k}}} = \frac{{\partial {{\bar x}^{{\nu _1}}}}}{{\partial {x^{{\mu _1}}}}} 
\cdots \frac{{\partial {{\bar x}^{{\nu _k}}}}}{{\partial {x^{{\mu _k}}}}}{y^{{\mu _1} \cdots {\mu _k}}}, \nonumber \\
&{\bar z^{{J_1}{\nu _2} \cdots {\nu _k}}} = k!{z^{{I_1}{\mu _2} \cdots {\mu _k}}}
\frac{{\partial {{\bar x}^{{J_1}}}}}{{\partial {x^{{I_1}}}}}
\frac{{\partial {{\bar x}^{{\nu _2}}}}}{{\partial {x^{{\mu _2}}}}} 
\cdots \frac{{\partial {{\bar x}^{{\nu _k}}}}}{{\partial {x^{{\mu _k}}}}} 
+ {y^{{\mu _1} \cdots {\mu _k}}}\frac{{\partial {{\bar y}^{{J_1}}}}}{{\partial {x^{{\mu _1}}}}}
\frac{{\partial {{\bar x}^{{\nu _2}}}}}{{\partial {x^{{\mu _2}}}}} 
\cdots \frac{{\partial {{\bar x}^{{\nu _k}}}}}{{\partial {x^{{\mu _k}}}}}, \nonumber \\
&{{{\bar{z}}}^{{{J}_{1}}{{J}_{2}}{{\nu }_{3}}\cdots {{\nu }_{k}}}}
={{(k!)}^{2}}{{z}^{{{I}_{1}}{{I}_{2}}{{\mu }_{3}}\cdots {{\mu }_{k}}}}
\frac{\partial {{{\bar{x}}}^{{{J}_{1}}}}}{\partial {{x}^{{{I}_{1}}}}}
\frac{\partial {{{\bar{x}}}^{{{J}_{2}}}}}{\partial {{x}^{{{I}_{2}}}}}
\frac{\partial {{{\bar{x}}}^{{{\nu }_{3}}}}}{\partial {{x}^{{{\mu }_{3}}}}}
\cdots \frac{\partial {{{\bar{x}}}^{{{\nu }_{k}}}}}{\partial {{x}^{{{\mu }_{k}}}}} \nonumber \\ 
& \quad +2!2!k!{{z}^{{{I}_{1}}{{\mu }_{2}}\cdots {{\mu }_{k}}}}
\frac{\partial {{{\bar{x}}}^{[{{J}_{1}}}}}{\partial {{x}^{{{I}_{1}}}}}
\frac{\partial {{{\bar{y}}}^{{{J}_{2}}]}}}{\partial {{x}^{{{\mu }_{2}}}}}
\frac{\partial {{{\bar{x}}}^{{{\nu }_{3}}}}}{\partial {{x}^{{{\mu }_{3}}}}}
\cdots \frac{\partial {{{\bar{x}}}^{{{\nu }_{k}}}}}{\partial {{x}^{{{\mu }_{k}}}}}
+\,{{y}^{{{\mu }_{1}}\cdots {{\mu }_{k}}}}\frac{\partial {{{\bar{y}}}^{{{J}_{1}}}}}{\partial {{x}^{{{\mu }_{1}}}}}
\frac{\partial {{{\bar{y}}}^{{{J}_{2}}}}}{\partial {{x}^{{{\mu }_{2}}}}}
\frac{\partial {{{\bar{x}}}^{{{\nu }_{3}}}}}{\partial {{x}^{{{\mu }_{3}}}}}
\cdots \frac{\partial {{{\bar{x}}}^{{{\nu }_{k}}}}}{\partial {{x}^{{{\mu }_{k}}}}}, \nonumber  \\
& \vdots \nonumber  \\
& {{{\bar{z}}}^{{{J}_{1}}\cdots {{J}_{l}}{{\nu }_{l+1}}\cdots {{\nu }_{k}}}} 
 ={{(k!)}^{l}}{{z}^{{{I}_{1}}\cdots {{I}_{l}}{{\mu }_{l+1}}\cdots {{\mu }_{k}}}}
\frac{\partial {{{\bar{x}}}^{{{J}_{1}}}}}{\partial {{x}^{{{I}_{1}}}}}
\cdots \frac{\partial {{{\bar{x}}}^{{{J}_{l}}}}}{\partial {{x}^{{{I}_{l}}}}}
\frac{\partial {{{\bar{x}}}^{{{\nu }_{l+1}}}}}{\partial {{x}^{{{\mu }_{l+1}}}}}
\cdots \frac{\partial {{{\bar{x}}}^{{{\nu }_{k}}}}}{\partial {{x}^{{{\mu }_{k}}}}} \nonumber \\ 
& \quad +{}_{l}{{C}_{l-1}}l!{{(k!)}^{l-1}}{{z}^{{{I}_{1}}\cdots {{I}_{l-1}}{{\mu }_{l}}
\cdots {{\mu }_{k}}}}\frac{\partial {{{\bar{x}}}^{[{{J}_{1}}}}}{\partial {{x}^{{{I}_{1}}}}}
\cdots \frac{\partial {{{\bar{x}}}^{{{J}_{l-1}}}}}{\partial {{x}^{{{I}_{l-1}}}}}
\frac{\partial {{{\bar{y}}}^{{{J}_{l}}]}}}{\partial {{x}^{{{\mu }_{l}}}}}
\frac{\partial {{{\bar{x}}}^{{{\nu }_{l+1}}}}}{\partial {{x}^{{{\mu }_{l+1}}}}}
\cdots \frac{\partial {{{\bar{x}}}^{{{\nu }_{k}}}}}{\partial {{x}^{{{\mu }_{k}}}}} \nonumber \\ 
& \quad +{}_{l}{{C}_{l-2}}l!{{(k!)}^{l-2}}{{z}^{{{I}_{1}}\cdots {{I}_{l-2}}{{\mu }_{l-1}}
\cdots {{\mu }_{k}}}}\frac{\partial {{{\bar{x}}}^{[{{J}_{1}}}}}{\partial {{x}^{{{I}_{1}}}}}
\cdots \frac{\partial {{{\bar{x}}}^{{{J}_{l-2}}}}}{\partial {{x}^{{{I}_{l-2}}}}}
\frac{\partial {{{\bar{y}}}^{{{J}_{l-1}}}}}{\partial {{x}^{{{\mu }_{l-1}}}}}
\frac{\partial {{{\bar{y}}}^{{{J}_{l}}]}}}{\partial {{x}^{{{\mu }_{l}}}}}
\frac{\partial {{{\bar{x}}}^{{{\nu }_{l+1}}}}}{\partial {{x}^{{{\mu }_{l+1}}}}}
\cdots \frac{\partial {{{\bar{x}}}^{{{\nu }_{k}}}}}{\partial {{x}^{{{\mu }_{k}}}}} \nonumber \\ 
& \quad +\cdots +{}_{l}{{C}_{1}}l!(k!)\,{{z}^{{{I}_{1}}{{\mu }_{2}}\cdots {{\mu }_{k}}}}
\frac{\partial {{{\bar{x}}}^{[{{J}_{1}}}}}{\partial {{x}^{{{I}_{1}}}}}
\frac{\partial {{{\bar{y}}}^{{{J}_{2}}}}}{\partial {{x}^{{{\mu }_{2}}}}}
\cdots \frac{\partial {{{\bar{y}}}^{{{J}_{l}}]}}}{\partial {{x}^{{{\mu }_{l}}}}}
\frac{\partial {{{\bar{x}}}^{{{\nu }_{l+1}}}}}{\partial {{x}^{{{\mu }_{l+1}}}}}
\cdots \frac{\partial {{{\bar{x}}}^{{{\nu }_{k}}}}}{\partial {{x}^{{{\mu }_{k}}}}} \nonumber \\ 
& \quad +\,{{y}^{{{\mu }_{1}}\cdots {{\mu }_{k}}}}
\frac{\partial {{{\bar{y}}}^{{{J}_{1}}}}}{\partial {{x}^{{{\mu }_{1}}}}}
\cdots \frac{\partial {{{\bar{y}}}^{{{J}_{l}}}}}{\partial {{x}^{{{\mu }_{l}}}}}
\frac{\partial {{{\bar{x}}}^{{{\nu }_{l+1}}}}}{\partial {{x}^{{{\mu }_{l+1}}}}}
\cdots \frac{\partial {{{\bar{x}}}^{{{\nu }_{k}}}}}{\partial {{x}^{{{\mu }_{k}}}}}, \nonumber  \\
& \vdots \nonumber \\
& {{{\bar{z}}}^{{{J}_{1}}\cdots {{J}_{k}}}}  
={{(k!)}^{k}}{{z}^{{{I}_{1}}\cdots {{I}_{k}}}}
\frac{\partial {{{\bar{x}}}^{{{J}_{1}}}}}{\partial {{x}^{{{I}_{1}}}}}
\cdots \frac{\partial {{{\bar{x}}}^{{{J}_{k}}}}}{\partial {{x}^{{{I}_{k}}}}}
+{{(k!)}^{k}}{}_{k}{{C}_{k-1}}{{z}^{{{I}_{1}}\cdots {{I}_{k-1}}{{\mu }_{k}}}}
\frac{\partial {{{\bar{x}}}^{[{{J}_{1}}}}}{\partial {{x}^{{{I}_{1}}}}}
\cdots \frac{\partial {{{\bar{x}}}^{{{J}_{k-1}}}}}{\partial {{x}^{{{I}_{k-1}}}}}
\frac{\partial {{{\bar{y}}}^{{{J}_{k}}]}}}{\partial {{x}^{{{\mu }_{k}}}}} \nonumber \\ 
& \quad +\cdots +{{(k!)}^{l+1}}{}_{k}{{C}_{l}}{{z}^{{{I}_{1}}\cdots {{I}_{l}}{{\mu }_{l+1}}
\cdots {{\mu }_{k}}}}\frac{\partial {{{\bar{x}}}^{[{{J}_{1}}}}}{\partial {{x}^{{{I}_{1}}}}}
\cdots \frac{\partial {{{\bar{x}}}^{{{J}_{l}}}}}{\partial {{x}^{{{I}_{l}}}}}
\frac{\partial {{{\bar{y}}}^{{{J}_{l+1}}}}}{\partial {{x}^{{{\mu }_{l+1}}}}}
\cdots \frac{\partial {{{\bar{y}}}^{{{J}_{k}}]}}}{\partial {{x}^{{{\mu }_{k}}}}} \nonumber \\ 
& \quad +\cdots +{{(k!)}^{2}}{}_{k}{{C}_{1}}{{z}^{{{I}_{1}}{{\mu }_{2}}\cdots {{\mu }_{k}}}}
\frac{\partial {{{\bar{x}}}^{[{{J}_{1}}}}}{\partial {{x}^{{{I}_{1}}}}}
\frac{\partial {{{\bar{y}}}^{{{J}_{2}}}}}{\partial {{x}^{{{\mu }_{2}}}}
}\cdots \frac{\partial {{{\bar{y}}}^{{{J}_{k}}]}}}{\partial {{x}^{{{\mu }_{k}}}}}+{{y}^{{{\mu }_{1}}
\cdots {{\mu }_{k}}}}\frac{\partial {{{\bar{y}}}^{{{J}_{1}}}}}{\partial {{x}^{{{\mu }_{1}}}}}
\cdots \frac{\partial {{{\bar{y}}}^{{{J}_{k}}}}}{\partial {{x}^{{{\mu }_{k}}}}}.  \label{2nd_k_coordtrans}
\end{align}
These transformations are smooth, and the charts form a smooth atlas on ${({\Lambda ^k}T)^2}M$. 
The natural embedding 
\begin{align}
&({x^\mu },{y^{{\mu _1} \cdots {\mu _k}}},{z^{{I_1}{\mu _2} 
\cdots {\mu _k}}},{z^{{I_1}{I_2}{\mu _3} \cdots {\mu _k}}}, 
\cdots ,{z^{{I_1} \cdots {I_k}}}) \nonumber \\
&\hspace{1cm} \to ({x^\mu },{y^{{\mu _1} 
\cdots {\mu _k}}},{y^{{\mu _1} \cdots {\mu _k}}},{z^{{I_1}{\mu _2} 
\cdots {\mu _k}}},{z^{{I_1}{I_2}{\mu _3} \cdots {\mu _k}}}, \cdots ,{z^{{I_1} \cdots {I_k}}})
\end{align}
shows 
that ${({\Lambda ^k}T)^2}M$ is a submanifold of ${\Lambda ^k}T{\Lambda ^k}TM$. 
To save the space, we will frequently use the abbreviation such as 
$\displaystyle{{\bar y^I} = k!\frac{{\partial {{\bar x}^I}}}{{\partial {x^J}}}{y^J}}$, 
instead of $\displaystyle{{\bar y^{{\nu _1} \cdots {\nu _k}}} 
= \frac{{\partial {{\bar x}^{{\nu _1}}}}}{{\partial {x^{{\mu _1}}}}} 
\cdots \frac{{\partial {{\bar x}^{{\nu _k}}}}}{{\partial {x^{{\mu _k}}}}}{y^{{\mu _1} \cdots {\mu _k}}}}$. 
\\
In case of $k = 2$,
\begin{flalign}
&w = \frac{1}{2}{y^{{\mu _1}{\mu _2}}}\frac{\partial }{{\partial {x^{{\mu _1}}}}} \wedge 
\frac{\partial }{{\partial {x^{{\mu _2}}}}} + {w^{{I_1}{\mu _2}}}
\frac{\partial }{{\partial {y^{{I_1}}}}} \wedge \frac{\partial }{{\partial {x^{{\mu _2}}}}} 
+ \frac{1}{2}{w^{{I_1}{I_2}}}\frac{\partial }{{\partial {y^{{I_1}}}}} \wedge 
\frac{\partial }{{\partial {y^{{I_2}}}}} \nonumber \\
&   = \frac{1}{2}{{\bar y}^{{\mu _1}{\mu _2}}}\frac{{\partial {x^{{\nu _1}}}}}{{\partial {{\bar x}^{{\mu _1}}}}}
   \frac{{\partial {x^{{\nu _2}}}}}{{\partial {{\bar x}^{{\mu _2}}}}}
   \frac{\partial }{{\partial {x^{{\nu _1}}}}} \wedge \frac{\partial }{{\partial {x^{{\nu _2}}}}} 
   + \left( {{{\bar y}^{{\mu _1}{\mu _2}}}\frac{{\partial {y^{{J_1}}}}}{{\partial {{\bar x}^{{\mu _1}}}}}
   \frac{{\partial {x^{{\nu _2}}}}}{{\partial {{\bar x}^{{\mu _2}}}}} 
   + 2{{\bar w}^{{I_1}{\mu _2}}}\frac{{\partial {x^{{J_1}}}}}{{\partial {{\bar x}^{{I_1}}}}}
   \frac{{\partial {x^{{\nu _2}}}}}{{\partial {{\bar x}^{{\mu _2}}}}}} \right)
   \frac{\partial }{{\partial {y^{{J_1}}}}} \wedge \frac{\partial }{{\partial {x^{{\nu _2}}}}} \nonumber \\
&  \quad  + \frac{1}{2}\left( {{{\bar y}^{{\mu _1}{\mu _2}}}
  \frac{{\partial {y^{{J_1}}}}}{{\partial {{\bar x}^{{\mu _1}}}}}
  \frac{{\partial {y^{{J_2}}}}}{{\partial {{\bar x}^{{\mu _2}}}}} 
  + 4{{\bar w}^{{I_1}{\mu _2}}}\frac{{\partial {x^{{J_1}}}}}{{\partial {{\bar x}^{{I_1}}}}}
  \frac{{\partial {y^{{J_2}}}}}{{\partial {{\bar x}^{{\mu _2}}}}} 
  + 4{{\bar w}^{{I_1}{I_2}}}\frac{{\partial {x^{{J_1}}}}}{{\partial {{\bar x}^{{I_1}}}}}
  \frac{{\partial {x^{{J_2}}}}}{{\partial {{\bar x}^{{I_2}}}}}} \right)
  \frac{\partial }{{\partial {y^{{J_1}}}}} \wedge \frac{\partial }{{\partial {y^{{J_2}}}}}  \nonumber  \\
&{\bar y^{{\nu _1}{\nu _2}}} = \frac{{\partial {{\bar x}^{{\nu _1}}}}}{{\partial {x^{{\mu _1}}}}}
\frac{{\partial {{\bar x}^{{\nu _2}}}}}{{\partial {x^{{\mu _2}}}}}{y^{{\mu _1}{\mu _2}}}, \nonumber \\ 
&{\bar z^{{J_1}{\nu _2}}} = \,2{z^{{I_1}{\mu _2}}}\frac{{\partial {{\bar x}^{{J_1}}}}}{{\partial {x^{{I_1}}}}}
\frac{{\partial {{\bar x}^{{\nu _2}}}}}{{\partial {x^{{\mu _2}}}}} 
+ {y^{{\mu _1}{\mu _2}}}\frac{{\partial {{\bar y}^{{J_1}}}}}{{\partial {x^{{\mu _1}}}}}
\frac{{\partial {{\bar x}^{{\nu _2}}}}}{{\partial {x^{{\mu _2}}}}}, \nonumber \\
&{{\bar{z}}^{{{J}_{1}}{{J}_{2}}}}=4{{z}^{{{I}_{1}}{{I}_{2}}}}
\frac{\partial {{{\bar{x}}}^{{{J}_{1}}}}}{\partial {{x}^{{{I}_{1}}}}}
\frac{\partial {{{\bar{x}}}^{{{J}_{2}}}}}{\partial {{x}^{{{I}_{2}}}}}
+8{{z}^{{{I}_{1}}{{\mu }_{2}}}}\frac{\partial {{{\bar{x}}}^{[{{J}_{1}}}}}{\partial {{x}^{{{I}_{1}}}}}
\frac{\partial {{{\bar{y}}}^{{{J}_{2}}]}}}{\partial {{x}^{{{\mu }_{2}}}}}+{{y}^{{{\mu }_{1}}{{\mu }_{2}}}}
\frac{\partial {{{\bar{y}}}^{{{J}_{1}}}}}{\partial {{x}^{{{\mu }_{1}}}}}
\frac{\partial {{{\bar{y}}}^{{{J}_{2}}}}}{\partial {{x}^{{{\mu }_{2}}}}}.
\end{flalign}

Now ${\Lambda ^k}\tau _M^{2,1}$ is a surjective submersion by definition, 
so it remains to check the local trivialisation. 
The local trivialisation of 
$({({\Lambda ^k}T)^2}M,{\Lambda ^k}\tau _M^{2,1},{\Lambda ^k}TM)$ around any point 
$p \in {\Lambda ^k}TM$ is given by 
$({V_p},{\mathbb{R}^n},{t_p})$, 
${t_p}:{({\Lambda ^k}\tau _M^{2,1})^{ - 1}}({V_p}) \to {V_p} \times {\mathbb{R}^l}$, 
$l = {}_n{C_k}$, $p \in {V_p}$, where ${V_p}$ is an open set of 
${\Lambda ^k}TM$, which in chart expression for any 
$\xi  \in {({\Lambda ^k}\tau _M^{2,1})^{ - 1}}({V_p}) \subset {({\Lambda ^k}T)^2}M$ is 
\begin{align}
{t_p}(\xi ) 
= ({\Lambda ^k}\tau _M^{2,1}(\xi ),{z^{{I_1}{\mu _2} \cdots {\mu _k}}}(\xi ),{z^{{I_1}{I_2}{\mu _3} 
\cdots {\mu _k}}}(\xi ), \cdots ,{z^{{I_1} \cdots {I_k}}}(\xi )).
\end{align} 
Therefore, $({({\Lambda ^k}T)^2}M,{\Lambda ^k}\tau _M^{2,1},{\Lambda ^k}TM)$ is indeed a bundle.

\begin{defn} Second order $k$-multivector bundle \\
Similarly as in the case of mechanics $(k = 1)$, the triple 
$({({\Lambda ^k}T)^2}M,{\Lambda ^k}\tau _M^{2,0},M)$ 
with ${\Lambda ^k}\tau _M^{2.0} = {\Lambda ^k}{\tau _M} \circ {\Lambda ^k}\tau _M^{2,1}$, 
${\Lambda ^k}\tau _M^{2,1}: = {\left. {{\Lambda ^k}{\tau _{{\Lambda ^k}TM}}} 
\right|_{{{({\Lambda ^k}T)}^2}M}}$ 
is also a bundle with the trivialisation 
$({U_p},{\mathbb{R}^{N - n}},{t_p})$, 
${t_p}:{({\Lambda ^k}\tau _M^{2,0})^{ - 1}}({U_p}) \to {U_p} \times {\mathbb{R}^N}$, 
where ${U_p}$ is an open set of $M$, and 
\begin{align}
N = {}_{(n + l)}{C_k} = l + {}_n{C_{k - 1}} \times {}_l{C_1} + {}_n{C_{k - 2}} \times {}_l{C_2} 
+  \cdots  + {}_l{C_k},
\end{align} 
where $l = {}_n{C_k}$, around any $p \in {U_p} \subset M$. 
We call this a {\it Second order $k$-multivector bundle over $M$} or simply, {\it Second order $k$-multivector bundle}. 
\end{defn}

In the second order field theory, the dynamical variables are the section of the second order 
$k$-multivector bundle, and 
${({\Lambda ^k}T)^2}M$ will be the space where the Lagrangian should be defined. 

\subsection{Higher order $k$-multivector bundle} \label{higher_kTM}
Here we will briefly introduce the higher order $k$-multivector bundles.
The construction is the same as the Section \ref{higher_TM}, 
and the $r$-th order $k$-multivector bundle over ${({\Lambda ^k}T)^{r-1}}M$;
namely, $({({\Lambda ^k}T)^r}M, {\Lambda ^k}\tau _M^{r,r-1},{({\Lambda ^k}T)^{r-1}}M)$ 
can be constructed by induction. 
The projection map is defined by 
\begin{align}
{\Lambda ^k} \tau _M^{r,r - 1}: = {\left. {{\Lambda ^k}{\tau _{{{({\Lambda ^k}T)}^{r - 1}}M}}} \right|_{{{({\Lambda ^k}T)}^r}M}}.
\end{align} 
Consider the bundle morphism 
$(({\Lambda ^k}T){\Lambda ^k}\tau _M^{r - 1,r - 2}, \lb[3] {\Lambda ^k}\tau _M^{r - 1,r - 2})$ 
from 
$({\Lambda ^k}T({({\Lambda ^k}T)^{r - 1}}M), \lb[3] {\Lambda ^k}{\tau _{{{({\Lambda ^k}T)}^{r - 1}}M}},  \lb[3] 
{({\Lambda ^k}T)^{r - 1}}M)$ to $({\Lambda ^k}T({({\Lambda ^k}T)^{r - 2}}M),  \lb[3] 
{\Lambda ^k}{\tau _{{{({\Lambda ^k}T)}^{r - 2}}M}},  \lb[3] {({\Lambda ^k}T)^{r - 2}}M)$. 
Then we will define the total space ${({\Lambda ^k}T)^r}M$ by 
\begin{align}
{({\Lambda ^k}T)^r}M: = \{ u \in {\Lambda ^k}T({({\Lambda ^k}T)^{r - 1}}M){\mkern 1mu} |{\mkern 1mu} 
({\Lambda ^k}T){\Lambda ^k}\tau _M^{r - 1,r - 2}(u) = {\iota _{r - 1}} 
\circ {\Lambda ^k}{\tau _{{{({\Lambda ^k}T)}^{r - 1}}M}}(u)\} ,  \label{higherbundle_k_r}
\end{align}
where ${\iota _{r - 1}}:{({\Lambda ^k}T)^{r - 1}}M \to {\Lambda ^k}T({({\Lambda ^k}T)^{r - 2}}M)$  
is the inclusion map.
$({({\Lambda ^k}T)^r}M, \lb[3] {\Lambda ^k}\tau _M^{r,r - 1}, \lb[3] {({\Lambda ^k}T)^{r - 1}}M)$ 
is a sub-bundle of $({\Lambda ^k}T({({\Lambda ^k}T)^{r - 1}}M), \lb[3] 
{\Lambda ^k}{\tau _{{{({\Lambda ^k}T)}^{r - 1}}M}}, \lb[3] {({\Lambda ^k}T)^{r - 1}}M)$.

\section{Integration of differential forms} \label{sec_integrationofd-forms}

For the calculus of variations on the Kawaguchi manifold, 
we need to integrate a Lagrangian which is a $k$-form on a $k$-dimensional submanifold. 
In this section, we will describe how to implement such integrals, 
and begin by introducing the method to define the integration of a 
$n$-form on a $n$-dimensional compact oriented manifold $M$. The integration of 
$k$-form $(k < n)$ on a $k$-dimensional compact oriented submanifold would then be given, 
first in the case where there is an immersion map (which is called parameterisation) from the parameter space 
to the total space $M$, and second for the case where there is no such map. 
We will begin with basic definitions and theorems. 

\begin{defn}   Orientation of two charts \\
Let $(U,\varphi ),\varphi  = ({x^i})$ and $(\bar U,\bar \varphi ),\bar \varphi  = ({y^j})$ be two charts on $M$ 
such that $U \cap \bar U \ne \emptyset $. 
We say that {\it $(U,\varphi )$ and $(\bar U,\bar \varphi )$ has the same orientation}, 
when 
\begin{align}
\det \left( {\frac{{\partial {y^i}}}{{\partial {x^j}}}} \right) > 0. 
\end{align}
\end{defn}
\pagebreak[3]
\begin{defn}   Orientable manifold \\ 
We say that a manifold {\it $M$ is orientable}, 
if there exist an atlas ${\mathcal{A}} = {\{ ({U_\iota },{\varphi _\iota })\} _{\iota  \in I}}$ 
such that for any pair of intersecting charts from ${\mathcal{A}}$, 
has the same orientation. The manifold which is given such an atlas is called an 
{\it oriented manifold}, and it has the {\it orientation} associated to this atlas. 
\end{defn}
Suppose we have an oriented manifold with an atlas ${\mathcal{A}}$, 
then by definition, this manifold has a specific orientation associated to 
${\mathcal{A}}$. Transfer one axis in opposite direction for every chart in 
${\mathcal{A}}$, and denote this new atlas by 
$\tilde{\mathcal{A}}$. Every chart in $\tilde{\mathcal{A}}$ still has the same orientation, but it is opposite to the orientation of ${\mathcal{A}}$. In this way, orientable manifold can always have two orientations. 
\begin{defn}Orientation preserving (reversing) map \\
Let $M$and $N$ be two smooth $n$-dimensional oriented manifolds, 
$\alpha :M \to N$ a diffeomorphism. In particular, the tangent mapping 
${T_p}\alpha :{T_p}M \to {T_{\alpha (p)}}N$ has constant rank $n$ for every $p \in M$. 
Choose a chart $(U,\varphi )$ on $M$ and a chart $(V,\psi )$ on $N$ such that $\alpha (U) \subset V$. 
We define a number ${\varepsilon _\alpha }$, equal to 1 or $ - 1$, by the chart expressions 
\begin{align}
|\det D\psi \alpha {\varphi ^{ - 1}}|\; = {\varepsilon _\alpha } \cdot \det D\psi \alpha {\varphi ^{ - 1}}.
\end{align}
This number is independent of the choice of charts that has the same orientation. 
To see this, let $(\bar U,\bar \varphi )$ be a chart on $M$ which has the same orientation as 
$(U,\varphi )$, and $(\bar V,\bar \psi )$ a chart on $N$ 
which has the same orientation as $(V,\psi )$. Then,  
\begin{align}
&  {\varepsilon _\alpha } = \operatorname{sgn} \det D\psi \alpha {\varphi ^{ - 1}} \nonumber \\
&  \quad  = \operatorname{sgn} (\det D\psi {{\bar \psi }^{ - 1}} 
\cdot \det D\bar \psi \alpha {{\bar \varphi }^{ - 1}} 
\cdot \det D\bar \varphi {\varphi ^{ - 1}}) \nonumber \\
&  \quad  = \operatorname{sgn} (\det D\psi {{\bar \psi }^{ - 1}} 
\cdot \det D\bar \varphi {\varphi ^{ - 1}}) \cdot \operatorname{sgn} 
\det D\bar \psi \alpha {{\bar \varphi }^{ - 1}} \nonumber \\
&  \quad  = \operatorname{sgn} \det D\bar \psi \alpha {{\bar \varphi }^{ - 1}}.  
\end{align}
Thus, the number ${\varepsilon _\alpha }$ is independent of the charts with the same orientation. 
We say that $\alpha$ is {\it orientation preserving map} (resp. {\it orientation-reversing map}), 
if ${\varepsilon _\alpha } = 1$ (${\varepsilon _\alpha } =  - 1$). 
\end{defn}
\begin{ex} 
The special case is when $M = U$, $N = {\mathbb{R}^n}$ and $\alpha  = \varphi $, 
where $\varphi  = ({x^\mu })$ are the coordinate functions of the domain $U$. 
We assume that we always choose the map $\varphi $ 
as to be orientation preserving with respect to the canonical coordinates on ${\mathbb{R}^n}$.  
\end{ex}
\begin{defn}  support of $k$-form$ (k \leqslant n)$ $\alpha $ \\
The {\it support} of $\alpha $ is the closure of $\{ p \in M|{\alpha _p} \ne 0\} $. 
We denote it by $\operatorname{supp} (\alpha )$
\end{defn}
\begin{defn}  $\sigma $-compact \\
We say that the topological space $X$ is $\sigma $-compact when $X$ 
is a union of countable number of compact subsets, ${X_i},{\kern 1pt} (i = 1,2,3, \cdots )$. 
Namely, $X = \bigcup\limits_{i = 1}^\infty  {{X_i}} $. 
\end{defn}
\begin{theorem}  Partition of Unity  \label{thm_partitionofunity}\\ 
Let $M$ be a $\sigma $-compact ${C^\infty }$-manifold, and 
${\{ {U^\iota }\} _{\iota  \in I}}$ an arbitrary open covering of $M$. 
$\iota  \in I$ is an index taken from countable index set $I$. 
Then, there exists countable number of ${C^\infty }$-functions ${h_j}:M \to \mathbb{R},\;j = 1,2,3, \cdots $, 
such that satisfies the following conditions: \\
1. $0 \leqslant {h_j} \leqslant 1$, for $\forall j = 1,2,3, \cdots $ \\
2. The family ${\{ \operatorname{supp} ({h_j})\} _{j \in I}}$ is a locally finite open covering of $M$ and also a refinement of ${\{ {U^\iota }\} _{\iota  \in I}}$  \\
3. $\displaystyle{\sum\limits_{j \in I} {{h_j}}  = 1.}$ \\
\end{theorem}
The sum in 3. is over finite functions, since ${\{ \operatorname{supp} ({h_j})\} _{j \in I}}$ is a locally finite set. 
The proof could be found in the standard textbooks, e.g., Spivak~\cite{Spivak1}. 

\begin{defn}
The family ${\{ {h_i}\} _{i \in I}}$ satisfying the conditions 1., 2., 3. 
in Theorem \ref{thm_partitionofunity} is called a {\it partition of unity subordinate to} ${\{ {U^\iota }\} _{\iota  \in I}}$.
\end{defn}

\begin{defn} Rectangular region \\
Let $V$ be an open subset of $U$. We call {\it $V$ a rectangular region}, 
when there exists a chart $(U,\varphi )$, such that by appropriate shrinking of the domain, 
induces a chart $(V,\tilde \varphi )$, 
$\tilde \varphi (V) = \{ p \in {\mathbb{R}^n}|{\tilde \varphi ^i}(p) 
\in ( - {a^i},{a^i}),\;{a^i} \in \mathbb{R},\;i = 1, \cdots ,n\} $. 
We call $(V,\tilde \varphi )$ a {\it rectangle chart}. 
\end{defn}

Now we will introduce the integral of an $n$-form on a $n$ dimensional manifold $M$. 
\begin{defn} Integration of a top form \\
Let $\omega $ be a $n$-form on $M$, and $(U,\varphi ),\varphi  = ({x^1}, \cdots ,{x^n})$ a chart on $M$. 
First, suppose that $\operatorname{supp} (\omega ) \subset V$, where $V$ is a rectangular region of $U$. 
Let the local coordinate expression of $\omega $ be $\omega  = fd{x^1} \wedge  \cdots  \wedge d{x^n}$, with $f \in {C^\infty }(M)$. 
The integration of $\omega $ is defined by, 
\begin{align}
\int_M \omega   = \int_V \omega   = \int_{\tilde \varphi (V)} {{{\tilde \varphi }^*}(fd{x^1} \wedge  
\cdots  \wedge d{x^n}):}  = \int_{ - {a_1}}^{{a_1}} { \cdots \int_{ - {a_n}}^{{a_n}} {f({x^1}, 
\cdots ,{x^n}){\kern 1pt} d{x^1} \cdots d{x^n}} } , \label{integral1} 
\end{align}
where for simplicity, we denoted the pull back of the coordinates on $U$ to 
${\mathbb{R}^n}$ also as $({x^\mu })$. The right hand side is the standard multiple integral. 
\end{defn}

\begin{lemma}  
The Right hand side of (\ref{integral1}) does not depend on the choice of the rectangular region which 
contains the support of $\omega $, provided that it has the same orientation. 
\end{lemma}
\begin{proof}  
Suppose there is another chart $(\bar U,\bar \varphi ),\bar \varphi  = ({y^1}, \cdots ,{y^n})$, 
such that $U \cap \bar U \ne \emptyset $, and $V \subset U \cap \bar U$. 
By definition, the rectangle chart of $V$ on the first chart is 
\[\tilde \varphi (V) = \{ p \in {\mathbb{R}^n}|{\tilde \varphi ^i}(p) 
\in ( - {a^i},{a^i}),\;{a^i} \in \mathbb{R},\;i = 1, \cdots ,n\} .\] 
The second chart, $\bar \varphi (V)$ is not a rectangle chart in general, 
however in ${\mathbb{R}^n}$ we can always choose a open rectangle such that contains 
$\bar \varphi (V)$. For instance, choose 
\[L: = \{ ( - {b^1},{b^1}) \times  \cdots  \times ( - {b^n},{b^n})|{b^i} 
= \sup \{ \left| {{{\bar \varphi }^i}(p)} \right|,p \in V\}  \in \mathbb{R}\}, \] 
then since $\bar \varphi $ is a diffeomorphism, we can always choose a 
$\phi  \in \mbox{\it Diff}({\mathbb{R}^n})$ such that 
$\phi  \circ \bar \varphi (V) = L$. Let $\varphi ' = \phi  \circ \bar \varphi $, 
and the local coordinate expression of $\omega $ 
in this chart be $\omega  = gd{y^1} \wedge  \cdots  \wedge d{y^n}$, with $g \in {C^\infty }(M)$. 
Since 
\begin{align}
&  d{y^1} \wedge  \cdots  \wedge d{y^n} = \frac{1}{{k!}}{\varepsilon _{{i_1} \cdots {i_n}}}d{y^{{i_1}}} \wedge  
\cdots  \wedge d{y^{{i_n}}} = \frac{1}{{k!}}{\varepsilon _{{i_1} 
\cdots {i_n}}}\frac{{\partial {y^{{i_1}}}}}{{\partial {x^{{j_1}}}}} 
\cdots \frac{{\partial {y^{{i_n}}}}}{{\partial {x^{{j_n}}}}}d{x^{{j_1}}} \wedge  
\cdots  \wedge d{x^{{j_n}}} \nonumber \\
&   = \frac{1}{{k!}}{\varepsilon _{{i_1} \cdots {i_n}}}{\varepsilon ^{{j_1} 
\cdots {j_n}}}\frac{{\partial {y^{{i_1}}}}}{{\partial {x^{{j_1}}}}} 
\cdots \frac{{\partial {y^{{i_n}}}}}{{\partial {x^{{j_n}}}}}d{x^1} \wedge  
\cdots  \wedge d{x^n} = {\varepsilon ^{{j_1} \cdots {j_n}}}\frac{{\partial {y^1}}}{{\partial {x^{{j_1}}}}} 
\cdots \frac{{\partial {y^n}}}{{\partial {x^{{j_n}}}}}d{x^1} \wedge  \cdots  \wedge d{x^n} \nonumber \\
&   = \det \left( {\frac{{\partial {y^i}}}{{\partial {x^j}}}} \right)d{x^1} \wedge  
\cdots  \wedge d{x^n}, 
\end{align}
we will have 
\begin{align}
f = g\det \left( {\frac{{\partial {y^i}}}{{\partial {x^j}}}} \right).
\end{align}
We suppose our manifold $M$ is orientable, 
then it is always possible to choose two charts such that 
$\displaystyle{\det \left( {\frac{{\partial {y^i}}}{{\partial {x^j}}}} \right) > 0}$ on 
$U \cap \bar U$. Then, 
\begin{align}
\int_V \omega  &= \int_{\tilde \varphi (V)} {{{\tilde \varphi }^*}(fd{x^1} \wedge  
\cdots  \wedge d{x^n})}  = \int_{ - {a_1}}^{{a_1}} { \cdots \int_{ - {a_n}}^{{a_n}} {f({x^1}, 
\cdots ,{x^n}){\kern 1pt} d{x^1} \cdots d{x^n}} }  \nonumber \\
   &= \int_{ - {a_1}}^{{a_1}} { \cdots \int_{ - {a_n}}^{{a_n}} {g({y^1}, 
   \cdots ,{y^n}){\kern 1pt} \det \left( {\frac{{\partial {y^i}}}{{\partial {x^j}}}} \right)d{x^1} \cdots d{x^n}} }  \nonumber \\
   &= \int_{ - {a_1}}^{{a_1}} { \cdots \int_{ - {a_n}}^{{a_n}} {g({y^1}, \cdots ,{y^n}){\kern 1pt} {\kern 1pt} 
   \left| {\det \left( {\frac{{\partial {y^i}}}{{\partial {x^j}}}} \right)} \right|d{x^1} \cdots d{x^n}} }  \nonumber \\
   &= \int_{ - {b_1}}^{{b_1}} { \cdots \int_{ - {b_n}}^{{b_n}} {g({y^1}, \cdots ,{y^n}){\kern 1pt} d{y^1} 
   \cdots d{y^n}} } .  \label{integral2}
\end{align}
Accordingly, the value does not depend on the choice of charts. 
\end{proof}

Now, suppose $M$ is orientable and compact. Since every point $p \in M$ is in some rectangle region, 
$M$ could be covered by finite number of rectangle regions 
$\{ {V^1}, \cdots ,{V^s}\} $. Then, choose the atlas of $M$ as 
${\mathcal{A}} = \{ ({V^1},{\varphi ^1}), \cdots ,({V^s},{\varphi ^s})\} $, 
where all charts in ${\mathcal{A}}$ have the same orientation. 
Take a partition of unity ${\{ {h_i}\} _{i \in I}}$, subordinate to $\{ {V^1}, \cdots ,{V^s}\} $. 
Then the integral of $n$-form $\omega $ on $M$ can be calculated by 
\begin{align}
\int_M \omega  : = \int_M {\sum\limits_{i = 1}^\infty  {{h_i}\omega } }  
= \sum\limits_{i = 1}^\infty  {\int_M {{h_i}\omega } }  
= \sum\limits_{i = 1}^\infty  {\int_{\operatorname{supp} ({h_i})} {{h_i}\omega } }  
= \sum\limits_{j = 1}^s {\sum\limits_{i:\operatorname{supp} ({h_i}) \subset {V^j}} 
{\int_{{V^j}} {{h_i}\omega } } },   \label{integral3}
\end{align}
since $\operatorname{supp} ({h_i}\omega ) \subseteq \operatorname{supp} ({h_i})$, and 
${\{ \operatorname{supp} ({h_i})\} _{i \in N}}$ is a refinement of 
$\{ {V^1}, \cdots ,{V^s}\} $, for any $i \in \mathbb{N}$, 
$\operatorname{supp} ({h_i}\omega ) \subset {V^j}$ for some $j$, 
$1 \leqslant j \leqslant s$. 
We can calculate the {\it r.h.s.} of (\ref{integral3}) by (\ref{integral2}), 
\begin{align}
&\sum\limits_{j = 1}^s {\sum\limits_{i:\operatorname{supp} ({h_i}) \subset {V^j}} {\int_{{V^j}} {{h_i}\omega } } }  
= \sum\limits_{j = 1}^s {\sum\limits_{i:\operatorname{supp} ({h_i}) \subset {V^j}} 
{\int_{{{\tilde \varphi }_j}({V^j})} {{{\tilde \varphi }_j}^*({h_i}{f_j}dx_j^1 \wedge  
\cdots  \wedge dx_j^n)} } }  \nonumber \\
&= \sum \limits_{j = 1}^s \sum\limits_{i:\operatorname{supp} ({h_i}) \subset {V^j}} 
\int_{ - {a^j}_1}^{{a^j}_1}  \cdots \int_{ - {a^j}_n}^{{a^j}_n} 
({h_i} \circ {{\tilde \varphi }_j}) 
({f_j} \circ {{\tilde \varphi }_j}) ({x^1}, \cdots ,{x^n})
\det \left( {\frac{{\partial x_j^b}}{{\partial {x^a}}}} \right)d{x^1} \cdots d{x^n} , 
\end{align}
where 
\begin{align}
{f_j}dx_j^1 \wedge  \cdots  \wedge dx_j^n 
\end{align} 
is a local expression of $\omega $ on the rectangle chart $({V^j},{\tilde \varphi ^j}),{\tilde \varphi ^j} = (x_j^1, \cdots ,x_j^n)$, induced by the chart $({V^j},{\varphi ^j})$. 
This expression does not depend on the choice of partition of unity. 
To see this, consider just two coordinate patches ${V_1},{V_2}$. 
\begin{align}
\int_{{V^1} \cup {V^2}} \omega  : = \int_{{V^1} \cup {V^2}} {\sum\limits_{i \in I} {{h_i}\omega } } 
\end{align}
For ${V^1} \cap {V^2} = \emptyset $, 
\begin{align}
&\int_{{V^1} \cup {V^2}} {\sum\limits_{i = 1}^\infty  {{h_i}\omega } } 
 = \sum\limits_{i:\operatorname{supp} ({h_i}) \subset {V^1}} {\int_{{V^1}} {{h_i}\omega } }  
 + \sum\limits_{i:\operatorname{supp} ({h_i}) \subset {V^2}} {\int_{{V^2}} {{h_i}\omega } }  \nonumber \\
&= \sum\limits_{i \in I} {\int_{{V^1}} {{h_i}\omega } }  + \sum\limits_{i \in I} {\int_{{V^2}} {{h_i}\omega } }  
= \int_{{V^1}} {\sum\limits_{i \in I} {{h_i}\omega } }  + \int_{{V^2}} {\sum\limits_{i \in I} {{h_i}\omega } }  
= \int_{{V^1}} \omega   + \int_{{V^2}} \omega  . 
\end{align}
For ${V^1} \cap {V^2} \ne \emptyset $, 
\begin{align}
  \int_{{V^1} \cup {V^2}} \omega  &:= \int_{{V^1} \cup {V^2}} {\sum\limits_{i = 1}^\infty  {{h_i}\omega } }  \nonumber \\
   &= \sum\limits_{i:\operatorname{supp} ({h_i}) \subset {V^1}} {\int_{{V^1}{ \setminus }({V^1} \cap {V^2})} 
{{h_i}\omega } }  + \sum\limits_{i:\operatorname{supp} ({h_i}) \subset {V^2}}
 {\int_{{V^2}{ \setminus }({V^1} \cap {V^2})} {{h_i}\omega } }  
 + \sum\limits_{i \in I} {{h_i}} \int_{{V^1} \cap {V^2}} \omega   \nonumber \\
 &= \sum\limits_{i \in I} {\int_{{V^1}{ \setminus }({V^1} \cap {V^2})} {{h_i}\omega } }  
+ \sum\limits_{i \in I} {\int_{{V^2}{ \setminus }({V^1} \cap {V^2})} {{h_i}\omega } }  
+ \sum\limits_{i \in I} {{h_i}} \int_{{V^1} \cap {V^2}} \omega   \nonumber \\
  &= \int_{{V^1}{ \setminus }({V^1} \cap {V^2})} {\sum\limits_{i \in I} {{h_i}\omega } }  
+ \int_{{V^2}{ \setminus }({V^1} \cap {V^2})} {\sum\limits_{i \in I} {{h_i}\omega } }  
+ \int_{({V^1} \cap {V^2})} {\sum\limits_{i \in I} {{h_i}\omega } }  \nonumber \\
&= \int_{{V^1}{ \setminus }({V^1} \cap {V^2})} \omega   + \int_{{V^2}{ \setminus }({V^1} \cap {V^2})} \omega  
+ \int_{({V^1} \cap {V^2})} \omega  . 
\end{align}
The extension to the arbitrary number of covers is apparent. 

Now we introduce two ways of defining an integration of a $k$-form on a 
$k$-dimensional compact subset $S$ of $M$. 
The first is when $S$ is given by an inclusion (injective immersion) 
of $k$-dimensional manifold $P$ into $M$, namely $S$ is an immersed submanifold of $M$, 
and the second is when $S$ is an embedded submanifold of $M$. 
Though in the further discussion we always consider the case when we have the inclusion map, 
we will also introduce the definition for the second case as well. 

\begin{defn} Integration of $k$-form on $n$-dimensional manifold \\
Let $S$ be a immersed submanifold of $M$, given by 
$S=\iota (P) \subset M$, 
where $P$ is a compact $k$-dimensional manifold and $\iota$ an inclusion map.
Suppose we have a $k$-form $\omega $ on $M$.
The {\it integration of $k$-form on $M$} is given by, 
\begin{align}
\int_{\iota (P)} \omega = \int_P {{\iota ^*}\omega } .
\end{align} 
Then since the {\it r.h.s.} is an integration of a $k$-form over a $k$-dimensional manifold, 
the previous definition could be applied to calculate the integral. 
\end{defn}
$S$ will be diffeomorphic to $P$, and will inherit the topology of $P$ by $\iota$. 
In general this is not the same as the subset topology. 
In the later chapters, we will consider especially when $P$ a closed $k$-rectangle.
Then $P$ is called a {\it parameter space} and the map $\iota $ is called a 
{\it parameterisation}.  

Now, suppose we have a submanifold $S$ of $M$, 
we can still define the integration of $k$-form on $S$ in terms of pieces and 
adapted charts of a submanifold. 

\begin{defn} Submanifold and adapted charts \\
Let $M$ be a n-dimensional ${C^\infty }$-manifold and choose a non-empty subset $S \subset M$. 
$S$ is said to be a $k$-dimensional submanifold of $M$, if to every point ${p_0} \in S$ there exists a chart 
$(U,\varphi ),\varphi  = ({x^\mu })$ on $M$ at ${p_0}$ such that the set $S \cap U$ is given by the equation, 
\begin{align}
{x^{k + 1}}(p) = 0,{x^{k + 2}}(p) = 0, \cdots ,{x^n}(p) = 0.
\end{align} 
for $\forall p \in S \cap U$. 
Coordinate system $(U,\varphi ),\varphi  = ({x^\mu })$ with the above properties is said to be 
{\it adapted to the submanifold $S$ at ${p_0}$}. 
\end{defn}

\begin{defn} Half space \\
Let ${t^1}, \cdots ,{t^n}$ be the canonical coordinates of ${\mathbb{R}^n}$, 
and $\mathbb{R}_ - ^n = \{ y \in {\mathbb{R}^n}|{t^1}(y) \geqslant 0\} $ and 
$\partial \mathbb{R}_ - ^n = \{ y \in \mathbb{R}_ - ^n|{t^n}(y) = 0\} $. 
The subset $\mathbb{R}_ - ^n$ of ${\mathbb{R}^n}$, considered with the subspace topology, 
is called the {\it half space} of${\mathbb{R}^n}$, and $\partial \mathbb{R}_ - ^n$ 
considered as the topological subspace of $\mathbb{R}_ - ^n$, 
is canonically isomorphic with ${\mathbb{R}^{n - 1}}$, and is called the {\it boundary} of $\mathbb{R}_ - ^n$. 
\end{defn}

\begin{defn} Open sets in $\mathbb{R}_ - ^n$ \\
Let $U$ be a subset of $\mathbb{R}_ - ^n$. 
We say that $U$ is open in $\mathbb{R}_ - ^n$ when $U = \Omega  \cap \mathbb{R}_ - ^n$, 
where $\Omega $ is some open subset of ${\mathbb{R}^n}$. 
\end{defn}

Let $M$ be a n-dimensional ${C^\infty }$-differentiable manifold, 
and $\Omega $ be a nonempty compact subset of $M$. 
Let ${p_0} \in \Omega $ be a fixed point, and $(U,\varphi ), \varphi  = ({x^\mu })$ a chart on $M$ such that ${p_0} \in U$. 
We say that the chart $(U,\varphi )$ is adapted to $\Omega $ at ${p_0}$, 
if $\varphi (\Omega  \cap U)$ is open in $\mathbb{R}_ - ^n$. 
If the chart $(U,\varphi )$ is adapted to $\Omega $ at ${p_0}$, 
it is adapted at every point of $p \in \Omega  \cap U$, and we say a chart $(U,\varphi )$ 
is adapted to $\Omega $. Clearly, since $\Omega $ is supposed to be compact, 
there exist finitely many points ${p_1},{p_2}, \cdots ,{p_N}$ of 
$\Omega $ and adapted charts $({U_1},{\varphi _1}), \cdots ,({U_N},{\varphi _N})$, 
such that ${p_1} \in {U_1},{p_2} \in {U_2}, \cdots ,{p_N} \in {U_N}$, 
and $\Omega  \subset \bigcup\limits_\iota  {{U_\iota }} $. 

\begin{defn} Pieces of a manifold \\
$\Omega  \subset M$ is called a {\it piece} of $M$, 
if it is compact and to each point $p \in \Omega $ there exists a chart at $p$ adapted to $\Omega $. 
\end{defn}

Let $\operatorname{int} \Omega $ be the set of interior points of $\Omega $, 
and set $\partial \Omega  = \Omega { \setminus }\operatorname{int} \Omega $. 
Let ${q_0} \in \partial \Omega $, and let $(U,\varphi ),\varphi  = ({x^\mu })$ 
be a chart adapted to $\Omega $ at ${q_0}$. Then for every $q \in \partial \Omega  \cap U$,
${x^n}(q) = 0$, and the set $\partial \Omega $ has on $\partial \Omega  \cap U$ 
the equation ${x^n} = 0$. Thus by the definition $\partial \Omega $ is a submanifold of $\Omega $, 
of dimension $\dim \partial \Omega  = n - 1$. 
The submanifold $\partial \Omega $ is called the {\it boundary} of $\Omega $. 
It is easily seen that $\partial \Omega $ is compact. 

\begin{ex} 
A ball ${B_1} = \{ (x,y) \in {\mathbb{R}^2}|{x^2} + {y^2} \leqslant 1\} $ is a piece of ${\mathbb{R}^2}$. 
${B_1}$ is compact so it suffices to show that to each point $p \in {B_1}$ there exists a chart at $p$ adapted to ${B_1}$. 
For the points $p \in \operatorname{int} {B_1}$, it is apparent we can find an open chart on 
${\mathbb{R}^2}$ which is adapted to ${B_1}$. 
For the points on the boundary, $p \in \partial {B_1}$, let $(U,\phi ), \phi  = (r,\theta )$ 
be a chart on ${\mathbb{R}^2}$, with $ - 1 < r < R$, $c < \varphi  < c + \pi $ with $R$ 
some constant greater than 0, and $c$ a constant such that $p \in U$. 
Then $\phi ({B_1} \cap U) = ( - 1, 0] \times (c,c + \pi ) = \{ ( - 1, R) \times (c, c + \pi )\}  \cap \mathbb{R}_ - ^2$ 
is open in $\mathbb{R}_ - ^n$. By choosing appropriate $c$, 
we can always find such adapted chart for any $p \in \partial {B_1}$. 
\end{ex} 

One can also show that if $\Omega $ is orientable, $\partial \Omega $ is also orientable. 
Let $\Omega $ be a piece of $M$, ${p_0} \in \partial \Omega $ a point. 
We say that a vector $\xi  \in {T_{{p_0}}}M$ is oriented outwards $\Omega $, 
if there exists a chart $(U,\varphi ),\varphi  = ({x^\mu })$ on $M$ adapted to $\Omega 
$ at ${p_0}$, such that the chart expression 
$\displaystyle{\xi  = {\xi ^i}{\left( {\frac{\partial }{{\partial {x^i}}}} \right)_{{p_0}}}}$ 
satisfies the condition ${\xi ^n} > 0$. 
We show that this definition is independent of the choice of adapted chart which has the same orientation. 
Let $(\bar U,\bar \varphi ),\bar \varphi  = ({\bar x^\mu })$ 
be the second adapted chart which has the same orientation as $(U,\varphi ), \varphi  = ({x^\mu })$. 
Then on this second chart, 
\begin{align}
\xi  = {\bar \xi ^i}{\left( {\frac{\partial }{{\partial {{\bar x}^i}}}} \right)_{{p_0}}}, 
\quad {\bar \xi ^n} = {\left( {\frac{{\partial {{\bar x}^n}}}{{\partial {x^i}}}} \right)_{{p_0}}}{\xi ^i}, 
\end{align}
but on $U \cap \bar U \cap \partial \Omega $, 
${\bar x^n}({x^1},{x^2}, \cdots ,{x^{n - 1}},0) = 0$. \\
Hence, $\displaystyle{\frac{{\partial {{\bar x}^n}}}{{\partial {x^1}}} = 0, 
\frac{{\partial {{\bar x}^n}}}{{\partial {x^2}}} = 0, \cdots , 
\frac{{\partial {{\bar x}^n}}}{{\partial {x^{n - 1}}}} = 0}$ 
at $\varphi ({p_0})$, and we have 
\begin{align}
{\bar \xi ^n} = \frac{{\partial {{\bar x}^n}}}{{\partial {x^n}}}{\xi ^n}.
\end{align} 
However, since both charts are adapted and has the same orientation, 
$\displaystyle{\frac{{\partial {{\bar x}^n}}}{{\partial {x^n}}} > 0}$, and therefore 
${\bar \xi ^n} > 0$. 
In particular if $(U,\varphi ),\varphi  = ({x^i})$ 
is a chart adapted to $\Omega $ at ${p_0} \in \partial \Omega $, 
then the vector $\displaystyle{{\left( {\frac{\partial }{{\partial {x^n}}}} \right)_{{p_0}}}}$ 
is oriented outwards $\Omega $.
Now suppose $\Omega $ is oriented with the orientation $S$. 
Then by the compactness of $\Omega $, 
we have finitely many charts $({U_k}, {\varphi _k}), {\varphi _k} = (x_k^i)$, 
$1 \leqslant k \leqslant N$, adapted to $\Omega $ with the same orientation, 
and $\Omega  \subset \bigcup\limits_k {{U_k}} $. We set ${V_k} = {U_k} \cap \partial \Omega ,
\quad {\psi _k} = (x_k^1,x_k^2, \cdots ,x_k^{n - 1})$, 
where the coordinate functions $x_k^1, \cdots ,x_k^{n - 1}$ are considered to be restricted to ${V_k}$.
The pairs $({V_k}, {\psi _k}),{\psi _k} = (x_k^i)$, $1 \leqslant k \leqslant N$ form a smooth atlas on 
$\partial \Omega $. 
Then by definition and from $\displaystyle{\frac{{\partial {{\bar x}^n}}}{{\partial {x^1}}} = 0, 
\frac{{\partial {{\bar x}^n}}}{{\partial {x^2}}} = 0, 
\cdots ,\frac{{\partial {{\bar x}^n}}}{{\partial {x^{n - 1}}}} = 0}$, 
for any pair of $1 \leqslant k,l \leqslant N$, $\displaystyle{ \det D{\varphi _k}{\varphi _l}^{ - 1} 
= \frac{{\partial x_k^n}}{{\partial x_l^n}} \cdot \det D{\psi _k}{\psi _l}^{ - 1}}$. 
But since $\det D{\varphi _k}{\varphi _l}^{ - 1} > 0$ and 
$\displaystyle{\frac{{\partial x_k^n}}{{\partial x_l^n}} > 0}$, $\det D{\psi _k}{\psi _l}^{ - 1} > 0$. 
Therefore by definition, $\partial \Omega $ is an orientable manifold. 
The orientation of $\partial \Omega $ defined by the atlas $({V_k},{\psi _k})$, 
$1 \leqslant k \leqslant N$ is said to be associated with the given orientation of $M$, 
defined by the charts $({U_k},{\varphi _k})$, $1 \leqslant k \leqslant N$. 

\begin{defn} Integration of a $k$-form on $n$-dimensional manifold (by submanifold chart) \\
Consider a manifold $M$ and an orientable submanifold $S \subset M$ with $\dim S = k$.
Suppose we have a $k$-form $\rho $ on the neighbourhood of $S$. 
We will now introduce the integration of this form on $S$ by means of adapted charts of a submanifold. 
Let $\Omega  \subset S$ be a compact piece (compact submanifold with boundary) which is covered by the adapted chart 
$(U,\varphi ), \varphi  = ({x^1}, \cdots ,{x^k},0, \cdots ,0)$. 
In this chart, the local expression of $\rho $ is 
\begin{align}
\rho  = \frac{1}{{k!}}{\rho _{{i_1} \cdots {i_k}}}d{x^{{i_i}}} \wedge  \cdots  \wedge d{x^{{i_k}}},
\end{align} 
where ${\rho _{{i_1} \cdots {i_k}}}$ is a function of ${x^1}, \cdots ,{x^k}$.
Then we can define the integral of $\rho $ on the piece $\Omega $ by 
\begin{align}
\int_\Omega  \rho   = \int_{\varphi (\Omega )} {{{({\varphi ^{ - 1}})}^*}} \rho .
\end{align}
If we choose $\Omega $ and $\varphi $ appropriately, this can be expressed as 
\begin{align}
\int_\Omega  \rho   = \int_{\varphi (\Omega )} {{{({\varphi ^{ - 1}})}^*}} 
\rho  = \int_{{a_1}}^{{b_1}} {\int_{{a_2}}^{{b_2}} { \cdots \int_{{a_k}}^{{b_k}} 
{\frac{1}{{k!}}{\varepsilon ^{{i_1} \cdots {i_k}}}{\rho _{{i_1} \cdots {i_k}}}}  
\circ {\varphi ^{ - 1}}} } \,d{x^1} \cdots d{x^k},  \label{submfd1}
\end{align}
since we assumed that the orientation are preserved by the mapping by $\varphi $. 

To consider the integral over whole $S$, consider the finite family of submanifold charts on X, 
$\{ ({U_1},{\varphi _1}),({U_2},{\varphi _2}), \ldots ,({U_N},{\varphi _N})\} $, 
such that the family of open sets (in $S$) $\{ {U_1} \cap S,{U_2} \cap S, \ldots ,{U_N} \cap S\} $ covers $S$. 
Let $\{ {\chi _1},{\chi _2}, \ldots ,{\chi _N}\} $ be a partition of unity, subordinate to this covering. 
Since by definition $\operatorname{supp} {\chi _j} \cap S$ is a closed subset of the compact set $S$, 
it must be compact; on the other hand, $\operatorname{supp} {\chi _j} \cap S$ is a subset of the set ${U_j} \cap S$, 
the chart neighbourhood of the induced chart by $({U_j},{\varphi _j})$ on $S$. 
In particular, the integral 
\begin{align}
\int_{\operatorname{supp} {\chi _j} \cap S} {{\chi _j}\eta } 
\end{align}
is defined by formula (\ref{submfd1}) for each $j$. We set 
\begin{align}\int_S \eta   = \sum\limits_{j = 1}^N {\int_{\operatorname{supp} {\chi _j} \cap S} {{\chi _j}\eta } } .  \label{submfd2}
\end{align}
The real number given by the formula (\ref{submfd2}), 
is called the {\it integral of the form $\eta $ on $S$}. 
\end{defn}
We show that the right-hand side of (\ref{submfd2}) is independent of the choice of the family of submanifold charts
 $\{ ({U_1},{\varphi _1}),({U_2},{\varphi _2}), \ldots ,({U_N},{\varphi _N})\} $ 
 and the partition of unity 
 $\{ {\chi _1},{\chi _2}, \ldots ,{\chi _N}\} $. 
 Let $\{ ({\bar U_1},{\bar \varphi _1}), ({\bar U_2},{\bar \varphi _2}), \ldots ,({\bar U_M},{\bar \varphi _M})\} $ 
 be another family of charts and $\{ {\bar \chi _1},{\bar \chi _2}, \ldots ,{\bar \chi _M}\} $ 
 the corresponding partition of unity. By (\ref{submfd2}), for each $j$,  
\begin{align}
\int_{\operatorname{supp} {\chi _j} \cap S} {{\chi _j}\eta }  
= \sum\limits_{i = 1}^M {\int_{\operatorname{supp} {{\bar \chi }_i} \cap 
\operatorname{supp} {\chi _j} \cap S} {{{\bar \chi }_i}{\chi _j}\eta } } ,
\end{align}
and similarly for each $i$, 
\begin{align}
\int_{\operatorname{supp} {{\bar \chi }_i} \cap S} {{{\bar \chi }_i}\eta }  
= \sum\limits_{j = 1}^N {\int_{\operatorname{supp} {\chi _j} \cap \operatorname{supp} {{\bar \chi }_i} \cap S} 
{{\chi _j}{{\bar \chi }_i}\eta } } .
\end{align}
Thus 
\begin{align}
\sum\limits_{j = 1}^N {\int_{\operatorname{supp} {\chi _j} \cap S} {{\chi _j}\eta } }  
&= \sum\limits_{j = 1}^N {\sum\limits_{i = 1}^M {\int_{\operatorname{supp} {{\bar \chi }_i} 
\cap \operatorname{supp} {\chi _j} \cap S} {{{\bar \chi }_i}{\chi _j}\eta } } }  
= \sum\limits_{i = 1}^M \sum\limits_{j = 1}^N \int_{\operatorname{supp} {{\bar \chi }_i} 
\cap \operatorname{supp} {\chi _j} \cap S} {{\chi _j}{{\bar \chi }_i}\eta }   \nonumber \\
&=  \sum\limits_{i = 1}^M \int_{\operatorname{supp} {{\bar \chi} _i} \cap S} {{{\bar \chi} _i}\eta } 
\end{align}
as required.

%% file: thesis2012_chap3.tex
\chapter{Basics of Finsler geometry and parameterisation} \label{chap_3} 
In this chapter \ref{chap_3},  
we will briefly introduce the properties of Finsler geometry and some related structures 
that we will use for the considerations of calculus of variations. 
Since our motivation is to construct a theory applicable to concrete models of physics, 
there are certain aspects that may differ from the standard approach of a geometer, 
whose main interest lies on the construction or understanding of the geometrical structure itself. 
In particular, some of the standard definitions that allows further inquiries into problems of 
geometry may simply be unsuitable for tackling problems of physics. 
Therefore, in such cases, we have to modify or loosen some conditions. 
For instance, if we require strong convexity for the definition of Finsler manifold, 
most of the standard physical problems would be out of the scope. 
Also, if we require convexity (or ``regularity'' in some references), 
no gauge theories can be handled. 
We therefore propose to use only the minimal definitions, 
and use the name {\it Finsler manifold} in such broad sense. 
Nevertheless, for the construction of the theory of calculus of variation, 
the minimal definitions turn out to be sufficient, 
and no additional structures such as connections and curvature are required. 
We will begin with a very short historical review on how Finsler geometry was introduced, 
and then give the basic structures of Finsler geometry in the more modern terms, 
namely the tangent bundle and Finsler function, 
and introduce the Finsler length and its parameterisation. 
Then we will introduce the important concept of Finsler-Hilbert form, 
which is directly connected with Finsler length. 
These objects are the main tools for the calculus of variation, 
discussed in chapter 5. 

\section{Introduction to Finsler geometry} \label{sec_introFinsler}
	Historically, in his inaugural lecture, Riemann already mentioned on the special case of 
Finsler metric by stating, 
`{\it …the line element can be an arbitrary homogeneous function of first degree in the quantities 
$dx$ which remains the same when all the quantities 
$dx$ change the sign, and in which the arbitrary constants are function of the quantities $x$.}'~\cite{Spivak1}, 
where he referred to $dx$ as an ``infinitesimal displacement'' 
from the position $x$. This statement can be translated to the formula 
\begin{align}
&F({x^1},{x^2}, \cdots ,{x^n},\lambda d{x^1},\lambda d{x^2}, \cdots ,\lambda d{x^n}) 
= \left| \lambda  \right|F({x^1},{x^2}, \cdots ,{x^n},d{x^1},d{x^2}, \cdots ,d{x^n}), \nonumber  \\
&  \hspace{10cm} \lambda  \in \mathbb{R}.   \label{Riemann-Finsler}
\end{align}
Then in the subsequent discussion, 
`{\it ...and consequently $ds$ equals the square root of an everywhere positive homogeneous function 
of the second degree in the quantities $dx$, 
in which the coefficients are continuous functions of the quantities $x$.}'
Therefore, by this statement he restricted this function $F = ds$ 
to a more special case where it is given by, 
$F({x^i},d{x^i}) = \sqrt {{g_{ij}}(x)d{x^i}d{x^j}}$, 
which is the infinitesimal length of a curve on a Riemannian manifold. 
After the development of tensor analysis and exterior differential calculus, 
the theory was reformed in such a way that infinitesimal displacements $dx$ was  
replaced by one forms, 
therefore allowing the concept to be treated in the realm of linear algebra, 
and the square of the structure $ds$ was replaced by a symmetric tensor 
$g={g_{ij}}(x)d{x^i} \odot d{x^j}$, 
which gives an inner product of tangent vectors at each point $p$ on $M$. 
In other words, the concept of infinitesimal displacement was in a sense abandoned; 
instead of considering a length of infinitesimal piece of a curve, 
the new structure $g$ defines the length of any finite sized vector at a point $p$ on $M$. 
Also, since $g$ is an inner product, 
it will give not only the length but also defines the angle between the two vectors. 

In this modern view, the geometry is described by considering the tangent bundle, 
the arena of ``vectors'', not only by its base manifold $M$ and the curve $C$, 
as in the original idea of Riemann. 
The additional structure at each point of $M$ is called a fibre in mathematics, 
and in physics it is frequently called the ``internal space''. 

	Finsler geometry also has these two perspectives, 
one from the study on the properties of infinitesimal length on $M$, 
that is, a view as a geometry of calculus of variations (Finsler, Carath{\' e}odory), 
and then another view from the study of bundles and tensor analysis (Synge, Berwald, Cartan). 
Now it has become more standard to understand the Riemann geometry in this latter perspective, 
and likewise can be said for Finsler geometry. 
In this and the following chapter \ref{chap_4}, we will also discuss on Finsler geometry 
and the further extensions of Kawaguchi geometry using this perspective. 
However, it is also very useful to remember the original Riemann's idea as well. 
Finsler geometry is simply a consideration of an arc length in more general setting than Riemann geometry. 
It has no inner product structure, and therefore more fundamental. 
The arc length of the line element $ds$ is a homogeneous function of degree one, 
homogeneous with respect to the infinitesimal dislocation $dx$, 
where $ds$ and $dx$ are simply functions on $M$. 
There are researches on Finsler and Kawaguchi geometry in this direction as well, 
especially in cases where calculus of variations are important~\cite{Oo1,OTY3}, 
and indeed we take such works as an inspiring reference to our later discussions in chapter 5. 

The main reference used in this section is Matsumoto~\cite{matsumoto2}, 
Chern, Chen and Lam~\cite{ChernChenLam}, and Tamassy~\cite{Tamassy3}.  


\section{Basic definitions of Finsler geometry}
The geometric structure that defines the Finsler manifold is a function on the total space of 
a tangent bundle $(TM,{\tau _M},M)$. 
This structure is called Finsler function or Finsler metric in some references. 

Let $M$ be a ${C^\infty }$-differentiable manifold, 
$(TM, \tau_M, M)$ its tangent bundle, ${T^0}M: = TM{ \setminus }0$ the slit tangent bundle excluding the 
zero section from $TM$, and $(U,\varphi ), \varphi  = ({x^\mu },{y^\mu }),\mu  = 1, \cdots ,n$ an induced chart on $TM$. 

\begin{defn}   Finsler manifold \\
The {\it $n$-dimensional Finsler manifold} is a pair $(M,F)$ 
where $F$ is a ${C^0}$ function on $TM$ and ${C^\infty }$ function on ${T^0}M$, 
satisfying the following conditions. \\ 
(I) Homogeneity 
\begin{eqnarray}
F(\lambda v) = \lambda F(v),v \in {T^0}M,  \lambda  > 0   \label{hom1}
\end{eqnarray}
or in coordinate expression,  
\begin{eqnarray}
F({x^\mu },\lambda {v^\mu }) = \lambda F({x^\mu },{v^\mu }), \lambda  > 0.
\end{eqnarray} 
\end{defn}

This condition (I) also implies the condition of {\it Euler's homogeneous function theorem}, 
\begin{eqnarray}
\frac{{\partial F}}{{\partial {y^\mu }}}{y^\mu } = F.   \label{cond_hom2}
\end{eqnarray}
Function with such properties is called a {\it Finsler function}.

\begin{remark}
Depending on the authors, usually there are several additional properties required for a Finsler manifold. 
Besides the above homogeneity condition, other requirements are such as, \\
(I')  absolute homogeneity  
\begin{eqnarray}
F(\lambda v) = \left| \lambda  \right|F(v),\lambda  \in \mathbb{R}
\end{eqnarray}
This corresponds to the condition (\ref{Riemann-Finsler}) suggested by Riemann, 
sometimes it is called {\it symmetric Finsler manifold}.\\ 
(II)  non-negativity of $F$ 
\begin{eqnarray}
F:TM \to {\mathbb{R}^ + }.
\end{eqnarray}
(III)  Convexity (Regularity) \\
The determinant of the matrix 
\begin{eqnarray}
{g_{ij}}: = \frac{1}{2}\frac{{{\partial ^2}F}}{{\partial {y^i}\partial {y^j}}},
\end{eqnarray} 
$i = 1, \cdots ,n$ is non-zero. \\
The structure ${g_{ij}} \in {C^\infty }(TM)$ is usually called a 
{\it fundamental tensor}, or {\it fundamental form}. 
(Refer to :~\cite{matsumoto2,ChernChenLam,Chern2,Tamassy3})\\
(IV)  Strong Convexity 
\begin{eqnarray}
{g_{ij}}(v){v^i}{v^j}: = \frac{1}{2}{\left. 
{\frac{{{\partial ^2}F}}{{\partial {y^i}\partial {y^j}}}} \right|_v}{v^i}{v^j} > 0, 
\end{eqnarray}
$\forall v \ne 0 \in {T_p}M, \forall p \in M.$
(Refer to :~\cite{ChernChenLam,Chern2,Tamassy3}) \\
Together with (II) this equation is equivalent to the triangular inequality, 
\begin{eqnarray}
F(v) + F(w) > F(v + w),
\end{eqnarray}
for $\forall p \in M$, $v,w \in {T_p}M,\;v \ne w$.
However, the condition (II), (III), (IV) is too restrictive for the application to physics. 
Therefore, we will only require the minimal condition (I). 
This condition is sufficient for our following discussions. 
\end{remark}

The Finsler function $F$ is the main structure of Finsler geometry, 
and this definition gives the view to Finsler geometry as geometry of tangent bundle endowed with 
a specific feature. 
In the following, we will define the Finsler length of a curve on $M$, 
and show that the homogeneity condition is equivalent to the parameterisation 
invariant property of this arc length. 
In this way the two perspectives of Finsler geometry become connected. 

\section{Parameterisation and Finsler length} \label{sec_paramFinslerlength}
Here we introduce the curves, arc segments, their parameterisations and Finsler length. 
The notion of the Finsler length is naturally extended to 
$1$-dimensional immersed submanifolds. 

\begin{defn} $C^r$-curves \label{def_cr_curve}\\
Let $M$ be a smooth manifold, and $\sigma: I \to M$ a $C^r$-mapping, 
where $I$ is an open interval of a real line $\mathbb{R}$. 
We denote the image of $I$ by $C$; $C:=\sigma(I) \subset M$, 
and call $C$, the {\it $C^r$-curve} on $M$. 
\end{defn}

\begin{defn} Lift of $C^r$-curves \label{def_lift_curve}\\
Consider a tangent bundle $(TM,{\tau _M},M)$ where ${\tau _M}$ is the natural projection. 
By differentiating the map $\sigma: I \to M$, we get a natural mapping 
$\hat \sigma :I \to TM$. Denote the image of $I$ by $\hat C := \hat \sigma (I) \subset TM$, 
where ${\tau _M}(\hat C) = C$, and $\hat \sigma (t)$, $t \in I$ is a tangent vector at the point 
$\sigma (t) \in M$, namely 
\begin{align}
\hat \sigma (t): = {T_t}\sigma \left( {\frac{d}{{dt}}} \right) 
= {\left. {\frac{{d({x^\mu } \circ \sigma )}}{{dt}}} \right|_t}{\left( {\frac{\partial }{{\partial {x^\mu }}}} 
\right)_{\sigma (t)}}.
\end{align}
The map $\hat \sigma$ and its image $\hat C$ is called the {\it lift}, or the {\it tangent lift} of $\sigma$ (resp. $C$).   
\end{defn}

\begin{defn} Regularity of $\sigma$ \\
The $C^r$-map $\sigma: I \to M$ is called {\it regular}, 
if its lift $\hat \sigma$ is nowhere $0$. 
\end{defn}

\begin{defn} Parameterisation of an immersed curve \\
The $C^r$-map $\sigma: I \to M$ is called an {\it immersion}, 
if its lift $\hat \sigma$ is injective,  
and the image $C=\sigma(I)$ is called an {\it immersed curve}. 
The map $\sigma$ is called a {\it parameterisation} of immersed curve $C$, 
and $I$ is called a {\it parameter space}, 
when $\sigma$ is an immersion, and preserves orientation.
\end{defn}

\begin{defn} Parameterisation \label{def_parameterisation} \\
Let $\sigma: I \to M$ be a $C^r$-map, and $C=\sigma (I)$.
The map $\sigma $ is called a {\it parameterisation} of $C$, 
and $I$ is called a {\it parameter space}, 
when $\sigma$ is injective, and preserves orientation.
\end{defn}

\begin{defn} Lift of a parameterisation \label{def_liftedparameterisation} \\
We call $\hat \sigma$ the {\it lift of parameterisation} $\sigma$, 
when $\sigma$ is a parameterisation. 
\end{defn} 

Given a curve $C$ on $M$, more than one map and open interval may exist, 
namely for cases such as $C=\sigma(I)=\rho(J)$, where $\sigma, \, \rho$ are the maps, 
and $I, \, J$ are the open intervals in $\mathbb{R}$. 
We can classify the curves by considering the properties of these maps. 
Some example of the curves are shown in Figure \ref{fig_curves}.

\begin{defn} Regular curve \\ 
Let $C$ be a $C^r$-curve on $M$. 
$C$ is called a {\it regular curve} on $M$, 
if there exists a regular $C^r$ map $\sigma$ 
and an open interval $I$ of $\mathbb{R}$ such that 
$\sigma(I) = C$. 
\end{defn}

\begin{defn} Parameterisable immersed curve \\ 
Let $C$ be a $C^r$-curve on $M$. 
$C$ is called a {\it parameterisable immersed curve} on $M$, 
if there exists an immersion $\sigma$ 
and an open interval $I$ of $\mathbb{R}$ such that 
$\sigma(I) = C$.  
\end{defn}

\begin{defn} Parameterisable curve \\ 
Let $C$ be a $C^r$-curve on $M$. 
$C$ is called a {\it parameterisable curve} on $M$, 
if there exists an injective $C^r$-map $\sigma$ 
and an open interval $I$ of $\mathbb{R}$ such that 
$\sigma(I) = C$.  
\end{defn}

\begin{figure}
  \centering
  \includegraphics[width=12cm]{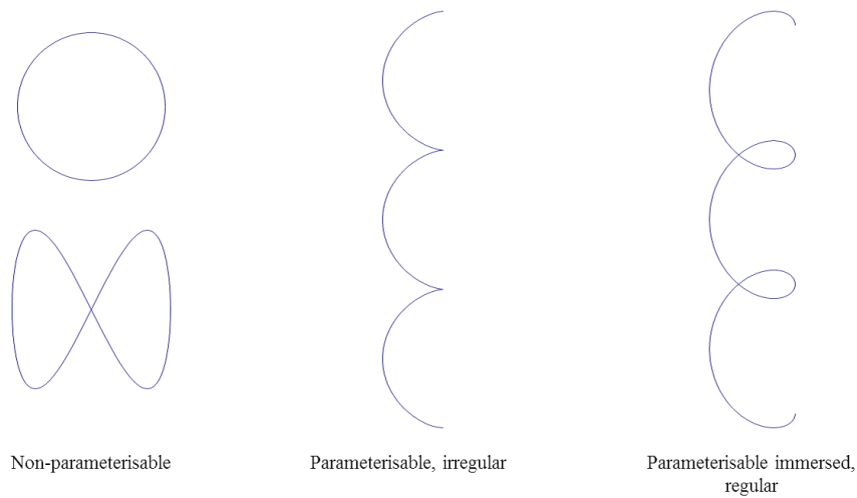}
  \caption{Example of curves in $\mathbb{R}^2$} \label{fig_curves}
\end{figure}

\begin{ex}
$S^1$ embedded in $\mathbb{R}^2$ is neither a parameterisable curve nor a pararameterisable immersed curve. 
For instance, consider a map $\sigma : \mathbb{R} \to \mathbb{R}^2$, which in coordinates are given by 
$\sigma(t) := (\cos(t), \sin(t))$. The lift $\hat \sigma$ will be $\hat \sigma = (-\sin(t), \cos(t))$, 
so nowhere zero, meaning it is a regular curve, but neither $\sigma$ nor $\hat \sigma$ are injective. 
Indeed we can consider different parameterisations, 
but since this is a closed curve, no injective map from an open subset of $\mathbb{R}$ exists. 
Furthermore, since it is also a smooth closed curve, no immersion from an open subset of $\mathbb{R}$ exists.
Nevertheless, by restricting $\sigma$ to finite interval $(0, 2 \pi)$, we can have the parameterisation 
of the corresponding part of the circle. 
\end{ex}

\begin{ex}
A curve ``$\infty$'' in $\mathbb{R}^2$ 
is neither a parameterisable curve nor a pararameterisable immersed curve, 
since it is smooth and closed.  
Nevertheless, we may consider a map such as $\sigma(t) := (\cos(t), \sin(t))$, 
and by restricting the open interval to finite interval $(0, 2 \pi)$, 
we can have the parameterisation of the corresponding part of the circle. 
\end{ex}

\begin{ex}
A curve defined by a map $\sigma : \mathbb{R} \to \mathbb{R}^2$, 
which in coordinates are given by 
$\sigma(t) := (\cos(t), \sin(t) - t)$, is a parameterisable curve, but not regular. 
The lift $\hat \sigma$ will be $\hat \sigma = (-\sin(t), \cos(t)-1)$, 
which becomes $0$ at $t=2 n \pi$, for integer $n$. 
\end{ex}

\begin{ex}
A spiral curve in $\mathbb{R}^2$ (on the right of Fig. \ref{fig_curves}) is a regular, 
and parameterisable immersed curve. 
There exists a regular parameterisation 
defined by a map $\sigma : \mathbb{R} \to \mathbb{R}^2$, 
which in coordinates are given by 
$\sigma(t) := (\cos(t), \sin(t) - t/2)$. 
The lift $\hat \sigma$ will be $\hat \sigma = (-\sin(t), \cos(t)-1/2)$, 
$\sigma$ is not injective since it gives the same point for $t$ such that $t-\sin(t)=\pi$. 
\end{ex}

\begin{ex}
In some cases, we can find a regular parameterisation of a curve that 
was originally given by a map which is not a regular parameterisation. 
Consider a curve defined by a map  $\sigma : I \to \mathbb{R}^2$, $I=(0, 2 \pi)$,  
which in coordinates are given by 
$\sigma(t) := (\sin(t), -\cos(2t))$. 
$\sigma$ is not regular, since its lift $\hat \sigma$ becomes $0$ at $t=\pi/2$.
However, there exists a regular parameterisation of this curve $C= \sigma(I)$ 
by the map $\tilde \sigma: J \to \mathbb{R}^2$, $J=(-1, 1)$,  
which in coordinates are given by 
$\hat \sigma(s) := (s, 2s^2 -1)$, and its lift is nowhere $0$.  
\end{ex}

Occasionally, we implicitly refer to the pair $(\sigma, I)$ by the parameterisation $\sigma$.
In the following discussion, we will only consider regular, parameterisable curves.

\begin{figure}
  \centering
  \includegraphics[width=5cm]{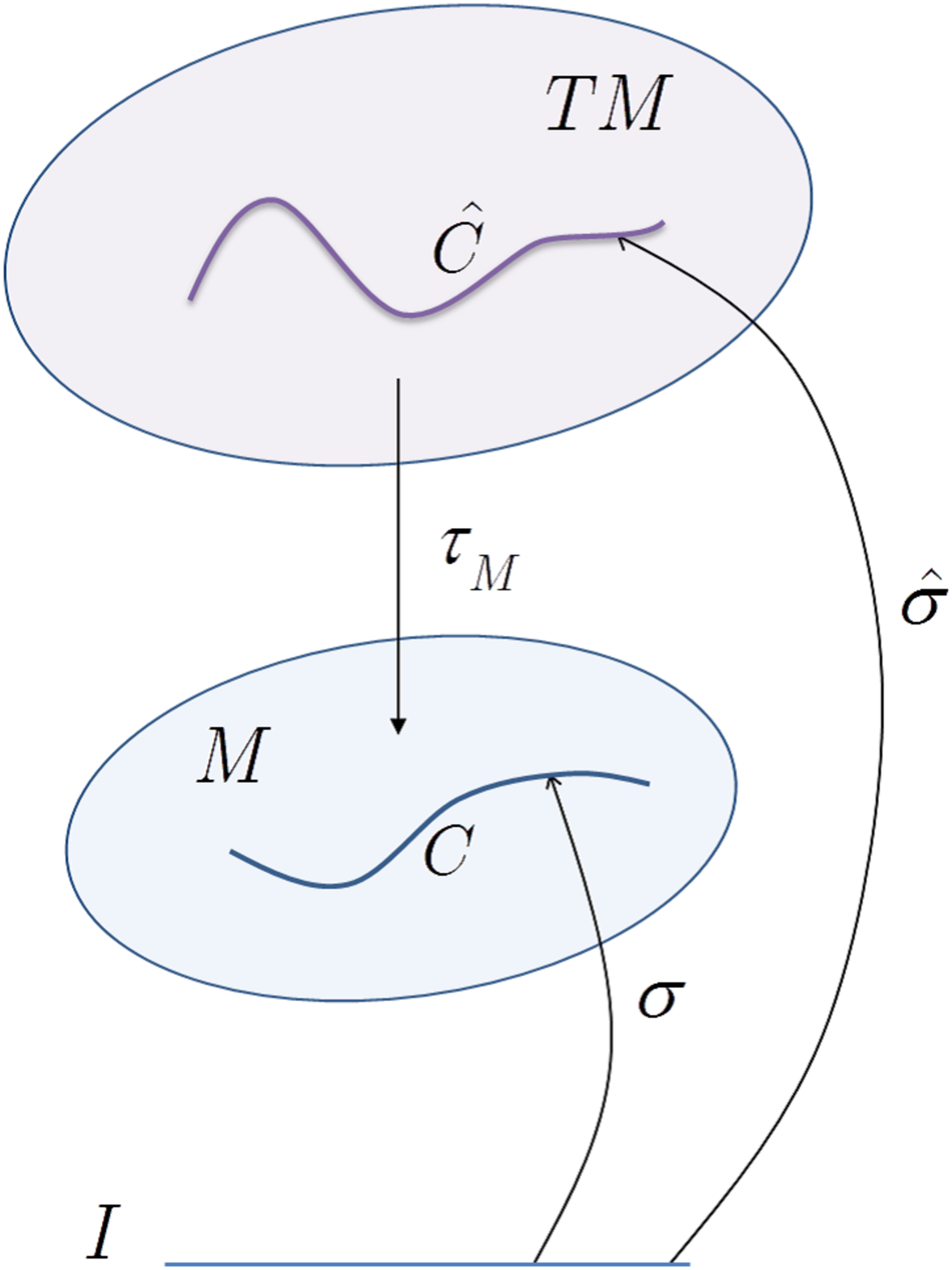}
  \caption{Parameterisation of a curve}
\end{figure}


Now we will introduce the concept of a length of a curve by integration on 
the parameter space. 
For simplicity, we will restrict ourselves to curves that are parameterisable, 
and consider its closed subset, which we define below. 

\begin{defn} arc segment \label{def_arc} \\ 
Let $\tilde{C}$ be a parameterisable curve on $M$ with some parameterisation $\sigma :I\to M$.
A subset of $\tilde{C}$ given by $C:=\sigma ([{{t}_{i}},{{t}_{f}}]) \subset \tilde{C}$, where $[{{t}_{i}},{{t}_{f}}]\subset I$ 
is called the {\it arc segment} on $M$, 
$\sigma$ is called the {\it parameterisation of the arc segment}
and the closed interval $[{{t}_{i}},{{t}_{f}}]$ is called the 
{\it parameter space of an arc segment}. 
\end{defn}

The Finsler function defines a geometrical length of an arc segment $C$ on $M$. 

\begin{defn} Finsler length  \label{def_Finserlength} \\
Let $(M,F)$ be the $n$-dimensional Finsler manifold, 
and $C$ the arc segment on $M$ such that $C=\sigma ([{{t}_{i}},{{t}_{f}}])$. 
We assign to $C$ the following integral  
\begin{align}
{{l}^{F}}(C)=\int_{{{t}_{i}}}^{{{t}_{f}}}{F\left( \hat{\sigma }(t) \right)dt}.   \label{def_FinslerCurve}
\end{align}
We call this number ${{l}^{F}}(C)$ the {\it Finsler length} of $C$.
\end{defn}

Let $(U,\varphi )$, $\varphi =({{x}^{\mu }},{{y}^{\mu }})$, 
$\mu =1,\cdots ,n$ be the induced chart on $TM$. 
By chart expression, (\ref{def_FinslerCurve}) is, 
\begin{align}
{{l}^{F}}(C)=\int_{{{t}_{i}}}^{{{t}_{f}}}{F\left( {{x}^{\mu }}(\hat{\sigma }(t)),{{y}^{\mu }}
(\hat{\sigma }(t)) \right)dt}
=\int_{{{t}_{i}}}^{{{t}_{f}}}{F\left( {{x}^{\mu }}(\sigma (t)),\frac{d{{x}^{\mu }}(\sigma (t))}{dt} \right)dt},  \label{FinslerCurve}
\end{align}
where we used the definition of $\hat{\sigma }$, and definition of induced coordinates of $TM$, 
\begin{align}
({{x}^{\mu }}\circ \hat{\sigma })(t)=({{x}^{\mu }}\circ \sigma )(t), \quad
({{y}^{\mu }}\circ \hat{\sigma })(t)={{\left. \frac{d({{x}^{\mu }}\circ \sigma )}{dt} \right|}_{t}}. 
\end{align}

Let $\rho :J \to C$, $J  \subset \mathbb{R}$ be another parameterisation of $C$. 
When there exists a diffeomorphism $\phi :J \to I$ 
such that $\rho  = \sigma  \circ \phi $, this gives an equivalence relation 
$\sigma \~\rho $. We are able to find an important property of the Finsler length, 
which is the following lemma: 
\begin{lemma} Reparameterisation invariance of Finsler length \label{lem_repinv_Finsler}\\
The Finsler length does not change by the reparameterisation $\rho  = \sigma  \circ \phi $, $\phi :J \to I$, 
where $\phi $ is a diffeomorphism such that preserves the orientation. 
\end{lemma}
\begin{proof}
Dividing the interval $[{{t}_{i}},{{t}_{f}}]$ if necessary into smaller closed sub-intervals 
we can suppose without loss of generality that the set $C=\sigma ([{{t}_{i}},{{t}_{f}}])$
lies in the coordinate neighbourhood of a chart $(U,\varphi )$, $\varphi =({{x}^{\mu }})$.
Then the lift of $\rho $ becomes, 
\begin{align}
\hat \rho (s) = {\left. {\frac{{d({x^\mu } \circ \sigma  \circ \phi )}}{{ds}}} 
\right|_s}{\left( {\frac{\partial }{{\partial {x^\mu }}}} \right)_{\sigma  \circ \phi (s)}} 
= {\left. {\frac{{d({x^\mu } \circ \sigma )}}{{dt}}} \right|_{\phi (s)}}{\left. {\frac{{d\phi }}{{ds}}} 
\right|_s}{\left( {\frac{\partial }{{\partial {x^\mu }}}} \right)_{\sigma  \circ \phi (s)}} 
= {\left. {\frac{{d\phi }}{{ds}}} \right|_s}\hat \sigma (\phi (s)),
\end{align}
for $s \in J$, and since $\rho $ is a regular parameterisation that preserves orientation, 
$\displaystyle{\frac{{d\phi }}{{ds}} > 0}$. 
It is easy to see that the length of $C$ is preserved by
\begin{align}
{l^F}(C) &= \int_{{s_i}}^{{s_f}} {F(\hat \rho (s))\,ds}  
= \int_{{s_i}}^{{s_f}} {F(\widehat {\sigma  \circ \phi }(s))\,ds} \hfill \nonumber \\
   &= \int_{{s_i}}^{{s_f}} {F\left( {\frac{{d\phi }}{{ds}}(s)\, 
\hat \sigma (\phi (s))} \right)\,ds}  
= \int_{{\phi ^{ - 1}}({t_i})}^{{\phi ^{ - 1}}({t_f})} 
{F(\hat \sigma (\phi (s)))\,\frac{{d\phi }}{{ds}}(s)ds}  \hfill \nonumber \\
&= \int_{{t_i}}^{{t_f}} {F(\hat \sigma (t)))\,dt} . \hfill  \label{Finsler_length}
\end{align}
${{s}_{i}},\,{{s}_{f}}$ are the pre-image of the boundary points ${{t}_{i}},\,{{t}_{f}}$ by $\phi$.
\end{proof} 
In the second line of (\ref{Finsler_length}), we have used the homogeneity condition of $F$. 
The homogeneity of $F$ and parameterisation invariance of Finsler length is an equivalent property. 

\begin{remark}
The ``Finsler length'' does not have the properties of a
``standard'' length, considered by Euclid or Riemannian geometry,  
since we require only homogeneity condition of the Finsler function. 
For instance, when one changes the orientation of the curve, in general, 
it gives different values (not just signatures). 
However, in our following discussion of the calculus of variations, 
we can still use this concept to obtain extremals and equations of motion, 
and it maybe also an interesting tool for considering differential geometry of submanifolds,
and possible generalisations of mechanics. 
\end{remark}

\section{Finsler-Hilbert form} \label{sec_FinslerHilbertform}

	Given a Finsler manifold, we can obtain a important geometrical structure which is called a 
{\it Hilbert form} by some authors~\cite{ChernChenLam}. 
In this thesis, we sometimes call them {\it Finsler-Hilbert form}, 
just to stress it is for the first order mechanics. 
In chapter \ref{chap_4}, 
we will generalise this concept to second order and higher dimensional parameter space. 

\begin{defn} Hilbert form \label{def_Hilbert_from} \\
Let $(V,\psi )$, $\psi  = ({x^\mu },{y^\mu })$, $\mu  = 1, \cdots ,n$ be an induced chart on $TM$. 
Consider the following $1$-form on ${T^0}M$,
which in local coordinates are expressed by 
\begin{align}
{\mathcal{F}} = \frac{{\partial F}}{{\partial {y^\mu }}}d{x^\mu }. \label{Hilbertform}
\end{align}
This form is invariant with respect to the coordinate transformations by 
\begin{align}
{x^\mu } \to {\tilde x^\mu } = {\tilde x^\mu }({x^\nu }), 
{y^\mu } \to {\tilde y^\mu } = \frac{{\partial {{\tilde x}^\mu }}}{{\partial {x^\nu }}}{y^\nu },
\end{align}
therefore, it is a globally defined form on ${T^0}M$.
We will call this global form with the local coordinate expression (\ref{Hilbertform}), 
{\it  Hilbert form}. 
\end{defn} 

\begin{lemma} 
Let $\mathcal{F}$ be the Hilbert $1$-form on ${{T}^{0}}M$, 
$C=\sigma (\bar I)$ the arc segment on $M$, with 
$\bar I=[{{t}_{i}},{{t}_{f}}]$ a closed interval in $\mathbb{R}$. Then, 
\begin{align}
\int_{{\hat{C}}}{\mathcal{F}}={{l}^{F}}(C).  \label{Hilbertform_length_id}
\end{align}
\end{lemma}

\begin{proof}
The simple calculation leads, 
\begin{align}
\int_{\hat C} {\mathcal{F}} &= \int_{\hat \sigma (\bar I)} 
  {\frac{{\partial F}}{{\partial {y^\mu }}}d{x^\mu }}  
  = \int_{{t_i}}^{{t_f}} {\frac{{\partial F}}{{\partial {y^\mu }}} 
  \circ \hat \sigma \,d({x^\mu } \circ \hat \sigma )}  \hfill \nonumber \\
&= \int_{{t_i}}^{{t_f}} {\frac{{\partial F}}{{\partial {y^\mu }}}(\hat \sigma (t))\, 
   \frac{{d({x^\mu }(\sigma (t))}}{{dt}}dt}  = \int_{{t_i}}^{{t_f}} 
   {\frac{{\partial F}}{{\partial {y^\mu }}}(\hat \sigma (t))\,{y^\mu }(\hat \sigma (t))dt}  \hfill \nonumber \\
&= \int_{{t_i}}^{{t_f}} {F(\hat \sigma (t))\,dt}  = {l^F}(C) \hfill  
\end{align}
where we used the pulled back homogeneity condition 
\begin{eqnarray}
\frac{{\partial F}}{{\partial {y^\mu }}} \circ \hat \sigma  
\cdot {y^\mu } \circ \hat \sigma  = F \circ \hat \sigma. 
\end{eqnarray}
\end{proof}

\begin{remark}
This lemma extends the notion of Finsler length given by (\ref{def_FinslerCurve}). 
Namely, since the Hilbert form can be integrated over {\it any} 
$1$-dimensional submanifold of $M$, the identity (\ref{Hilbertform_length_id}) 
suggests that it is possible to extend the integration 
over arc segments to {\it arbitrary} $1$-dimensional submanifold of $M$, 
by considering the Hilbert form.  
\end{remark}

\begin{remark}
Now that we showed that Finsler-Hilbert $1$-form gives the Finsler length 
(and in a more general situation of a submanifold), 
we can redefine the pair $(M,{\mathcal{F}})$ as the Finsler manifold instead of taking $(M,F)$. 
This is a more geometrical definition of a Finsler manifold, 
and also in close analogy to the case of Riemannian geometry, 
where the geometric structure is given by a tensor $g$, and not by a function. 
This observation is important also for the consideration of Kawaguchi geometry. 
\end{remark}

\begin{remark}\label{rem_HilbertCartan}
When given a Hilbert form, we can obtain the Cartan form, which is a one form defined
 on ${J^1}Y$, where $Y$ is a $n+1$-dimensional manifold, 
and ${J^1}Y$ is the prolongation of the bundle $(Y, \pi, \mathbb{R})$. 
( In most cases, $Y = \mathbb{R} \times Q$ is considered, and is called an 
extended configuration space. $Q$ is a configuration space of dimension $n$. 
In such case, the bundle $(Y, pr_1, \mathbb{R})$ becomes a trivial bundle. ) 
Let $(U, \psi) $, $\psi  = (t,q^i)$, $i = 1, \cdots , n$ 
be the adapted chart on $Y$, and the induced chart on $\mathbb{R}$ be 
$(\pi(U),t)$.
We denote the induced chart on ${J^1}Y$ by 
$({({\pi ^{1,0}})^{ - 1}}(U),{\psi ^1}),{\psi ^1} = (t,{q^i},{\dot q^i})$, 
where $\pi^{1,0}: J^1 Y \rightarrow Y$ is the prolongation of $\pi$.  
Suppose we have a Hilbert form on $TY$. Take the induced chart on $TY$ as 
$({({\tau _Y})^{ - 1}}(U),{\psi ^1}),\;{\psi ^1} = (x^0,x^i,y^0,y^i)$, $i = 1, \cdots ,n$. 
(In order to avoid confusion we use different symbols, 
but clearly ${x^0} = t \circ {\tau _Y},{x^i} = {q^i} \circ {\tau _Y}$.) 
Since both ${J^1}Y$and $TY$ are bundles over $Y$, around every $p \in Y$, 
there exists a local trivialisation. Take $p \in U$, and let $(U,{F_p},{t_p})$, 
${t_p}:{({\tau _Y})^{ - 1}}(U) \to U \times {F_p}$ 
be the local trivialisation of $TY$ and   $(U,{G_p},{\tilde t_p})$,
 ${\tilde t_p}:{({\pi ^{1,0}})^{ - 1}}(U) \to U \times {G_p}$ 
be the local trivialisation of ${J^1}Y$ where ${F_p} = {\mathbb{R}^{2(n + 1)}}$ and 
${G_p} = {\mathbb{R}^{2n + 1}}$.
Then there exists a natural inclusion $\iota :{({\tilde t_p})^{ - 1}}(U \times {G_p}) \hookrightarrow TM,\;$ 
$\iota ({({\tilde t_p})^{ - 1}}(U \times {G_p})) = {({t_p})^{ - 1}}(U \times {F_p})$, 
which in coordinate functions are given by  
\[{x^0} \circ \iota  = t,\;{x^i} \circ \iota  = {q^i},\;{y^0} \circ \iota  = 1,\;{y^i} \circ \iota  = {\dot q^i},\]
so that the submanifold equation is ${y^0} = 1$. 
The local coordinate expression of the Hilbert form is, 
\begin{eqnarray}
{\mathcal{F}} = \frac{{\partial F}}{{\partial {y^0}}}d{x^0} 
+ \frac{{\partial F}}{{\partial {y^i}}}d{x^i} = \frac{1}{{{y^0}}}
\left( {F - \frac{{\partial F}}{{\partial {y^i}}}{y^i}} \right)d{x^0} 
+ \frac{{\partial F}}{{\partial {y^i}}}d{x^i},
\end{eqnarray}
where the second equality holds by using the Euler's homogeneity condition. 
The pull back of this Hilbert form by $\iota$ is 
\begin{eqnarray}
{\iota ^*}{\mathcal{F}} = \frac{1}{{{y^0}}} \circ \iota 
\left( {F \circ \iota  - \frac{{\partial F}}{{\partial {y^i}}} \circ \iota  
\cdot {y^i} \circ \iota } \right)d({x^0} \circ \iota ) + \frac{{\partial F}}
{{\partial {y^i}}} \circ \iota  \cdot d({x^i} \circ \iota ) \hfill 
\end{eqnarray}
However, by the submanifold equation ${y^0} = 1$, this becomes exactly the Cartan form, 
\begin{eqnarray}
{\Theta _C} =  {\iota ^*}{\mathcal{F}} = \left( {{\mathcal{L}} 
- \frac{{\partial \user2{\mathcal{L}}}}{{\partial {{\dot q}^i}}}
{{\dot q}^i}} \right)dt + \frac{{\partial \user2{\mathcal{L}}}}
{{\partial {{\dot q}^i}}}d{x^i}, \label{HilbertCartanform}
\end{eqnarray}
where ${\mathcal{L}}: = F \circ \iota $ should be regarded as the ``conventional'' Lagrange 
function\footnote{We simply use the term 
``conventional'' to distinguish the Lagrangian function over $J^1 Y$, 
from our Lagrangian which is the Hilbert $1$-form over $TY$.} 
defined on ${J^1}Y$. 
In fact, the inclusion map $\iota$ can be globalised for all $J^1 Y$, and the 
relation (\ref{HilbertCartanform}) is global. 
In other words, Cartan form is a restriction of a Hilbert form to a submanifold ${J^1}Y$ in $TY$. 
Also, the base manifold of the bundle $(Y, \pi, \mathbb{R})$ 
is naturally considered as a parameter space, and for Cartan form such structure was needed, 
while Hilbert form does not need such fibre bundle structure of $Y$. 
In this sense, Hilbert form is a generalisation of the Cartan form that does not depend on 
specific bundle structures. 
\end{remark}
In chapter 5, we will also consider the converse and discuss how to obtain the 
Hilbert form when a conventional Lagrangian is given.

%% file: thesis2012_chap4.tex
\chapter{Basics of Kawaguchi geometry and parameterisation} \label{chap_4} 
In this chapter \ref{chap_4}, we will introduce a geometry 
which was originally considered by A. Kawaguchi as a extension of Finsler geometry.
In contrast to Finsler geometry, 
Kawaguchi geometry still does not have a well-developed consensus yet, 
and it may be a bit early to be called as ``geometry''. 
Nevertheless, it follows the same line of thought that originates from Riemann, 
and with the hope of its future establishment, we will call so in this thesis. 

\section{Introduction to Kawaguchi geometry} \label{sec_IntroKawaguchi}
A. Kawaguchi considered the generalisation of Finsler geometry in two directions, 
one for the case of higher order derivative, and another for the case of $k$-dimensional parameter space, 
from the viewpoint of calculus of variations~\cite{AK5}. 
The latter was referred to as Areal space. 
In either case, the theory was presented in means of local expressions.  
In this thesis, we will make an original exposition of Kawaguchi geometry by 
using multivector bundle and differential forms, and extend its validity to global 
expressions for the second order $1$-dimensional parameter space and first order $k$-dimensional parameter space.  
The higher order $k$-dimensional parameter space is left for future research. 

In the case of Finsler geometry, the definition for the Finsler-Hilbert form 
was such as it gives an invariant length of a parameterisable $1$-dimensional submanifold. 
Namely, the homogeneity of the Finsler function and the parameterisation invariance of Finsler length 
was equivalent. 
For the Kawaguchi geometry, we can also consider a similar property as the main pillar for 
setting up the foundation. 

\section{Second order, $1$-dimensional parameter space} \label{sec_2nd_Kawaguchi}
Here in this section we will reconstruct the first direction of generalisation of Finsler geometry 
originally considered by Kawaguchi in a modern fashion, 
where the second order derivatives (acceleration) are considered. 
Kawaguchi considered the homogeneity condition of higher order functions by 
requiring the following invariance on 
integration under the change of parameterisation, 
\begin{align}
\int_{{t_1}}^{{t_2}} {K\left( {{x^\mu },\frac{{d{x^\mu }}}{{dt}},\frac{{{d^2}{x^\mu }}}{{d{t^2}}}} \right)dt}  
&= \int_{{t_1}}^{{t_2}} {K\left( {{x^\mu },\frac{{d{x^\mu }}}{{ds}}\frac{{ds}}{{dt}},
\frac{{{d^2}{x^\mu }}}{{d{s^2}}}{{\left( {\frac{{ds}}{{dt}}} \right)}^2} + \frac{{d{x^\mu }}}{{ds}}
\frac{{{d^2}s}}{{d{t^2}}}} \right)dt}  \nonumber \\
&= \int_{{t_1}}^{{t_2}} {K\left( {{x^\mu },\frac{{d{x^\mu }}}{{ds}},\frac{{{d^2}{x^\mu }}}{{d{s^2}}}} \right)ds}.
\end{align}
This requirement gives a condition that the length of a curve defined by such function $K$ becomes invariant, 
in other words it is a geometrical length. 
The above expression is in a single local chart, but we can prove that the condition could be extended globally. 
Below we will give a definition of the manifold with such properties.

\subsection{Basic definitions of Finsler-Kawaguchi geometry} \label{subsec_2nd_kawaguchi} 
We will first define the geometric structure of the second order Finsler-Kawaguchi manifold 
by a function on the total space of 
a second order tangent bundle $({T^2}M, \tau _M^{2,0}, \lb[3] M)$, 
such that gives a geometrical length to a curve on $M$.   
We will call this structure a second order Finsler-Kawaguchi function. 

\begin{defn}  Second order Finsler-Kawaguchi manifold (Second order $1$-dimensional parameter space) \label{def_2ndFinsler} \\
Let $(M,K)$ be a pair of $n$-dimensional ${C^\infty }$-differentiable manifold 
$M$ and a function $K \in {C^\infty }({T^2}M)$, which for a adapted chart on ${T^2}M$, $(V^2,\psi^2 ), 
\psi^2  = ({x^\mu },{y^\mu },{z^\mu })$, $\mu  = 1, \cdots ,n$, satisfies the 
{\it second order homogeneity condition}, 
\begin{align}
K({x^\mu },\lambda {y^\mu },{\lambda ^2}{z^\mu } + \rho {y^\mu }) = \lambda K({x^\mu },{y^\mu },{z^\mu }),
\quad \lambda  \in {\mathbb{R}^ + },\quad \rho  \in \mathbb{R}.  \label{Zermelo_2nd_mech} 
\end{align}
We will call the function with such properties, a {\it second order Finsler-Kawaguchi function}, 
and the pair $(M,K)$ a {\it $n$-dimensional second order Finsler-Kawaguchi manifold}. 
\end{defn}
Compared to the case of first order Finsler, since ${T^2}M$ is not a vector space, 
we do not have an expression such as $\lambda v$, the vector multiplied by a constant. 
The condition (\ref{Zermelo_2nd_mech}) implies the following, 
\begin{align}
\left\{ 
\begin{array}{l}
  \displaystyle{{y^\mu }\frac{{\partial K}}{{\partial {y^\mu }}} 
  + 2{z^\mu }\frac{{\partial K}}{{\partial {z^\mu }}} = K}, \hfill \\
  \displaystyle{{y^\mu }\frac{{\partial K}}{{\partial {z^\mu }}} = 0}. \hfill  	  \label{2nd_mech_Euler}
\end{array} \right.
\end{align}
which is called the {\it Zermelo's condition}. The proofs can be found in~\cite{Krupka-Urban}. 

These conditions (\ref{2nd_mech_Euler}) are coordinate independent. 
Take another chart $(\bar U;\bar \varphi)$, 
$\bar \varphi = ({\bar x^\mu },{\bar y^\mu },{\bar z^\mu })$ on $M$, 
and then from the coordinate transformation rules, we have, 
\begin{align}
&{\bar y^\mu } = \frac{{\partial {{\bar x}^\mu }}}{{\partial {x^\nu }}}{y^\nu }, \nonumber \\
&\frac{{\partial K}}{{\partial {{\bar y}^\mu }}} = \frac{{\partial K}}{{\partial {y^\rho }}}
\frac{{\partial {y^\rho }}}{{\partial {{\bar y}^\mu }}} 
+ \frac{{\partial K}}{{\partial {z^\rho }}}\frac{{\partial {z^\rho }}}{{\partial {{\bar y}^\mu }}} 
= \frac{{\partial K}}{{\partial {y^\rho }}}\frac{{\partial {x^\rho }}}{{\partial {{\bar x}^\mu }}} 
+ 2\frac{{\partial K}}{{\partial {z^\rho }}}\frac{{{\partial ^2}{x^\rho }}}{{\partial {{\bar x}^\alpha }
\partial {{\bar x}^\mu }}}{\bar y^\alpha }, \nonumber \\
&{\bar z^\mu } = \frac{{\partial {{\bar x}^\mu }}}{{\partial {x^\nu }}}{z^\nu } 
+ \frac{{{\partial ^2}{{\bar x}^\mu }}}{{\partial {x^\alpha }\partial {x^\beta }}}{y^\alpha }{y^\beta }, \nonumber \\
&\frac{{\partial K}}{{\partial {{\bar z}^\mu }}} 
= \frac{{\partial K}}{{\partial {z^\rho }}}\frac{{\partial {z^\rho }}}{{\partial {{\bar z}^\mu }}},
\end{align}
and since 
\[\frac{{\partial {{\bar x}^\mu }}}{{\partial {x^\nu }}}\frac{{{\partial ^2}{x^\rho }}}
{{\partial {{\bar x}^\alpha }\partial {{\bar x}^\mu }}}{\bar y^\alpha } 
= \frac{\partial }{{\partial {{\bar x}^\alpha }}}\left( {\frac{{\partial {{\bar x}^\mu }}}{{\partial {x^\nu }}}
\frac{{\partial {x^\rho }}}{{\partial {{\bar x}^\mu }}}{{\bar y}^\alpha }} \right) 
- \frac{{{\partial ^2}{{\bar x}^\mu }}}{{\partial {x^\sigma }\partial {x^\nu }}}
\frac{{\partial {x^\sigma }}}{{\partial {{\bar x}^\alpha }}}\frac{{\partial {x^\rho }}}
{{\partial {{\bar x}^\mu }}}{\bar y^\alpha } =  - \frac{{{\partial ^2}{{\bar x}^\mu }}}
{{\partial {x^\sigma }\partial {x^\nu }}}\frac{{\partial {x^\rho }}}{{\partial {{\bar x}^\mu }}}{y^\sigma },
\] 
we obtain 
\begin{align}
&{\bar y^\mu }\frac{{\partial K}}{{\partial {{\bar y}^\mu }}} 
+ 2{\bar z^\mu }\frac{{\partial K}}{{\partial {{\bar z}^\mu }}} 
= {y^\mu }\frac{{\partial K}}{{\partial {y^\mu }}} 
+ 2{z^\mu }\frac{{\partial K}}{{\partial {z^\mu }}} = K, \nonumber \\
&{\bar y^\mu }\frac{{\partial K}}{{\partial {{\bar z}^\mu }}} 
= {y^\mu }\frac{{\partial K}}{{\partial {z^\mu }}} = 0.
\end{align}
\subsection{Parameterisation invariant length of Finsler-Kawaguchi geometry} \label{subsec_paraminv_2nd_mech}
In Finsler geometry, the homogeneity condition of Finsler function implied the invariance of Finsler length 
and vice versa. We would similarly define the Finsler-Kawaguchi length, 
and then show that the condition (\ref{Zermelo_2nd_mech}) and the parameter independence of Finsler-Kawaguchi length 
is equivalent. We will begin by introducing the lift of a parameterisation. 

\begin{figure}
  \centering
  \includegraphics[width=5cm]{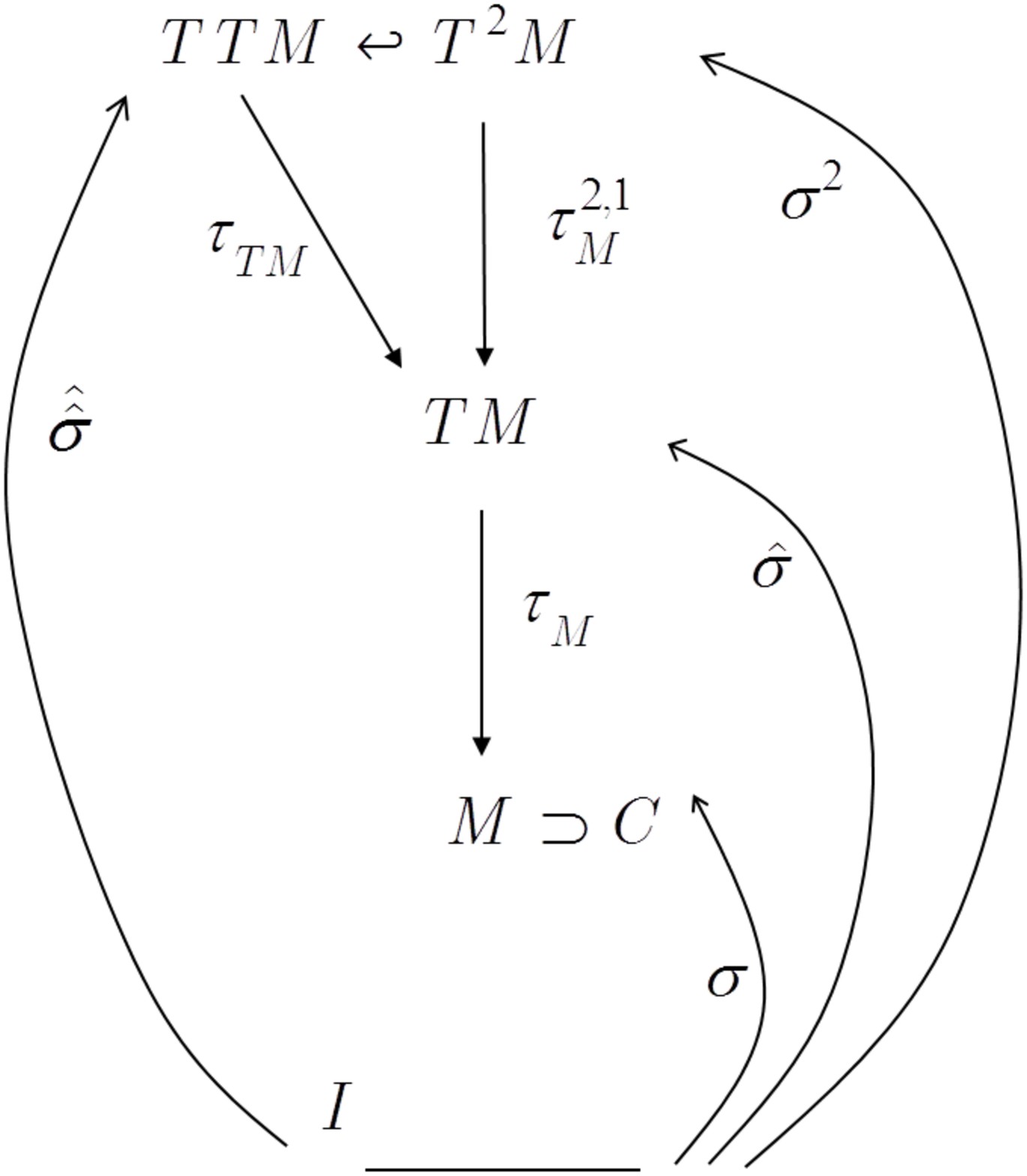}
  \caption{Second order lift of parameterisation}
\end{figure}

\break
\begin{defn} Second order lift of parameterisation \label{def_2nd_param}\\
Consider a second order tangent bundle $({T^2}M,\tau _M^{2,0},M)$ 
defined in Section \ref{subsec_second_TM}, 
and the induced chart $({V^2},{\psi ^2})$, 
${\psi ^2} = ({x^\mu },{y^\mu },{z^\mu })$, $\mu  = 1, \cdots ,n$, on ${T^2}M$. 
Let $\sigma $ be a parameterisation of $C$, namely $C = \sigma (I)$ 
defined in Section \ref{sec_paramFinslerlength}, and 
$I$ be an open interval in $\mathbb{R}$. 
We call the map ${\sigma ^2}:I \to {T^2}M$, such that its local expression is given by 
\begin{align}
{\sigma ^2}(t) = {\left. {{{\left. {\frac{{d({x^\mu } \circ \sigma )}}{{dt}}} \right|}_t}{{\left( {\frac{\partial }{{\partial {x^\mu }}}} \right)}_{\hat \sigma (t)}} + \frac{{{d^2}({x^\mu } \circ \sigma )}}{{d{t^2}}}} \right|_t}
{\left( {\frac{\partial }{{\partial {y^\mu }}}} \right)_{\hat \sigma (t)}}, \label{2nd_orderparameterisation}
\end{align}
the {\it second order lift of parameterisation $\sigma $} to $T^2 M$. 
The image ${{C}^{2}}={{\sigma }^{2}}(I)$ is called the {\it second order lift of $C$}. 
\end{defn}
Clearly, $\tau _{M}^{2,0}({{C}^{2}})=C$. 
The second order lift of parameterisation $\sigma $ is constructed by considering the subset of 
iterated tangent lift. 
Namely, construct the tangent lift $\hat{\hat{\sigma}} :I \to TTM$ of the parameterisation 
$\hat \sigma :I \to TM$, and then take its subset by ${\sigma ^2}: 
= \{ \hat{ \hat{ \sigma}} |{T_{\hat \sigma (t)}}{\tau _M}(\hat{ \hat{ \sigma}} (t)) 
= {\tau _{TM}}(\hat{ \hat{ \sigma}} (t)),t \in I\} $. 
The iterated tangent lift $\hat{ \hat{ \sigma}} (t)$ has the local coordinate expressions 
\begin{align}
\hat{ \hat{ \sigma}} (t) = {T_t}\hat \sigma \left( {\frac{d}{{dt}}} \right) 
= {\left. {\frac{{d({x^\mu } \circ \hat \sigma )}}{{dt}}} \right|_t}
{\left( {\frac{\partial }{{\partial {x^\mu }}}} \right)_{\hat \sigma (t)}} 
+ {\left. {\frac{{d({y^\mu } \circ \hat \sigma )}}{{dt}}} \right|_t}
{\left( {\frac{\partial }{{\partial {y^\mu }}}} \right)_{\hat \sigma (t)}}
\end{align}
and the condition for $\hat{ \hat{ \sigma}} (t)$ to be in ${T^2}M$ 
by the Definition \ref{def_2ndorder_TMoverTM} will give us the coordinates of ${\sigma ^2}(t)$, 
\begin{align}
&  ({x^\mu } \circ {\sigma ^2})(t) = ({x^\mu } \circ \hat \sigma )(t) = ({x^\mu } \circ \sigma )(t), \hfill \nonumber \\
&  ({y^\mu } \circ {\sigma ^2})(t) = {\left. {\frac{{d({x^\mu } \circ \hat \sigma )}}{{dt}}} \right|_t} 
= {\left. {\frac{{d({x^\mu } \circ \sigma )}}{{dt}}} \right|_t} = ({y^\mu } \circ \hat \sigma )(t), \hfill \nonumber \\
&  ({z^\mu } \circ {\sigma ^2})(t) = {\left. {\frac{{d({y^\mu } \circ \hat \sigma )}}{{dt}}} \right|_t} 
= {\left. {\frac{{{d^2}({x^\mu } \circ \sigma )}}{{d{t^2}}}} \right|_t}. \hfill
\end{align}
From above we conclude ${\sigma ^2}(t)$ has the expression (\ref{2nd_orderparameterisation}). 

Let the induced chart on $TTM$ be $({\tilde V^2},{\tilde \psi ^2})$, 
${\tilde \psi ^2} = ({x^\mu },{y^\mu },{\dot x^\mu },{\dot y^\mu })$, 
then the parameterisation $\sigma $ is called {\it second order}, 
when ${\dot y^\mu }(\hat{\hat{\sigma}} (t)) \ne 0$ for $\forall t \in I$. 

In the discussions concerning second order Finsler-Kawaguchi geometry, 
we will only consider the regular parameterisation of second order. 

The $r$-th order parameterisation ${\sigma ^r}:I \to {T^r}M$ can be obtained by iterative process. 
Namely, construct the lift $\widehat {{{(\sigma )}^{r - 1}}}:I \to T{T^{r - 1}}M$ 
of the parameterisation ${\sigma ^{r - 1}}:I \to {T^{r - 1}}M$, 
and then regarding the construction on the higher-order tangent bundle (\ref{higherbundle_r}), 
take its subset by 
\begin{align}
{\sigma ^r}: = \{ \widehat {{{(\sigma )}^{r - 1}}}|{\mkern 1mu} \, 
{T_{{\sigma ^{r - 1}}(t)}}\tau _M^{r - 1,r - 2}(\widehat {{{(\sigma )}^{r - 1}}}(t)) 
= {\iota _{r-1}} \circ {\tau _{{T^{r - 1}}M}}(\widehat {{{(\sigma )}^{r - 1}}}(t)),\,t \in I\}, \label{r_th_param}
\end{align}
where  ${\iota _{r-1}}:{T^{r-1}}M \to T{T^{r - 2}}M$ is the inclusion map. 

\begin{defn}$r$-th order parameterisation \\
Let $\sigma $ be a parameterisation of the curve $C$ on $M$. 
The map ${\sigma ^r}:I \to {T^r}M$ given by (\ref{r_th_param}) 
is called the {\it $r$-th order lift of parameterisation $\sigma$}.
\end{defn}

\begin{defn} Finsler-Kawaguchi length (second order) \\
The Finsler-Kawaguchi function defines a geometrical length for an arc segment 
$C=\sigma([{t_i},{t_f}])$ on $M$ 
by the lifted parameterisation ${\sigma ^2}$ of order $2$ as, 
\begin{align}
{l^K}(C) = \int_{{t_i}}^{{t_f}} {K\left( {{\sigma ^2}(t)} \right)dt} .
\end{align}
By chart expression, this is, 
\begin{align}
  {l^K}(C) &= \int_{{t_i}}^{{t_f}} {K\left( {{x^\mu }({\sigma ^2}(t)),{y^\mu }({\sigma ^2}(t)),{z^\mu }
  ({\sigma ^2}(t))} \right)dt}  \hfill \nonumber \\
  &= \int_{{t_i}}^{{t_f}} {K\left( {{x^\mu }(\sigma (t)),\frac{{d({x^\mu }(\sigma (t)))}}{{dt}},
  \frac{{{d^2}({x^\mu }(\sigma (t)))}}{{d{t^2}}}} \right)dt},  \hfill   \label{FinslerKawaguchiCurve}
\end{align}
where we used the definition of ${\sigma ^2}$, and its coordinates expressed by the induced coordinates of ${T^2}M$, 
\begin{align}
&({x^\mu } \circ {\sigma ^2})(t) = ({x^\mu } \circ \sigma )(t),\; \nonumber \\
&({y^\mu } \circ {\sigma ^2})(t) = {\left. {\frac{{d({x^\mu } \circ \sigma )}}{{dt}}} \right|_t},\; \nonumber \\
&({z^\mu } \circ {\sigma ^2})(t) = {\left. {\frac{{{d^2}({x^\mu } \circ \sigma )}}{{d{t^2}}}} \right|_t}. 
\end{align}
We call this ${l^K}(C)$, the (second order) {\it Finsler-Kawaguchi length} of the curve $C$.
\end{defn}

For the case of first order Finsler geometry, there was an important property of parameterisation 
independence of the Finsler length (Section \ref{sec_paramFinslerlength}). 
We will show that the second order Finsler-Kawaguchi length also has the same property. 
\begin{lemma} Reparameterisation invariance of Finsler-Kawaguchi length \label{lem_repinv_FinslerKawaguchi}\\
The second order Finsler-Kawaguchi length does not change by the reparameterisation 
$\rho  = \sigma  \circ \phi $, $\phi :J \to I$, 
where $\phi $ is a diffeomorphism such that preserves the orientation, 
and $I, J$ are open intervals in $\mathbb{R}$. 
\end{lemma}
\begin{proof}
The second order lift of $\rho $ becomes, 
\begin{align}
{\rho ^2}(s) &= {\left. {{{\left. {\frac{{d({x^\mu } \circ \sigma  \circ \phi )}}{{ds}}} \right|}_s}
{{\left( {\frac{\partial }{{\partial {x^\mu }}}} \right)}_{\widehat {\sigma  \circ \phi }(s)}} 
+ \frac{{{d^2}({x^\mu } \circ \sigma  \circ \phi )}}{{d{s^2}}}} \right|_s}{\left( {\frac{\partial }
{{\partial {y^\mu }}}} \right)_{\widehat {\sigma  \circ \phi }(s)}} \hfill \nonumber \\
&= {\left. {\frac{{d({x^\mu } \circ \sigma )}}{{dt}}} \right|_{\phi (s)}}{\left. {\frac{{d\phi }}{{ds}}} 
\right|_s}{\left( {\frac{\partial }{{\partial {x^\mu }}}} \right)_{\widehat {\sigma  \circ \phi }(s)}} 
+ \frac{d}{{ds}}\left( {{{\left. {\frac{{d({x^\mu } \circ \sigma )}}{{dt}}} \right|}_{\phi (s)}}{{\left. 
{\frac{{d\phi }}{{ds}}} \right|}_s}} \right){\left( {\frac{\partial }{{\partial {y^\mu }}}} 
\right)_{\widehat {\sigma  \circ \phi }(s)}} \hfill \nonumber \\
&= {\left. {\frac{{d({x^\mu } \circ \sigma )}}{{dt}}} \right|_{\phi (s)}}{\left. {\frac{{d\phi }}{{ds}}} 
\right|_s}{\left( {\frac{\partial }{{\partial {x^\mu }}}} \right)_{\widehat {\sigma  \circ \phi }(s)}} 
+ {\left. {\frac{{{d^2}({x^\mu } \circ \sigma )}}{{d{t^2}}}} \right|_{\phi (s)}}{\left( {{{\left. 
{\frac{{d\phi }}{{ds}}} \right|}_s}} \right)^2}{\left( {\frac{\partial }{{\partial {y^\mu }}}} 
\right)_{\widehat {\sigma  \circ \phi }(s)}} \nonumber \\
&+ {\left. {\frac{{d({x^\mu } \circ \sigma )}}{{dt}}} 
\right|_{\phi (s)}}{\left. {\frac{{{d^2}\phi }}{{d{s^2}}}} \right|_s}{\left( {\frac{\partial }
{{\partial {y^\mu }}}} \right)_{\widehat {\sigma  \circ \phi }(s)}}, \hfill  
\end{align}
for $s \in J$.
Its coordinates in ${T^2}M$ are, 
\begin{align}
({x^\mu } \circ {\rho ^2})(s) &= ({x^\mu } \circ \sigma )(\phi (s)) = ({x^\mu } \circ {\sigma ^2})(\phi (s)), \hfill \nonumber \\
({y^\mu } \circ {\rho ^2})(s) &= \left( {{{\left. {\frac{{d({x^\mu } \circ \sigma )}}{{dt}}} \right|}_{\phi ( \cdot )}}
\frac{{d\phi }}{{ds}}} \right)(s) = {\left. {\frac{{d\phi }}{{ds}}} \right|_s} \cdot ({y^\mu } \circ {\sigma ^2})(\phi (s)), \hfill \nonumber \\
({z^\mu } \circ {\rho ^2})(s) &= \left( {{{\left. {\frac{{{d^2}({x^\mu } 
\circ \sigma )}}{{d{t^2}}}} \right|}_{\phi ( \cdot )}}{{\left( {\frac{{d\phi }}{{ds}}} \right)}^2} 
+ {{\left. {\frac{{d({x^\mu } \circ \sigma )}}{{dt}}} \right|}_{\phi ( \cdot )}}\frac{{{d^2}\phi }}{{d{s^2}}}} 
\right)(s) \hfill \nonumber \\
&= {\left( {{{\left. {\frac{{d\phi }}{{ds}}} \right|}_s}} \right)^2} 
\cdot ({z^\mu } \circ {\sigma ^2})(\phi (s)) + {\left. {\frac{{{d^2}\phi }}{{d{s^2}}}} 
\right|_s} \cdot ({y^\mu } \circ {\sigma ^2})(\phi (s)) \hfill 
\end{align}
and since $\rho $ is a regular parameterisation that preserves orientation, 
$\displaystyle{\frac{{d\phi }}{{ds}} > 0}$. 

Let $[{{s}_{i}},{{s}_{f}}] \subset J$ and $[{{t}_{i}},{{t}_{f}}]\subset I$ 
be closed intervals, where $\phi ({{s}_{i}})={{t}_{i}},\,\, \phi ({{s}_{f}})={{t}_{f}}$. 
Now we can see that the length of $C = \sigma ([{t_i},{t_f}])$ is preserved by
\begin{align}
&{l^K}(C) = \int_{{s_i}}^{{s_f}} {K\left( {{\rho ^2}(s)} \right)ds}  
= \int_{{s_i}}^{{s_f}} {K\left( {{x^\mu }({\rho ^2}(s)),{y^\mu }({\rho ^2}(s)),{z^\mu }({\rho ^2}(s))} \right)ds}  \hfill \nonumber \\
& = \int_{{s_i}}^{{s_f}} {K\left( {{x^\mu }({\sigma ^2}(\phi (s))),
{{\left. {\frac{{d\phi }}{{ds}}} \right|}_s}{y^\mu }({\sigma ^2}(\phi (s))),
{{\left( {{{\left. {\frac{{d\phi }}{{ds}}} \right|}_s}} \right)}^2}{z^\mu }({\sigma ^2}(\phi (s))) 
+ {{\left. {\frac{{{d^2}\phi }}{{d{s^2}}}} \right|}_s}{y^\mu }({\sigma ^2}(\phi (s)))} \right)ds}  \hfill \nonumber \\
& = \int_{{\phi^{-1}(t_i)}}^{{\phi^{-1}(t_f)}} {K({\sigma ^2}(\phi (s)))\,\frac{{d\phi }}{{ds}}(s)ds}  \hfill \nonumber \\
& = \int_{{t_i}}^{{t_f}} {K({\sigma ^2}(t)))\,dt} . \hfill  \label{FinslerKawaguchi_length}
\end{align}
\end{proof} 
In the second line of (\ref{FinslerKawaguchi_length}), 
we have used the pull-back of second order homogeneity condition (\ref{Zermelo_2nd_mech}). 
We can conclude that the homogeneity of $K$ and parameterisation invariance of Finsler-Kawaguchi length 
is an equivalent property. 

\subsection{Finsler-Kawaguchi form}
	Given a Finsler-Kawaguchi manifold $(M,K)$, we can obtain an important geometrical structure, 
	which we call a {\it Finsler-Kawaguchi form}. 
	It is a form which is constructed in accord to the second order homogeneity condition, 
	and gives a conventional Lagrangian when pulled back to the parameter space by a certain parameterisation. 
	As such as the Hilbert form reduced to Cartan form by choosing a specific fibration, 
	the Finsler-Kawaguchi form is expected to give a form which would correspond to a 
	``Higher order Cartan form'' in the similar context, by fixing the fibration. 
	\footnote{This is not equivalent to the same higher order Cartan form which are derived in the setting of jet bundles.}

\begin{defn} Second order Finsler-Kawaguchi form \\
Let $(V^2,\psi^2),\psi^2  = ({x^\mu },{y^\mu },{z^\mu })$, $\mu  = 1, \cdots ,n$ be a chart on ${T^2}M$. 
The {\it second order Finsler-Kawaguchi form} ${\mathcal{K}}$ is a $1$-form on ${T^2}M$, 
which in local coordinates are expressed by 
\begin{align}
{\mathcal{K}} = \frac{{\partial K}}{{\partial {y^\mu }}}d{x^\mu } 
+ 2\frac{{\partial K}}{{\partial {z^\mu }}}d{y^\mu }.   \label{KawaguchiForm_2nd_mech}
\end{align}
This corresponds to the first formula in (\ref{2nd_mech_Euler}). 
\end{defn}
The Finsler-Kawaguchi form is invariant with respect to the coordinate transformations by, 
\begin{align}
&{x^\mu } \to {\bar x^\mu } = {\bar x^\mu }({x^\nu }), \nonumber \\
&{y^\mu } \to {\bar y^\mu }({x^\mu },{y^\mu }) 
= \frac{{\partial {{\bar x}^\mu }}}{{\partial {x^\nu }}}{y^\nu }, \nonumber \\
&{z^\mu } \to {\bar z^\mu }({x^\mu },{y^\mu },{z^\mu }) 
= \frac{{\partial {{\bar x}^\mu }}}{{\partial {x^\nu }}}{z^\nu } + 
\frac{{{\partial ^2}{{\bar x}^\mu }}}{{\partial {x^\alpha }\partial {x^\beta }}}{y^\alpha }{y^\beta }
\end{align}

Note that from the second formula of (\ref{Zermelo_2nd_mech}), we can similarly consider a form 
$\displaystyle{\frac{{\partial K}}{{\partial {z^\mu }}}d{x^\mu }}$, which is also coordinate independent. 
Adding such form to (\ref{KawaguchiForm_2nd_mech}) does not contribute to the conventional 
Lagrangian or equation of motion which its only dynamical variable is the time $t$, 
since it becomes $0$ when pulled back to the parameter space $P$. 
Still, it changes the Lagrangian and equation of motion before the pull-back, 
giving some ambiguity in the choice of such forms. 
Nevertheless, for the arbitrary order of derivatives, only one condition of Zermelo relates the derivatives of 
$K$ in the formula to $K$, and we can always use this condition to construct the Finsler-Kawaguchi form. 

\begin{pr} 
Let ${\mathcal{K}}$ be the second order Finsler-Kawaguchi $1$-form on ${T^2}M$, 
$I$ an open interval in $\mathbb{R}$, $[{t_i},{t_f}] \subset I$, and ${\sigma ^2}$ the second order lifted 
parameterisation ${\sigma ^2}:I \to {T^2}M$. 
The integration of $\mathcal{K}$ along ${C^2}:={\sigma ^2}([t_i,t_f])$ is given by the 
Finsler-Kawaguchi length ${l^K}(C)$. 
\end{pr}
\begin{proof}
The simple calculation leads, 
\begin{align}
\int_{{C^2}} {\mathcal{K}} &= \int_{{\sigma ^2}(I)} 
{\frac{{\partial K}}{{\partial {y^\mu }}}d{x^\mu } 
+ 2\frac{{\partial K}}{{\partial {z^\mu }}}d{y^\mu }}  
= \int_{{t_i}}^{{t_f}} {\frac{{\partial K}}{{\partial {y^\mu }}} 
\circ {\sigma ^2}\,d({x^\mu } \circ {\sigma ^2}) 
+ 2\frac{{\partial K}}{{\partial {z^\mu }}} \circ {\sigma ^2}\,d({y^\mu } \circ {\sigma ^2})}  \hfill \nonumber \\
&= \int_{{t_i}}^{{t_f}} {\frac{{\partial K}}{{\partial {y^\mu }}}({\sigma ^2}(t))\,
\frac{{d({x^\mu }({\sigma ^2}(t))}}{{dt}} + 2\frac{{\partial K}}{{\partial {z^\mu }}}({\sigma ^2}(t))\,
\frac{{d({y^\mu }({\sigma ^2}(t))}}{{dt}}dt}  \hfill \nonumber \\
&= \int_{{t_i}}^{{t_f}} {\frac{{\partial F}}{{\partial {y^\mu }}}({\sigma ^2}(t))\,{y^\mu }({\sigma ^2}(t)) 
+ 2\frac{{\partial K}}{{\partial {z^\mu }}}({\sigma ^2}(t))\,{z^\mu }({\sigma ^2}(t))dt}  \hfill \nonumber \\
&= \int_{{t_i}}^{{t_f}} {K({\sigma ^2}(t))\,dt}  = {l^K}(C), \hfill  
\end{align}
where we used the pull-back second order homogeneity condition 
\begin{align}
\frac{{\partial K}}{{\partial {y^\mu }}} \circ {\sigma ^2} \cdot {y^\mu } \circ {\sigma ^2} 
+ 2\frac{{\partial K}}{{\partial {z^\mu }}} \circ {\sigma ^2} \cdot {z^\mu } \circ {\sigma ^2} 
= K \circ {\sigma ^2}.
\end{align}
\end{proof}

\begin{remark}
As we redefined the pair $(M,\user2{\mathcal{F}})$ as the Finsler manifold instead of the pair $(M,F)$, 
we can redefine the pair $(M,\user2{\mathcal{K}})$ as the $n$-dimensional 
Finsler-Kawaguchi manifold instead of the pair $(M,K)$. 
\end{remark}

\begin{remark}
We saw in the previous chapter, the relation between Hilbert form and the 
Cartan form (Remark \ref{rem_HilbertCartan}). 
Similary, we can consider what the Finsler-Kawaguchi form corresponds to, when specifying the fibration.
Let $Y$ be a $n+1$-dimensional manifold, and fix the bundle $(Y, \pi, \mathbb{R})$. 
The second order prolongation of $Y$ is denoted as $J^2 Y$. 
Let $(V, \psi) $, $\psi  = (t,{x^i})$, $i = 1, \cdots , n$ 
be the adapted chart on $Y$, and the induced chart on $\mathbb{R}$ be 
$(\pi(V),t)$.
We denote the induced chart on ${J^2}Y$ by $({({\pi ^{2,0}})^{-1}}(V),{\psi^2}), 
{\psi^2} = (t,{q^i},{\dot q^i}, \ddot{q}^i )$, 
where $\pi^{2,0}: J^2 Y \rightarrow Y$ is the prolongation of $\pi$.  
Suppose we have a second order Finsler-Kawaguchi form ${\mathcal{K}}$ on ${T^2}Y \subset TTY$.
Take the induced chart on $T^2 Y$ as 
$({({{\tau _Y^{2,0}}})^{-1}}(V), {\tilde \psi^2}),\;{\tilde \psi ^2} 
= ({x^0},{x^i},{y^0},{y^i},\dot{x}^0,\dot{x}^i,z^0,z^i)$, $i = 1, \cdots ,n$. 
Consider an inclusion map 
$\iota :{({\pi ^{2,0}})^{ - 1}}(V) \hookrightarrow {({{\tau _Y^{2,0}}})^{ - 1}}(V)$, 
which in coordinates are given by  
\[\iota :(t,{q^i},{\dot q^i},{\ddot q^i}) 
\hookrightarrow ({x^0} = t,{x^i} = {q^i},{y^0} = 1,{y^i} = {\dot q^i}, \dot{x}^0 = 1, \dot{x}^i = {\dot q^i}, 
z^0 = 0, z^i = \ddot{q}^i ), \]
so that the submanifold equations are now given by ${y^0} = \dot{x}^0 = 1$, $y^i = \dot{x}^i$, $z^0 = 0$. 
Expressing the Finsler-Kawaguchi form in these coordinates gives, 
\begin{align}
{\mathcal{K}} = \frac{{\partial K}}{{\partial {y^0}}}d{x^0} + \frac{{\partial K}}{{\partial {y^i}}}d{x^i} 
+ 2\frac{{\partial K}}{{\partial {z^0}}}d{y^0} + 2\frac{{\partial K}}{{\partial {z^i}}}d{y^i}, 
\end{align}
and using the homogeneity conditions: 
\begin{align}
\frac{{\partial K}}{{\partial {y^0}}} &= \left( {K - \frac{{\partial K}}{{\partial {y^i}}}{y^i} 
- 2\frac{{\partial K}}{{\partial {z^0}}}{z^0} - 2\frac{{\partial K}}{{\partial {z^i}}}{z^i}} \right)
\frac{1}{{{y^0}}}, \nonumber \\
\frac{{\partial K}}{{\partial {z^0}}} &=  - \frac{{\partial K}}{{\partial {z^i}}}{z^i}\frac{1}{{{y^0}}},
\end{align}
becomes 
\begin{align}
{\mathcal{K}} = \left( {K - \frac{{\partial K}}{{\partial {y^i}}}{y^i} - 2\frac{{\partial K}}{{\partial {z^0}}}{z^0} 
- 2\frac{{\partial K}}{{\partial {z^i}}}{z^i}} \right)\frac{1}{{{y^0}}}d{x^0} 
+ \frac{{\partial K}}{{\partial {y^i}}}d{x^i} - \frac{{\partial K}}{{\partial {z^i}}}{z^i}\frac{1}{{{y^0}}}d{y^0} 
+ 2\frac{{\partial K}}{{\partial {z^i}}}d{y^i}. \nonumber \\
\end{align}
The pull-back of this Finsler-Kawaguchi form by $\iota$ is 
\begin{align}
&  {\iota ^*}{\mathcal{K}} = \frac{1}{{{y^0}}} \circ \iota 
  \left( {K \circ \iota  - \frac{{\partial K}}{{\partial {y^i}}} \circ \iota  
  \cdot {y^i} \circ \iota  - 2\frac{{\partial K}}{{\partial {z^0}}} \circ \iota  
  \cdot {z^0} \circ \iota  - 2\frac{{\partial K}}{{\partial {z^i}}} \circ \iota  
  \cdot {z^i} \circ \iota } \right)d({x^0} \circ \iota ) \hfill \nonumber \\
&  \quad  + \frac{{\partial K}}{{\partial {y^i}}} \circ \iota \, d ({x^i} \circ \iota ) 
  - \frac{1}{{{y^0}}} \circ \iota  \cdot \frac{{\partial K}}{{\partial {z^i}}} 
  \circ \iota  \cdot {z^i} \circ \iota \,d({y^0} \circ \iota ) 
  + 2\frac{{\partial K}}{{\partial {z^i}}} \circ \iota \,d({y^i} \circ \iota ). \hfill 
\end{align}
By the submanifold equations, this becomes
\begin{align}
&  {\Theta _K}: = {\iota ^*}\user2{\mathcal{K}} = \left( {\user2{\mathcal{L}} 
- \frac{{\partial {\mathcal{L}}}}{{\partial {{\dot q}^i}}}{{\dot q}^i} 
- 2\frac{{\partial {\mathcal{L}}}}{{\partial {{\ddot q}^i}}}{{\ddot q}^i}} \right)dt 
+ \frac{{\partial {\mathcal{L}}}}{{\partial {{\dot q}^i}}}d{q^i} 
+ 2\frac{{\partial {\mathcal{L}}}}{{\partial {{\ddot q}^i}}}d{{\dot q}^i} \hfill \nonumber \\
&   = \user2{\mathcal{L}}dt + \frac{{\partial \user2{\mathcal{L}}}}{{\partial {{\dot q}^i}}}{\omega ^i} 
+ 2\frac{{\partial {\mathcal{L}}}}{{\partial {{\ddot q}^i}}}{{\dot \omega }^i}, \hfill 
\end{align}
where we set ${\mathcal{L}}: = K \circ \iota$, and
\begin{align}
{\omega ^i}: = d{q^i} - {\dot q^i}dt,\;{\dot \omega ^i}: = d{\dot q^i} - {\ddot q^i}dt.
\end{align}
The $1$-forms ${\omega ^i},\;{\dot \omega ^i}$ are called contact forms, 
and they disappear when pulled back to the base manifold of the bundle $(Y, \pi, \mathbb{R})$. 
${\Theta _K}$ is the second order form which corresponds to the Cartan form in our context. 
\end{remark}

\section{First order, $k$-dimensional parameter space } \label{sec_1st_k_Kawaguchi}
Here in this section we will consider the second direction of generalisation of Finsler geometry, 
to $k$-dimensional parameter space. 
We will begin with the first order case. For the construction, 
we will utilise the structure of multivector bundles, 
introduced in chapter \ref{chap_2} Section \ref{subsec_multivectorbundles}. 

\subsection{Basic definitions of Kawaguchi geometry (first order $k$-multivector bundle)}

We will first define the geometric structure on the total space of a $k$-multivector bundle 
$({\Lambda ^k}TM,{\Lambda ^k}{\tau _M},M)$. We will call this structure a first order $k$-areal Kawaguchi function. 
\begin{defn} Kawaguchi manifold (First order $k$-dimensional parameter space) \\
Let $(M,K)$ be a pair of $n$-dimensional ${C^\infty }$-differentiable manifold $M$ and a function 
$K \in {C^\infty }({\Lambda ^k}TM)$ with $k \leqslant n$ that satisfies the following homogeneity condition, 
\begin{align}
K(\lambda v) = \lambda K(v),\quad \lambda  > 0, \, \, \mbox{for} \, \, v \in {\Lambda ^k}TM.
\end{align}
We will call the function with such properties, a {\it first order $k$-areal Kawaguchi function}, 
and the pair $(M,K)$ a {\it $n$-dimensional $k$-areal Kawaguchi manifold}, 
or simply {\it Kawaguchi manifold}, if the subject of discussion is clear. 
\end{defn}
Let $(V, \psi ),\psi  = ({x^\mu },{y^{{\mu _1} \cdots {\mu _k}}})$, 
$\mu ,{\mu _1}, \cdots ,{\mu _k} = 1, \cdots ,n$ be a chart on ${\Lambda ^k}TM$, 
then the local expression of the above condition can be written as 
\begin{align}
K({x^\mu },\lambda {y^{{\mu _1} \cdots {\mu _k}}}) = 
\lambda K({x^\mu },{y^{{\mu _1} \cdots {\mu _k}}}), \quad \lambda  > 0.  \label{cond-k-hom}
\end{align}

The condition (\ref{cond-k-hom}) implies the following, 
\begin{align}
\frac{1}{{k!}}\frac{{\partial K}}{{\partial {y^{{\mu _1} \cdots {\mu _k}}}}}{y^{{\mu _1} \cdots {\mu _k}}} = K,
\label{1st_k_Euler}
\end{align}
Which corresponds to the Euler's homogeneous theorem. 

\subsection{Parameterisation invariant $k$-area of Kawaguchi geometry} \label{subsec_paraminv_k_kawaguchi}
In this section we will define the object $k$-curves, $k$-patchs, 
their parameterisations and Kawaguchi area. 
Kawaguchi area is invariant with respect to the reparameterisation we describe in the following, 
and the notion is naturally extended to $k$-dimensional immersed submanifolds. 
 Similarly as in the case of Finsler length and Finsler-Kawaguchi length, 
the reparameterisation invariance of Kawaguchi area is due to the homogeneity of $k$-areal Kawaguchi function. 

\begin{figure}
  \centering
  \includegraphics[width=5cm]{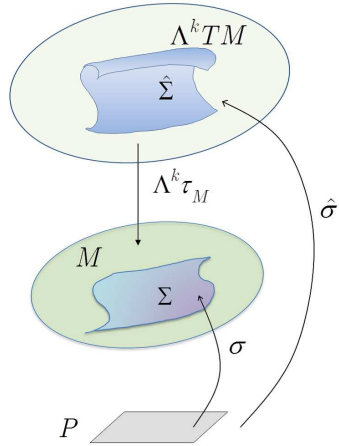}
  \caption{lift of parameterisation for Kawaguchi area}
\end{figure}

\begin{defn} ${C^r}$-$k$-curves \label{def_cr_kcurve}\\
Let $M$ be a smooth manifold, and $\sigma :P \to M$ a ${C^r}$-mapping, 
where $P$ is an open rectangle of ${\mathbb{R}^k}$. 
We denote the image of $P$ by $\Sigma $; $\Sigma : = \sigma (P) \subset M$, 
and call $\Sigma $, the {\it ${C^r}$-$k$-curve} on $M$. 
\end{defn}

\begin{defn} Lift of ${C^r}$-$k$-curves \label{def_lift_curve}\\
Consider a $k$-multivector bundle $({\Lambda ^k}TM,{\Lambda ^k}{\tau _M},M)$ where ${\Lambda ^k}{\tau _M}$ is the natural projection,  
and $(V, \psi )$, $\psi =({{x}^{\mu }},{{y}^{{{\mu }_{1}} \cdots {{\mu }_{k}}}})$, 
$\mu ,{\mu _1}, \cdots ,{\mu _k} = 1, \cdots ,n$ be the induced chart on ${\Lambda ^k}TM$, 
and $({t^1},{t^2}, \cdots ,{t^k})$ the global chart of ${\mathbb{R}^k}$. 
Let $\sigma :P \to M$ be a ${C^r}$-mapping where $P$ 
is an open rectangle of ${\mathbb{R}^k}$. 
Consider a map $\hat \sigma :P \to {\Lambda ^k}TM$, such that its image is denoted by $\hat \Sigma : = \hat \sigma (P) \subset {\Lambda ^k}TM$, 
where ${{\Lambda }^{k}}{{\tau }_{M}}(\hat{\Sigma })=\Sigma $,  
and its coordinate expressions given by 
\begin{align}
\hat{\sigma }(t):={{T}_{t}}\sigma \left( \frac{\partial }{\partial {{t}^{1}}} \right)\wedge 
\cdots \wedge {{T}_{t}}\sigma \left( \frac{\partial }{\partial {{t}^{k}}} \right)
={{\left. \frac{\partial ({{x}^{{{\mu }_{1}}}}\circ \sigma )}{\partial {{t}^{1}}} \right|}_{t}}{{\left. \cdots \frac{\partial ({{x}^{{{\mu }_{k}}}}\circ \sigma )}{\partial {{t}^{k}}} \right|}_{t}}{{\left( \frac{\partial }{\partial {{x}^{{{\mu }_{1}}}}}\wedge \cdots \wedge \frac{\partial }{\partial {{x}^{{{\mu }_{k}}}}} \right)}_{\sigma (t)}}
\label{lift_1st_k_param}
\end{align} 
for $t \in P$. 
The map $\hat \sigma $ and its image $\hat \Sigma $ is called the 
{\it lift} or the {\it multi-tangent lift } of $\sigma $ (resp. $\Sigma $). 
\end{defn}
In coordinate charts, (\ref {lift_1st_k_param}) is expressed as 
\begin{align}
({{x}^{\mu }}\circ \hat{\sigma })(t)=({{x}^{\mu }}\circ \sigma )(t), 
({{y}^{{{\mu }_{1}}\cdots {{\mu }_{k}}}}\circ \hat{\sigma })(t)
={{\left. \frac{\partial ({{x}^{[{{\mu }_{1}}}}\circ \sigma )}{\partial {{t}^{1}}} 
\right|}_{t}}\cdots {{\left. \frac{\partial ({{x}^{{{\mu }_{k}}]}}
\circ \sigma )}{\partial {{t}^{k}}} \right|}_{t}}. \label{lift_1st_k_param_coord}
\end{align}

\begin{defn} Regularity of $\sigma $ \\
The ${C^r}$-map $\sigma $ is called {\it regular}, 
if its lift $\hat \sigma $ is nowhere $0$. 
\end{defn}

\begin{defn} Parameterisation of an immersed curve \\
The ${C^r}$-map $\sigma :P \to M$is called an {\it immersion}, 
if its lift $\hat \sigma $ is injective,  
and the image $\Sigma $ is called an {\it immersed $k$-curve}. 
The map $\sigma $ is called a {\it parameterisation} of immersed $k$-curve $C$, 
and $P$ is called a {\it parameter space}, when $\sigma $ is an immersion, 
and preserves orientation.
\end{defn}

\begin{defn} Parameterisation \label{def_parameterisation} \\
Let $\sigma :P \to M$ be a ${C^r}$-map, and $\Sigma  = \sigma (P)$. 
The map $\sigma $ is called a {\it parameterisation} of $\Sigma $, 
and $P$ is called a {\it parameter space}, when $\sigma $ is injective, 
and preserves orientation. 
\end{defn}

\begin{defn} Lift of a parameterisation \label{def_liftedparameterisation_k} \\
We call $\hat \sigma $ the {\it lift of parameterisation} $\sigma $, 
when $\sigma $ is a parameterisation. 
\end{defn} 

Given a $k$-curve $\Sigma $ on $M$, more than one map and open rectangle may exist, 
namely for cases such as $\Sigma  = \sigma (P) = \rho (Q)$, 
where $\sigma ,\,\,\rho $ are the maps, 
and $P, Q$ are the open rectangles in ${\mathbb{R}^k}$. 
We can classify the $k$-curves by considering the properties of these maps. 

\begin{defn} Regular $k$-curve \\ 
Let $\Sigma $ be a ${C^r}$-$k$-curve on $M$. 
$\Sigma $ is called a {\it regular $k$-curve} on $M$, 
if there exists a regular ${C^r}$ map $\sigma $ 
and an open rectangle $P$ of ${\mathbb{R}^k}$ such that 
$\sigma (P) = \Sigma $. 
\end{defn}

\begin{defn} Parameterisable immersed $k$-curve \\ 
Let $\Sigma $ be a ${C^r}$-$k$-curve on $M$. 
$\Sigma $ is called a {\it parameterisable immersed $k$-curve} on $M$, 
if there exists an immersion $\sigma $ 
and an open rectangle $P$ of ${\mathbb{R}^k}$ such that 
$\sigma (P) = \Sigma $.  
\end{defn}

\begin{defn} Parameterisable curve \\ 
Let $\Sigma $ be a ${C^r}$-$k$-curve on $M$. 
$\Sigma $ is called a {\it parameterisable $k$-curve} on $M$, 
if there exists an injective ${C^r}$-map $\sigma $ 
and an open rectangle $P$ of  ${\mathbb{R}^k}$ such that 
$\sigma (P) = \Sigma $.  
\end{defn}

Occasionally, we implicitly refer to the pair $(\sigma, P)$ by the parameterisation $\sigma $. 
In the following discussion, we will only consider regular, parameterisable $k$-curves. 

Now we will introduce the concept of a $k$-dimensional area of a $k$-curve 
by integration on the parameter space. 
For simplicity, we will restrict ourselves to curves that are parameterisable, 
and consider its closed subset, which we define below. 

\begin{defn} $k$-patch \label{def_kpatch} \\ 
Let $\tilde \Sigma $ be a parameterisable $k$-curve on $M$ with some parameterisation 
$\sigma $.
A subset of $\tilde \Sigma $ given by 
$\Sigma : = \sigma ([t_i^1,t_f^1] \times [t_i^2,t_f^2] \times  
\cdots  \times [t_i^k,t_f^k]) \subset \tilde \Sigma $, 
where $[t_i^1,t_f^1] \times [t_i^2,t_f^2] \times  \cdots  \times [t_i^k,t_f^k] \subset P$ 
is called the {\it $k$-patch} on $M$, 
$\sigma $ is called the {\it parameterisation of the $k$-patch}
and the closed rectangle $[t_i^1,t_f^1] \times [t_i^2,t_f^2] \times  \cdots  \times [t_i^k,t_f^k]$ 
is called the {\it parameter space of the $k$-patch}. 
\end{defn}

The $k$-areal Kawaguchi function defines a geometrical area for a $k$-patch $\Sigma$ on $M$. 

\begin{defn} Kawaguchi $k$-area  \\
Let $(M, K)$ be the $n$-dimensional $k$-areal Kawaguchi manifold, 
and $\Sigma$ the $k$-patch on $M$ such that 
$\Sigma =\sigma ([t_{i}^{1},t_{f}^{1}]\times [t_{i}^{2},t_{f}^{2}]\times \cdots \times [t_{i}^{k},t_{f}^{k}])$. 
We assign to $\Sigma$ the following integral  
\begin{align}
{{l}^{K}}(\Sigma )=\int_{t_{i}^{1}}^{t_{f}^{1}}{d{{t}^{1}}
\int_{t_{i}^{2}}^{t_{f}^{2}}{d{{t}^{2}}\cdots \int_{t_{i}^{k}}^{t_{f}^{k}}{d{{t}^{k}}K 
\left( \hat{\sigma }(t) \right)}}}.  \label{def_KawaguchiArea_1st}
\end{align}
We call this number ${{l}^{K}}(\Sigma)$ the {\it Kawaguchi area} or {\it Kawaguchi $k$-area} of $\Sigma$.
\end{defn}

Let $(V,\psi )$, $\psi =({{x}^{\mu }}, {{y}^{{{\mu }_{1}}\cdots {{\mu }_{k}}}})$, 
$\mu ,{{\mu }_{1}},\cdots ,{{\mu }_{k}}=1,\cdots ,n$ 
be the induced chart on ${{\Lambda }^{k}}TM$.
By chart expression, (\ref{def_KawaguchiArea_1st}) is, 
\begin{align}
& {{l}^{K}}(\Sigma )=\int_{t_{i}^{1}}^{t_{f}^{1}}{d{{t}^{1}}
\int_{t_{i}^{2}}^{t_{f}^{2}}{d{{t}^{2}}\cdots \int_{t_{i}^{k}}^{t_{f}^{k}}{d{{t}^{k}}
K\left( {{x}^{\mu }}
(\hat{\sigma }(t)),{{y}^{{{\mu }_{1}}\cdots {{\mu }_{k}}}}(\hat{\sigma }(t)) \right)}}} \nonumber \\ 
& =\int_{t_{i}^{1}}^{t_{f}^{1}}{d{{t}^{1}}\int_{t_{i}^{2}}^{t_{f}^{2}}{d{{t}^{2}}
\cdots \int_{t_{i}^{k}}^{t_{f}^{k}}{d{{t}^{k}}K\left( {{x}^{\mu }}
(\sigma (t)),\frac{\partial ({{x}^{[{{\mu }_{1}}}}(\sigma (t)))}
{\partial {{t}^{1}}}\cdots 
\frac{\partial ({{x}^{{{\mu }_{k}}]}}(\sigma (t)))}{\partial {{t}^{k}}} \right)}}},  \label{KawaguchiArea}
\end{align} 
where we used the definition of $\hat{\sigma }$, 
and definition of induced coordinates of $\Lambda^k TM$ (\ref{lift_1st_k_param_coord}).  

Let $\rho :Q\to \Sigma $, $Q\subset {{\mathbb{R}}^{k}}$ be another parameterisation of $\Sigma $. 
When there exists a diffeomorphism $\phi :Q\to P$ such that $\rho =\sigma \circ \phi $, 
this gives an equivalence relation $\sigma \sim \rho $. 
As in the case of $1$-dimensional parameter space, in Kawaguchi geometry of 
$k$-dimensional parameter space, the Kawaguchi area defined above has the important property of 
reparameterisation invariance. 

\begin{lemma} Reparameterisation invariance of Kawaguchi $k$-area \label{lem_repinv_kawaguchi_1st}
The Kawaguchi $k$-area does not change by the reparameterisation 
$\rho  = \sigma  \circ \phi $, where $\phi :Q \to P$ is a diffeomorphism, and preserves the orientation. 
\end{lemma}
\begin{proof}
Dividing the rectangle $\bar{P} := [t_{i}^{1},t_{f}^{1}]\times [t_{i}^{2},t_{f}^{2}]\times 
\cdots \times [t_{i}^{k},t_{f}^{k}]$ if necessary into smaller closed sub-rectangles, 
we can suppose without loss of generality that the set 
$\Sigma =\sigma (\bar{P})$ lies in the coordinate neighbourhood of a chart 
$(U,\varphi )$, $\varphi =({{x}^{\mu }})$.
Then the lift of $\rho $ becomes, 
\begin{align}
  \hat \rho (s) &= {\left. {\frac{{\partial ({x^{{\mu _1}}} \circ \sigma  \circ \phi )}}{{\partial {s^1}}}} 
  \right|_s}{\left. { \cdots \frac{{\partial ({x^{{\mu _k}}} \circ \sigma  \circ \phi )}}{{\partial {s^k}}}} 
  \right|_s}{\left( {\frac{\partial }{{\partial {x^{{\mu _1}}}}} \wedge  \cdots  \wedge 
  \frac{\partial }{{\partial {x^{{\mu _k}}}}}} \right)_{\sigma  \circ \phi (s)}} \hfill \nonumber \\
   &= {\left. {\frac{{\partial ({x^{{\mu _1}}} \circ \sigma )}}{{\partial {t^{{a_1}}}}}} \right|_{\phi (s)}} 
   \cdots {\left. {\frac{{\partial ({x^{{\mu _k}}} \circ \sigma )}}{{\partial {t^{{a_k}}}}}} 
   \right|_{\phi (s)}}{\left. {\frac{{\partial ({t^{{a_1}}} \circ \phi )}}{{\partial {s^1}}}} \right|_s} 
   \cdots {\left. {\frac{{\partial ({t^{{a_k}}} \circ \phi )}}{{\partial {s^k}}}} \right|_s}
   {\left( {\frac{\partial }{{\partial {x^{{\mu _1}}}}} \wedge  \cdots  \wedge 
   \frac{\partial }{{\partial {x^{{\mu _k}}}}}} \right)_{\sigma  \circ \phi (s)}} \hfill \nonumber \\
   &= {\varepsilon _{{a_1} \cdots {a_k}}}{\left. {\frac{{\partial ({t^{{a_1}}} \circ \phi )}}
   {{\partial {s^1}}}} \right|_s} \cdots {\left. {\frac{{\partial ({t^{{a_k}}} \circ \phi )}}
   {{\partial {s^k}}}} \right|_s}\hat \sigma (\phi (s)), \nonumber \\ 
\end{align}
for $ s \in Q$, ${a_1}, \cdots ,{a_k} = 1,2, \cdots , k$,  
and since $\rho $ is a regular parameterisation that preserves orientation, 
\begin{align}
{\varepsilon _{{a_1} \cdots {a_k}}}{\left. {\frac{{\partial ({t^{{a_1}}} \circ \phi )}}{{\partial {s^1}}}} 
\right|_s} \cdots {\left. {\frac{{\partial ({t^{{a_k}}} \circ \phi )}}{{\partial {s^k}}}} \right|_s} > 0.
\end{align}
The $k$-dimensional area of $\Sigma $ is preserved by
\begin{align}
& {{l}^{K}}(\Sigma )=\int_{s_{i}^{1}}^{s_{f}^{1}}{d{{s}^{1}}
\int_{s_{i}^{2}}^{s_{f}^{2}}{d{{s}^{2}}\cdots \int_{s_{i}^{k}}^{s_{f}^{k}}{d{{s}^{k}}K
\left( \hat{\rho }(s) \right)}}} \nonumber \\ 
& \quad =\int_{s_{i}^{1}}^{s_{f}^{1}} ds^1
\cdots \int_{s_{i}^{k}}^{s_{f}^{k}} d s^k K\left( {{\varepsilon }_{{{a}_{1}}
\cdots {{a}_{k}}}}{{\left. \frac{\partial ({{t}^{{{a}_{1}}}}\circ \phi )}{\partial {{s}^{1}}} \right|}_{s}}
\cdots {{\left. \frac{\partial ({{t}^{{{a}_{k}}}}\circ \phi )}{\partial {{s}^{k}}} 
\right|}_{s}}\hat{\sigma }(\phi (s)) \right) \nonumber \\ 
& \quad =\int_{{{\phi }^{-1}}(s_{i}^{1})}^{{{\phi }^{-1}}(s_{f}^{1})}
\cdots \int_{{{\phi }^{-1}}(s_{i}^{k})}^{{{\phi }^{-1}}(s_{f}^{k})} ds^1
\wedge \cdots \wedge d s^k{{\varepsilon }_{{{a}_{1}}
\cdots {{a}_{k}}}}{{\left. \frac{\partial ({{t}^{{{a}_{1}}}}\circ \phi )}{\partial {{s}^{1}}} \right|}_{s}}
\cdots {{\left. \frac{\partial ({{t}^{{{a}_{k}}}}\circ \phi )}{\partial {{s}^{k}}} \right|}_{s}}K
\left( \hat{\sigma }(\phi (s)) \right) \nonumber \\ 
& \quad =\int_{t_{i}^{1}}^{t_{f}^{1}}{\cdots 
\int_{t_{i}^{k}}^{t_{f}^{k}}{d{{t}^{1}}\wedge d t^2 \wedge 
\cdots \wedge d{{t}^{k}}K\left( \hat{\sigma }(t) \right)}} \nonumber \\
& \quad =\int_{t_{i}^{1}}^{t_{f}^{1}}{dt^1}
\cdots \int_{t_{i}^{k}}^{t_{f}^{k}}{dt^k}K
\left( \hat{\sigma }(t) \right), \label{Kawaguchi_area}
\end{align}
where $s_{i}^{1},s_{f}^{1},s_{i}^{2},s_{f}^{2}, 
\cdots, s_{i}^{k},s_{f}^{k}$ are the pre-image of the boundary points 
$t_{i}^{1},t_{f}^{1},t_{i}^{2},t_{f}^{2},\cdots,\lb[3] t_{i}^{k},t_{f}^{k}$ by $\phi$.  
For the third equality of (\ref{Kawaguchi_area}), 
we have used the pulled back homogeneity condition of $K$, 
and the definition of integration of 
$k$-form in accord to Section \ref{sec_integrationofd-forms}. 
\end{proof} 
We conclude that the homogeneity of $K$ and parameterisation invariance of Kawaguchi $k$-area 
is an equivalent property. 

\begin{remark}
As similarly in the case of a curve, the ``Kawaguchi area'' does not have the properties of a
``standard'' area, considered by Euclid or Riemannian geometry, 
since we require only homogeneity condition of the Kawaguchi function. 
For instance, when one changes the orientation of the $k$-curve, in general, 
it gives different values (not just signatures). 
Nevertheless, in our following discussion of the calculus of variations, 
we can still use this concept to obtain extremals and equations of motion. 
It is especially designed to use for the applications for extensions of field theory, 
which unifies spacetime and fields. 
Such situation appears commonly in modern theoretical physics, 
and Kawaguchi area will be a good geometrical object to consider for constructing 
viable models. 
\end{remark}

\subsection{Kawaguchi $k$-form} \label{subsec_Kawaguchi_kform}
	Given a $n$-dimensional $k$-areal Kawaguchi manifold $(M,K)$, 
	we can obtain an important geometrical structure, which we will call a 
	{\it Kawaguchi $k$-form}. Kawaguchi $k$-form is constructed in accord with the homogeneity 
	condition, and gives the Lagrangian of a field theory when pulled back to the parameter space, 
	namely the spacetime, by a certain parameterisation. 

\begin{defn} Kawaguchi $k$-form (first order field theory) \\
Let $(V, \psi )$, $\psi  = ({x^\mu },{y^{{\mu _1} \cdots {\mu _k}}})$, 
$\mu ,{\mu _1}, \cdots ,{\mu _k} = 1, \cdots ,n$ be the induced chart on ${\Lambda ^k}TM$. 
The {\it Kawaguchi $k$-form} ${\mathcal{K}}$ is a $k$-form on ${\Lambda ^k}TM$, 
which in local coordinates are expressed by 
\begin{align}
{\mathcal{K}} = \frac{1}{{k!}}\frac{{\partial K}}{{\partial {y^{{\mu _1} 
\cdots {\mu _k}}}}}d{x^{{\mu _1}}} \wedge  \cdots  \wedge d{x^{{\mu _k}}}. \label{KawaguchiForm_1st_k}
\end{align}
This expression corresponds to the homogeneity condition (\ref{1st_k_Euler}). 
\end{defn}

\begin{pr} 
The Kawaguchi form is invariant with respect to the coordinate transformations. 
\end{pr}
\begin{proof} 
Let $\left( {\bar V, \bar \psi } \right),\bar \psi  
= \left( {{{\bar x}^\mu },{{\bar y}^{{\mu _1} \cdots {\mu _k}}}} \right)$ 
be another chart on ${\Lambda ^k}TM$ with intersection $V \cap \bar V \ne \emptyset $.
Then by the coordinate transformation 
\begin{align}
&{x^\mu } \to {\bar x^\mu } = {\bar x^\mu }({x^\nu }), \nonumber \\
&{y^{{\mu _1} \cdots {\mu _k}}} \to {\bar y^{{\mu _1} \cdots {\mu _k}}}({x^\mu },{y^{{\mu _1} \cdots {\mu _k}}}) 
= \frac{{\partial {{\bar x}^{{\mu _1}}}}}{{\partial {x^{{\nu _1}}}}} \cdots 
\frac{{\partial {{\bar x}^{{\mu _k}}}}}{{\partial {x^{{\nu _k}}}}}{y^{{\nu _1} \cdots {\nu _k}}} ,
\end{align}
we have 
\begin{align}
{\mathcal{K}} &= \frac{1}{{k!}}\frac{{\partial K}}{{\partial {y^{{\mu _1} \cdots {\mu _k}}}}}d{x^{{\mu _1}}} \wedge  
\cdots  \wedge d{x^{{\mu _k}}} = \frac{1}{{k!}}\frac{{\partial K}}{{\partial {{\bar y}^{{\nu _1} \cdots {\nu _k}}}}}
\frac{{\partial {{\bar y}^{{\nu _1} \cdots {\nu _k}}}}}{{\partial {y^{{\mu _1} 
\cdots {\mu _k}}}}}\frac{{\partial {x^{{\mu _1}}}}}{{\partial {{\bar x}^{{\nu _1}}}}} \cdots 
\frac{{\partial {x^{{\mu _k}}}}}{{\partial {{\bar x}^{{\nu _k}}}}}d{{\bar x}^{{\mu _1}}} \wedge  
\cdots  \wedge d{{\bar x}^{{\mu _k}}} \hfill \nonumber \\
&= \frac{1}{{k!}}\frac{{\partial K}}{{\partial {{\bar y}^{{\nu _1} \cdots {\nu _k}}}}} 
   \frac{{\partial {{\bar x}^{{\nu _1}}}}}{{\partial {x^{[{\mu _1}}}}} \cdots \frac{{\partial {{\bar x}^{{\nu _k}}}}}
   {{\partial {x^{{\mu _k}]}}}}\frac{{\partial {x^{{\mu _1}}}}}{{\partial {{\bar x}^{{\nu _1}}}}} \cdots 
   \frac{{\partial {x^{{\mu _k}}}}}{{\partial {{\bar x}^{{\nu _k}}}}}d{{\bar x}^{{\mu _1}}} \wedge  
   \cdots  \wedge d{{\bar x}^{{\mu _k}}} \hfill \nonumber \\
&= \frac{1}{{k!}}\frac{{\partial K}}{{\partial {{\bar y}^{{\mu _1} \cdots {\mu _k}}}}}d{{\bar x}^{{\mu _1}}} \wedge  
\cdots  \wedge d{{\bar x}^{{\mu _k}}}. \hfill 
\end{align}
\end{proof}

\begin{pr} 
Let $\mathcal{K}$ be the Kawaguchi $k$-form on ${{\Lambda }^{k}}TM$, 
$\Sigma=\sigma (\bar P)$ the $k$-patch on $M$, with 
$\bar{P}=[t_{i}^{1},t_{f}^{1}]\times [t_{i}^{2},t_{f}^{2}]\times \cdots \times [t_{i}^{k},t_{f}^{k}]$ 
a closed rectangle in $\mathbb{R}^k$. 
Then,
\begin{align}
\int_{{\hat{\Sigma }}}{\mathcal{K}}={{l}^{K}}(\Sigma ). \label{Kawaguchiform_area_id}
\end{align}
\end{pr}

\begin{proof}
The simple calculation leads, 
\begin{align}
\int_{\hat \Sigma } {\mathcal{K}}  
&= \int_{\hat \sigma (\bar P)} {\frac{1}{{k!}}\frac{{\partial K}}{{\partial {y^{{\mu _1} 
\cdots {\mu _k}}}}}d{x^{{\mu _1} \cdots {\mu _k}}} = } \int_{t_i^1}^{t_f^1} 
{ \cdots \int_{t_i^k}^{t_f^k} {\frac{1}{{k!}}\frac{{\partial K}}{{\partial {y^{{\mu _1} \cdots {\mu _k}}}}} 
\circ \hat \sigma \,d({x^{{\mu _1}}} \circ \hat \sigma )}  \wedge  
\cdots  \wedge d({x^{{\mu _k}}} \circ \hat \sigma )}  \hfill \nonumber \\
&= \int_{t_i^1}^{t_f^1} { \cdots \int_{t_i^k}^{t_f^k} {\frac{1}{{k!}} 
\frac{{\partial K}}{{\partial {y^{{\mu _1} \cdots {\mu _k}}}}}\left( {{x^\mu }(\sigma (t)),
{{\left. {\frac{{\partial ({x^{{\mu _1}}} \circ \sigma )}}{{\partial {t^1}}}} \right|}_t}
{{\left. { \cdots \frac{{\partial ({x^{{\mu _k}}} \circ \sigma )}}
{{\partial {t^k}}}} \right|}_t}} \right)} } {\kern 1pt}  \hfill \nonumber \\
&   \quad \quad \quad \quad \quad \quad \quad \quad \quad \quad \quad \quad \quad  \times 
\frac{{\partial ({x^{{\mu _1}}} \circ \sigma )}}{{\partial {t^1}}} \cdots \frac{{\partial ({x^{{\mu _k}}} 
\circ \sigma )}}{{\partial {t^k}}}d{t^1} \wedge  \cdots  \wedge d{t^k} \hfill \nonumber \\
&= \int_{t_i^1}^{t_f^1} { \cdots \int_{t_i^k}^{t_f^k} {\frac{1}{{k!}}
\frac{{\partial K}}{{\partial {y^{{\mu _1} \cdots {\mu _k}}}}} \left( {{x^\mu }(\sigma (t)),{y^{{\mu _1} 
\cdots {\mu _k}}}(\hat \sigma (t))} \right){y^{{\mu _1} \cdots {\mu _k}}}(\hat \sigma (t))} } {\kern 1pt} d{t^1} \wedge  \cdots  \wedge d{t^k} \hfill \nonumber \\
&= \int_{t_i^1}^{t_f^1} {d{t^1} \cdots \int_{t_i^k}^{t_f^k} {d{t^k}K 
\left( {{x^\mu }(\sigma (t)),\frac{{\partial ({x^{[{\mu _1}}}(\sigma (t))}}{{\partial {t^1}}} 
\cdots \frac{{\partial ({x^{{\mu _k}]}}(\sigma (t))}}{{\partial {t^k}}}} \right)} } {\kern 1pt} 
= {l^K}(\Sigma ) \hfill  
\end{align}
where we used the pull-back homogeneity condition 
\begin{align}
\frac{1}{{k!}}\frac{{\partial K}}{{\partial {y^{{\mu _1} \cdots {\mu _k}}}}} \circ \hat \sigma  
\cdot {y^{{\mu _1} \cdots {\mu _k}}} \circ \hat \sigma  = K \circ \hat \sigma .
\end{align}
\end{proof}

\begin{remark}
This lemma extends the notion of Kawaguchi area given by (\ref{def_KawaguchiArea_1st}). 
Namely, since the Kawaguchi form can be integrated over {\it any} 
$k$-dimensional submanifold of $M$, the identity (\ref{Kawaguchiform_area_id}) 
suggests to extend the integration
over $k$-patches to {\it arbitrary} $k$-dimensional submanifold of $M$.
\end{remark}

\begin{remark} 
Now that we showed that Kawaguchi $k$-form gives the Kawaguchi $k$-area 
(and in a more general situation of a submanifold), we redefine the pair 
$(M,\mathcal{K})$ as the $n$-dimensional $k$-areal Kawaguchi manifold instead of the pair $(M,K)$. 
This is a more geometrical definition of a Kawaguchi manifold, similar as in the case of Finsler geometry. 
\end{remark}

\section{Second order, $k$-dimensional parameter space } \label{sec_2nd_k_dim}
Here in this section, combining the previous two directions of generalisation, 
we will consider the case of second order, $k$-dimensional parameter space. 
However, unlike the previous discussions, 
we will consider only the case of $M=\mathbb{R}^n$, and leave the global construction for the future work. 
In this section, $M=\mathbb{R}^n$ is assumed. 
\subsection{Basic definitions of Kawaguchi space (second order $k$-multivector bundle) }
We will first define the geometric structure on the total space of a second order $k$-multivector bundle 
$({({\Lambda ^k}T)^2}M,{\Lambda ^k}\tau _M^{2,0},M)$ with 
${\Lambda ^k}\tau _M^{2.0} = {\Lambda ^k}{\tau _M} \circ {\Lambda ^k}\tau _M^{2,1}$, 
${\Lambda ^k}\tau _M^{2,1}: = {\left. {{\Lambda ^k}{\tau _{{\Lambda ^k}TM}}} \right|_{{{({\Lambda ^k}T)}^2}M}}$. 
We will call this structure a second order $k$-areal Kawaguchi function. 
For visibility, the multi-index notation (see Section \ref{subsec_2nd_k_bundle}) are used.

\begin{defn}  Kawaguchi space (Second order $k$-dimensional parameter space)
Let $(M,K)$ be a pair of $n$-dimensional Cartesian space $M=\mathbb{R}^n$ and a function 
$K \in {C^\infty }({({\Lambda ^k}T)^2}M)$, $k \leqslant n$, 
which for the induced global chart 
${{\varphi }^{2}}=({{x}^{\mu }}, {{y}^{{{\mu }_{1}} \cdots {{\mu }_{k}}}}, 
{{z}^{{{I}_{1}}{{\nu }_{2}}\cdots {{\nu }_{k}}}}, \lb[3]
{{z}^{{{I}_{1}}{{I}_{2}} {{\nu }_{3}}\cdots {{\nu }_{k}}}}, \cdots , 
{{z}^{{{I}_{1}}{{I}_{2}}\cdots {{I}_{k}}}})$, $\mu ,{{\mu }_{1}}, 
\cdots , 
{{\mu }_{k}},{{\nu }_{2}},\cdots ,
{{\nu }_{k}}=1,\cdots ,n$, 
${{I}_{j}} :=\mu _{j}^{{{i}_{1}}} \cdots \mu _{j}^{{{i}_{k}}}$, 
on ${({\Lambda ^k}T)^2}M$ satisfies the following {\it second order homogeneity condition}, 
\begin{align}
&K \left( x^\mu , \lambda {y^{\mu _1 \cdots \mu _k}}, 
 (\lambda)^2 z^{{I_1}{\nu_2 \cdots \nu_k}} 
+ \lambda^{\nu_2 \cdots \nu_k} y^{I_1} , {(\lambda)^3}{z^{{I_1}{I_2}{\nu _3} \cdots {\nu _k}}} 
+ \lambda^{I_2 \nu_3 \cdots \nu_k} y^{I_1} , \right. \nonumber \\
& \hspace{7cm} \dots , \left.
 (\lambda)^{k+1} z^{{I_1}{I_2} \cdots {I_k}} 
+ \lambda^{I_2 \cdots I_k} y^{I_1} \right) \hfill \nonumber \\
& \quad \mathsmaller{= \lambda K({x^\mu },{y^{{\mu _1} \cdots {\mu _k}}},{z^{{I_1}{\nu _2} \cdots {\nu _k}}},
{z^{{I_1}{I_2}{\nu _3} \cdots {\nu _k}}}, \dots ,{z^{{I_1}{I_2} \cdots {I_k}}})}, \hfill   
\label{cond-2nd-k-hom}
\end{align}
for $\lambda  > 0$, and ${\lambda ^{{\nu _2} \cdots {\nu _k}}}, 
{\lambda ^{I_2 {\nu _3} \cdots {\nu _k}}}, \dots ,{\lambda ^{I_1 \cdots I_k}}$ being arbitrary constants. 
We will call the function with such properties, a {\it second order $k$-areal Kawaguchi function} on $M=\mathbb{R}^n$, 
and the pair $(M,K)$ a {\it second order $n$-dimensional $k$-areal Kawaguchi space}, 
or simply {\it second order Kawaguchi space}, if the subject of discussion is clear.
\end{defn}
As in the case of second order Finsler-Kawaguchi geometry, ${({\Lambda ^k}T)^2}M$ is not a vector space. 
Therefore, we have only second order homogeneity conditions in chart expressions.
This condition (\ref{cond-2nd-k-hom}) implies the following conditions, 
\begin{align}
\left\{ \begin{array}{l}
  \displaystyle{\frac{{\partial K}}{{\partial {y^{{I_1}}}}}{y^{{I_1}}} 
  + \frac{2}{{(k - 1)!}}\frac{{\partial K}}{{\partial {z^{{I_1}{\nu _2} \cdots {\nu _k}}}}}{z^{{I_1}{\nu _2} 
  \cdots {\nu _k}}} 
  + \frac{3}{{2!(k - 2)!}}\frac{{\partial K}}{{\partial {z^{{I_1}{I_2}{\nu _3} 
  \cdots {\nu _k}}}}}{z^{{I_1}{I_2}{\nu _3} \cdots {\nu _k}}} }  \hfill \\
  \displaystyle{\quad  + \cdots  + \frac{k+1}{{k!}}\frac{{\partial K}}{{\partial {z^{{I_1}{I_2} \cdots {I_k}}}}}{z^{{I_1}{I_2} \cdots {I_k}}} = K,} \hfill \\
  \displaystyle{\frac{{\partial K}}{{\partial {z^{{I_1}{\nu _2} \cdots {\nu _k}}}}}{y^{{I_1}}} = 
  \frac{{\partial K}}{{\partial {z^{{I_1}{I_2}{\nu _3} \cdots {\nu _k}}}}}{y^{{I_1}}} 
  =  \cdots  = \frac{{\partial K}}{{\partial {z^{{I_1}{I_2} \cdots {I_k}}}}}{y^{{I_1}}} = 0.} \hfill  
\end{array}  \right. \label{2nd_k_Euler}
\end{align}
However, this expression is not coordinate invariant, and therefore, 
we cannot deduce the corresponding Kawaguchi $k$-form of the second order by just these expressions. 
In this text we will only consider the case of $M=\mathbb{R}^n$. 

\subsection{Parameterisation invariant $k$-area of second order \\Kawaguchi geometry} \label{subsec_paraminv_k_2_kawaguchi}
In this section we will define the $k$-dimensional area, 
which we will call the Kawaguchi area of second order. 
This area is invariant with respect to reparameterisation. 
Similarly as in the previous cases, this is due to the homogeneity of second order $k$-areal Kawaguchi function. 
We will begin by introducing the second order lift of parameterisation. 
In this section, $M=\mathbb{R}^n$ is assumed.
\begin{figure}
  \centering
  \includegraphics[width=5cm]{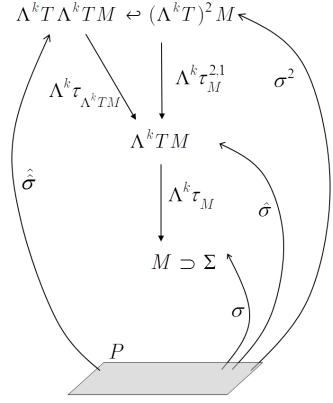}
  \caption{lift of parameterisation for Kawaguchi area}
\end{figure}
\begin{defn} Second order lift of parameterisation \label{def_liftedparameterisation_k_2nd}\\
Consider a second order $k$-multivector bundle $({({\Lambda ^k}T)^2}M,{\Lambda ^k}\tau _M^{2,0},M)$ 
defined in Section \ref{subsec_2nd_k_bundle}, where $M=\mathbb{R}^n$ 
and the induced global chart 
${{\varphi }^{2}}=({{x}^{\mu }}, {{y}^{{{\mu }_{1}} \cdots {{\mu }_{k}}}}, 
{{z}^{{{I}_{1}}{{\nu }_{2}}\cdots {{\nu }_{k}}}}, \lb[3]
{{z}^{{{I}_{1}}{{I}_{2}} {{\nu }_{3}}\cdots {{\nu }_{k}}}}, \cdots , 
{{z}^{{{I}_{1}}{{I}_{2}}\cdots {{I}_{k}}}})$, $\mu ,{{\mu }_{1}}, 
\cdots , 
{{\mu }_{k}},{{\nu }_{2}},\cdots ,
{{\nu }_{k}}=1,\cdots ,n$, 
${{I}_{j}} :=\mu _{j}^{{{i}_{1}}} \cdots \mu _{j}^{{{i}_{k}}}$, 
on ${({\Lambda ^k}T)^2}M$. 
Let $\sigma $ be a parameterisation of $\Sigma$, namely $\Sigma =\sigma(P)$ 
defined in Section \ref{subsec_paraminv_k_kawaguchi}, where $P$ is an open $k$-rectangle.  
We call the map ${\sigma ^2}:P \to {\Sigma ^2} \subset {({\Lambda ^k}T)^2}M$, 
such that its coordinate expression is given by 
\begin{align}
&{\sigma ^2}(t) \!=\! {\left. {\frac{{\partial ({x^{{\mu _1}}} \circ \sigma )}}{{\partial {t^1}}}} 
\right|_t}{\left. { \! \cdots \frac{{\partial ({x^{{\mu _k}}} \circ \sigma )}}{{\partial {t^k}}}} \right|_t}
{\left( {\frac{\partial }{{\partial {x^{{\mu _1}}}}} \wedge  \cdots  \wedge \frac{\partial }
{{\partial {x^{{\mu _k}}}}}} \right)_{\hat \sigma (t)}} \hfill \nonumber \\
&\quad + {\varepsilon ^{{a_1} \cdots {a_k}}}{\left. {\frac{{\partial ({y^{{I_1}}} \circ \hat \sigma )}}
   {{\partial {t^{{a_1}}}}}} \right|_t}{\left. {{{\left. {\frac{{\partial ({x^{{\mu _2}}} \circ \sigma )}}
   {{\partial {t^{{a_2}}}}}} \right|}_t} \cdots \frac{{\partial ({x^{{\mu _k}}} \circ \sigma )}}
   {{\partial {t^{{a_k}}}}}} \right|_t}{\left( {\frac{\partial }{{\partial {y^{{I_1}}}}} \wedge 
   \frac{\partial }{{\partial {x^{{\mu _2}}}}} \wedge  \cdots  \wedge \frac{\partial }
   {{\partial {x^{{\mu _k}}}}}} \right)_{\hat \sigma (t)}} \hfill \nonumber \\
&\quad + {\varepsilon ^{{a_1} \cdots {a_k}}}{\left. {\frac{{\partial ({y^{{I_1}}} \circ \hat \sigma )}}
   {{\partial {t^{{a_1}}}}}} \right|_t}{\left. {\frac{{\partial ({y^{{I_2}}} \circ \hat \sigma )}}
   {{\partial {t^{{a_2}}}}}} \right|_t}{\left. {{{\left. {\frac{{\partial ({x^{{\mu _3}}} \circ \sigma )}}
   {{\partial {t^{{a_3}}}}}} \right|}_t} \cdots \frac{{\partial ({x^{{\mu _k}}} \circ \sigma )}}
   {{\partial {t^{{a_k}}}}}} \right|_t}{\left( {\frac{\partial }{{\partial {y^{{I_1}}}}} \wedge 
   \frac{\partial }{{\partial {y^{{I_2}}}}} \wedge \frac{\partial }{{\partial {x^{{\mu _3}}}}} \wedge  
   \cdots  \wedge \frac{\partial }{{\partial {x^{{\mu _k}}}}}} \right)_{\hat \sigma (t)}} \hfill \nonumber \\
&\quad  + \, \cdots  + {\varepsilon ^{{a_1} \cdots {a_k}}}{\left. {\frac{{\partial ({y^{{I_1}}} \circ \hat \sigma )}}
   {{\partial {t^{{a_1}}}}}} \right|_t}{\left. {\frac{{\partial ({y^{{I_2}}} \circ \hat \sigma )}}
   {{\partial {t^{{a_2}}}}}} \right|_t}{\left. { \cdots \frac{{\partial ({y^{{I_k}}} \circ \hat \sigma )}}
   {{\partial {t^{{a_k}}}}}} \right|_t}{\left( {\frac{\partial }{{\partial {y^{{I_1}}}}} \wedge \frac{\partial }
   {{\partial {y^{{I_2}}}}} \wedge  \cdots  \wedge \frac{\partial }{{\partial {y^{{I_k}}}}}} 
   \right)_{\hat \sigma (t)}}, \hfill \nonumber \\
&\frac{{\partial ({y^{{I_j}}} \circ \hat \sigma )}}{{\partial {t^a}}} 
 = \frac{\partial }{{\partial {t^a}}}\left( {\frac{{\partial ({x^{[\mu _1^j}} \circ \sigma )}}{{\partial {t^1}}} 
 \cdots \frac{{\partial ({x^{\mu _k^j]}} \circ \sigma )}}{{\partial {t^k}}}} \right),\;\;\;a,{a_1}, 
 \cdots ,{a_k},{b_1}, \cdots ,{b_k} = 1, \cdots ,k, \hfill 
\label{2nd_order_k_parameterisation}
\end{align}
the {\it second order lift of parameterisation $\sigma $} to ${({\Lambda ^k}T)^2}M$. 
The image ${{\Sigma}^{2}}={{\sigma }^{2}}(P)$ is called the {\it second order lift of $\Sigma$}. 
\end{defn}
The above second order lift of parameterisation $\sigma $ is constructed similarly as in the case of 
$1$-dimensional parameter space~(Section \ref{subsec_paraminv_2nd_mech}), 
by considering the subset of iterated tangent lift. 
Namely, we first construct the tangent lift $\hat{ \hat{ \sigma}} : P \to {\Lambda ^k}T{\Lambda ^k}TM$ 
of the parameterisation $\hat \sigma :P \to {\Lambda ^k}TM$, and then take its subset by 
${\sigma ^2} := \{ \hat{\hat{\sigma}} |{\Lambda ^k}{T_{\hat \sigma (t)}}{\Lambda ^k}{\tau _M}(\hat{ \hat{ \sigma}} (t)) = {\Lambda ^k}{\tau _{TM}}(\hat{ \hat{ \sigma}} (t)),t \in P\} $.  \label{cond_2_k_bundle}
The iterated tangent lift $\hat{ \hat{ \sigma}} (t)$ has the coordinate expressions, 
\begin{align}
&\hat{\hat{\sigma}} (t) = {{T}_{t}}\hat{\sigma }\left( \frac{\partial }{\partial {{t}^{1}}} \right)\wedge \cdots \wedge {{T}_{t}}\hat{\sigma }\left( \frac{\partial }{\partial {{t}^{k}}} \right) 
= {\left. {\frac{{\partial ({x^{{\mu _1}}} \circ \hat \sigma )}}
{{\partial {t^1}}}} \right|_t}{\left. { \cdots \frac{{\partial ({x^{{\mu _k}}} \circ \hat \sigma )}}
{{\partial {t^k}}}} \right|_t}{\left( {\frac{\partial }{{\partial {x^{{\mu _1}}}}} \wedge  \cdots  \wedge 
\frac{\partial }{{\partial {x^{{\mu _k}}}}}} \right)_{\hat \sigma (t)}} \hfill \nonumber \\
&\quad    + {\varepsilon ^{{a_1} \cdots {a_k}}}{\left. {\frac{{\partial ({y^{{I_1}}} 
\circ \hat \sigma )}}{{\partial {t^{{a_1}}}}}} \right|_t}{\left. {{{\left. {\frac{{\partial ({x^{{\mu _2}}} 
\circ \hat \sigma )}}{{\partial {t^{{a_2}}}}}} \right|}_t} \cdots \frac{{\partial ({x^{{\mu _k}}} 
\circ \hat \sigma )}}{{\partial {t^{{a_k}}}}}} \right|_t}{\left( {\frac{\partial }{{\partial {y^{{I_1}}}}} \wedge 
\frac{\partial }{{\partial {x^{{\mu _2}}}}} \wedge  \cdots  \wedge \frac{\partial }{{\partial {x^{{\mu _k}}}}}} 
\right)_{\hat \sigma (t)}} \hfill \nonumber \\
&\quad    + {\varepsilon ^{{a_1} \cdots {a_k}}}{\left. 
{\frac{{\partial ({y^{{I_1}}} \circ \hat \sigma )}}{{\partial {t^{{a_1}}}}}} \right|_t}{\left. 
{\frac{{\partial ({y^{{I_2}}} \circ \hat \sigma )}}{{\partial {t^{{a_2}}}}}} \right|_t}{\left. 
{{{\left. {\frac{{\partial ({x^{{\mu _3}}} \circ \hat \sigma )}}{{\partial {t^{{a_3}}}}}} \right|}_t} 
\cdots \frac{{\partial ({x^{{\mu _k}}} \circ \hat \sigma )}}{{\partial {t^{{a_k}}}}}} \right|_t}
{\left( {\frac{\partial }{{\partial {y^{{I_1}}}}} \wedge \frac{\partial }{{\partial {y^{{I_2}}}}} \wedge
 \frac{\partial }{{\partial {x^{{\mu _3}}}}} \wedge  \cdots  \wedge \frac{\partial }{{\partial {x^{{\mu _k}}}}}} 
 \right)_{\hat \sigma (t)}} \hfill \nonumber \\
&\quad    +  \cdots  + {\varepsilon ^{{a_1} \cdots {a_k}}}{\left. {\frac{{\partial ({y^{{I_1}}} 
\circ \hat \sigma )}}{{\partial {t^{{a_1}}}}}} \right|_t}{\left. {\frac{{\partial ({y^{{I_2}}} 
\circ \hat \sigma )}}{{\partial {t^{{a_2}}}}}} \right|_t}{\left. { \cdots \frac{{\partial ({y^{{I_k}}} 
\circ \hat \sigma )}}{{\partial {t^{{a_k}}}}}} \right|_t}{\left( {\frac{\partial }{{\partial {y^{{I_1}}}}} \wedge 
\frac{\partial }{{\partial {y^{{I_2}}}}} \wedge  \cdots  \wedge \frac{\partial }{{\partial {y^{{I_k}}}}}} 
\right)_{\hat \sigma (t)}}, \hfill 
\end{align}
and the condition for $\hat{ \hat{ \sigma}} (t)$ to be in ${({\Lambda ^k}T)^2}M$ 
will give us the coordinates of ${\sigma ^2}(t)$, 
\begin{align}
&({x^\mu } \circ {\sigma ^2})(t) = ({x^\mu } \circ \hat \sigma )(t) = ({x^\mu } \circ \sigma )(t), \hfill \nonumber \\
& ({y^{{\mu _1} \cdots {\mu _k}}} \circ {\sigma ^2})(t) = 
{\left. {\frac{{\partial ({x^{[{\mu _1}}} \circ \hat \sigma )}}
{{\partial {t^1}}}} \right|_t} \cdots {\left. {\frac{{\partial ({x^{{\mu _k}]}} \circ \hat \sigma )}}
{{\partial {t^k}}}} \right|_t} = {\left. {\frac{{\partial ({x^{[{\mu _1}}} \circ \sigma )}}
{{\partial {t^1}}}} \right|_t}{\left. { \cdots \frac{{\partial ({x^{{\mu _k}]}} \circ \sigma )}}
{{\partial {t^k}}}} \right|_t} \nonumber \\
&\quad = ({y^{{\mu _1} \cdots {\mu _k}}} \circ \hat \sigma )(t), \hfill \nonumber \\
& ({z^{{I_1}{\nu _2} \cdots {\nu _k}}} \circ {\sigma ^2})(t) = {\varepsilon ^{{a_1} \cdots {a_k}}}
{\left. {\frac{{\partial ({y^{{I_1}}} \circ \hat \sigma )}}{{\partial {t^{{a_1}}}}}} \right|_t}
{\left. {\frac{{\partial ({x^{[{\nu _2}}} \circ \sigma )}}{{\partial {t^{{a_2}}}}}} \right|_t} 
\cdots {\left. {\frac{{\partial ({x^{{\nu _k}]}} \circ \sigma )}}{{\partial {t^{{a_k}}}}}} \right|_t}, \hfill \nonumber \\
&  \quad  = {\varepsilon ^{{a_1} \cdots {a_k}}}\frac{\partial }{{\partial {t^{{a_1}}}}}
{\left( \! {\frac{{\partial ({x^{[\mu _i^1}} \circ \sigma )}}{{\partial {t^1}}} 
\cdots \frac{{\partial ({x^{\mu _k^1]}} \circ \sigma )}}{{\partial {t^k}}}} \! \right)_{\! t}}
{\left. {\frac{{\partial ({x^{[{\nu _2}}} \circ \sigma )}}{{\partial {t^{{a_2}}}}}} \right|_t} 
\cdots {\left. {\frac{{\partial ({x^{{\nu _k}]}} \circ \sigma )}}{{\partial {t^{{a_k}}}}}} \right|_t}, \hfill \nonumber \\
&  ({z^{{I_1}{I_2}{\nu _3} \cdots {\nu _k}}} \circ {\sigma ^2})(t) \hfill \nonumber \\
& \quad = {\varepsilon ^{{a_1} \cdots {a_k}}}\frac{\partial }{{\partial {t^{{a_1}}}}}
{\left( \!{\frac{{\partial ({x^{[\mu _i^1}} 
\circ \sigma )}}{{\partial {t^1}}} \cdots \frac{{\partial ({x^{\mu _k^1]}} \circ \sigma )}}
{{\partial {t^k}}}} \! \right)_{\! t}} \nonumber  \\
&\hspace{5cm} \times \frac{\partial }{{\partial {t^{{a_2}}}}}{\left( \! {\frac{{\partial ({x^{[\mu _i^2}} 
\circ \sigma )}}{{\partial {t^1}}} \cdots \frac{{\partial ({x^{\mu _k^2]}} \circ \sigma )}}
{{\partial {t^k}}}} \! \right)_{\! t}}{\left. {\frac{{\partial ({x^{[{\nu _3}}} \circ \sigma )}}
{{\partial {t^{{a_3}}}}}} \right|_t} \cdots {\left. {\frac{{\partial ({x^{{\nu _k}]}} \circ \sigma )}}
{{\partial {t^{{a_k}}}}}} \right|_t}, \hfill \nonumber \\
& \vdots  \hfill \nonumber \\
& ({z^{{I_1}{I_2} \cdots {I_k}}} \circ {\sigma ^2})(t) 
= {\varepsilon ^{{a_1} \cdots {a_k}}}\frac{\partial }{{\partial {t^{{a_1}}}}}
{\left(\!  {\frac{{\partial ({x^{[\mu _i^1}} \circ \sigma )}}{{\partial {t^1}}} \cdots 
\frac{{\partial ({x^{\mu _k^1]}} \circ \sigma )}}{{\partial {t^k}}}} \! \right)_{\! t}} \cdots 
\frac{\partial }{{\partial {t^{{a_k}}}}}{\left(\!  {\frac{{\partial ({x^{[\mu _i^k}} \circ \sigma )}}
{{\partial {t^1}}} \cdots \frac{{\partial ({x^{\mu _k^k]}} \circ \sigma )}}{{\partial {t^k}}}} \! \right)_{\! t}}, \hfill 
\label{2nd-k-coord}
\end{align}
giving the formula (\ref{2nd_order_k_parameterisation}). 

The parameterisation $\sigma $ where its second order lift ${\sigma ^2}$ is nowhere $0$ is called a 
{\it regular parameterisation of order $2$}. Apparently, if ${\sigma ^2}$ is nowhere $0$, also the tangent lift 
$\hat \sigma $ is nowhere $0$. 
In the discussions concerning second order $k$-dimensional parameter space Kawaguchi geometry, 
we will only consider the regular parameterisation.

The $r$-th order parameterisation ${{\sigma }^{r}}:P\to {{({{\Lambda }^{k}}T)}^{r}}M$ 
can be obtained by iterative process. 
Namely, construct the lift 
$\widehat{{{(\sigma )}^{r-1}}}: P \to {{\Lambda }^{k}}T({{({{\Lambda }^{k}}T)}^{r-1}}M)$ 
of the parameterisation ${{\sigma }^{r-1}}:P\to {{({{\Lambda }^{k}}T)}^{r-1}}M$, 
and then regarding the construction on the higher-order multi tangent bundle (\ref{higherbundle_k_r}), 
take its subset by 
\begin{align}
{{\sigma }^{r}}:=\{\widehat{{{(\sigma )}^{r-1}}}|\,{{\Lambda }^{k}}{{T}_{{{\sigma }^{r-1}}(t)}}
\tau _{{{\Lambda }^{k}}TM}^{r-1,r-2}(\widehat{{{(\sigma )}^{r-1}}}(t))={{\iota }_{r-1}}
\circ {{\tau }_{{{({{\Lambda }^{k}}T)}^{r-1}}M}}(\widehat{{{(\sigma )}^{r-1}}}(t)),\,t\in P\}. \label{r_th_param}
\end{align}
\begin{defn}$r$-th order parameterisation \\
Let $\sigma $ be a parameterisation of the $k$-curve $\Sigma$ on $M$. 
The map ${{\sigma }^{r}}:P\to {{({{\Lambda }^{k}}T)}^{r}}M$ given by (\ref{r_th_param}) is called the {\it $r$-th order lift of parameterisation $\sigma $} .
\end{defn} 

The second order $k$-areal Kawaguchi function defines a geometrical area for a 
$k$-patch $\Sigma$ on $M$. 

\begin{defn} Kawaguchi $k$-area (second order)\\
Let $(M, K)$, $M=\mathbb{R}^n$ be the second order $n$-dimensional $k$-areal Kawaguchi space,
and $\Sigma$ the $k$-patch on $M$ such that 
$\Sigma =\sigma (\bar{P})$, 
$\bar{P}=[t_{i}^{1},t_{f}^{1}]\times [t_{i}^{2},t_{f}^{2}]\times \cdots \times [t_{i}^{k},t_{f}^{k}]$.  
We assign to $\Sigma$ the following integral 
\begin{align}
{l^K}(\Sigma ) = \int_{t_i^1}^{t_f^1} {d{t^1}\int_{t_i^2}^{t_f^2} {d{t^2} \cdots \int_{t_i^k}^{t_f^k} {d{t^k} 
K\left( {{\sigma ^2}(t)} \right)} } }, \quad \quad t \in R   \label{second_KawaguchiArea} 
\end{align}
We call this number ${{l}^{K}}(\Sigma)$ the (second order) {\it Kawaguchi area} 
or {\it Kawaguchi $k$-area} of $\Sigma$.
\end{defn}

Let ${{\varphi }^{2}}=({{x}^{\mu }}, {{y}^{{{\mu }_{1}} \cdots {{\mu }_{k}}}}, 
{{z}^{{{I}_{1}}{{\nu }_{2}}\cdots {{\nu }_{k}}}}, \lb[3]
{{z}^{{{I}_{1}}{{I}_{2}} {{\nu }_{3}}\cdots {{\nu }_{k}}}}, \cdots , 
{{z}^{{{I}_{1}}{{I}_{2}}\cdots {{I}_{k}}}})$, $\mu ,{{\mu }_{1}}, 
\cdots , 
{{\mu }_{k}},{{\nu }_{2}},\cdots ,
{{\nu }_{k}}=1,\cdots ,n$, 
${{I}_{j}} :=\mu _{j}^{{{i}_{1}}} \cdots \mu _{j}^{{{i}_{k}}}$, 
be the induced global chart on ${({\Lambda ^k}T)^2}M$.
Then by chart expression, this is, 
\begin{align}
{l^K}(\Sigma) \! = \!  \int_{t_i^1}^{t_f^1} \! {d{t^1} \cdots \! \int_{t_i^k}^{t_f^k} \! {d{t^k} K 
\! \left( {{x^\mu }({\sigma ^2}(t)),{y^{{\mu _1} \cdots {\mu _k}}}({\sigma ^2}(t)),{z^{{I_1}{\nu _2} 
\cdots {\nu _k}}}({\sigma ^2}(t)), \cdots ,{z^{{I_1} \cdots {I_k}}}({\sigma ^2}(t))} \right)} } 
\end{align}
with the components given by (\ref{2nd-k-coord}). 

Also for the second order, the Kawaguchi $k$-area defined above is reparameterisation invariant. 
\begin{lemma} Reparameterisation invariance of Kawaguchi $k$-area \label{lem_repinv_kawaguchi_2nd} \\
The second order Kawaguchi $k$-area does not change by the reparameterisation 
$\rho  = \sigma  \circ \phi $, 
where $\phi :Q \to P$ is a diffeomorphism, and preserves the orientation. 
\end{lemma}
\begin{proof}
By (\ref{2nd_order_k_parameterisation}), the second order lift of $\rho  = \sigma  \circ \phi $ is, 
\begin{align}
& {\rho ^2}(s) = {\left. {\frac{{\partial ({x^{{\mu _1}}} \circ \sigma  \circ \phi )}}{{\partial {s^1}}}} 
  \right|_s}{\left. { \cdots \frac{{\partial ({x^{{\mu _k}}} \circ \sigma  \circ \phi )}}{{\partial {s^k}}}} 
  \right|_s}{\left( {\frac{\partial }{{\partial {x^{{\mu _1}}}}} \wedge  \cdots  \wedge
   \frac{\partial }{{\partial {x^{{\mu _k}}}}}} \right)_{\widehat {\sigma  \circ \phi }(s)}} \hfill \nonumber \\
& \quad   + {\varepsilon ^{{a_1} \cdots {a_k}}}{\left. {\frac{{\partial ({y^{{I_1}}} \circ 
\widehat {\sigma  \circ \phi })}}{{\partial {s^{{a_1}}}}}} \right|_s}
{\left. {{{\left. {\frac{{\partial ({x^{{\mu _2}}} \circ \sigma  \circ \phi )}}
{{\partial {s^{{a_2}}}}}} \right|}_s} \cdots \frac{{\partial ({x^{{\mu _k}}} \circ \sigma  \circ \phi )}}
{{\partial {s^{{a_k}}}}}} \right|_s}{\left( {\frac{\partial }{{\partial {y^{{I_1}}}}} \wedge 
\frac{\partial }{{\partial {x^{{\mu _2}}}}} \wedge  \cdots  \wedge \frac{\partial }
{{\partial {x^{{\mu _k}}}}}} \right)_{\widehat {\sigma  \circ \phi }(s)}} \hfill \nonumber \\
& \quad + {\varepsilon ^{{a_1} \cdots {a_k}}}{\left. {\frac{{\partial ({y^{{I_1}}} \circ 
\widehat {\sigma  \circ \phi })}}{{\partial {s^{{a_1}}}}}} \right|_s}{\left. {\frac{{\partial ({y^{{I_2}}} 
\circ \widehat {\sigma  \circ \phi })}}{{\partial {s^{{a_2}}}}}} \right|_s}{\left. 
{{{\left. {\frac{{\partial ({x^{{\mu _3}}} \circ \sigma  \circ \phi )}}{{\partial {s^{{a_3}}}}}} 
\right|}_s} \cdots \frac{{\partial ({x^{{\mu _k}}} \circ \sigma  \circ \phi )}}{{\partial {s^{{a_k}}}}}} \right|_s} \nonumber \\
&\hspace{7cm} \times {\left( {\frac{\partial }{{\partial {y^{{I_1}}}}} \wedge \frac{\partial }{{\partial {y^{{I_2}}}}} 
\wedge \frac{\partial }{{\partial {x^{{\mu _3}}}}} \wedge  \cdots  \wedge \frac{\partial }
{{\partial {x^{{\mu _k}}}}}} \right)_{\widehat {\sigma  \circ \phi }(s)}} \hfill \nonumber \\
& \quad  +  \cdots  + {\varepsilon ^{{a_1} \cdots {a_k}}}{\left. 
{\frac{{\partial ({y^{{I_1}}} \circ \widehat {\sigma  \circ \phi })}}{{\partial {s^{{a_1}}}}}} \right|_s}
{\left. { \cdots \frac{{\partial ({y^{{I_k}}} \circ \widehat {\sigma  \circ \phi })}}{{\partial {s^{{a_k}}}}}} 
\right|_s}{\left( {\frac{\partial }{{\partial {y^{{I_1}}}}} \wedge \frac{\partial }{{\partial {y^{{I_2}}}}} \wedge  
\cdots  \wedge \frac{\partial }{{\partial {y^{{I_k}}}}}} \right)_{\widehat {\sigma  \circ \phi }(s)}},
\end{align}
with ${a_1}, \cdots ,{a_k},{b_1}, \cdots ,{b_k} = 1, \cdots ,k$. 
By the chain rule, we have relations such as 
\begin{align}
&\frac{{\partial ({x^{{\mu _1}}} \circ \sigma  \circ \phi )}}{{\partial {s^{{a_1}}}}} 
= \frac{{\partial ({t^c} \circ \phi )}}{{\partial {s^{{a_1}}}}}{\left. 
{\frac{{\partial ({x^{{\mu _1}}} \circ \sigma )}}{{\partial {t^c}}}} 
\right|_{\phi ( \cdot )}}, \nonumber \\
&\frac{{\partial ({y^{{I_j}}} \circ \widehat {\sigma  \circ \phi })}}{{\partial {s^{{a_1}}}}} 
= \mathcal{T} \frac{{\partial ({t^c} \circ \phi )}}{{\partial {s^{{a_1}}}}}
\frac{\partial }{{\partial {t^c}}}({y^{{I_j}}} \circ \hat \sigma ) 
+ \left( {\frac{\partial }{{\partial {s^{{a_1}}}}}\mathcal{T} } \right)({y^{{I_j}}} \circ \hat \sigma ), \label{paramrel_2_k_2}
\end{align}
where we put 
\begin{align}
\mathcal{T} : = {\varepsilon _{{b_1} \cdots {b_k}}}\frac{{\partial ({t^{{b_1}}} \circ \phi )}}{{\partial {s^1}}} 
\cdots \frac{{\partial ({t^{{b_k}}} \circ \phi )}}{{\partial {s^k}}} = {\varepsilon ^{{a_1} 
\cdots {a_k}}}\frac{{\partial ({t^1} \circ \phi )}}{{\partial {s^{{a_1}}}}} \cdots 
\frac{{\partial ({t^k} \circ \phi )}}{{\partial {s^{{a_k}}}}}.
\end{align}
The second relation of (\ref{paramrel_2_k_2}) follows from: 
\begin{align}
\frac{{\partial ({y^{{I_j}}} \circ \widehat {\sigma  \circ \phi })}}
  {{\partial {s^{{a_1}}}}} &= \frac{\partial }{{\partial {s^{{a_1}}}}}
  \left( {\frac{{\partial ({x^{[\mu _1^j}} \circ \sigma  \circ \phi )}}{{\partial {s^1}}} \cdots 
  \frac{{\partial ({x^{\mu _k^j]}} \circ \sigma  \circ \phi )}}{{\partial {s^k}}}} \right) \hfill \nonumber \\
&= \frac{\partial }{{\partial {s^{{a_1}}}}}\left( {\frac{{\partial ({t^{{b_1}}} \circ \phi )}}
   {{\partial {s^1}}} \cdots \frac{{\partial ({t^{{b_k}}} \circ \phi )}}{{\partial {s^k}}}
   {{\left. {\frac{{\partial ({x^{[\mu _1^j}} \circ \sigma )}}{{\partial {t^{{b_1}}}}}} \right|}_{\phi ( \cdot )}} 
   \cdots {{\left. {\frac{{\partial ({x^{\mu _k^j]}} \circ \sigma )}}
   {{\partial {t^{{b_k}}}}}} \right|}_{\phi ( \cdot )}}} \right) \hfill \nonumber \\
&= \frac{{\partial ({t^{{b_1}}} \circ \phi )}}{{\partial {s^1}}} \cdots 
   \frac{{\partial ({t^{{b_k}}} \circ \phi )}}{{\partial {s^k}}}\frac{{\partial ({t^c} \circ \phi )}}
   {{\partial {s^{{a_1}}}}}\frac{\partial }{{\partial {t^c}}}\left( {\frac{{\partial ({x^{[\mu _1^j}} 
   \circ \sigma )}}{{\partial {t^{{b_1}}}}} \cdots \frac{{\partial ({x^{\mu _k^j]}} \circ \sigma )}}
   {{\partial {t^{{b_k}}}}}} \right) \hfill \nonumber \\
&+ \frac{\partial }{{\partial {s^{a_1}}}}\left( {\frac{{\partial ({t^{{b_1}}} \circ \phi )}}
  {{\partial {s^1}}} \cdots \frac{{\partial ({t^{{b_k}}} \circ \phi )}}{{\partial {s^k}}}} \right)
  {\left. {\frac{{\partial ({x^{[\mu _1^j}} \circ \sigma )}}{{\partial {t^{{b_1}}}}}} 
  \right|_{\phi ( \cdot )}} \cdots {\left. {\frac{{\partial ({x^{\mu _k^j]}} \circ \sigma )}}
  {{\partial {t^{{b_k}}}}}} \right|_{\phi ( \cdot )}} \hfill \nonumber \\
&= \mathcal{T} \frac{{\partial ({t^c} \circ \phi )}}{{\partial {s^{{a_1}}}}}\frac{\partial }
   {{\partial {t^c}}}({y^{{I_j}}} \circ \hat \sigma ) + \left( {\frac{\partial }{{\partial {s^{{a_1}}}}}\mathcal{T} } 
   \right)({y^{{I_j}}} \circ \hat \sigma ). \hfill  
\end{align}
Considering each components of ${\rho ^2}$, and using the above relation, 
we can find the relations between the two parameterisations $\rho ,\sigma$, in its coordinate representation:  
\begin{align}
\left\{ \begin{gathered}
  ({x^\mu } \circ {\rho ^2})(s) = ({x^\mu } \circ {\sigma ^2})(\phi (s)) = ({x^\mu } \circ \sigma )(\phi (s)), 
  \hfill \\
  ({y^{{\mu _1} \cdots {\mu _k}}} \circ {\rho ^2})(s) = \mathcal{T} ({y^{{\mu _1} \cdots {\mu _k}}} 
  \circ {\sigma ^2})(\phi (s)), \hfill \\
  ({z^{{I_1}{\mu _2} \cdots {\mu _k}}} \circ {\rho ^2})(s) = {(\mathcal{T} )^2}({z^{{I_1}{\mu _2} 
  \cdots {\mu _k}}} \circ {\sigma ^2})(\phi (s)) + {\beta ^{{\mu _2} \cdots {\mu _k}}}({y^{{I_1}}} 
  \circ {\sigma ^2})(\phi (s)), \hfill \\
  \vdots \hfill  \\ 
 ({{z}^{{{I}_{1}}\cdots {{I}_{l}}{{\mu }_{l+1}}\cdots {{\mu }_{k}}}}
 \circ {{\rho }^{2}})(s)={{(\mathcal{T})}^{l+1}}({{z}^{{{I}_{1}}{{\mu }_{2}}
 \cdots {{\mu }_{k}}}}\circ {{\sigma }^{2}})(\phi (s))+{{\beta }^{{{I}_{2}}\cdots {{I}_{l}}{{\mu }_{l+1}}
 \cdots {{\mu }_{k}}}}({{y}^{{{I}_{1}}}}\circ {{\sigma }^{2}})(\phi (s)), \hfill \\ 
   \vdots  \hfill \\
 ({{z}^{{{I}_{1}}{{I}_{2}}\cdots {{I}_{k}}}}\circ {{\rho }^{2}})(s)
 ={{(\mathcal{T})}^{k+1}}({{z}^{{{I}_{1}}{{I}_{2}}\cdots {{I}_{k}}}} \circ {{\sigma }^{2}})(\phi (s))
 +{{\beta }^{{{I}_{2}}\cdots {{I}_{k}}}}({{y}^{{{I}_{1}}}}\circ {{\sigma }^{2}})(\phi (s)), \hfill \\ 
\end{gathered}  \right.
\label{paramrel_2_k_3}
\end{align}
where we set 
\begin{align}
\left\{ \begin{gathered}
  {\beta ^{{\mu _2} \cdots {\mu _k}}}: = \frac{1}{2}{\varepsilon ^{{a_1} \cdots {a_k}}}
  \frac{\partial }{{\partial {t^{{a_1}}}}}{\left( {{{(\user2{\mathcal{T}})}^2}
  \frac{{\partial ({x^{{\mu _2}}} \circ \sigma )}}{{\partial {t^{{a_2}}}}} \cdots 
  \frac{{\partial ({x^{{\mu _k}}} \circ \sigma )}}{{\partial {t^{{a_k}}}}}} \right)_{\phi ( \cdot )}}, \hfill \\
   \vdots  \hfill \\
  {\beta ^{{I_2} \cdots {I_l}{\mu _{l + 1}} \cdots {\mu _k}}}: = 
  \frac{l}{{l + 1}}{\varepsilon ^{{a_1} \cdots {a_k}}}\frac{\partial }{{\partial {t^{{a_1}}}}}
  {\left( {{{({\mathcal{T}})}^{l + 1}}\frac{{\partial ({y^{{I_2}}} \circ \hat \sigma )}}
  {{\partial {t^{{a_2}}}}} \cdots \frac{{\partial ({y^{{I_l}}} \circ \hat \sigma )}}
  {{\partial {t^{{a_l}}}}}\frac{{\partial ({x^{{\mu _{l + 1}}}} \circ \sigma )}}
  {{\partial {t^{{a_{l + 1}}}}}} \cdots \frac{{\partial ({x^{{\mu _k}}} \circ \sigma )}}
  {{\partial {t^{{a_k}}}}}} \right)_{\phi ( \cdot )}}, \hfill \\
   \vdots  \hfill \\
  {\beta ^{{I_2} \cdots {I_k}}}: = \frac{k}{{k + 1}}{\varepsilon ^{{a_1} 
  \cdots {a_k}}}\frac{\partial }{{\partial {t^{{a_1}}}}}{\left( {{{({\mathcal{T}})}^{k + 1}}
  \frac{{\partial ({y^{{I_2}}} \circ \hat \sigma )}}{{\partial {t^{{a_2}}}}} \cdots 
  \frac{{\partial ({y^{{I_k}}} \circ \hat \sigma )}}{{\partial {t^{{a_k}}}}}} \right)_{\phi ( \cdot )}}. \hfill \\ 
\end{gathered}  \right.
\end{align}
Below we will show some intermediate calculations. For example, 
the second formula of (\ref{paramrel_2_k_3}) is obtained by
\begin{align}
& \left( \! {\frac{{\partial ({x^{[{\mu _1}}} \circ \sigma  \circ \phi )}}{{\partial {s^1}}} 
  \cdots \frac{{\partial ({x^{{\mu _k}]}} \circ \sigma  \circ \phi )}}{{\partial {s^k}}}} \! \right)\!(s) 
  = \frac{{\partial ({t^{{c_1}}} \circ \phi )}}{{\partial {s^1}}} \cdots \frac{{\partial 
  ({t^{{c_k}}} \circ \phi )}}{{\partial {s^k}}}{\left. {\frac{{\partial ({x^{[{\mu _1}}} \circ \sigma )}}
  {{\partial {t^{{c_1}}}}}} \right|_{\phi (s)}} \cdots {\left. {\frac{{\partial ({x^{{\mu _k}]}} \circ \sigma )}}
  {{\partial {t^{{c_k}}}}}} \right|_{\phi (s)}} \hfill \nonumber \\
&  = \mathcal{T} ({y^{{\mu _1} \cdots {\mu _k}}} \circ \sigma )(\phi (s)). \hfill 
\end{align}
The third formula uses (\ref{paramrel_2_k_2}), namely  
\begin{align}
{\varepsilon ^{{a_1} \cdots {a_k}}}\frac{{\partial ({y^{{I_j}}} \circ \widehat {\sigma  \circ \phi })}}
{{\partial {s^{{a_1}}}}}\frac{{\partial ({x^{{\mu _2}}} \circ \sigma  \circ \phi )}}
{{\partial {s^{{a_2}}}}} \cdots \frac{{\partial ({x^{{\mu _k}}} \circ \sigma  \circ \phi )}}
{{\partial {s^{{a_k}}}}} = {(\user2{\mathcal{T}})^2}({z^{{I_j}{\mu _2} \cdots {\mu _k}}} \circ {\sigma ^2}) 
+ {\beta ^{{\mu _2} \cdots {\mu _k}}}({y^{{I_j}}} \circ {\sigma ^2}), \nonumber \\
\end{align}
The above follows from 
\begin{align}
&{\varepsilon ^{{a_1} \cdots {a_k}}}\frac{{\partial ({y^{{I_j}}} \circ \widehat {\sigma  \circ \phi })}}
{{\partial {s^{{a_1}}}}}\frac{{\partial ({x^{{\mu _2}}} \circ \sigma  \circ \phi )}}
{{\partial {s^{{a_2}}}}} \cdots \frac{{\partial ({x^{{\mu _k}}} \circ \sigma  \circ \phi )}}
{{\partial {s^{{a_k}}}}} \hfill \nonumber \\
&  = {\varepsilon ^{{a_1} \cdots {a_k}}}\left( \! {\mathcal{T} \frac{{\partial ({t^{{c_1}}} \circ \phi )}}
{{\partial {s^{{a_1}}}}}\frac{\partial }{{\partial {t^{{c_1}}}}}({y^{{I_j}}} \circ \hat \sigma ) 
+ \left( \! {\frac{\partial }{{\partial {s^{{a_1}}}}}\mathcal{T} }\!  \right)({y^{{I_j}}} \circ \hat \sigma )} \! \right) \nonumber \\
&\hspace{5cm}\times \frac{{\partial ({t^{{c_2}}} \circ \phi )}}{{\partial {s^{{a_2}}}}}{\left. {\frac{{\partial ({x^{{\mu _2}}} 
\circ \sigma )}}{{\partial {t^{{c_2}}}}}} \right|_{\phi ( \cdot )}} \cdots \frac{{\partial ({t^{{c_k}}} 
\circ \phi )}}{{\partial {s^{{a_k}}}}}{\left. {\frac{{\partial ({x^{{\mu _k}}} \circ \sigma )}}
{{\partial {t^{{c_k}}}}}} \right|_{\phi ( \cdot )}} \hfill \nonumber \\
& = \mathcal{T} {\varepsilon ^{{a_1} \cdots {a_k}}}\frac{{\partial ({y^{{I_j}}} \circ \hat \sigma )}}
{{\partial {t^{{c_1}}}}}\frac{{\partial ({t^{{c_1}}} \circ \phi )}}{{\partial {s^{{a_1}}}}} 
\frac{{\partial ({t^{{c_2}}} \circ \phi )}}{{\partial {s^{{a_2}}}}} \cdots \frac{{\partial ({t^{{c_k}}} 
\circ \phi )}}{{\partial {s^{{a_k}}}}}{\left. {\frac{{\partial ({x^{{\mu _2}}} \circ \sigma )}}
{{\partial {t^{{c_2}}}}}} \right|_{\phi ( \cdot )}} \cdots {\left. {\frac{{\partial ({x^{{\mu _k}}} 
\circ \sigma )}}{{\partial {t^{{c_k}}}}}} \right|_{\phi ( \cdot )}} \hfill \nonumber \\
& \quad   + {\varepsilon ^{{a_1} \cdots {a_k}}}\left( \! {\frac{\partial }{{\partial {s^{{a_1}}}}}\mathcal{T} } \! \right) 
({y^{{I_j}}} \circ \hat \sigma )\frac{{\partial ({t^{{c_2}}} \circ \phi )}}{{\partial {s^{{a_2}}}}} 
\cdots \frac{{\partial ({t^{{c_k}}} \circ \phi )}}{{\partial {s^{{a_k}}}}}{\left. {\frac{{\partial ({x^{{\mu _2}}} 
\circ \sigma )}}{{\partial {t^{{c_2}}}}}} \right|_{\phi ( \cdot )}} \cdots {\left. {\frac{{\partial ({x^{{\mu _k}}} 
\circ \sigma )}}{{\partial {t^{{c_k}}}}}} \right|_{\phi ( \cdot )}} \hfill \nonumber \\
&  = (\mathcal{T})^2  \cdot {\varepsilon^{{c_1} \cdots {c_k}}}\frac{{\partial ({y^{{I_j}}} 
\circ \hat \sigma )}}{{\partial {t^{{c_1}}}}}{\left. {\frac{{\partial ({x^{{\mu _2}}} \circ \sigma )}}
{{\partial {t^{{c_2}}}}}} \right|_{\phi ( \cdot )}} \cdots {\left. {\frac{{\partial ({x^{{\mu _k}}} 
\circ \sigma )}}{{\partial {t^{{c_k}}}}}} \right|_{\phi ( \cdot )}} \hfill \nonumber \\
&  \quad  + {\varepsilon ^{{a_1} \cdots {a_k}}}\left( \! {\frac{\partial }{{\partial {s^{{a_1}}}}}\mathcal{T} } \! \right) 
\frac{{\partial ({t^{{c_2}}} \circ \phi )}}{{\partial {s^{{a_2}}}}} \cdots \frac{{\partial ({t^{{c_k}}} \circ \phi )}}
{{\partial {s^{{a_k}}}}}{\left. {\frac{{\partial ({x^{{\mu _2}}} \circ \sigma )}}{{\partial {t^{{c_2}}}}}} 
\right|_{\phi ( \cdot )}} \cdots {\left. {\frac{{\partial ({x^{{\mu _k}}} \circ \sigma )}}{{\partial {t^{{c_k}}}}}} 
\right|_{\phi ( \cdot )}}({y^{{I_j}}} \circ \hat \sigma ) \hfill \nonumber \\
&  = {(\mathcal{T} )^2}({z^{{I_j}{\mu _2} \cdots {\mu _k}}} \circ {\sigma ^2}) \nonumber \nolinebreak[3] 	\\
& \nolinebreak[3] \quad + {\varepsilon ^{{a_1} \cdots {a_k}}}\left( \! {\frac{\partial }{{\partial {s^{{a_1}}}}}\mathcal{T} } \! \right) 
\frac{{\partial ({t^{{c_2}}} \circ \phi )}}{{\partial {s^{{a_2}}}}} \cdots \frac{{\partial 
({t^{{c_k}}} \circ \phi )}}{{\partial {s^{{a_k}}}}}{\left. {\frac{{\partial ({x^{{\mu _2}}} \circ \sigma )}}
{{\partial {t^{{c_2}}}}}} \right|_{\phi ( \cdot )}} \cdots {\left. {\frac{{\partial ({x^{{\mu _k}}} \circ \sigma )}}
{{\partial {t^{{c_k}}}}}} \right|_{\phi ( \cdot )}}({y^{{I_j}}} \circ \hat \sigma ).  
\end{align}
Since we assumed $\rho $ is a regular parameterisation that preserves orientation, 
$\mathcal{T}  > 0$, and we see  
that (\ref{paramrel_2_k_3}) are in the form of the homogeneity conditions (\ref{cond-2nd-k-hom}),  
the second order $k$-dimensional area of $\Sigma $ is preserved by
\begin{align}
\begin{gathered}
  {l^K}(\Sigma ) = \int_{s_i^1}^{s_f^1} {d{s^1}\int_{s_i^2}^{s_f^2} {d{s^2} \cdots \int_{s_i^k}^{s_f^k} {d{s^k}K\left( {{\rho ^2}(s)} \right)} } }  \hfill \\
   = \int_{s_i^1}^{s_f^1} {d{s^1} \cdots \int_{s_i^k}^{s_f^k} {d{s^k}K\left( {{x^\mu }({\sigma ^2}(\phi (s))),{\mathcal{T}}{y^{{\mu _1} \cdots {\mu _k}}}({\sigma ^2}(\phi (s))),} \right.} }  \hfill \\
  \quad {({\mathcal{T}})^2}({z^{{I_1}{\mu _2} \cdots {\mu _k}}} \circ {\sigma ^2})(\phi (s)) + {\beta ^{{\mu _2} \cdots {\mu _k}}}({y^{{I_1}}} \circ {\sigma ^2})(\phi (s)), \hfill \\
  \left. {\quad  \cdots ,{{({\mathcal{T}})}^{k + 1}}({z^{{I_1}{I_2} \cdots {I_k}}} \circ {\sigma ^2})(\phi (s)) + {\beta ^{{I_2} \cdots {I_k}}}({y^{{I_1}}} \circ {\sigma ^2})(\phi (s))} \right) \hfill \\
   = \int_{{\phi ^{ - 1}}(t_i^1)}^{{\phi ^{ - 1}}(t_f^1)} { \cdots \int_{{\phi ^{ - 1}}(t_i^k)}^{{\phi ^{ - 1}}(t_f^k)} {d{s^1} \wedge d{s^2} \wedge  \cdots  \wedge d{s^k}{\mathcal{T}}K\left( {{\sigma ^2}(\phi (s))} \right)} }  \hfill \\
   = \int_{t_i^1}^{t_f^1} { \cdots \int_{t_i^k}^{t_f^k} {d{t^1} \wedge d{t^2} \wedge  \cdots  \wedge d{t^k}K\left( {{\sigma ^2}(t)} \right)} }  \hfill \\
   = \int_{t_i^1}^{t_f^1} {d{t^1}\int_{t_i^2}^{t_f^2} {d{t^2} \cdots \int_{t_i^k}^{t_f^k} {d{t^k}K\left( {{\sigma ^2}(t)} \right)} } } , \hfill \\ 
\end{gathered}
\label{Kawaguchi_area_2nd}
\end{align}
where $s_{i}^{1},s_{f}^{1},s_{i}^{2},s_{f}^{2}, 
\cdots, s_{i}^{k},s_{f}^{k}$ are the pre-image of the boundary points 
$t_{i}^{1},t_{f}^{1},t_{i}^{2},t_{f}^{2},\cdots, \lb[3] t_{i}^{k},t_{f}^{k}$ by $\phi$.  
We have used the homogeneity condition of $K$, and the definition of integration of 
$k$-form in accord to Section \ref{sec_integrationofd-forms}. 
\end{proof} 
We can conclude that in the second order case, 
homogeneity of $K$ and parameterisation invariance of Kawaguchi $k$-area is an equivalent property, 
provided that we are considering the case of $M={{\mathbb{R}}^{n}}$.

\subsection{Second order Kawaguchi $k$-form}
	Now we will turn to defining a second order Kawaguchi $k$-form, 
which we construct to have the same property as the previous cases, 
namely, it should be constructed by referring to the conditions given by (\ref{2nd_k_Euler}); 
and should be equivalent to giving a second order $k$-dimensional area, 
when its pull back is integrated over the parameter space. 
However, since the conditions (\ref{2nd_k_Euler}) depend on coordinates, 
from these alone we cannot construct a global form for general manifolds. 
For this reason we also restrict our model for $M=\mathbb{R}^n$ case, 
and leave the general case for future research. 
Nevertheless, the obtained form could be used for the consideration of second order field theories 
with the restriction of $M=\mathbb{R}^n$.

\begin{defn} Second order Kawaguchi $k$-form \\
Let ${{\varphi }^{2}}=({{x}^{\mu }},{{y}^{{{\mu }_{1}}\cdots {{\mu }_{k}}}},{{z}^{{{I}_{1}};{{\nu }_{2}}\cdots {{\nu }_{k}}}},{{z}^{{{I}_{1}}{{I}_{2}};{{\nu }_{3}}\cdots {{\nu }_{k}}}},\cdots ,{{z}^{{{I}_{1}}{{I}_{2}}\cdots {{I}_{k}}}})$, $\mu ,{{\mu }_{1}},\cdots ,{{\mu }_{k}},{{\nu }_{2}},\cdots ,{{\nu }_{k}}=1,\cdots ,n$, ${{I}_{j}},:=\mu _{j}^{{{i}_{1}}}\cdots \mu _{j}^{{{i}_{k}}}$, 
be the induced global chart on ${{({{\Lambda }^{k}}T)}^{2}}M$, where $M=\mathbb{R}^n$. The {\it second order Kawaguchi $k$-form} $\mathcal{K}$ is a $k$-form on ${{({{\Lambda }^{k}}T)}^{2}}M$, which in coordinates are expressed by 
\begin{align}
&\mathcal{K} = \frac{1}{k!} \frac{\partial K}{{\partial {y^{{\mu _1} \cdots {\mu _k}}}}}d{x^{{\mu _1} 
\cdots {\mu _k}}} + \frac{2}{{(k - 1)!}}\frac{{\partial K}}{{\partial {z^{{I_1}{\nu _2} 
\cdots {\nu _k}}}}}d{y^{{I_1}}} \wedge d{x^{{\nu _2} \cdots {\nu _k}}} \hfill \nonumber \\
&  \quad  + \frac{3}{{2!(k - 2)!}} \frac{{\partial K}}{{\partial {z^{{I_1}{I_2}{\nu _3} \cdots {\nu _k}}}}}d{y^{{I_1}}} \wedge d{y^{{I_2}}} \wedge d{x^{{\nu _3} \cdots {\nu _k}}} 
 +  \cdots  + \frac{k+1}{{k!}}\frac{{\partial K}}{{\partial {z^{I_1 I_2 \cdots I_k }}}}dy^{I_1} \wedge  
 \cdots  \wedge dy^{I_k}. \hfill \label{KawaguchiForm_2nd_k}
\end{align}
\end{defn}
We used the abbreviation such as 
\[d{x^{{\mu _1} \cdots {\mu _k}}}: = d{x^{{\mu _1}}} \wedge  \cdots  \wedge d{x^{{\mu _k}}}, \quad
d{y^{{I_1}}} \wedge d{x^{{\nu _2} \cdots {\nu _k}}}: = d{y^{{I_1}}} \wedge d{x^{{\nu _2}}} \wedge 
\cdots  \wedge d{x^{{\nu _k}}}. \]
This expression (\ref{KawaguchiForm_2nd_k}) corresponds to the first homogeneity condition in (\ref{2nd_k_Euler}). 

As we already mentioned, in general the above form is not invariant 
with respect to the coordinate transformations given by (\ref{2nd_k_coordtrans}).
  
\begin{pr} 
Let $\mathcal{K}$ be the second order Kawaguchi $k$-form on ${{({{\Lambda }^{k}}T)}^{2}}M$, 
$\Sigma=\sigma (\bar P)$ the $k$-patch on $M$, with 
$\bar{P}=[t_{i}^{1},t_{f}^{1}]\times [t_{i}^{2},t_{f}^{2}]\times \cdots \times [t_{i}^{k},t_{f}^{k}]$ 
a closed rectangle in $\mathbb{R}^k$. Then,
\begin{align}
\int_{{{\Sigma }^2}}{\mathcal{K}}={{l}^{K}}(\Sigma ). \label{Kawaguchiform_area_id2}
\end{align}
\end{pr}
\begin{proof}
The simple calculation leads, 
\begin{align}
&  \int_{{\Sigma ^2}} {\mathcal{K}}  = \int_{{\sigma ^2}(P)} {\frac{1}{{k!}}\frac{{\partial K}}
{{\partial {y^{{\mu _1} \cdots {\mu _k}}}}}d{x^{{\mu _1} \cdots {\mu _k}}}}  
+ \int_{{\sigma ^2}(P)} {\frac{2}{{(k - 1)!}}\frac{{\partial K}}{{\partial {z^{{I_1}{\nu _2} 
\cdots {\nu _k}}}}}d{y^{{I_1}}} \wedge d{x^{{\nu _2} \cdots {\nu _k}}}}  \hfill \nonumber  \\
&  \quad  + \int_{{\sigma ^2}(P)} {\frac{3}{{2!(k - 2)!}}\frac{{\partial K}}{{\partial {z^{{I_1}{I_2}{\nu _3} 
\cdots {\nu _k}}}}}d{y^{{I_1}}} \wedge d{y^{{I_2}}} \wedge d{x^{{\nu _3} \cdots {\nu _k}}}}  \hfill \nonumber  \\
&  \quad  +  \cdots  + \int_{{\sigma ^2}(P)} {\frac{k+1}{{k!}}\frac{{\partial K}}{{\partial {z^{{I_1}{I_2} 
\cdots {I_k}}}}}d{y^{{I_1}}} \wedge  \cdots  \wedge d{y^{{I_k}}}}  \hfill \nonumber  \\
&   = \int_{t_i^1}^{t_f^1} { \cdots \int_{t_i^k}^{t_f^k} {\frac{1}{{k!}}\frac{{\partial K}}{{\partial {y^{{\mu _1} 
\cdots {\mu _k}}}}} \circ {\sigma ^2} \cdot d({x^{{\mu _1}}} \circ {\sigma ^2})}  \wedge  
\cdots  \wedge d({x^{{\mu _k}}} \circ {\sigma ^2})}  \hfill \nonumber \\
&  \quad  + \int_{{\sigma ^2}(P)} {\frac{2}{{(k - 1)!}}\frac{{\partial K}}{{\partial {z^{{I_1}{\nu _2} 
\cdots {\nu _k}}}}} \circ {\sigma ^2} \cdot d({y^{{I_1}}} \circ {\sigma ^2}) \wedge d(x^{\nu _2} \circ \sigma^2 ) 
\wedge \cdots \wedge d(x^{\nu _k} \circ \sigma^2) }   \hfill \nonumber \\
&  \quad  +  \cdots  + \int_{{\sigma ^2}(P)} {\frac{k+1}{{k!}}\frac{{\partial K}}{{\partial {z^{{I_1}{I_2} 
\cdots {I_k}}}}} \circ {\sigma ^2}\,d({y^{{I_1}}} \circ {\sigma ^2}) \wedge  \cdots  \wedge d({y^{{I_k}}} 
\circ {\sigma ^2})}  \hfill \nonumber \\
&   = \int_{t_i^1}^{t_f^1} { \cdots \int_{t_i^k}^{t_f^k} \frac{1}{k!}{\frac{{\partial K}}{{\partial {y^{{\mu _1} 
\cdots {\mu _k}}}}}\left( {{\sigma ^2}(t)} \right){y^{{\mu _1} \cdots {\mu _k}}}({\sigma ^2}(t)){\kern 1pt} d{t^1} \wedge  
\cdots  \wedge d{t^k}} }  \hfill \nonumber \\
&  \quad  + \int_{t_i^1}^{t_f^1} { \cdots \int_{t_i^k}^{t_f^k} {\frac{2}{(k-1)!}\frac{{\partial K}}
{{\partial {z^{{I_1}{\nu _2} \cdots {\nu _k}}}}}\left( {{\sigma ^2}(t)} \right){z^{{I_1}{\nu _2} 
\cdots {\nu _k}}}(\hat \sigma (t))d{t^1} \wedge  \cdots  \wedge d{t^k}} } {\kern 1pt}  \hfill \nonumber \\
&  \quad  +  \cdots  + \int_{{\sigma ^2}(P)} {\frac{k+1}{k!}\frac{{\partial K}}{{\partial {z^{{I_1}{I_2} \cdots {I_k}}}}}
\left( {{\sigma ^2}(t)} \right){z^{{I_1}{I_2} \cdots {I_k}}}({\sigma ^2}(t)){\kern 1pt} d{t^1} \wedge  
\cdots  \wedge d{t^k}}  \hfill \nonumber \\
&   = \int_{t_i^1}^{t_f^1} {d{t^1} \cdots \int_{t_i^k}^{t_f^k} {d{t^k}K\left( {{\sigma ^2}(t)} \right)} } 
{\kern 1pt}  \hfill \nonumber \\
&   = {l^K}(\Sigma ) \hfill
\end{align}
where we used the pull-back homogeneity condition 
\begin{align}
& \frac{{\partial K}}{{\partial {y^{{I_1}}}}} \circ {\sigma ^2} \cdot {y^{{I_1}}} 
\circ {\sigma ^2} + \frac{2}{{(k - 1)!}}\frac{{\partial K}}{{\partial {z^{{I_1}{\nu _2} \cdots {\nu _k}}}}} 
\circ {\sigma ^2} \cdot {z^{{I_1}{\nu _2} \cdots {\nu _k}}} \circ {\sigma ^2} \hfill \nonumber \\
&  \quad  + \frac{3}{{2!(k - 2)!}}\frac{{\partial K}}{{\partial {z^{{I_1}{I_2}{\nu _3} 
\cdots {\nu _k}}}}} \circ {\sigma ^2} \cdot {z^{{I_1}{I_2}{\nu _3} \cdots {\nu _k}}} \circ {\sigma ^2} 
+  \cdots  + \frac{k+1}{{k!}}\frac{{\partial K}}{{\partial {z^{{I_1}{I_2} \cdots {I_k}}}}} \circ {\sigma ^2} 
\cdot {z^{{I_1}{I_2} \cdots {I_k}}} \circ {\sigma ^2} \hfill \nonumber \\
&   = K \circ {\sigma ^2}. \hfill 
\end{align}
\end{proof}

\begin{remark}
Similarly as in the case of Finsler manifold and second order Finsler-Kawaguchi manifold, 
we can redefine the pair $(M,\user2{\mathcal{K}})$ as the second order $n$-dimensional 
$k$-areal Kawaguchi manifold, instead of the pair $(M,K)$, where $M=\mathbb{R}^n$. 
\end{remark}

%% file: thesis2012_chap5.tex
\chapter{Lagrangian formulation of Finsler and Kawaguchi geometry} \label{chap_5}

In the preceding chapters, we have prepared the foundations for considering the 
parameterisation invariant theory of calculus of variation. 
Here in this chapter, we will interpret the Finsler length as the Lagrangian, 
and derive the Euler-Lagrange equations and conservation laws 
by considering the calculus of variation. 
We interpret the Finsler manifold as a dynamical system, and in this context 
we will also call the Euler-Lagrange equations the {\it equations of motion}. 
We will begin with the first order mechanics, 
which will be based on Finsler geometry (Chapter 3), 
and then extend it to second order mechanics, based on Finsler-Kawaguchi 
geometry (Chapter 4, Section \ref{sec_2nd_Kawaguchi}), 
and finally to the field theory (first and second order), 
based on Kawaguchi geometry (Chapter 4, Section 
\ref{sec_1st_k_Kawaguchi}, \ref{sec_2nd_k_dim}). 
In all cases, the dynamics of the object (particle, field) is described as a motion of 
a $k$-patch (arc segment) $\Sigma$ of $n$-dimensional manifold $M$. 
Such parameterisation invariant property gives us a freedom of choosing convenient parameters, 
which may reduce the effort to solve the equations of motion, 
and a non-trivial example for Brachistochrone is shown for the case of Finsler.


\section{First order mechanics} \label{sec_1st_mech}
	Here we introduce the theory of first order mechanics, in terms of Finsler geometry. 
By the term {\it first order}, we mean that the total space we are considering is the tangent bundle, 
 and by {\it mechanics}, we mean that we are considering the arc segment on $M$.
 
	The basic structure we consider in this section is introduced in Chapter 2 and 3, 
namely the $n$-dimensional Finsler manifold $(M, {\mathcal{F}})$, 
the tangent bundle $(TM,{\tau_M},M)$, and a $1$-dimensional curve (arc segment) $C$ on $M$, 
parameterised by $\sigma$.  
The curve (arc segment) describes the trajectory of the object on $M$. 
 
In our setting, the Hilbert $1$-form is the Lagrangian, and 
the action will be defined by considering the integration over the 
lift of the parameterisable curve (arc segment) $C$. 
The Euler-Lagrange equations are derived by taking the variation of the action with respect 
to the flow on $M$ that deforms the arc segment $C$, and fixed on the boundary. 
We can show that the action and consequently the Euler-Lagrange equations are 
independent with respect to the parameterisation belonging to the same equivalent class. 

\subsection{Action} 
Suppose we have a dynamical system (differential equations expressing motions) 
where the trajectory of the point particle (or any object which dynamics 
could be considered as a point) is expressed by an arc segment $C$ of a parameterisable curve, 
such that $C = \sigma (I) \subset M$, where $I$ is a closed interval $I = [{t_i},{t_f}] \subset \mathbb{R}$. 

When we can express this system by Finsler geometry, namely the pair $(M, \mathcal{F})$ 
where $\mathcal{F}$ is a Finsler-Hilbert $1$-form, 
we refer to this dynamical system as {\it first order mechanics}, 
and conversely call the pair $(M,\mathcal{F})$ a {\it dynamical system}, 
and $\mathcal{F}$ a {\it Lagrangian}. 
Though this terminology may sound strange or gives an impression of restricted cases of mechanics at the beginning, 
we will show in this section (Remark \ref{convL-Fins-rel}), 
that given a conventional Lagrangian\footnote{We simply use the term 
``conventional'' to distinguish the Lagrangian function on a local chart of $J^1 Y$, 
from our Lagrangian which is the Hilbert $1$-form over $TY$.}, 
one can always construct a Finsler function (and therefore obtain the Finsler-Hilbert $1$-form) 
on the corresponding local chart of $TM$. 

The action of first order mechanics is defined as follows.

%
\begin{defn} Action of first order mechanics \\
Let  $(M, {\mathcal{F}})$ be a $n$-dimensional Finsler manifold, 
$(U,\varphi ),\;\varphi  = ({x^\mu })$ be a chart on $M$, 
and $(V,\psi )$, $V={{\tau }_{M}}^{-1}(U)$, $\psi =({{x}^{\mu }},{{y}^{\mu }})$ 
the induced chart on $TM$. 
The local coordinate expression of the Finsler-Hilbert form 
${\mathcal{F}} \in {\Omega ^1}(TM)$ is given by
$\displaystyle{{\mathcal{F}} = \frac{{\partial F}}{{\partial {y^\mu }}}d{x^\mu }}$, 
where $F$ is the Finsler function. 
Let $C$ be an arc segment on $M$, $\sigma$ its parameterisation, 
$\sigma(I) = C \subset M$ with $I = [{t_i},{t_f}] \subset \mathbb{R}$, 
and $\hat{\sigma}$ the tangent lift of $\sigma$, 
defined in Chapter 3 (Definition \ref{def_liftedparameterisation}).
We call the functional ${S}^{\mathcal{F}}( C )$ defined by 
\begin{eqnarray}
{S}^{\mathcal{F}}( C ) := {l}^{F}(C) = \int_{\hat C} {\mathcal{F}} 
= \int_{\hat \sigma (I)} {\frac{{\partial F}}{{\partial {y^\mu }}}d{x^\mu }},  \label{FinslerAction}
\end{eqnarray}
the {\it action of first order mechanics associated with $\mathcal{F}$.} 
\end{defn}

As we have seen in Section \ref{sec_paramFinslerlength}, Lemma \ref{lem_repinv_Finsler},

Finsler length is invariant with respect to the reparameterisation, 
therefore the action is also invariant.  

\subsection{Extremal and equations of motion} \label{subsec_extremal}

	Having defined the action, we are able to derive the equations of motion by 
considering the extremal of the action. 
To make the discussion simple, we only consider global flows in this thesis. 
Nevertheless, with some details added, the formulation can be set up similarly with local flows. 

Consider a ${C^\infty }$-flow, $\alpha :\mathbb{R} \times M \to M$, and its associated 
$1$-parameter group of transformations ${\{ {\alpha _s}\} _{s \in \mathbb{R}}}$ . 
The $1$-parameter group ${\alpha _s}:M \to M$ induces a tangent 
$1$-parameter group $T{\alpha _s}:TM \to TM$ on $TM$. 
This will also deform the curve (arc segment) $C$ to $C'=\alpha_s (C)$, 
and since this is a smooth deformation, it again becomes a parameterisable curve. 
By the reparameterisation independence, 
we can always choose the parameterisation of this deformed $C'$ by a new $\sigma': I \to M$, $\sigma'(I)=C'$, 
so that it has the same parameter space as $C$.  
The variation of the action will be expressed by the small deformations made to the action by 
${\alpha _s}$. 
\begin{figure}
  \centering
  \includegraphics[width=5cm]{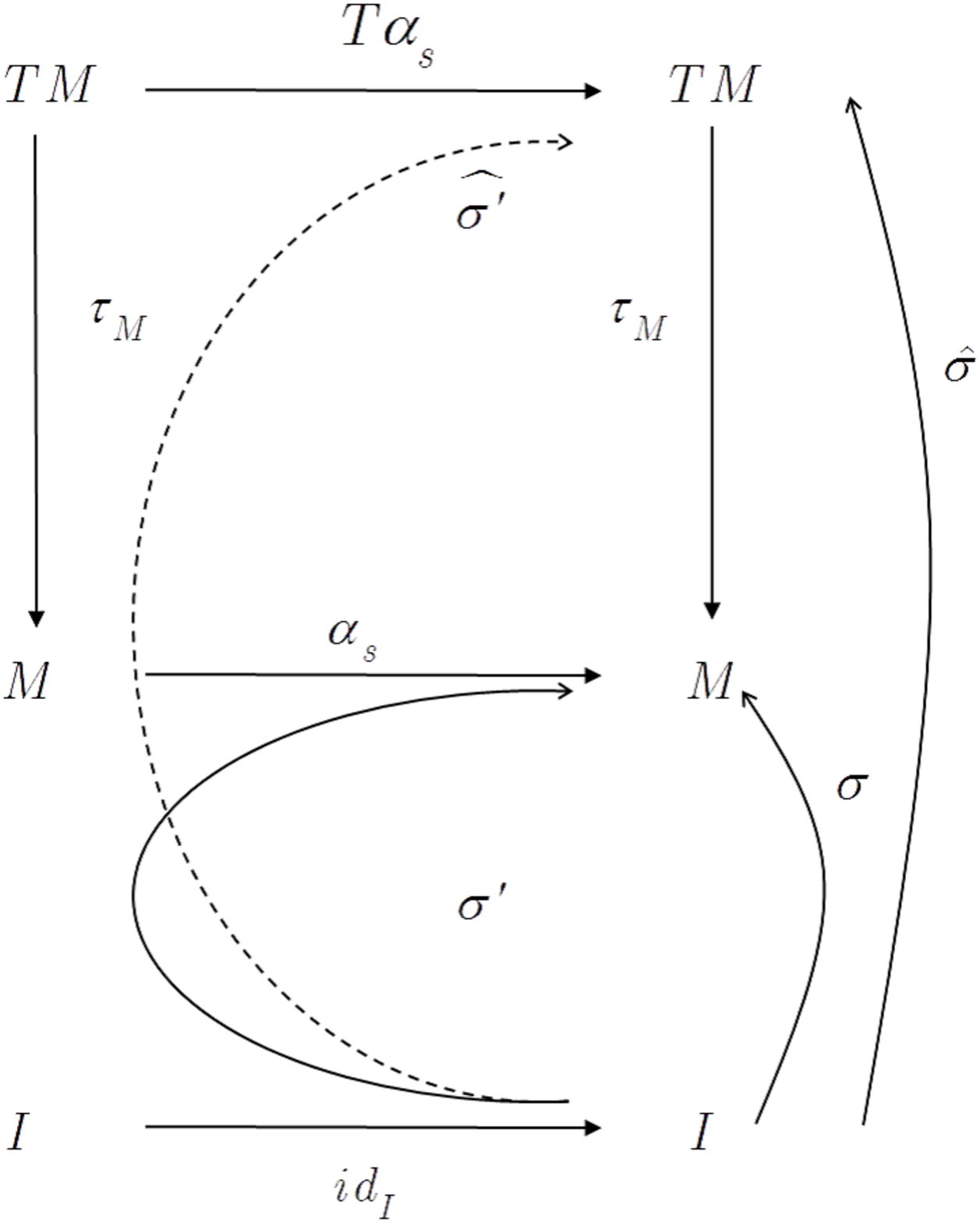}
  \caption{First order mechanics} \label{fig_1st_mech}
\end{figure}

\begin{defn} Variation of the action \\
Let $\xi$ be a vector field on $M$ which generates the $1$-parameter group $\alpha_s$, 
i.e., $\displaystyle{\xi  = {\left. {\frac{{d{\alpha _s}}}{{ds}}} \right|_{s = 0}}}$. 
We call the functional 
\begin{align}
  {\delta _\xi }{S^{\mathcal{F}}}(C)& : = \mathop {\lim }\limits_{s \to 0} 
  \frac{1}{s}\left\{ {{S^\user2{\mathcal{F}}}({\alpha _s}(C)) - {S^{\mathcal{F}}}(C)} \right\} \hfill \nonumber \\
 & = \mathop {\lim }\limits_{s \to 0} \frac{1}{s}\left\{ {\int_{\widehat {{\alpha _s}(\sigma)}(I)} {\mathcal{F}}  
   - \int_{\hat{\sigma}(I)} {\mathcal{F}} } \right\} \hfill 
\end{align}
the {\it variation of the action ${S^{\mathcal{F}}}(C)$ with respect to the flow $\alpha$, 
associated to $\mathcal{F}$ }.
\end{defn}

It is easy to see that the lift of this modified parameterisation 
$\sigma'$ is given by, 
$\hat \sigma ' = \widehat {{\alpha _s} \circ \sigma } 
= T{\alpha _s} \circ \hat \sigma  = T{\alpha _s} \circ \hat \sigma  \circ i{d_I}^{ - 1}$, 
which we show in Figure \ref{fig_1st_mech}.

We will get, 
\begin{eqnarray}
  \delta_\xi {{S}^{\mathcal{F}}}(C) &:=& \mathop {\lim }\limits_{s \to 0} \frac{1}{s}
  \left\{ {\int_{T{\alpha _s} \circ \hat \sigma (I)} {\mathcal{F}}  - \int_{\hat \sigma (I)} 
  {\mathcal{F}} } \right\} = \mathop {\lim }\limits_{s \to 0} \frac{1}{s}\left\{ {\int_{\hat \sigma (I)} 
  {{{(T{\alpha _s})}^*}{\mathcal{F}}}  - \int_{\hat \sigma (I)} {\mathcal{F}} } \right\} \hfill \nonumber \\
   &=& \int_{\hat \sigma (I)} {{L_X}{\mathcal{F}}}  
   = \int_{\hat C} {{L_X}{\mathcal{F}}} . \hfill 
\end{eqnarray}
where $X$ is a vector field on $TM$ that generates the tangent 
$1$-parameter group $T{\alpha _s}$, i.e., 
$\displaystyle{X = {\left. {\frac{{d(T{\alpha _s})}}{{ds}}} \right|_{s = 0}}}$, 
and $L_X$ is a Lie derivative with respect to $X$.

Let us calculate the vector field $X$ in local coordinates. 
As usual, let $(U,\varphi )$, $\varphi  = ({x^\mu })$ be a chart on $M$, 
and the induced chart of $TM$; $(V,\psi )$, $V = {\tau _M}^{ - 1}(U)$, 
$\psi  = ({x^\mu },{y^\mu })$. 
Let $\xi$ be be a generator of the $1$-parameter group $\alpha_s$
and its local expression $\displaystyle{\xi  = {\xi ^\mu }\frac{\partial }{{\partial {x^\mu }}}}$, 
where ${\xi ^\mu } \in {C^\infty }(M)$. 
The tangent map $T{\alpha _s}$ at $p \in M$ sends the vector 
$v \in {T_p}M$ to ${T_{{\alpha _s}(p)}}M$ by 
\begin{align}
{T_p}{\alpha _s}(v) 
= {\left. \frac{{\partial (x^\mu \circ \alpha _s \circ \varphi^{- 1} ) }}{{\partial {x^\nu }}} \right|_{\varphi (p)}}
{v^\nu }{\left( {\frac{\partial }{{\partial {x^\mu }}}} \right)_{{\alpha _s}(p)}},
\end{align}
and since $(T{\alpha _s},{\alpha _s})$ is a bundle morphism and from the definition of induced 
coordinates of a tangent bundle, we have for its coordinate expressions, 
\begin{align}
&{x^\mu } \circ {T_p}{\alpha _s}(v) = {x^\mu } \circ {\alpha _s} \circ {\tau _M}(v), \nonumber  \\
&{y^\mu } \circ {T_p}{\alpha _s}(v) = {\left. 
\frac{{\partial (x^\mu \circ \alpha _s \circ \varphi^{- 1} ) }}{{\partial {x^\nu }}} \right|_{\varphi ({\tau _M}(v))}}{y^\nu }(v). 
\end{align}%
By these observations, the vector field $X$ at a point $q \in TM$ has a local expression, 
\begin{align}
{X_q} &= {\left. {\frac{{d({x^\mu } \circ T{\alpha _s})}}{{ds}}} \right|_{s = 0}}
{\left( {\frac{\partial }{{\partial {x^\mu }}}} \right)_q} 
+ {\left. {\frac{{d({y^\mu } \circ T{\alpha _s})}}{{ds}}} \right|_{s = 0}}
{\left( {\frac{\partial }{{\partial {y^\mu }}}} \right)_q} \hfill \nonumber \\
&= {\left. {\frac{d}{{ds}}({x^\mu } \circ {\alpha _s} \circ {\tau _M})} \right|_{s = 0}}
{\left( {\frac{\partial }{{\partial {x^\mu }}}} \right)_q} + {y^\nu }(q)\frac{d}{{ds}}
{\left( {{{\left. \frac{{\partial (x^\mu \circ \alpha _s \circ \varphi^{- 1} ) }}{{\partial {x^\nu }}} \right|}_{\varphi ({\tau _M}(q))}}} \right)_{s = 0}}{\left( {\frac{\partial }
{{\partial {y^\mu }}}} \right)_q} \hfill \nonumber \\
&= \left( {{\xi ^\mu } \circ {\tau _M}} \right)(q){\left( {\frac{\partial }{{\partial {x^\mu }}}} \right)_q} 
+ \left( {\frac{{\partial {\xi ^\mu }}}{{\partial {x^\nu }}} \circ {\tau _M} \cdot {y^\nu }} \right)(q){
\left( {\frac{\partial }{{\partial {y^\mu }}}} \right)_q}, \hfill 
\end{align}
therefore, 
\begin{align}
X = {\xi ^\mu } \circ {\tau _M}\left( {\frac{\partial }{{\partial {x^\mu }}}} \right) 
+ \frac{{\partial {\xi ^\mu }}}{{\partial {x^\nu }}} \circ {\tau _M} 
\cdot {y^\nu }\left( {\frac{\partial }{{\partial {y^\mu }}}} \right).   \label{pro.v.f.onTM}
\end{align}
We will call $X$, the {\it induced vector field} by $\xi$, on $TM$.    

The Lie derivative ${L_X}{\mathcal{F}}$ in coordinate expression is 
\begin{align}
&{L_X}{\mathcal{F}} = {L_X}\left( {\frac{{\partial F}}{{\partial {y^\rho }}}d{x^\rho }} \right) 
= X\left( {\frac{{\partial F}}{{\partial {y^\rho }}}} \right)d{x^\rho } 
+ \frac{{\partial F}}{{\partial {y^\rho }}}d{L_X}{x^\rho } \hfill \nonumber \\
& \quad  = \left\{ {{\xi ^\mu } \circ {\tau _M}\left( {\frac{{{\partial ^2}F}}
{{\partial {x^\mu }\partial {y^\rho }}}} \right) + \frac{{\partial {\xi ^\mu }}}
{{\partial {x^\nu }}} \circ {\tau _M} \cdot {y^\nu }\left( {\frac{{{\partial ^2}F}}
{{\partial {y^\mu }\partial {y^\rho }}}} \right)} \right\}d{x^\rho } + \frac{{\partial F}}
{{\partial {y^\rho }}}d\left( {{\xi ^\rho } \circ {\tau _M}} \right) \hfill \nonumber  \\
& \quad  = {\xi ^\mu } \circ {\tau _M}\left\{ {\frac{{{\partial ^2}F}}
{{\partial {x^\mu }\partial {y^\rho }}}d{x^\rho } - d\left( {\frac{{\partial F}}
{{\partial {y^\mu }}}} \right)} \right\} + \frac{{\partial {\xi ^\mu }}}{{\partial {x^\nu }}} 
\circ {\tau _M} \cdot {y^\nu }\left( {\frac{{{\partial ^2}F}}{{\partial {y^\mu }\partial {y^\rho }}}} \right)
d{x^\rho } + d\left( {\frac{{\partial F}}{{\partial {y^\rho }}} \cdot {\xi ^\rho } \circ {\tau _M}} \right). \label{Lie_1st_mech}
\end{align}
The result of (\ref{Lie_1st_mech}) is called the {\it infinitesimal first variation formula} for the Hilbert form $\mathcal{F}$.

The variation of action becomes 
\begin{align}
&  {\delta _\xi }{S^\user2{\mathcal{F}}}(C) = \int_{\hat \sigma (I)} {{L_X}\user2{\mathcal{F}}}  
= \int_I {{{\hat \sigma }^*}{L_X}\user2{\mathcal{F}}}  \hfill \nonumber \\
&  \quad  = \int_I {{{\hat \sigma }^*}\left( {{\xi ^\mu } \circ {\tau _M}
\left\{ {\frac{{{\partial ^2}F}}{{\partial {x^\mu }\partial {y^\rho }}}d{x^\rho } 
- d\left( {\frac{{\partial F}}{{\partial {y^\mu }}}} \right)} \right\} 
+ d\left( {\frac{{\partial F}}{{\partial {y^\rho }}} \cdot {\xi ^\rho } \circ {\tau _M}} \right)} \right)}  \hfill \nonumber \\
&  \quad  = \int_{\hat C} { {{\xi ^\mu } \circ {\tau _M}
\left\{ {\frac{{{\partial ^2}F}}{{\partial {x^\mu }\partial {y^\rho }}}d{x^\rho } 
- d\left( {\frac{{\partial F}}{{\partial {y^\mu }}}} \right)} \right\} 
+ d\left( {\frac{{\partial F}}{{\partial {y^\rho }}} \cdot {\xi ^\rho } \circ {\tau _M}} \right)},}  
\hfill \label{variation_1st_mech} 
\end{align}
which is called the {\it integral first variation formula}. 
We have used the homogeneity condition: 
\begin{align}
\left( { {\frac{{{\partial ^2}F}}
{{\partial {y^\mu }\partial {y^\rho }}}} \cdot {y^\rho }} \right) 
\circ \hat \sigma  = 0. \label{cond_hom3}
\end{align}
%
(\ref{cond_hom3}) is obtained by taking the derivative of 
(\ref{cond_hom2}) with respect to $y^\mu$,  and then taking the pull back. 

Now we can proceed to find the equations of motion to this system. 
We will first give the definition of an extremal. 

\break
\begin{defn} Extremal of an action \label{def_extremal}  
\begin{enumerate}
\item We say that an arc segment $C$ is {\it stable} with respect to a flow $\alpha$, when it satisfies 
\begin{align}
{\delta _\xi }{S^ {\mathcal{F}}}(C) = 0, \label{c_stable}
\end{align}
where $\xi$ is the generator of $\alpha$. 
\item We say that an arc segment $C$ is an {\it extremal} of the action $S^{\mathcal{F}}$, 
when it satisfies (\ref{c_stable}) for any $\alpha$ 
such that its associated $1$-parameter group $\alpha_s$ satisfies 
${\alpha _s}(\partial C) = \partial C$, $\forall s \in \mathbb{R}$, 
where $\partial C$ is the boundary of $C$. 
\end{enumerate}
\end{defn}
With this concept of an extremal, we can obtain the following theorem. 
\begin{theorem} Extremals  \label{thm_1st_var} \\
Let $C$ be an arc segment. 
The following statements are equivalent. 
\begin{enumerate}
\item $C$ is an extremal. 
\item The equation 
\begin{eqnarray}
&&{{\mathcal{EL}}^F}_\mu  \circ \hat \sigma  = 0, \hfill \nonumber  \\
&&{{\mathcal{EL}}^F}_\mu  := \frac{{{\partial ^2}F}}{{\partial {x^\mu }
\partial {y^\rho }}}d{x^\rho } - d\left( {\frac{{\partial F}}{{\partial {y^\mu }}}} \right), \hfill \label{1st_EL_pullback}
\end{eqnarray}
holds for arbitrary parameterisation $\sigma$.
\end{enumerate} 
\end{theorem}
\begin{proof}
Suppose $C$ is an extremal. Then, by definition, for all ${\alpha _s}:M \to M$, 
such that does not change the boundary of $C$, we have ${\delta _\xi }S(C) = 0$. 
On the other hand, the last term in (\ref{variation_1st_mech}) becomes $0$, 
since it is the boundary term. Therefore, 
we have, 
\begin{align}
\int_{\hat C} {\left( {{\xi ^\mu } \circ {\tau _M}\left\{ {\frac{{{\partial ^2}F}}
{{\partial {x^\mu }\partial {y^\rho }}}d{x^\rho } 
- d\left( {\frac{{\partial F}}{{\partial {y^\mu }}}} \right)} \right\}} \right) = 0.} 
\end{align}
Since this relation must be true for all $\xi $, which is the generator of $\alpha_s$, 
we have 
\begin{align}
\left( {\frac{{{\partial ^2}F}}{{\partial {x^\mu }\partial {y^\rho }}}d{x^\rho } 
- d\left( {\frac{{\partial F}}{{\partial {y^\mu }}}} \right)} \right) \circ \hat \sigma  = 0,
\end{align}
for any parameterisation $\sigma$. 
To prove the converse, it is sufficient to take the similar steps backwards.
\end{proof}

The equations (\ref{1st_EL_pullback}) are called the {\it Euler-Lagrange equations} or 
{\it equations of motion} of the Lagrangian $\mathcal{F}$.  

\begin{defn} Symmetry of the dynamical system \\ 
Let $u$ be a vector field over $M$, and $Y$ an induced vector field by $u$ over $TM$.  
We say that {\it $\mathcal{F}$ is invariant with respect to $u$}, if 
\begin{align}
{L_Y}{\mathcal{F}} = 0, 
\end{align}
and $u$ is called a {\it symmetry} of the dynamical system $(M,\mathcal{F})$. 
We also say that $u$ generates the invariant transformations 
on the Finsler manifold $(M, \mathcal{F})$.
\end{defn}
Now we will have the following important relation between the symmetry and a conserved quantity. 
\begin{theorem} Noether \\
Suppose we are given a symmetry of $(M, \mathcal{F})$. 
Then there exists a function $f$ on $TM$, 
which along the extremal $\gamma$ of $S^\mathcal{F}$ satisfies, 
\begin{align}
\int_{\hat \gamma } {df}  = 0, \label{eq_conserved}
\end{align}
for any parameterisation $\sigma$ which parameterise $\gamma$.
\end{theorem}

\begin{proof}
Let the symmetry be $u$, with its local coordinate expression 
$\displaystyle{u = {u^\mu }\frac{\partial }{{\partial {x^\mu }}}}$, 
and the induced vector field $Y$. 
Then from (\ref{variation_1st_mech}), we have 
\begin{align}
  0 &= \int_{\hat \gamma } {{L_Y}\user2{\mathcal{F}}}  \hfill \nonumber \\
&    = \int_{\hat \gamma } {\left( {{u^\mu } \circ {\tau _M}\left\{ {\frac{{{\partial ^2}F}}
  {{\partial {x^\mu }\partial {y^\rho }}}d{x^\rho } - d\left( {\frac{{\partial F}}{{\partial {y^\mu }}}} \right)} \right\} 
  + d\left( {\frac{{\partial F}}{{\partial {y^\rho }}} \cdot {u^\rho } \circ {\tau _M}} \right)} \right)}  \hfill \nonumber \\
&    = \int_{\hat \gamma } {d\left( {\frac{{\partial F}}{{\partial {y^\rho }}} \cdot {u^\rho } 
\circ {\tau _M}} \right)} , \hfill 
\end{align}
The third equality comes from the fact we consider along the extremal $\gamma$. 
Therefore we have a function on $TM$, 
\begin{align}
f = \frac{{\partial F}}{{\partial {y^\rho }}} \cdot {u^\rho } \circ {\tau _M},
\end{align}
such that satisfies the condition. 
\end{proof}

We call the relation (\ref{eq_conserved}), the {\it conservation law}. 

We can express the conservation law (\ref{eq_conserved}) by 
taking arbitrary parameterisation for this $\gamma$. 
For instance, by $\sigma:[t_i,t_f] \to M$
\begin{align}
0 = \int_{\partial \hat \gamma } {\frac{{\partial F}}{{\partial {y^\rho }}} \cdot {u^\rho } \circ {\tau _M}}  
= \frac{{\partial F}}{{\partial {y^\rho }}} \cdot {u^\rho } \circ {\tau _M}(\hat \sigma ({t_i})) 
- \frac{{\partial F}}{{\partial {y^\rho }}} \cdot {u^\rho } \circ {\tau _M}(\hat \sigma ({t_f})).
\end{align}
\begin{defn} Noether current \\
The quantity $f$ is called the {\it Noether current associated with $u$}.  
\end{defn}
Besides the symmetry defined above, we can also consider the symmetry of ${\mathcal{F}}$
directly by the vector field on $TM$. 
Such symmetries will also induce a conservation law along the extremal. 
\begin{defn} Symmetry of the Finsler-Hilbert form on $TM$\\ 
Let $X$ be a vector field over $TM$.  
We say that $\mathcal{F}$ is invariant with respect to $X$, if 
\begin{align}
{L_X} {\mathcal{F}} = 0,
\end{align} 
and call $X$ a {\it symmetry} of Finsler-Hilbert form $\mathcal{F}$ on $TM$. 
\end{defn}
The generalised conservation law will be the following. 
\begin{theorem} Conservation law \\
Suppose we are given a symmetry of $ {\mathcal{F}}$ on $TM$. 
Then there exists a function $f$ on $TM$, 
which along the extremal $\gamma$ of $S^\mathcal{F}$ satisfies, 
\begin{align}
\int_{{\hat{\gamma }}}{df}=0, 
\end{align}
for any parameterisation $\sigma$ which parameterise $\gamma$.
\end{theorem}

\begin{proof}
Let the symmetry be $X$, with its local coordinate expression 
$\displaystyle{X = {X^\mu }\frac{\partial }{{\partial {x^\mu }}} 
+ {\tilde X^\mu }\frac{\partial }{{\partial {y^\mu }}}}$, 
Then we have 
\begin{align}
& 0 = \int_{\hat \gamma } {{L_X} {\mathcal{F}}}  \hfill \nonumber \\
& \quad  = \int_{\hat \gamma } {\left( {{X^\mu }\left\{ {\frac{{{\partial ^2}F}}{{\partial {x^\mu }\partial {y^\rho }}}d{x^\rho } - d\left( {\frac{{\partial F}}{{\partial {y^\mu }}}} \right)} \right\} + {{\tilde X}^\mu }\frac{{{\partial ^2}F}}{{\partial {y^\mu }\partial {y^\rho }}}d{x^\rho } + d\left( {\frac{{\partial F}}{{\partial {y^\rho }}} \cdot {X^\rho }} \right)} \right)}  \hfill \nonumber \\
& \quad  = \int_{\hat \gamma } {d\left( {\frac{{\partial F}}{{\partial {y^\rho }}} \cdot {X^\rho }} \right)}  \hfill  \label{eq_conserved_2}
\end{align}
The third equality comes from the fact we consider along the extremal $\gamma$. 
Therefore we have a function on $TM$, 
\begin{align}
f = \frac{{\partial F}}{{\partial {y^\rho }}} \cdot {X^\rho }
\end{align} 
such that satisfies the condition. 
\end{proof}
It can be said that in the special case when there exists a symmetry on $TM$, 
such that is an induced vector field on $M$ denoted by $u$,  
then there exists a symmetry $u$ on $(M,\mathcal{F})$. 

By the coordinate transformation 
\begin{align}
&{x^\mu } \to {\tilde x^\mu } = {\tilde x^\mu }({x^\nu }), \nonumber \\
&{y^\mu } \to {\tilde y^\mu } 
= \frac{{\partial {{\tilde x}^\mu }}}{{\partial {x^\nu }}}{y^\nu }, 
\end{align}
the differential $1$-form ${\mathcal{EL}^F}_{\mu}$ in (\ref{1st_EL_pullback}) transforms as  
\begin{align}
\frac{{{\partial ^2}F}}{{\partial {{\tilde x}^\mu }\partial {{\tilde y}^\rho }}}d{\tilde x^\rho } 
- d\left( {\frac{{\partial F}}{{\partial {{\tilde y}^\mu }}}} \right) 
= \left( {\frac{{\partial {x^\nu }}}{{\partial {{\tilde x}^\mu }}}} \right)
\left( {\frac{{{\partial ^2}F}}{{\partial {x^\nu }\partial {y^\rho }}}d{x^\rho } 
- d\left( {\frac{{\partial F}}{{\partial {y^\nu }}}} \right)} \right).
\end{align}
This observation leads us to define a new coordinate invariant form.
\begin{lemma} Euler-Lagrange form \\
There exist a global two form on $TM$, which in local coordinates are expressed by
\begin{align}
{\mathcal{E}}{ {\mathcal{L}}^F}: = d{x^\mu } \wedge  {\mathcal{E}}{ {\mathcal{L}}^F}_\mu  
= \left( {\frac{{{\partial ^2}F}}{{\partial {x^\mu }\partial {y^\rho }}}d{x^\mu } 
+ d\left( {\frac{{\partial F}}{{\partial {y^\rho }}}} \right)} \right) \wedge d{x^\rho } . \label{EL_form}
\end{align}
\end{lemma} 
From the previous coordinate transformations, this form is obviously coordinate independent. 

There is a direct relation between the exterior derivative of $\mathcal{F}$ and $\mathcal{EL}$, 
\begin{align}
d{\mathcal{F}} = {\mathcal{E}}{{\mathcal{L}}^F} 
- \frac{{{\partial ^2}F}}{{\partial {x^\mu }\partial {y^\nu }}}d{x^\mu } \wedge d{x^\nu }. \label{F-EL_rel}
\end{align}
It can be also checked easily that this is also a coordinate invariant relation.

\begin{remark} \label{convL-Fins-rel}
In Chapter 3, Remark \ref {rem_HilbertCartan}, we showed that when given a Hilbert form, 
we can obtain the Cartan form by taking an inclusion map from $J^1 Y$ to $TY$. 
Here we will show that given a ``conventional'' Lagrange function on $J^1 Y$, 
we can also construct its homogeneous counterpart. However, unlike in the case of the former, 
in general, this cannot be done globally. ( It is possible only when $Y=\mathbb{R}\times Q$. )
As in the Remark \ref {rem_HilbertCartan}, 
let $(U,\psi )$, $\psi  = (t,{q^i})$, $i = 1, \cdots ,n$ 
be the adapted chart on a $(n+1)$-dimensional manifold $Y$, 
and the induced chart on $\mathbb{R}$ be 
$(\pi(U),t)$. We denote the induced chart on ${J^1}Y$ by $({({\pi ^{1,0}})^{ - 1}}(U),{\psi ^1}),{\psi ^1} 
= (t,{q^i},{\dot q^i})$
$({({\pi ^{1,0}})^{ - 1}}(U),{\psi ^1}),{\psi ^1} = (t,{q^i},{\dot q^i})$. 
Take the induced chart on $TY$ as $(V,{\tilde \psi ^1})$, $V = {({\tau _Y})^{ - 1}}(U)$, 
${\tilde \psi ^1} = ({x^0},{x^i},{y^0},{y^i})$, $i = 1, \cdots ,n$, 
such that ${y^0} \ne 0$. It is always possible to choose such coordinates for a single chart. 
(In order to avoid confusion we use different symbols, 
but clearly ${x^0} = t \circ {\tau _Y},{x^i} = {q^i} \circ {\tau _Y}$.) 
Now consider a map $\rho :V \hookrightarrow {J^1}Y$, $\rho (V) = {({\pi ^{1,0}})^{ - 1}}(U)$, 
which in coordinates are defined by 
\begin{align}
t \circ \rho  = {x^0},\;{q^i} \circ \rho  = {x^i},\;{\dot q^i} \circ \rho  = \frac{y^i}{y^0}.
\end{align}
Let $F$ be a function on $V$, defined by, 
\begin{align}
F/{y^0} = {\rho ^*}{\mathcal{L}}.
\end{align}
In coordinates, 
\begin{align}
F({x^\mu },{y^\mu }) = {\mathcal{L}}(t \circ \rho ,{q^i} \circ \rho ,{\dot q^i} \circ \rho ){y^0} 
= {\mathcal{L}}({x^0},{x^i},\frac{{{y^i}}}{{{y^0}}}){y^0},  \label{convL-FinsL}
\end{align}
for $\mu  = 0, \cdots ,n$, $i = 1, \cdots ,n$. 
Then on $V$, $F$ satisfies the homogeneity function. 
We now have $\displaystyle{{\mathcal{F}} = \frac{{\partial F}}{{\partial {y^\mu }}}d{x^\mu }}$ on $V$. 
In this way, for a local coordinate chart, (or for the case of $Y=\mathbb{R} \times Q$, 
also globally) we can construct a local Hilbert $1$-form from a Lagrangian, 
which also can be used as an reparameterisation invariant action provided that the arc segment 
where the integration is carried out is covered by this single chart. 
\end{remark}

\begin{remark} 
Locally (on a single chart), we can also construct a different Finsler function and 
local Hilbert form from $\mathcal{L}$ by choosing an appropriate map for $\rho$. 
Though these constructions are coordinate dependent, if the arc segment where the integration is carried out is 
covered by this single chart, we can take it as a reparameterisation invariant action. 
\end{remark}

\begin{remark}
We can also check that the equations of motion given by (\ref{1st_EL_pullback}) reduce to the 
conventional Euler-Lagrange equations by considering the same inclusion map 
$\iota :{{J}^{1}}Y\to TY$, $\dim Y=n+1$, in Chapter 3, Remark \ref {rem_HilbertCartan}, 
which in coordinate expression were given by 
\begin{align}
{x^0} \circ \iota  = t,\;{x^i} \circ \iota  = {q^i},\;{y^0} \circ \iota  = 1,\;{y^i} \circ \iota  = {\dot q^i}, 
\end{align}
$i = 1, \cdots ,n$.
We will use Greek indices such as $\mu ,\rho  = 0,1,2, \cdots ,n$, 
and Latin indices such as $i,j,k = 1,2, \cdots ,n$. 
Rewrite the $1$-form ${\mathcal{E}}{{\mathcal{L}}^F}_\mu $, 
\begin{align}
 {\mathcal{E}}{{\mathcal{L}}^F}_\mu  & = \frac{{{\partial ^2}F}}{{\partial {x^\mu }\partial {y^\rho }}}d{x^\rho } 
- d\left( {\frac{{\partial F}}{{\partial {y^\mu }}}} \right) 
= \frac{{{\partial ^2}F}}{{\partial {x^\mu }\partial {y^0}}}d{x^0} 
+ \frac{{{\partial ^2}F}}{{\partial {x^\mu }\partial {y^i}}}d{x^i} 
- d\left( {\frac{{\partial F}}{{\partial {y^\mu }}}} \right) \hfill \nonumber \\
&   = \frac{1}{{{y^0}}}\left( {\frac{{\partial F}}{{\partial {x^\mu }}} 
- \frac{{{\partial ^2}F}}{{\partial {x^\mu }\partial {y^j}}}{y^j}} \right)d{x^0} 
+ \frac{{{\partial ^2}F}}{{\partial {x^\mu }\partial {y^i}}}d{x^i} 
- d\left( {\frac{{\partial F}}{{\partial {y^\mu }}}} \right), \hfill 
\end{align}
where the components of $\mu=0$, and $\mu=1, 2, \cdots , n$ are
\begin{align}
 {\mathcal{E}}{{\mathcal{L}}^F}_0 & = \frac{{{\partial ^2}F}}{{\partial {x^0}\partial {y^\rho }}}d{x^\rho } 
- d\left( {\frac{{\partial F}}{{\partial {y^0}}}} \right) \hfill \nonumber \\
&    = \frac{1}{{{y^0}}}\left( {\frac{{\partial F}}{{\partial {x^0}}} 
- \frac{{{\partial ^2}F}}{{\partial {x^0}\partial {y^j}}}{y^j}} \right)d{x^0} 
+ \frac{{{\partial ^2}F}}{{\partial {x^0}\partial {y^i}}}d{x^i} 
- \frac{1}{{{y^0}}}d\left( {F - \frac{{\partial F}}{{\partial {y^j}}}{y^j}} \right) \hfill \nonumber \\
&   \quad  + \frac{1}{{{{\left( {{y^0}} \right)}^2}}}\left( {F 
- \frac{{\partial F}}{{\partial {y^j}}}{y^j}} \right)d{y^0}, \hfill \nonumber \\
  {\mathcal{E}}{{\mathcal{L}}^F}_i &= \frac{{{\partial ^2}F}}{{\partial {x^i}\partial {y^\rho }}}d{x^\rho } 
- d\left( {\frac{{\partial F}}{{\partial {y^i}}}} \right) \hfill \nonumber \\
&    = \frac{1}{{{y^0}}}\left( {\frac{{\partial F}}{{\partial {x^i}}} 
- \frac{{{\partial ^2}F}}{{\partial {x^i}\partial {y^j}}}{y^j}} \right)d{x^0} 
+ \frac{{{\partial ^2}F}}{{\partial {x^i}\partial {y^i}}}d{x^i} 
- d\left( {\frac{{\partial F}}{{\partial {y^i}}}} \right). \hfill 
\end{align}
The pull back to ${J^1}Y$ becomes
\begin{align}
  {\iota ^*}{{\mathcal{EL}}^F}_0 &= \left( {\frac{1}{{{y^0}}}\left( {\frac{{\partial F}}{{\partial {x^0}}} 
- \frac{{{\partial ^2}F}}{{\partial {x^0}\partial {y^j}}}{y^j}} \right)d{x^0} 
+ \frac{{{\partial ^2}F}}{{\partial {x^0}\partial {y^j}}}d{x^j} - \frac{1}{{{y^0}}}d\left( 
{F - \frac{{\partial F}}{{\partial {y^j}}}{y^j}} \right)} \right) \circ \iota  \hfill \nonumber \\
&   \quad + \left( {\frac{1}{{{{\left( {{y^0}} \right)}^2}}}\left( 
{F - \frac{{\partial F}}{{\partial {y^j}}}{y^j}} \right)d{y^0}} \right) \circ \iota  \hfill \nonumber \\
&   = \left( {\frac{{\partial F}}{{\partial {x^0}}} 
- \frac{{{\partial ^2}F}}{{\partial {x^0}\partial {y^j}}}{y^j}} \right) \circ \iota \,dt 
+ \frac{{{\partial ^2}F}}{{\partial {x^0}\partial {y^j}}} \circ \iota \,d{q^j} 
- d\left( {\left( {F - \frac{{\partial F}}{{\partial {y^j}}}{y^j}} \right) \circ \iota } \right) \hfill \nonumber \\
&   = \left( {\frac{{\partial {\mathcal{L}}}}{{\partial t}} 
- \frac{{{\partial ^2}{\mathcal{L}}}}{{\partial t\partial {{\dot q}^j}}}{{\dot q}^j}} \right)dt 
+ \frac{{{\partial ^2}{\mathcal{L}}}}{{\partial t\partial {{\dot q}^j}}}d{q^j} 
- d\left( {{\mathcal{L}} - \frac{{\partial {\mathcal{L}}}}{{\partial {{\dot q}^j}}}{{\dot q}^j}} \right), \hfill \nonumber \\
 {\iota ^*}{{\mathcal{EL}}^F}_i & = \left( {\frac{1}{{{y^0}}}
\left( {\frac{{\partial F}}{{\partial {x^i}}} - \frac{{{\partial ^2}F}}{{\partial {x^i}\partial {y^j}}}{y^j}} 
\right)d{x^0} + \frac{{{\partial ^2}F}}{{\partial {x^i}\partial {y^j}}}d{x^j} 
- d\left( {\frac{{\partial F}}{{\partial {y^i}}}} \right)} \right) \circ \iota  \hfill \nonumber \\
&   = \left( {\left( {\frac{{\partial F}}{{\partial {x^i}}} 
- \frac{{{\partial ^2}F}}{{\partial {x^i}\partial {y^j}}}{y^j}} \right) \circ \iota \,dt 
+ \frac{{{\partial ^2}F}}{{\partial {x^i}\partial {y^j}}} \circ \iota \,d{q^j} 
- d\left( {\frac{{\partial F}}{{\partial {y^i}}} \circ \iota } \right)} \right) \hfill \nonumber \\
&   = \left( {\left( {\frac{{\partial {\mathcal{L}}}}{{\partial {q^i}}} 
- \frac{{{\partial ^2}{\mathcal{L}}}}{{\partial {q^i}\partial {{\dot q}^j}}}{{\dot q}^j}} \right)\,dt 
+ \frac{{{\partial ^2}{\mathcal{L}}}}{{\partial {q^i}\partial {{\dot q}^j}}}\,d{q^j} 
- d\left( {\frac{{\partial {\mathcal{L}}}}{{\partial {{\dot q}^i}}}} \right)} \right). \hfill  
\end{align}
Now suppose we have a map ${\gamma ^1}:\mathbb{R} \to {J^1}Y$ such that satisfies 
\begin{align}
\iota  \circ {\gamma ^1} = \hat \sigma .
\end{align}
Then, 
\begin{align}
\left( {{\mathcal{EL}}^F}_\mu  \circ \iota \right) \circ {\gamma ^1} 
= {{\mathcal{EL}}^F}_\mu  \circ \left( {\iota  \circ {\gamma ^1}} \right) 
= {{\mathcal{EL}}^F}_\mu  \circ \hat \sigma ,
\end{align}
therefore, 
the equation of motion ${{\mathcal{EL}}^F}_\mu  \circ \hat \sigma  = 0$ 
where ${{\mathcal{EL}}^F}_\mu $ is a form on $TM$, 
can be interpreted as a equation of motion 
$\left( {{\mathcal{EL}}^F}_\mu  \circ \iota \right) \circ {\gamma ^1} = 0$, 
where ${{\mathcal{EL}}^F}_\mu  \circ \iota $ is a form on $J^1 Y$. 

In the special case where $Y = \mathbb{R} \times Q$, where $Q$ is the $n$-dimensional configuration space, 
and ${J^1}Y$ is the prolongation of the bundle $(Y,p{r_1},\mathbb{R})$, 
we can consider a section $\gamma $ of $(Y,p{r_1},\mathbb{R})$, and take its prolongation 
${J^1}\gamma $ as ${\gamma ^1}$. 
In such case, the pull back equation 
$\left( {{\mathcal{E}}{{\mathcal{L}}^F}_\mu  \circ \iota } \right) \circ {\gamma ^1} = 0$ becomes, 
\begin{align}
& {{\mathcal{EL}}^F}_0 \circ \iota  \circ {J^1}\gamma  
= \left( {\left( {\frac{{\partial {\mathcal{L}}}}{{\partial t}} 
- \frac{{{\partial ^2}{\mathcal{L}}}}{{\partial t\partial {{\dot q}^j}}}{{\dot q}^j}} \right)dt 
+ \frac{{{\partial ^2}{\mathcal{L}}}}{{\partial t\partial {{\dot q}^j}}}d{q^j} 
- d\left( {{\mathcal{L}} - \frac{{\partial {\mathcal{L}}}}{{\partial {{\dot q}^j}}}{{\dot q}^j}} \right)} \right) 
\circ {J^1}\gamma  \hfill \nonumber \\
&  = \left( {\left( {\frac{{\partial {\mathcal{L}}}}{{\partial t}} 
- \frac{{{\partial ^2}{\mathcal{L}}}}{{\partial t\partial {{\dot q}^j}}}{{\dot q}^j}} \right) 
\circ {J^1}\gamma \, + \left( {\frac{{{\partial ^2}{\mathcal{L}}}}{{\partial t\partial {{\dot q}^j}}}{{\dot q}^j}} \right) 
\circ {J^1}\gamma \, - \frac{d}{{dt}}\left( {\left( {{\mathcal{L}} 
- \frac{{\partial {\mathcal{L}}}}{{\partial {{\dot q}^j}}}{{\dot q}^j}} \right) \circ {J^1}\gamma } \right)} 
\right)dt \hfill \nonumber \\
&   = \left( {\frac{{\partial {\mathcal{L}}}}{{\partial t}}} \right) \circ {J^1}\gamma \,\,dt \hfill \nonumber \\
&   \quad - \left( {\frac{{\partial {\mathcal{L}}}}{{\partial t}} 
+ \frac{{\partial {\mathcal{L}}}}{{\partial {q^j}}}{{\dot q}^j} 
+ \frac{{\partial {\mathcal{L}}}}{{\partial {{\dot q}^j}}}{{\ddot q}^j} 
- \frac{{{\partial ^2}{\mathcal{L}}}}{{\partial t\partial {{\dot q}^j}}}{{\dot q}^j} 
- \frac{{{\partial ^2}{\mathcal{L}}}}{{\partial {q^k}\partial {{\dot q}^j}}}{{\dot q}^k}{{\dot q}^j} 
- \frac{{{\partial ^2}{\mathcal{L}}}}{{\partial {{\dot q}^k}\partial {{\dot q}^j}}}{{\ddot q}^k}{{\dot q}^j} 
- \frac{{\partial {\mathcal{L}}}}{{\partial {{\dot q}^j}}}{{\ddot q}^j}} \right) \circ {J^1}\gamma \,dt \hfill \nonumber \\
&   =  - \,\left( {\frac{{\partial {\mathcal{L}}}}{{\partial {q^j}}}{{\dot q}^j} 
- \frac{{{\partial ^2}{\mathcal{L}}}}{{\partial t\partial {{\dot q}^j}}}{{\dot q}^j} 
- \frac{{{\partial ^2}{\mathcal{L}}}}{{\partial {q^k}\partial {{\dot q}^j}}}{{\dot q}^k}{{\dot q}^j} 
- \frac{{{\partial ^2}{\mathcal{L}}}}{{\partial {{\dot q}^k}\partial {{\dot q}^j}}}{{\ddot q}^k}{{\dot q}^j}} 
\right) \circ {J^1}\gamma \,dt \hfill \nonumber \\
&  =  - \left( {\frac{{\partial {\mathcal{L}}}}{{\partial {q^j}}} \circ {J^1}\gamma  
- \frac{d}{{dt}}\left( {\frac{{\partial {\mathcal{L}}}}{{\partial {{\dot q}^j}}} \circ {J^1}\gamma } \right)} 
\right)\left( {{{\dot q}^j} \circ {J^1}\gamma } \right)\,dt = 0, \hfill \nonumber \\ 
& {{\mathcal{EL}}^F}_i \circ \iota  \circ {J^1}\gamma  
= \left( {\left( {\frac{{\partial {\mathcal{L}}}}{{\partial {q^i}}} 
- \frac{{{\partial ^2}{\mathcal{L}}}}{{\partial {q^i}\partial {{\dot q}^j}}}{{\dot q}^j}} \right)\,dt 
+ \frac{{{\partial ^2}{\mathcal{L}}}}{{\partial {q^i}\partial {{\dot q}^j}}}\,d{q^j} 
- d\left( {\frac{{\partial {\mathcal{L}}}}{{\partial {{\dot q}^i}}}} \right)} \right) \circ {J^1}\gamma  \hfill \nonumber \\
&   = \left( {\left( {\frac{{\partial {\mathcal{L}}}}{{\partial {q^i}}} 
- \frac{{{\partial ^2}{\mathcal{L}}}}{{\partial {q^i}\partial {{\dot q}^j}}}{{\dot q}^j}} \right)\, 
\circ {J^1}\gamma \, + \left( {\frac{{{\partial ^2}{\mathcal{L}}}}{{\partial {q^i}\partial {{\dot q}^j}}}\,{{\dot q}^j}} \right) 
\circ {J^1}\gamma \,\, - \frac{d}{{dt}}\left( {\frac{{\partial {\mathcal{L}}}}{{\partial {{\dot q}^i}}} \circ {J^1}\gamma } \right)} \right)dt \hfill \nonumber \\
&  = \left( {\left( {\frac{{\partial {\mathcal{L}}}}{{\partial {q^i}}}} \right)\, \circ {J^1}\gamma \, 
- \frac{d}{{dt}}\left( {\frac{{\partial {\mathcal{L}}}}{{\partial {{\dot q}^i}}} \circ {J^1}\gamma } \right)} 
\right)dt = 0, \hfill \label{Fins-conveq-rel}
\end{align}
therefore, giving us the well-known form of Euler-Lagrange equations. 
By fixing the bundle $(Y,p{r_1},\mathbb{R})$, 
which means a specific choice of parameterisation has been made (in this case by $\iota \circ J^1 \gamma = \hat{\sigma}$), 
the parameterisation independent equation of motion 
reduces to the standard notion, and the number of equations degenerates to $n$. 
\end{remark}  

\begin{ex} \label{ex_Newton}
Every Newtonian mechanics in the form ${\mathcal{L}} = K - V$ with $K$ the 
kinematical energy and $V$ a potential term could be expressed in the setting of 
Finsler manifold. 

Let $(M,{\mathcal{F}})$ be a Finsler manifold with $\dim M = n + 1$, and the induced chart 
$(U,\psi ),\psi  = ({x^\mu },{y^\mu })$ on $TM$. The conventional 
Lagrangian function of the particle moving in 
$n$ dimension space with mass $m$ is given on the space $J^1 M$, 
with $M=\mathbb{R}\times \mathbb{R}^n$. Let the induced chart on $J^1 M$ be $(\tilde U,\tilde \psi )$, 
$\tilde \psi  = (t,{q^i},{\dot q^i})$, then the local expression of the Lagrangian function is, 
\begin{eqnarray}
{\mathcal{L}} = \frac{m}{2}\left( {{{({{\dot q}^1})}^2} + {{({{\dot q}^2})}^2} +  \dots  
+ {{({{\dot q}^n})}^2}} \right) - V({q^1}, \dots ,{q^n}).
\end{eqnarray} 
Then the local expression of the Finsler function obtained by (\ref{convL-FinsL}) is , 
\begin{eqnarray}
F = \frac{m}{2}\frac{{{{({y^1})}^2} + {{({y^2})}^2} +  \dots  + {{({y^n})}^2}}}{{{y^0}}} 
- V({x^1}, \dots ,{x^n}){y^0}.
\end{eqnarray} 
It is apparent that $F$ satisfies the homogeneity condition. 
In the case of $M = \mathbb{R} \times \mathbb{R}^n$, this construction could be done globally, 
and we can create a Hilbert form by 
$\displaystyle{{\mathcal{F}} = \frac{{\partial F}}{{\partial {y^i}}}d{x^i}}$, which we take as our Lagrangian. 
The parameterisation independent equation of motion is obtained by the equation 
$\left( {{{\mathcal{EL}}^F}_\mu } \right) \circ \hat \sigma  = 0$, 
for any parameterisation $\sigma $, and the explicit coordinate expression can be calculated as, 
\begin{align}
&  {{\mathcal{EL}}^F}_0 \circ \hat \sigma  
= \left( {\frac{{{\partial ^2}F}}{{\partial {x^0}\partial {y^\rho }}}d{x^\rho } 
- d\left( {\frac{{\partial F}}{{\partial {y^0}}}} \right)} \right) \circ \hat \sigma  \hfill \nonumber \\
&  \quad  = \left( {\frac{{\partial V}}{{\partial {x^k}}}d{x^k} 
- m\frac{{{{({y^1})}^2} +  \dots  
+ {{({y^n})}^2}}}{{{(y^0)^3}}}d{y^0} + m\frac{{{y^k}}}{{{(y^0)^2}}}d{y^k}} \right) \circ \hat \sigma  = 0, \hfill \nonumber \\
&  {{\mathcal{EL}}^F}_i \circ \hat \sigma  = \left( {\frac{{{\partial ^2}F}}{{\partial {x^i}\partial {y^\rho }}}d{x^\rho } -
 d\left( {\frac{{\partial F}}{{\partial {y^i}}}} \right)} \right) \circ \hat \sigma  \hfill \nonumber \\
&  \quad  = \left( { - \frac{{\partial V}}{{\partial {x^i}}}d{x^0} + m\frac{{{y^i}}}{{{{({y^0})}^2}}}d{y^0} 
- m\frac{1}{{{y^0}}}d{y^i}} \right) \circ \hat \sigma  = 0, \hfill  
\end{align}
for $i = 1, \cdots ,n$. 
We have used the relation such as, 
\begin{align}
&  \frac{{\partial F}}{{\partial {y^0}}} =  - \frac{1}{2}m\frac{{{{({y^1})}^2} +  \dots  
+ {{({y^n})}^2}}}{{{{({y^0})}^2}}},\quad \frac{{\partial F}}{{\partial {y^i}}} 
= m\frac{{{y^i}}}{{{y^0}}},\quad \frac{{{\partial ^2}F}}{{\partial {x^i}\partial {y^0}}} 
=  - \frac{{\partial V}}{{\partial {x^i}}}, \hfill \nonumber \\
& \frac{{{\partial ^2}F}}{{\partial {y^0}\partial {y^0}}} 
= m\frac{{{{({y^1})}^2} +  \dots  + {{({y^n})}^2}}}{{{{({y^0})}^3}}},
\quad \frac{{{\partial ^2}F}}{{\partial {y^0}\partial {y^i}}} 
=  - m\frac{{{y^i}}}{{({y^0})^2}},\quad \frac{{{\partial ^2}F}}{{\partial {y^i}\partial {y^j}}} 
= m\frac{1}{{{y^0}}}\delta^{ij}, \hfill
\end{align}
during the calculation. 
To see that these equations give the conventional Euler-Lagrange equations, 
let us choose some convenient parameterisation 
$\sigma :I \to M$, such that $\sigma :t \to (t = {x^0},{x^1}, \dots ,{x^n})$. 
${x^0}$ corresponds to the time variable, and now the ${y^0}$ coordinate of the tangent lift 
$\hat \sigma $ becomes $1$. Consider the pull back of these equations, and we get 
\begin{eqnarray}
&&{{\mathcal{EL}}^F}_0 \circ \hat \sigma  
= \left( {\frac{{\partial V}}{{\partial {x^i}}} 
\circ \hat \sigma \, + {m^i}\frac{{d({y^i} \circ \hat \sigma )}}{{dt}}} \right)({y^i} \circ \hat \sigma )dt = 0, \nonumber \\
&&{{\mathcal{EL}}^F}_i \circ \hat \sigma  
=  - \left( {\frac{{\partial V}}{{\partial {x^i}}} \circ \hat \sigma \, + m\frac{{d({y^i} 
\circ \hat \sigma )}}{{dt}}} \right)dt = 0,
\end{eqnarray}
which is the usual Newton's equation of motion. 
Apparently, the first equation can be derived from the second by multiplying by 
${y^i}$ and summing up, therefore dependent, and we only have $n$ equations. 
\end{ex}

The choice of $\sigma$ is arbitrary, and it is not at all necessary to choose the one in 
Example \ref{ex_Newton}. This freedom of choosing the parameterisation could be useful 
when one tries to find a good variable for solving equations.

\begin{ex} Brachistochrone \\
The reparameterisation invariance give us the freedom to choose a 
``good'' parameterisation that let us solve the equation of motion easily. 
One such example given here is the classical problem of brachistochrone. 
Consider two points fixed in $\mathbb{R}^2$, 
and suppose a constant gravitational field. 
The problem is to find the curve connecting these two points, 
which the particle with mass $m$ will roll at the least amount of time. 
We will take the global coordinates on $\mathbb{R}^2$ as $(x,y)$, 
where the $y$ coordinate is parallel and is in the same direction to the 
constant gravitational acceleration $g (g>0)$, and the global coordinates on 
$T{\mathbb{R}^2}$ as $(x,y,\dot x,\dot y)$. 
The Finsler Lagrangian function that measures the elapsed time is given by 
\begin{align}
F = \sqrt {\frac{{{{(\dot x)}^2} + {{(\dot y)}^2}}}{v}} ,
\end{align} 
where $v$ is the velocity of the particle. This satisfies the homogeneity condition.
Using the condition of energy conservation, we have$v = \sqrt {2gy} $, 
assuming the initial value of energy to be zero. 
Then, 
\begin{align}
F = \sqrt {\frac{{{{(\dot x)}^2} + {{(\dot y)}^2}}}{{2gy}}} ,
\end{align} 
and by (\ref{1st_EL_pullback}) the equation of motion on $TM$ becomes, 
\begin{align}
\left\{ \begin{gathered}
  d\left( {\frac{{\partial F}}{{\partial \dot x}}} \right) = 0, \hfill \\
   {\frac{{{\partial ^2}F}}{{\partial y\partial \dot x}}dx 
   + \frac{{{\partial ^2}F}}{{\partial y\partial \dot y}}dy 
   - d\left( {\frac{{\partial F}}{{\partial \dot y}}} \right)}  = 0. \hfill \\ 
\end{gathered}  \right. \label{ex_brachistochrone1}
\end{align}
From the first equation of (\ref{ex_brachistochrone1}), we have 
\begin{align}
\frac{{\partial F}}{{\partial \dot x}} = \tilde C,\;
\end{align}
where $\tilde{C}$ is a constant. This gives us the equation,
\begin{align}
{(\dot x)^2} = 2gyC\left( {{{(\dot x)}^2} + {{(\dot y)}^2}} \right),  \label{ex_brachistochrone2}
\end{align}
where $C$ is also a constant. 
The second equation gives, 
\begin{align}
-\dot{x}({{\dot{x}}^{2}}+{{\dot{y}}^{2}})dx+2\dot{x}\dot{y}yd\dot{x}-2{{\dot{x}}^{2}}yd\dot{y}=0. \label{ex_brachistochrone3}
\end{align}
Now, since these expression of Euler-Lagrange equations does not depend on the choice of parameterisation, 
we have the freedom of choosing a convenient parameterisation $\sigma$. 
This is equivalent to saying that $F$ is a Finsler function, which implies the Hessian of $F$ to be singular, 
and in coordinates, it means the coordinate functions regarding the first derivatives are dependent. 
In other words, we have the freedom of fixing one of these coordinate functions, 
provided such choice comply with the condition for regular parameterisation. 
Let us choose one such parameterisation, $\dot x = 2gyC$. 
Then (\ref{ex_brachistochrone2}), (\ref{ex_brachistochrone3}) becomes,
\begin{align}
& \dot x = \left( {{{(\dot x)}^2} + {{(\dot y)}^2}} \right), \nonumber \\
& - gCdx + \dot yd\dot x - \dot xd\dot y = 0.
\end{align}
The first equation represents a circle, so we will introduce a parameter $\theta $ and put,
\begin{align}
\left\{ \begin{gathered}
  \dot x = \frac{1}{2} \pm \frac{1}{2}\cos \theta,  \hfill \\
  \dot y =  \pm \frac{1}{2}\sin \theta.  \hfill \\ 
\end{gathered}  \right. 
\end{align}
Suppose $y=0$ when $\theta=0$, 
then from $\dot x = 2gyC$, $\dot x = 0$ when $\theta  = 0$, 
so we should take 
\begin{align}
\dot x = \frac{1}{2} - \frac{1}{2}\cos \theta.
\end{align} 
Then, from the second equation, we have 
\begin{align}
dx =  \pm \frac{1}{{4gC}}(1 - \cos \theta )d\theta, 
\end{align}
and get,
\begin{align}
\left\{ \begin{gathered}
  x =  \pm \frac{1}{{4gC}}(\theta  - \sin \theta ) + {C_2}, \hfill \\
  y = \frac{1}{{4gC}}(1 - \cos \theta ). \hfill \\ 
\end{gathered}  \right.
\end{align}
Suppose $x = 0$ when $\theta  = 0$, then ${C_2} = 0$. 
Also, suppose the final condition to be $({x_\pi },{y_\pi })$ at $\theta  = \pi $. 
Then $\displaystyle{{y_\pi } = \frac{1}{{2gC}}}$, $\displaystyle{{x_\pi } =  \pm \frac{{{y_\pi }}}{2}\pi}$, 
and since $x \geqslant 0, \, y \geqslant 0$, we get, 
\begin{align}
\left\{ \begin{gathered}
  x = \frac{{{y_\pi }}}{2}(\theta  - \sin \theta ), \hfill \\
  y = \frac{{{y_\pi }}}{2}(1 - \cos \theta ). \hfill \\ 
\end{gathered}  \right. 
\end{align}
which is the cycloid. 
\end{ex}

Notice the choice of parameterisation we made led us to discover the curve in a most elementary way.
This is one major advantage of using the parameter invariant equations of motion. 
Since we have the freedom of choosing parameters, we can choose one that matches our need.

\section{Second order mechanics} \label{sec_2nd_mech}
	Here we will present the Lagrange formulations for the higher order case of mechanics, 
in terms of Finsler-Kawaguchi geometry introduced in Section \ref{sec_2nd_Kawaguchi}. 
We discuss especially for the second order. 
Higher order should follow in the similar extension.  
By the term {\it second order}, we 
mean that the total space we are considering is 
the second order tangent bundle, and by {\it mechanics}, 
we mean that we are considering the arc segment on $M$. 

The basic structure we consider in this section is introduced in 
Chapter 2 and 4 (Section \ref{sec_2nd_Kawaguchi}), 
namely the $n$-dimensional second order Finsler-Kawaguchi manifold $(M,{\mathcal{K}})$, 
the second order tangent bundle $({T^2}M,\tau _M^{2,0},M)$, 
and a $1$-dimensional curve (arc segment) $C$ on $M$, which is parameterised by $\sigma$. 
The curve (arc segment) describes the trajectory of the object on $M$.

We take the second-order Finsler-Kawaguchi form ${\mathcal{K}}$ as the Lagrangian, 
and the action will be defined by considering the integration 
over the second order lift of the parameterisable curve (arc segment) $C$.
The Euler-Lagrange equations are derived by taking the variation of the action with respect 
to the flow on $M$ that deforms the arc segment $C$, and fixed on the boundary.
Similarly as in the previous first order case, we can show that the action 
and consequently the Euler-Lagrange equations are 
independent with respect to the parameterisation belonging to the same equivalent class. 

\subsection{Action}   \label{subsec_2nd_Mech_Action}
Suppose we have a dynamical system (differential equations expressing motions) 
where the trajectory of the point particle (or any object which dynamics 
could be considered as a point) is expressed by an arc segment $C$ of a parameterisable curve, 
such that $C = \sigma (I) \subset M$, where $I$ is a closed interval $I = [{t_i},{t_f}] \subset \mathbb{R}$. 

When we can express this system by second order Finsler-Kawaguchi geometry, 
namely the pair $(M, \mathcal{K})$ where $\mathcal{K}$ is a second order Finsler-Kawaguchi $1$-form, 
we refer to this dynamical system as {\it second order mechanics}, 
and conversely call the pair $(M,\mathcal{K})$ a {\it dynamical system}, 
and $\mathcal{K}$ the {\it Lagrangian} of second order. 
As we did in the first order case, we will also show in this section (Remark \ref{convL-2ndFins-rel}),
that given a conventional Lagrangian of second order\footnote{We simply use the term 
``conventional'' to distinguish the Lagrangian function on a local chart of $J^2 Y$, 
from our Lagrangian which is the Finsler-Kawaguchi $1$-form over $T^2 Y$.}, 
one can always construct a Finsler-Kawaguchi function (and therefore obtain the Finsler-Kawaguchi $1$-form) 
on the corresponding local chart of $T^2M$.

The action of second order mechanics is defined as follows. 

%
\begin{defn} Action of second order mechanics \\
Let  $(M, {\mathcal{K}})$ be a $n$-dimensional second order Finsler-Kawaguchi manifold, 
$(U,\varphi ),\;\varphi  = ({x^\mu })$ be a chart on $M$, 
and $(V^2, \psi^2 )$, ${{V}^{2}}={{(\tau _{M}^{2,0})}^{-1}}(U)$, $\psi = ({x^\mu}, {y^\mu}, z^\mu)$
the induced chart on $T^2 M$. 
The local coordinate expression of the Finsler-Hilbert form 
 ${\mathcal{K}} \in {\Omega ^1}(T^2M)$ is given by
\begin{eqnarray}
{\mathcal{K}} = \frac{{\partial K}}{{\partial {y^\mu }}}d{x^\mu } 
+ 2\frac{{\partial K}}{{\partial {z^\mu }}}d{y^\mu },
\end{eqnarray}
where $K$ is the Finsler-Kawaguchi function.  
Let $C$ be an arc segment on $M$, and $\sigma$ its parameterisation, 
$\sigma(I) = C \subset M$ with $I = [{t_i},{t_f}] \subset \mathbb{R}$.
We call the functional ${S^{\mathcal{K}}}(C)$ defined by 
\begin{eqnarray}
{S^{\mathcal{K}}}(C): = {l^K}(C) = \int_{{C^2}} {\mathcal{K}}  
= \int_{{\sigma ^2}(I)} {\frac{{\partial K}}{{\partial {y^\mu }}}d{x^\mu } 
+ 2\frac{{\partial K}}{{\partial {z^\mu }}}d{y^\mu }} , \label{Finsler-KawaguchiAction}
\end{eqnarray}
the {\it action of second order mechanics associated with $\mathcal{K}$.} 
\end{defn}
As we have seen in Section \ref{subsec_paraminv_2nd_mech}, 
Lemma \ref{lem_repinv_FinslerKawaguchi}, 
Finsler-Kawaguchi length is invariant with respect to the reparameterisation,
therefore the action is also invariant.  

\subsection{Total derivative}
Here we will introduce an operator called the {\it total derivative} 
that is an identity map on the total space of the bundle we consider. 
It becomes the derivative with respect to the parameter on the parameter space, 
namely $\displaystyle{\frac{d}{{dt}}}$ in this section, but later we will generalise this concept to the case of 
$k$-dimensional parameter space. 

Consider the bundle morphism $(Tf,f)$ from $(TE,{\tau _E},E)$ 
to $(TM,{\tau _M},M)$, introduced in 
(Example \ref{ex_tangentbundlemorphism}). 
For the case $E = TM$, 
there is an identity map called the {\it total derivative}. 
\begin{figure}
  \centering
  \includegraphics[width=7cm]{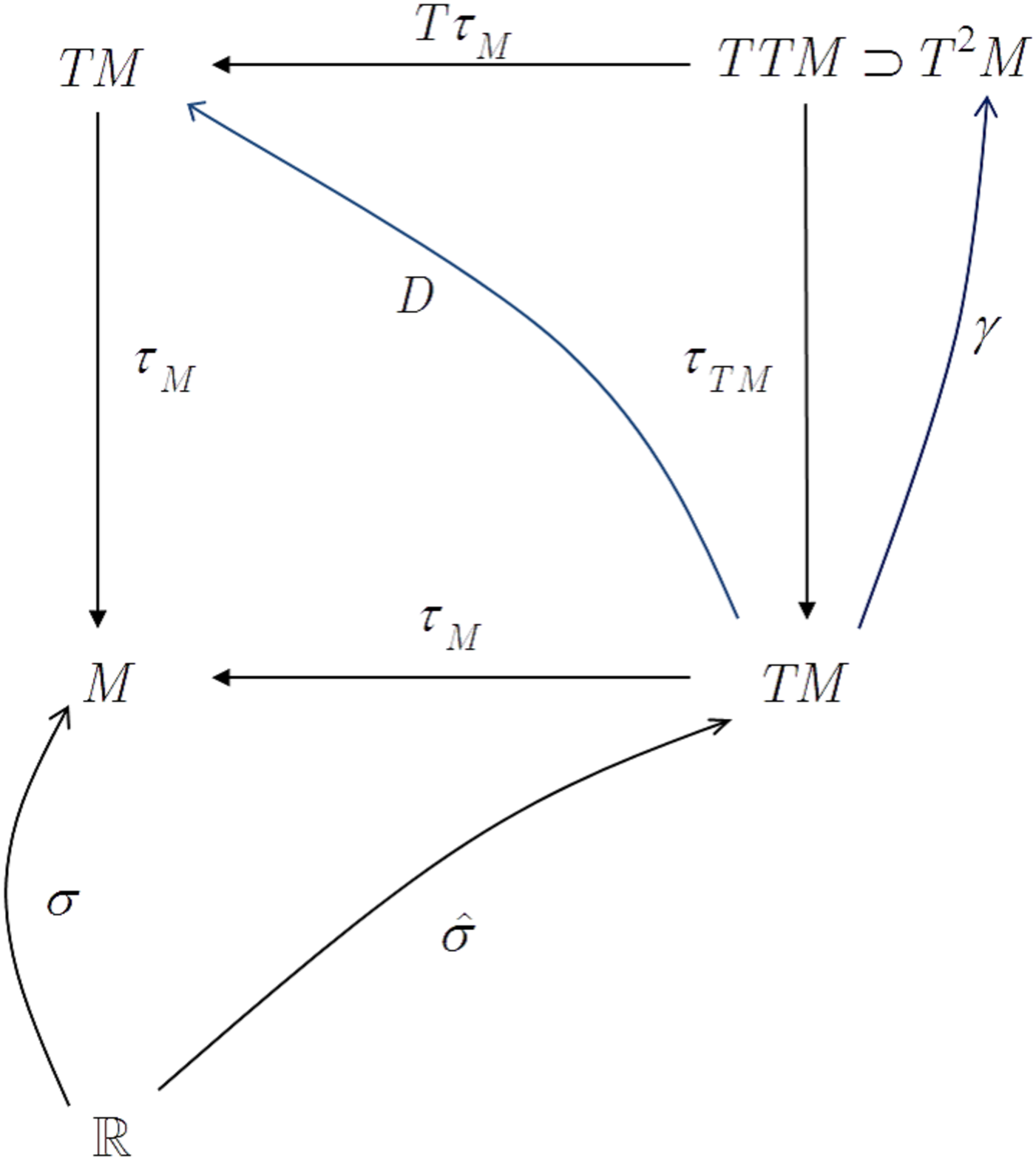}
  \caption{Total derivative}
\end{figure}

\begin{pr} \label{pr_totalderiv}
Consider the bundle morphism $(T{\tau _M},{\tau _M})$ from $(TTM,{\tau _{TM}},TM)$ to $(TM,{\tau _M},M)$. 
Let $\gamma $ be a section of the sub-bundle 
${{\left. {{\tau }_{TM}} \right|}_{{{T}^{2}}M}}$ of ${{\tau }_{TM}}$, 
and define a map $D:TM\to TM$ by $D=T{{\tau }_{M}}\circ \gamma $. 
Then $D$ is an identity map which its local coordinate expression is given by 
\begin{align}
D={{y}^{\mu }}\cdot {{\left( \frac{\partial }{\partial {{x}^{\mu }}} \right)}_{{{\tau }_{M}}(\cdot )}}.
\end{align} 
\end{pr}
\begin{proof} 
By the definition of ${{T}^{2}}M$, we have 
$T{{\tau }_{M}}(\gamma (p))={{\tau }_{TM}}(\gamma (p))$ for $\forall p\in TM$. 
Then since ${{\left. {{\tau }_{TM}} \right|}_{{{T}^{2}}M}}\circ \gamma 
= {{\tau }_{TM}}\circ \gamma =i{{d}_{TM}}$, 
we have $T{\tau _M} \circ \gamma  = i{d_{TM}}$.  
This is an identity map of $TM$,  and the coordinate expression becomes, 
$\displaystyle{D(p)=p={{y}^{\mu }}(p){{\left( \frac{\partial }{\partial {{x}^{\mu }}} \right)}_{{{\tau }_{M}}(p)}}}$, 
therefore, $\displaystyle{D={{y}^{\mu }}\cdot {{\left( \frac{\partial }{\partial {{x}^{\mu }}} \right)}_{{{\tau }_{M}}(\cdot )}}}$. 
\end{proof}
\begin{defn} Total derivative \label{def_totalderiv1} \\
We call the map $D$, the {\it total derivative} of first order.
\end{defn}

\begin{remark} 
The pull-back bundle of $(TM, {\tau _M}, M)$ by ${\bar \tau _M}:TM \to M$ is 
$({({\bar \tau _M})^*}TM, \lb[3] {({\bar \tau _M})^*}{\tau _M}, \lb[2] TM)$, and 
the total derivative $D$ defines a unique section 
$\delta $ of ${({\bar \tau _M})^*}{\tau _M}$,  
which is called a {\it vector field along} 
${\bar \tau _M}$, by $D = {({\tau _M})^*}{\bar \tau _M} \circ \delta $. 
By definition of the pull-back bundle, ${({\bar \tau _M})^*}{\tau _M}(p,q) = q$, 
${({\tau _M})^*}{\bar \tau _M}(p,q) = p$ where $p$ is a point of the total space of the bundle 
${\tau _M}$, and $q$ is the point of the base space of the bundle ${({\bar \tau _M})^*}{\tau _M}$. 
Suppose we have $\delta (q) = (p,q) \in {({\bar \tau _M})^*}TM = TM{ \times _M}TM$, then 
\begin{align}
D(q) = {({\tau _M})^*}{\bar \tau _M} \circ \delta (q) = {({\tau _M})^*}{\bar \tau _M}(p,q) = p, 
\end{align}
but since $D$ is an identity, we must have $p = q$, for all $q$, which means $\delta $ must be unique.  
\end{remark}

Indeed, consider a map $\sigma :\mathbb{R} \to M$, and then its tangent map $T\sigma $ 
will send the total derivative vector field $\displaystyle{\frac{d}{{dt}}}$ at 
$s \in \mathbb{R}$ to $TM$ by 
\begin{align}
{T_s}\sigma (\frac{d}{{dt}}) = {\left. {\frac{{d({x^\mu } \circ \sigma )}}{{dt}}} 
\right|_s}{\left( {\frac{\partial }{{\partial {x^\mu }}}} \right)_{\sigma (s)}},
\end{align} 
where in induced coordinates the components are the ${y^\mu }$ coordinates. 
To see this in the converse way, consider a smooth function $f \in {C^\infty }(M)$. 
$D(p)$ has the coordinate expression 
\begin{align}
D(p) = p = {y^\mu }(p){\left( {\frac{\partial }{{\partial {x^\mu }}}} \right)_{{\tau _M}(p)}},
\end{align} 
so we can define a new smooth function $Df$ on $TM$ by 
\begin{align}
Df(p): = D(p)f = {y^\mu }(p){\left( {\frac{{\partial f}}{{\partial {x^\mu }}}} \right)_{{\tau _M}(p)}},
\end{align}
for $\forall p \in TM$, and therefore, 
\begin{align}
Df = {y^\mu } \cdot \left( {\left( {\frac{{\partial f}}{{\partial {x^\mu }}}} \right) 
\circ {\tau _M}} \right).
\end{align} 

Consider a parameterisable curve $C$ on $M$, and let the parameterisation be 
$\sigma :\mathbb{R} \to M$, and its lift $\hat \sigma :\mathbb{R} \to TM$. 
Then, along this curve, 
\begin{align}
{\hat \sigma ^*}Df = Df \circ \hat \sigma  
= ({y^\mu } \circ \hat \sigma ) \cdot 
\left( {\frac{{\partial f}}{{\partial {x^\mu }}} \circ {\tau _M}} \right) \circ \hat \sigma  
= \frac{{d({x^\mu } \circ \sigma )}}{{dt}} \cdot \left( {\frac{{\partial f}}{{\partial {x^\mu }}} 
\circ {\tau _M}} \right) \circ \hat \sigma  = \frac{d}{{dt}}(f \circ \sigma ). \label{eq_deriv1stpb}
\end{align}

\begin{figure}
  \centering
  \includegraphics[width=7cm]{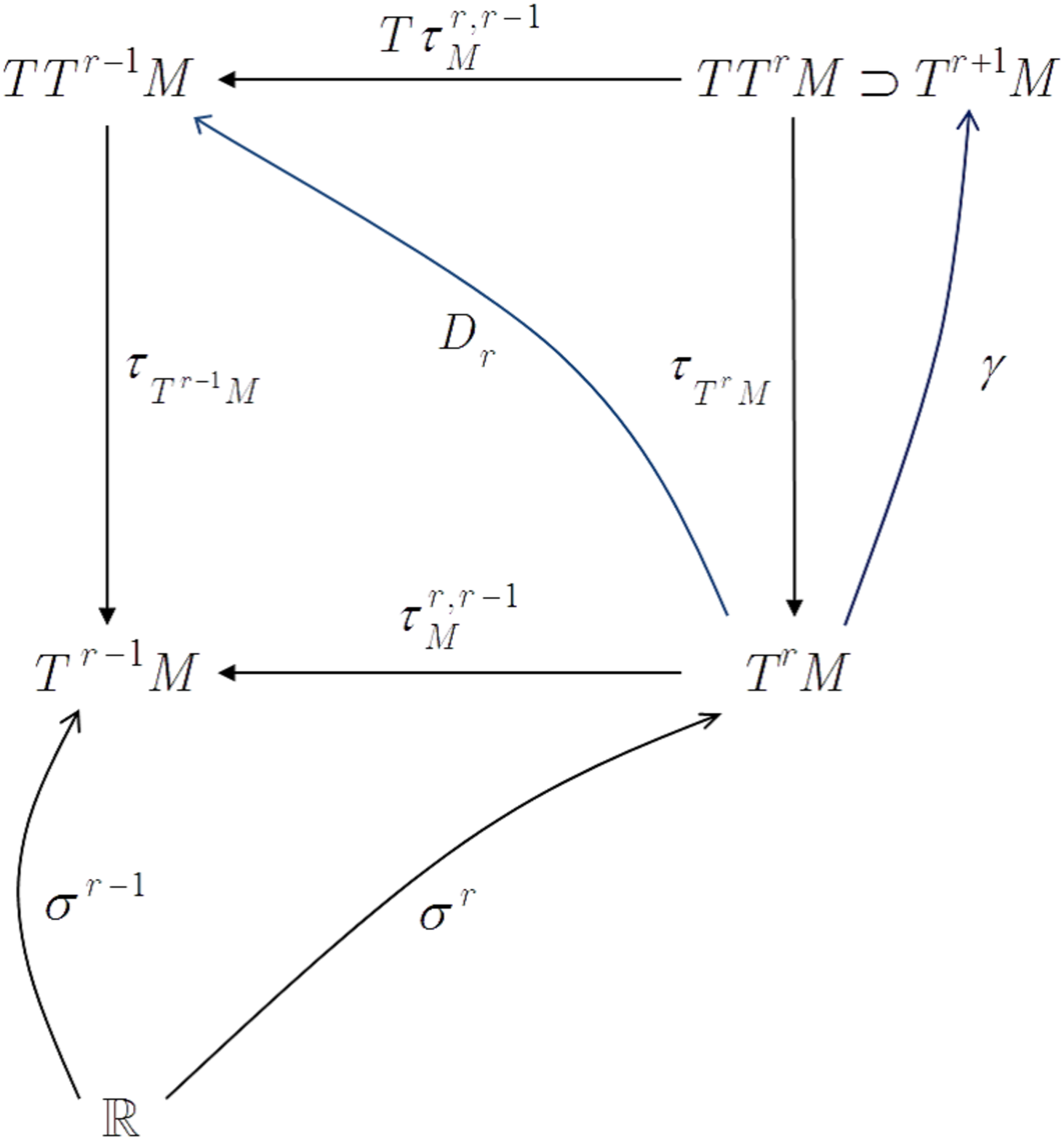}
  \caption{$r$-th order total derivative}\label{fig_r_totalderiv}
\end{figure}

The total derivative can be also introduced for the higher order cases (Figure \ref{fig_r_totalderiv}). 
Let $({V^r},{\psi ^r})$, ${\psi ^r} = (x_1^\mu ,x_2^\mu , \cdots ,x_{r + 1}^\mu )$ the induced chart on ${T^r}M$.
Similar to the Proposition \ref{pr_totalderiv}, we have the following. 
\begin{pr} \label{pr_totalderivr}
Consider the bundle morphism $(T\tau _M^{r,r - 1},\tau _M^{r,r - 1})$ from 
$(T{T^r}M, \lb[3] {\tau _{{T^r}M}}, \lb[3] {T^r}M)$ to $(T{T^{r - 1}}M,{\tau _{{T^{r - 1}}M}},{T^{r - 1}}M)$. 
Let $\gamma $ be the section of the sub-bundle 
$\tau _{M}^{r+1,r}:={{\left. {{\tau }_{{{T}^{r}}M}} \right|}_{{{T}^{r+1}}M}}$ 
of ${{\tau }_{{{T}^{r}}M}}$, 
and define a map ${{D}_{r}}:{{T}^{r}}M \to T{{T}^{r-1}}M$ by ${{D}_{r}}=T\tau _{M}^{r,r-1}\circ \gamma$. 
Then ${{D}_{r}}$ is an inclusion map  ${D_r} = {\iota _r}$, ${\iota _r}:{T^r}M \to T{T^{r-1}}M$, such that 
its local coordinate expression is given by 
\begin{align}
{D_r} = x_2^\mu  \cdot {\left( {\frac{\partial }{{\partial x_1^\mu }}} \right)_{\tau _M^{r,r - 1}( \cdot )}} 
+ x_3^\mu  \cdot {\left( {\frac{\partial }{{\partial x_2^\mu }}} \right)_{\tau _M^{r,r - 1}( \cdot )}} +  \cdots  
+ x_{r + 1}^\mu  \cdot {\left( {\frac{\partial }{{\partial x_r^\mu }}} \right)_{\tau _M^{r,r - 1}( \cdot )}}.
\end{align} 
\end{pr}

\begin{proof}
By the definition of ${T^{r + 1}}M$, we have
$T\tau _M^{r,r - 1}(\gamma (p)) = {\iota _r} \circ {\tau _{{T^r}M}}(\gamma (p))$ for $\forall p \in T^r M$. 
Then since ${{\left. {{\tau }_{{{T}^{r}}M}} \right|}_{{{T}^{r+1}}M}}\circ \gamma 
= {{\tau }_{{{T}^{r}}M}}\circ \gamma =i{{d}_{{{T}^{r}}M}}$, 
we have $T \tau _M^{r,r - 1} \circ \gamma  = \iota_r  \circ i{d_{{T^r}M}}$, 
which is an identity of ${{T}^{r}}M$ expressed as the element of $T{{T}^{r}}M$. 
Therefore, we can deduce its coordinate expression, 
\begin{align*}
{{D}_{r}}(p)=x_{2}^{\mu }(p){{\left( \frac{\partial }{\partial x_{1}^{\mu }} 
\right)}_{\tau _{M}^{r,r-1}(p)}}+x_{3}^{\mu }(p){{\left( \frac{\partial }{\partial x_{2}^{\mu }} 
\right)}_{\tau _{M}^{r,r-1}(p)}}+\cdots +x_{r+1}^{\mu }(p)
{{\left( \frac{\partial }{\partial x_{r}^{\mu }} \right)}_{\tau _{M}^{r,r-1}(p)}}, 
\end{align*}
for $\forall p \in {T^r}M$, and obtain,
\begin{align*}
{{D}_{r}}=x_{2}^{\mu }\cdot {{\left( \frac{\partial }{\partial x_{1}^{\mu }} 
\right)}_{\tau _{M}^{r,r-1}(\cdot )}}
+x_{3}^{\mu }\cdot {{\left( \frac{\partial }{\partial x_{2}^{\mu }} \right)}_{\tau _{M}^{r,r-1}(\cdot )}}
+\cdots +x_{r+1}^{\mu }\cdot {{\left( \frac{\partial }{\partial x_{r}^{\mu }} 
\right)}_{\tau _{M}^{r,r-1}(\cdot )}}. \label{totalderivr}
\end{align*}
\end{proof}

\begin{defn} $r$-th order total derivative  \label{def_totalderivr} \\ 
The identity map $D_r$ is called the {\it $r$-th order total derivative}. 
\end{defn}

The map $D_r$ for $r \ge 2$ cannot be expressed by a section of a vector bundle, 
therefore it is not possible to understand it as a vector field as in the case of $D$. 
As in the previous case, we consider the operation of the inclusion map ${D_r}:{T^r}M \to T{T^{r - 1}}M$ 
to a smooth function $g$ on $T^{r-1} M$, and define a new smooth function ${D_r}g$ on $T^r M$ by,  
\begin{align}
{D_r}g(w) &:= {D_r}(w)g = x_2^\mu (w){\left( {\frac{{\partial g}}{{\partial x_1^\mu }}} 
\right)_{\tau _M^{r,r - 1}(w)}} 
+ x_3^\mu (w){\left( {\frac{{\partial g}}{{\partial x_2^\mu }}} \right)_{\tau _M^{r,r - 1}(w)}} \nonumber  \\
&+  \cdots  + x_{r + 1}^\mu (w){\left( {\frac{{\partial g}}{{\partial x_r^\mu }}} \right)_{\tau _M^{r,r - 1}(w)}}, 
\end{align}
for $\forall w \in {T^r}M$, 
where $\tau _M^{r,r - 1}$ is a projection: $\tau _M^{r,r - 1}:{T^r}M \to {T^{r - 1}}M$, 
and defined iteratively by $\tau _M^{r,r - 1}: = {\left. {{\tau _{{T^{r - 1}}M}}} \right|_{{T^r}M}}$. 
Now we get, 
\begin{align}
{D_r}g = x_2^\mu \left( {\frac{{\partial g}}{{\partial x_1^\mu }} \circ \tau _M^{r,r - 1}} \right) 
+ x_3^\mu \left( {\frac{{\partial g}}{{\partial x_2^\mu }} \circ \tau _M^{r,r - 1}} \right) 
+  \cdots  + x_{r + 1}^\mu \left( {\frac{{\partial g}}{{\partial x_r^\mu }} \circ \tau _M^{r,r - 1}} \right).
\end{align}

Let us see more details, for the case of $r=2$ for simplicity.  
We will take the induced chart 
$({V^2},{\psi ^2})$, ${\psi ^2} = ({x^\mu },{y^\mu },{z^\mu })$ on $T^2 M$ for some readability, and denote, 
\begin{align}
{D_2} = {y^\mu }\frac{\partial }{{\partial {x^\mu }}} + {z^\mu }\frac{\partial }{{\partial {y^\mu }}}. \label{def_totalderiv2}
\end{align}
\begin{align}
{D_2}g = {y^\mu } \cdot \left( {\frac{{\partial g}}{{\partial {x^\mu }}} \circ \tau _M^{2,1}} \right) + {z^\mu } \cdot \left( {\frac{{\partial g}}{{\partial {y^\mu }}} \circ \tau _M^{2,1}} \right).
\end{align}
For the case where $g = Df$, we will have 
\begin{align}
&  {D_2}(Df) = {y^\mu } \cdot \left( {\frac{\partial }{{\partial {x^\mu }}}Df \circ \tau _M^{2,1}} \right) 
+ {z^\mu } \cdot \left( {\frac{\partial }{{\partial {y^\mu }}}Df \circ \tau _M^{2,1}} \right) \hfill \nonumber \\
&   = {y^\mu } \cdot \left( {\frac{\partial }{{\partial {x^\mu }}}
\left( {{y^\rho }\frac{{\partial f}}{{\partial {x^\rho }}} \circ {\tau _M}} \right) \circ \tau _M^{2,1}} \right) 
+ {z^\mu } \cdot \left( {\frac{\partial }{{\partial {y^\mu }}}\left( {{y^\rho }\frac{{\partial f}}
{{\partial {x^\rho }}} \circ {\tau _M}} \right) \circ \tau _M^{2,1}} \right) \hfill \nonumber \\
&   = {y^\mu }{y^\rho } \cdot \left( {\frac{{{\partial ^2}f}}{{\partial {x^\mu }\partial {x^\rho }}} 
\circ \tau _M^{2,0}} \right) + {z^\mu } \cdot \left( {\frac{{\partial f}}{{\partial {x^\mu }}} 
\circ \tau _M^{2,0}} \right).
\end{align}

To see this becomes the total derivative with respect to the parameterisation space, 
consider a parameterisable curve $C$ on $M$, and let the parameterisation 
be $\sigma :\mathbb{R} \to M$, its lift $\hat \sigma :\mathbb{R} \to TM$, 
and its second order lift ${\sigma ^2}:\mathbb{R} \to {T^2}M$. Then, along this curve, 
\begin{align}
&  {({\sigma ^2})^*}{D_2}g = {D_2}g \circ {\sigma ^2} = ({y^\mu } \circ {\sigma ^2}) \cdot 
\left( {\frac{{\partial g}}{{\partial {x^\mu }}} \circ \tau _M^{2,1}} \right) \circ {\sigma ^2} 
+ ({z^\mu } \circ {\sigma ^2}) \cdot \left( {\frac{{\partial g}}{{\partial {y^\mu }}} 
\circ \tau _M^{2,1}} \right) \circ {\sigma ^2} \hfill \nonumber \\
&   = ({y^\mu } \circ \hat \sigma ) \cdot \left( {\frac{{\partial g}}{{\partial {x^\mu }}}} \right)
 \circ \hat \sigma  + ({z^\mu } \circ {\sigma ^2}) \cdot \left( {\frac{{\partial g}}{{\partial {y^\mu }}}} \right) 
 \circ \hat \sigma  \hfill \nonumber \\
&   = \frac{d}{{dt}}(g \circ \hat \sigma ). \hfill 
\end{align}
In the case of $g = Df$, we can further use the relation (\ref{eq_deriv1stpb}), and get 
\begin{align}
{({\sigma ^2})^*}{D_2}(Df) = \frac{d}{{dt}}(Df \circ \hat \sigma ) 
= \frac{{{d^2}}}{{d{t^2}}}(f \circ \sigma ).
\end{align}
In order to see the connection to the derivatives with respect to the parameter space, we used the pull-back. 
Nevertheless, the total derivative on the total space itself can be defined by the bundle morphism only, 
and no consideration of curves on $M$ or its parameterisation is required. 

For general $r$, similar discussions could be made. For instance, 
we will have 
\begin{align}
&  {({\sigma ^r})^*}{D_r}({D_{r - 1}}( \cdots Df) \cdots ) = \frac{d}{{dt}}(({D_{r - 1}}( \cdots Df) \cdots ) 
\circ {\sigma ^{r - 1}}) \nonumber  \\
& \quad= \frac{{{d^2}}}{{d{t^2}}}(({D_{r - 2}}( \cdots Df) \cdots ) \circ {\sigma ^{r - 2}}) \hfill 
 = \cdots  = \frac{{{d^{r - 1}}}}{{d{t^{r - 1}}}}(Df \circ \hat \sigma ) = \frac{{{d^r}}}{{d{t^r}}}(f \circ \sigma ), \hfill 
 \label{pullback_Dr}
\end{align}
by iteration.

\subsection{Extremal and equations of motion} \label{subsec_extremal_2nd_mech}

	Having defined the action, we are able to derive the equations of motion by 
considering the extremal of the action. 
As in the previous section, we only consider global flows. 
Nevertheless with some details added; the formulation can be made similarly with local flows. 
Consider a ${C^\infty }$-flow, $\alpha :\mathbb{R} \times M \to M$, and its associated 
$1$-parameter group of transformations ${\{ {\alpha _s}\} _{s \in \mathbb{R}}}$ . 
The $1$-parameter group ${\alpha_s}:M \to M$ induces a $1$-parameter group 
$TT{\alpha _s}:{T^2}M \to {T^2}M$ generated by the induced tangent mapping. 
This will also modify the curve (arc segment) $C$ to $C' = {\alpha _s}(C)$, 
which will be now parameterised by $\sigma'$. 
As in the first order case, the variation will be expressed by the small deformations made to the action by 
${\alpha _s}$. 

Before proceeding, we will first check that the induced $1$-parameter group 
$TT{\alpha _s}:TTM \to TTM$ is also a $1$-parameter group on $T^2 M$. 

Let $(U,\varphi )$, $\varphi  = ({x^\mu })$ be a chart on $M$, 
$(V,\psi )$, $V = {\tau _M}^{ - 1}(U)$, $\psi  = ({x^\mu },{y^\mu })$ the induced chart on $TM$, 
$(\tilde V,\tilde \psi )$, 
$\tilde V = {\tau _{TM}}^{ - 1}(V)$, $\tilde \psi  = ({x^\mu },{y^\mu },{\dot x^\mu },{\dot y^\mu })$ 
the induced chart on $TTM$, 
and $({V^2},{\psi ^2})$ ${V^2} = {\left. {\tilde V} \right|_{{T^2}M}}$, 
${\psi ^2} = ({x^\mu },{y^\mu },{z^\mu })$ the induced chart on $T^2 M$. 
The projection maps are denoted by ${\tau _M}:TM \to M$, ${\tau _{TM}}:TTM \to TM$, 
$\tau _M^2: = {\tau _M} \circ {\tau _{TM}}$, $\tau _M^{2,1}:{T^2}M \to TM$, 
$\tau _M^{2,0}: = {\tau _M} \circ \tau _M^{2,1}$. 
The local expression of ${w_q} \in {T_q}TM$, $q \in TM$ is 
\begin{align}
{w_q} = {w^\mu }{\left( {\frac{\partial }{{\partial {x^\mu }}}} \right)_q} 
+ {\tilde w^\mu }{\left( {\frac{\partial }{{\partial {y^\mu }}}} \right)_q}.
\end{align} 
Then, $TT{\alpha _s}$ maps ${w_q}$by 
\begin{align}
&  TT{\alpha _s}({w_q}) = {\left. {\frac{{\partial ({x^\nu } \circ T{\alpha _s} 
  \circ {\psi ^{ - 1}})}}{{\partial {x^\mu }}}} \right|_{\psi (q)}}{w^\mu }
  {\left( {\frac{\partial }{{\partial {x^\nu }}}} \right)_{T{\alpha _s}(q)}} \hfill \nonumber \\
&   + {\left. {\frac{{\partial ({y^\nu } \circ T{\alpha _s} \circ {\psi ^{ - 1}})}}
{{\partial {x^\mu }}}} \right|_{\psi (q)}}{w^\mu }{\left( {\frac{\partial }{{\partial {y^\nu }}}} 
\right)_{T{\alpha _s}(q)}} + {\left. {\frac{{\partial ({y^\nu } \circ T{\alpha _s} 
\circ {\psi ^{ - 1}})}}{{\partial {y^\mu }}}} \right|_{\psi (q)}}{{\tilde w}^\mu }
{\left( {\frac{\partial }{{\partial {y^\nu }}}} \right)_{T{\alpha _s}(q)}}, \hfill  
\end{align}
Where $T{\alpha _s}$ is the induced $1$-parameter group on $TM$ by ${\alpha _s}$. 

In components of coordinates of $TTM$, this is 
\begin{align}
&  {x^\mu } \circ TT{\alpha _s}({w_q}) = {x^\mu } \circ T{\alpha _s}(q) = {x^\mu } \circ {\alpha _s} 
\circ {\tau _M}(q) = {x^\mu } \circ {\alpha _s} \circ {\tau _M} \circ {\tau _{TM}}({w_q}) = {x^\mu } 
\circ {\alpha _s} \circ \tau _M^2({w_q}), \hfill \nonumber \\
&  {y^\mu } \circ TT{\alpha _s}({w_q}) = {y^\mu } \circ T{\alpha _s}(q) 
= \left( {{{\left. {\frac{{\partial ({x^\mu } \circ {\alpha _s} \circ {\varphi ^{ - 1}})}}{{\partial {x^\nu }}}} 
\right|}_{\varphi (\tau _M^2( \cdot ))}}{y^\nu }} \right)({w_q}), \hfill \nonumber \\
&  {{\dot x}^\mu } \circ TT{\alpha _s}({w_q}) = {\left. {\frac{{\partial ({x^\mu } 
\circ T{\alpha _s} \circ {\psi ^{ - 1}})}}{{\partial {x^\nu }}}} \right|_{\psi ({\tau _{TM}}({w_q}))}}{w^\nu } 
= \left( {{{\left. {\frac{{\partial ({x^\mu } \circ {\alpha _s} \circ {\varphi ^{ - 1}})}}{{\partial {x^\nu }}}} 
\right|}_{\varphi (\tau _M^2( \cdot ))}}{{\dot x}^\nu }} \right)({w_q}), \hfill \nonumber \\
&  {{\dot y}^\mu } \circ TT{\alpha _s}({w_q}) = {\left. {\frac{{\partial ({y^\mu } 
\circ T{\alpha _s} \circ {\psi ^{ - 1}})}}{{\partial {x^\nu }}}} \right|_{\psi ({\tau _{TM}}({w_q}))}}{w^\nu } 
+ {\left. {\frac{{\partial ({y^\mu } \circ T{\alpha _s} \circ {\psi ^{ - 1}})}}{{\partial {y^\nu }}}} 
\right|_{\psi ({\tau _{TM}}({w_q}))}}{{\tilde w}^\nu } \hfill \nonumber \\
&  \quad  = {\left. {\frac{\partial }{{\partial {x^\nu }}}\left( {\left( {{{\left. {\frac{{\partial ({x^\mu } 
\circ {\alpha _s} \circ {\varphi ^{ - 1}})}}{{\partial {x^\rho }}}} \right|}_{\varphi ({\tau _M}( \cdot ))}}
{y^\rho }} \right) \circ {\psi ^{ - 1}}} \right)} \right|_{\psi ({\tau _{TM}}({w_q}))}}{{\dot x}^\nu }({w_q}) \hfill \nonumber \\
&  \quad \quad  + {\left. {\frac{\partial }{{\partial {y^\nu }}}\left( {\left( {{{\left. 
{\frac{{\partial ({x^\mu } \circ {\alpha _s} \circ {\varphi ^{ - 1}})}}{{\partial {x^\rho }}}} 
\right|}_{\varphi ({\tau _M}( \cdot ))}}{y^\rho }} \right) \circ {\psi ^{ - 1}}} \right)} 
\right|_{\psi ({\tau _{TM}}({w_q}))}}{{\dot y}^\nu }({w_q}) \hfill \nonumber \\
&  \quad  = \left( {{{\left. {\frac{{{\partial ^2}({x^\mu } \circ {\alpha _s} 
\circ {\varphi ^{ - 1}})}}{{\partial {x^\nu }\partial {x^\rho }}}} \right|}_{\varphi 
(\tau _M^2( \cdot ))}}{y^\rho }{{\dot x}^\nu } + {{\left. {\frac{{\partial ({x^\mu } \circ {\alpha _s} 
\circ {\varphi ^{ - 1}})}}{{\partial {x^\nu }}}} \right|}_{\varphi (\tau _M^2( \cdot ))}}{{\dot y}^\nu }} 
\right)({w_q}). \hfill 
\end{align}
In the case when ${w_q} \in {T^2}M$, we have 
${w^\mu } = {y^\mu }({w_q}) = {\dot x^\mu }({w_q})$, and the expressions become,
\begin{align}
&  {x^\mu } \circ TT{\alpha _s}({w_q}) = {x^\mu } \circ {\alpha _s} \circ \tau _M^{2,0}({w_q}), \hfill \nonumber \\
&  {y^\mu } \circ TT{\alpha _s}({w_q}) = \left( {{{\left. {\frac{{\partial ({x^\mu } \circ {\alpha _s} 
\circ {\varphi ^{ - 1}})}}{{\partial {x^\nu }}}} \right|}_{\varphi (\tau _M^{2,0}( \cdot ))}}{y^\nu }} 
\right)({w_q}), \hfill \nonumber \\
&  {{\dot x}^\mu } \circ TT{\alpha _s}({w_q}) = \left( {{{\left. {\frac{{\partial ({x^\mu } \circ {\alpha _s} 
\circ {\varphi ^{ - 1}})}}{{\partial {x^\nu }}}} \right|}_{\varphi (\tau _M^{2,0}( \cdot ))}}{y^\nu }} 
\right)({w_q}), \hfill \nonumber \\
&  {{\dot y}^\mu } \circ TT{\alpha _s}({w_q}) = \left( {{{\left. {\frac{{{\partial ^2}({x^\mu } 
\circ {\alpha _s} \circ {\varphi ^{ - 1}})}}{{\partial {x^\nu }\partial {x^\rho }}}} 
\right|}_{\varphi (\tau _M^{2,0}( \cdot ))}}{y^\rho }{y^\nu } 
+ {{\left. {\frac{{\partial ({x^\mu } \circ {\alpha _s} \circ {\varphi ^{ - 1}})}}{{\partial {x^\nu }}}} 
\right|}_{\varphi (\tau _M^{2,0}( \cdot ))}}{{\dot y}^\nu }} \right)({w_q}). \hfill  \label{comp_TTa}
\end{align}
Therefore, ${y^\mu } \circ TT{\alpha _s}({w_p}) = {\dot x^\mu } \circ TT{\alpha _s}({w_p})$, 
and the $1$-parameter group $TT{\alpha _s}$will take the elements of ${T^2}M$ to ${T^2}M$. 

\begin{figure}
  \centering
  \includegraphics[width=6cm]{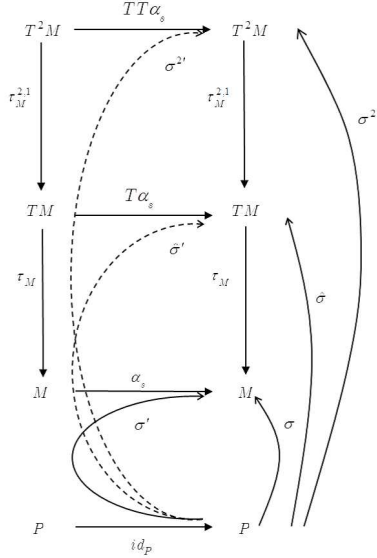}
  \caption{Second order mechanics} \label{fig_2nd_mech}
\end{figure}

With the above considerations, we define the following. 
\begin{defn} Variation of the action \\
Let $\xi$ be a vector field on $M$ which generates the $1$-parameter group $\alpha_s$, i.e., 
$\displaystyle{\xi  = {\left. {\frac{{d{\alpha _s}}}{{ds}}} \right|_{s = 0}}}$. 
We call the functional 
\begin{align}
&  {\delta _\xi }{S^\user2{\mathcal{K}}}(C): = \mathop {\lim }\limits_{s \to 0} \frac{1}{s}
  \left\{ {{S^\user2{\mathcal{K}}}({\alpha _s}(C)) - {S^{\mathcal{K}}}(C)} \right\} \hfill \nonumber \\
&   = \mathop {\lim }\limits_{s \to 0} \frac{1}{s}\left\{ {\int_{{{({\alpha _s}(\sigma ))}^2}(I)} 
{\mathcal{K}}  - \int_{{\sigma ^2}(I)} {\mathcal{K}} } \right\} \hfill  
\end{align}
the {\it variation of the action ${S^\user2{\mathcal{K}}}(C)$ with respect to the flow $\alpha$, 
associated to $\mathcal{K}$ }. 
\end{defn}

The second order lift of this modified parameterisation $\sigma '$ is given by, 
\begin{align}
{(\sigma ')^2} = {({\alpha _s} \circ \sigma  \circ i{d_I}^{ - 1})^2} = TT{\alpha _s} \circ {\sigma ^2} \circ i{d_I}^{ - 1}, 
\end{align}
which we show in figure \ref{fig_2nd_mech}. 
We can also check this easily. 
We will get, 
\begin{align}
&  {\delta _\xi }{S^{\mathcal{K}}}(C) = \mathop {\lim }\limits_{s \to 0} \frac{1}{s}\left\{ {\int_{TT{\alpha _s} 
\circ {\sigma ^2} \circ i{d_I}^{ - 1}(I)} {\mathcal{K}}  - \int_{{\sigma ^2}(I)} {\mathcal{K}} } \right\} \hfill \nonumber \\
&   = \mathop {\lim }\limits_{s \to 0} \frac{1}{s}\left\{ {\int_{{\sigma ^2}(I)} {TT{\alpha _s}^*{\mathcal{K}}} 
- \int_{{\sigma ^2}(I)} {\mathcal{K}} } \right\} = \int_{{\sigma ^2}(I)} {{L_X}{\mathcal{K}}}  \hfill \nonumber \\
&   = \int_{{C^2}} {{L_X}{\mathcal{K}}}  \hfill 
\end{align}
where $X$ is a vector field on ${{T}^{2}}M$ that generates the tangent 
$1$-parameter group $TT{{\alpha }_{s}}$, i.e., 
$\displaystyle{X={{\left. \frac{d(TT{{\alpha }_{s}})}{ds} \right|}_{s=0}}}$, and 
${L_X}$ is a Lie derivative with respect to $X$. 

We use the same definition of extremal given by Definition \ref{def_extremal}.
We will calculate the vector field $X$ and the equation of motion in local coordinates. 
Let $\xi $ be a vector field related to the $1$-parameter group ${\alpha _s}$, $\displaystyle{\xi  = \frac{{d{\alpha _s}}}{{ds}}}$, 
and its local expression $\displaystyle{\xi  = {\xi ^\mu }\frac{\partial }{{\partial {x^\mu }}}}$, 
where ${\xi ^\mu } \in {C^\infty }(M)$. 
We have already calculated the components of the map $TT{\alpha _s}$ in (\ref{comp_TTa}), 
which in the chart of $T^2 M$, becomes
\begin{align}
&  {x^\mu } \circ TT{\alpha _s}({w_q}) = {x^\mu } \circ {\alpha _s} \circ \tau _M^{2,0}({w_q}), \hfill \nonumber \\
&  {y^\mu } \circ TT{\alpha _s}({w_q}) = \left( {{{\left. {\frac{{\partial ({x^\mu } \circ {\alpha _s} 
\circ {\varphi ^{ - 1}})}}{{\partial {x^\nu }}}} \right|}_{\varphi (\tau _M^{2,0}( \cdot ))}}{y^\nu }} \right)
({w_q}), \hfill \nonumber \\
&  {z^\mu } \circ TT{\alpha _s}({w_q}) = \left( {{{\left. {\frac{{{\partial ^2}({x^\mu } 
\circ {\alpha _s} \circ {\varphi ^{ - 1}})}}{{\partial {x^\nu }\partial {x^\rho }}}} 
\right|}_{\varphi (\tau _M^{2,0}( \cdot ))}}{y^\nu }{y^\rho } 
+ {{\left. {\frac{{\partial ({x^\mu } \circ {\alpha _s} \circ {\varphi ^{ - 1}})}}{{\partial {x^\nu }}}} 
\right|}_{\varphi (\tau _M^{2,0}( \cdot ))}}{z^\nu }} \right)({w_q}). \hfill \label{comp_TTa2}
\end{align}
By these observations, the vector field $X$ on ${T^2}M$ at a point $w \in {T^2}M$ 
has a local expression, 
\begin{align}
&  {X_w} = {\left. {\frac{{d({x^\mu } \circ TT{\alpha _s})}}{{ds}}} \right|_{s = 0}}
{\left( {\frac{\partial }{{\partial {x^\mu }}}} \right)_w} 
+ {\left. {\frac{{d({y^\mu } \circ TT{\alpha _s})}}{{ds}}} \right|_{s = 0}}
{\left( {\frac{\partial }{{\partial {y^\mu }}}} \right)_w} 
+ {\left. {\frac{{d({z^\mu } \circ TT{\alpha _s})}}{{ds}}} 
\right|_{s = 0}}{\left( {\frac{\partial }{{\partial {z^\mu }}}} \right)_w} \hfill \nonumber \\
&  \quad  = {\left. {\frac{d}{{ds}}({x^\mu } \circ {\alpha _s} \circ \tau _M^{2,0})} 
\right|_{s = 0}}{\left( {\frac{\partial }{{\partial {x^\mu }}}} \right)_w} 
+ {y^\nu }(w)\frac{d}{{ds}}{\left( {{{\left. {\frac{{\partial ({x^\mu } \circ {\alpha _s} 
\circ {\varphi ^{ - 1}})}}{{\partial {x^\nu }}}} \right|}_{\varphi (\tau _M^{2,0}(w))}}} \right)_{s = 0}}
{\left( {\frac{\partial }{{\partial {y^\mu }}}} \right)_w} \hfill \nonumber \\
&  \quad \quad  + \frac{d}{{ds}}\left( {{{\left. {\frac{{{\partial ^2}({x^\mu } \circ {\alpha _s} 
\circ {\varphi ^{ - 1}})}}{{\partial {x^\nu }\partial {x^\rho }}}} \right|}_{\varphi (\tau _M^{2,0}(w))}}
{y^\nu }{y^\rho } + {{\left. {\frac{{\partial ({x^\mu } \circ {\alpha _s} \circ {\varphi ^{ - 1}})}}
{{\partial {x^\nu }}}} \right|}_{\varphi (\tau _M^{2,0}(w))}}{z^\nu }} \right)
{\left( {\frac{\partial }{{\partial {z^\mu }}}} \right)_w} \hfill \nonumber \\
&  \quad  = \left( {{\xi ^\mu } \circ \tau _M^{2,0}} \right)(w){\left( {\frac{\partial }
{{\partial {x^\mu }}}} \right)_w} + \left( {\frac{{\partial {\xi ^\mu }}}{{\partial {x^\nu }}} 
\circ \tau _M^{2,0} \cdot {y^\nu }} \right)(w){\left( {\frac{\partial }{{\partial {y^\mu }}}} \right)_w} \hfill \nonumber \\
&  \quad \quad  + \left( {\frac{{{\partial ^2}{\xi ^\mu }}}
{{\partial {x^\nu }\partial {x^\rho }}} \circ \tau _M^{2,0} \cdot {y^\nu }{y^\rho } 
+ \frac{{\partial {\xi ^\mu }}}{{\partial {x^\nu }}} \circ \tau _M^{2,0} 
\cdot {z^\nu }} \right)(w){\left( {\frac{\partial }{{\partial {z^\mu }}}} \right)_w}. \hfill 
\end{align}
therefore, 
\begin{align}
X = ({\xi ^\mu } \circ \tau _M^{2,0})\frac{\partial }{{\partial {x^\mu }}} 
+ \left( {\frac{{\partial {\xi ^\mu }}}{{\partial {x^\nu }}} \circ \tau _M^{2,0} \cdot {y^\nu }} \right)
\frac{\partial }{{\partial {y^\mu }}} + \left( {\frac{{{\partial ^2}{\xi ^\mu }}}
{{\partial {x^\nu }\partial {x^\rho }}} \circ \tau _M^{2,0} \cdot {y^\nu }{y^\rho } 
+ \frac{{\partial {\xi ^\mu }}}{{\partial {x^\nu }}} \circ \tau _M^{2,0} \cdot {z^\nu }} \right)
\frac{\partial }{{\partial {z^\mu }}}. 
\end{align}
We can make this expression shorter by using the total derivatives we defined in Definition \ref{def_totalderiv1} and 
\ref{def_totalderivr} (or (\ref{def_totalderiv2}) ), 
\begin{align}
X = {\xi ^\mu } \circ \tau _M^{2,0}\frac{\partial }{{\partial {x^\mu }}} 
+ D{\xi ^\mu } \circ \tau _M^{2,1}\frac{\partial }{{\partial {y^\mu }}} 
+ {D_2}(D{\xi ^\mu })\frac{\partial }{{\partial {z^\mu }}}. \label{pro.v.f.onT2M}
\end{align}

The Lie derivative ${L_X}{\mathcal{K}}$ in coordinate expression becomes, 
\begin{align}
&  {L_X}{\mathcal{K}} = {L_X}\left( {\frac{{\partial K}}{{\partial {y^\rho }}}d{x^\rho } + 2\frac{{\partial K}}{{\partial {z^\rho }}}d{y^\rho }} \right) \hfill \nonumber\\
&  \quad  = X\left( {\frac{{\partial K}}{{\partial {y^\rho }}}} \right)d{x^\rho } + \frac{{\partial K}}{{\partial {y^\rho }}}d{L_X}{x^\rho } + 2X\left( {\frac{{\partial K}}{{\partial {z^\rho }}}} \right)d{y^\rho } + 2\frac{{\partial K}}{{\partial {z^\rho }}}d{L_X}{y^\rho } \hfill \nonumber\\
&  \quad  = \left\{ {{\xi ^\mu } \circ \tau _M^{2,0} \cdot \left( {\frac{{{\partial ^2}K}}{{\partial {x^\mu }\partial {y^\rho }}}} \right) + D{\xi ^\mu } \circ \tau _M^{2,1} \cdot \left( {\frac{{{\partial ^2}K}}{{\partial {y^\mu }\partial {y^\rho }}}} \right) + {D_2}(D{\xi ^\mu })\frac{{{\partial ^2}K}}{{\partial {z^\mu }\partial {y^\rho }}}} \right\}d{x^\rho } + \frac{{\partial K}}{{\partial {y^\rho }}}d{\xi ^\rho } \hfill \nonumber\\
&  \quad \quad  + 2\left\{ {{\xi ^\mu } \circ \tau _M^{2,0} \cdot 
\left( {\frac{{{\partial ^2}K}}{{\partial {x^\mu }\partial {z^\rho }}}} \right) 
+ D{\xi ^\mu } \circ \tau _M^{2,1} \cdot 
\left( {\frac{{{\partial ^2}K}}{{\partial {y^\mu }\partial {z^\rho }}}} \right) 
+ {D_2}(D{\xi ^\mu })\frac{{{\partial ^2}K}}{{\partial {z^\mu }\partial {z^\rho }}}} \right\}d{y^\rho } \nonumber \\
&  \quad \quad + 2\frac{{\partial K}}{{\partial {z^\rho }}}d\left( {D{\xi ^\mu }} \right) \hfill \nonumber\\
&  \quad  = {\xi ^\mu } \circ \tau _M^{2,0} \cdot \left\{ {\frac{{{\partial ^2}K}}{{\partial {x^\mu }\partial {y^\rho }}}d{x^\rho } - d\left( {\frac{{\partial K}}{{\partial {y^\mu }}}} \right) + 2\frac{{{\partial ^2}K}}{{\partial {x^\mu }\partial {z^\rho }}}d{y^\rho }} \right\} \hfill \nonumber\\
&  \quad \quad  + D{\xi ^\mu } \circ \tau _M^{2,1} \cdot \left( {\frac{{{\partial ^2}K}}{{\partial {y^\mu }\partial {y^\rho }}}d{x^\rho } + 2\frac{{{\partial ^2}K}}{{\partial {y^\mu }\partial {z^\rho }}}d{y^\rho } - 2d\left( {\frac{{\partial K}}{{\partial {z^\mu }}}} \right)} \right) \hfill \nonumber\\
&  \quad \quad  + {D_2}(D{\xi ^\mu })\left( {\frac{{{\partial ^2}K}}{{\partial {z^\mu }\partial {y^\rho }}}d{x^\rho } + 2\frac{{{\partial ^2}K}}{{\partial {z^\mu }\partial {z^\rho }}}d{y^\rho }} \right) \hfill \nonumber\\
&  \quad \quad  + d\left( {{\xi ^\mu } \circ \tau _M^{2,0} \cdot \frac{{\partial K}}{{\partial {y^\mu }}} 
+ 2D{\xi ^\mu } \circ \tau _M^{2,1} \cdot \frac{{\partial K}}{{\partial {z^\mu }}}} \right). \hfill 
 \label{Lie_2nd_mech}
\end{align}
The result of (\ref{Lie_2nd_mech}) is called the {\it infinitesimal first variation formula} 
for the Finsler-Kawaguchi form $\mathcal{K}$.

The variation of action becomes, 
\begin{align}
&  {\delta _\xi }{S^{\mathcal{K}}}(C) = \int_{{\sigma ^2}(I)} {{L_X}{\mathcal{K}}}  = \int_I {{\sigma ^2}^*{L_X}{\mathcal{K}}}  \hfill \nonumber\\
&  \quad  = \int_{{C^3}} {{\xi ^\mu } \circ \tau _M^{3,0} \cdot \left[ {\left\{ {\frac{{{\partial ^2}K}}{{\partial {x^\mu }\partial {y^\rho }}}d{x^\rho } + 2\frac{{{\partial ^2}K}}{{\partial {x^\mu }\partial {z^\rho }}}d{y^\rho } - d\left( {\frac{{\partial K}}{{\partial {y^\mu }}}} \right)} \right\} \circ \tau _M^{3,2} + d\left( {{D_3}\frac{{\partial K}}{{\partial {z^\mu }}}} \right)} \right]}  \hfill \nonumber\\
&  \quad \quad  + \int_{{C^3}} {d\left( {{\xi ^\mu } \circ \tau _M^{3,0} \cdot \frac{{\partial K}}{{\partial {y^\mu }}} \circ \tau _M^{3,2} + 2D{\xi ^\mu } \circ \tau _M^{3,1} \cdot \frac{{\partial K}}{{\partial {z^\mu }}} \circ \tau _M^{3,2} - {D_3}\left( {{\xi ^\mu } \circ \tau _M^{2,0} \cdot \frac{{\partial K}}{{\partial {z^\mu }}}} \right)} \right)}, \label{variation_2nd_mech} 
\end{align}
which is called the {\it integral first variation formula}. 
$D_3$ is the total derivative defined by Definition \ref{def_totalderivr}. 
We used the homogeneity condition: 	
\begin{align}
&  K = \frac{{\partial K}}{{\partial {y^\mu }}}{y^\mu } 
+ 2\frac{{\partial K}}{{\partial {z^\mu }}}{z^\mu },
\quad \,\frac{{\partial K}}{{\partial {z^\mu }}}{y^\mu } = 0 \hfill \nonumber \\
&  \frac{{{\partial ^2}K}}{{\partial {x^\rho }\partial {y^\mu }}}{y^\mu } 
+ 2\frac{{{\partial ^2}K}}{{\partial {x^\rho }\partial {z^\mu }}}{z^\mu } 
= \frac{{\partial K}}{{\partial {x^\rho }}}, \hfill \nonumber \\
&  \frac{{{\partial ^2}K}}{{\partial {y^\rho }\partial {y^\mu }}}{y^\mu } 
+ 2\frac{{{\partial ^2}K}}{{\partial {y^\rho }\partial {z^\mu }}}{z^\mu } = 0, \hfill \nonumber \\
&  \frac{{{\partial ^2}K}}{{\partial {y^\rho }\partial {z^\mu }}}{y^\mu } 
+ \frac{{\partial K}}{{\partial {z^\rho }}} = 0,\quad \frac{{{\partial ^2}K}}
{{\partial {y^\rho }{z^\mu }}}{y^\mu } + \frac{{\partial K}}{{\partial {y^\rho }}} = 0, \hfill \label{cond_2nd_hom}
\end{align}
and its pull-back relations on the parameter space.

The detailed calculations are shown in the Appendix. 

Now we can obtain the following theorem. 
\begin{theorem} Extremals \label{thm_1st_var_2nd} \\
Let $C$ be an arc segment. 
The following statements are equivalent. 
\begin{enumerate}
\item $C$ is an extremal. 
\item The equation 
\begin{align}
&  {\mathcal{E}}{{\mathcal{L}}^K}_\mu  \circ {\sigma ^3} = 0, \hfill \nonumber \\
&  {\mathcal{E}}{{\mathcal{L}}^K}_\mu : = 
\left\{ {\frac{{{\partial ^2}K}}{{\partial {x^\mu }\partial {y^\rho }}}d{x^\rho } 
+ 2\frac{{{\partial ^2}K}}{{\partial {x^\mu }\partial {z^\rho }}}d{y^\rho } 
- d\left( {\frac{{\partial K}}{{\partial {y^\mu }}}} \right)} \right\} \circ \tau _M^{3,2} 
+ d\left( {{D_3}\frac{{\partial K}}{{\partial {z^\mu }}}} \right), \hfill \label{2nd_EL_pullback}
\end{align}
holds for arbitrary parameterisation $\sigma$. 
\end{enumerate}
\end{theorem}
\begin{proof}
Suppose $C$ is an extremal. Then, by definition, for all ${\alpha _s}:M \to M$, 
such that does not change the boundary of $C$, we have ${\delta _\xi }S(C) = 0$. 
On the other hand, the last term in (\ref{variation_2nd_mech}) becomes $0$, 
since it is the boundary term. Therefore, 
we have, 
\begin{align}
\int_{{C^3}} {{\xi ^\mu } \circ \tau _M^{3,0}\left[ {\left\{ {\frac{{{\partial ^2}K}}{{\partial {x^\mu }\partial {y^\rho }}}d{x^\rho } + 2\frac{{{\partial ^2}K}}{{\partial {x^\mu }\partial {z^\rho }}}d{y^\rho } - d\left( {\frac{{\partial K}}{{\partial {y^\mu }}}} \right)} \right\} \circ \tau _M^{3,2} + d\left( {{D_3}\frac{{\partial K}}{{\partial {z^\mu }}}} \right)} \right]}  = 0.
\end{align}
Since this relation must be true for all $\xi $, which is the generator of $\alpha$, we have 
\begin{align}
\left[ {\left\{ {\frac{{{\partial ^2}K}}{{\partial {x^\mu }\partial {y^\rho }}}d{x^\rho } 
+ 2\frac{{{\partial ^2}K}}{{\partial {x^\mu }\partial {z^\rho }}}d{y^\rho } 
- d\left( {\frac{{\partial K}}{{\partial {y^\mu }}}} \right)} \right\} \circ \tau _M^{3,2} 
+ d\left( {{D_3}\frac{{\partial K}}{{\partial {z^\mu }}}} \right)} \right] \circ {\sigma ^3} = 0,
\end{align}
for any parameterisation $\sigma$. 
To prove the converse, take the similar steps backwards. 
\end{proof}

The equations (\ref{2nd_EL_pullback}) are called the {\it Euler-Lagrange equations} or 
{\it equations of motion} of the second order Lagrangian $\mathcal{K}$.  

\begin{defn} Symmetry of the dynamical system \\ 
Let $u$ be a vector field over $M$, and $Y$ an induced vector field by $u$ over $T^2 M$.  
We say that {\it $\mathcal{K}$ is invariant with respect to $u$}, if
\begin{align}
{L_Y}{\mathcal{K}} = 0, 
\end{align}
and $u$ is called a {\it symmetry} of the dynamical system $(M,\mathcal{K})$. 
We also say that $u$ generates the invariant transformations on the 
Finsler-Kawaguchi manifold $(M, \mathcal{K})$.
\end{defn}
Now we will have the following conservation law. 
\begin{theorem}  Noether (second order)\\
Suppose we are given a symmetry of second order Finsler-Kawaguchi manifold $(M, \mathcal{K})$. 
Then there exists a function $f$ on $T^3 M$, 
which along the extremal $\gamma$ of $S^\mathcal{K}$ satisfies, 
\begin{align}
\int_{{\gamma ^3}} {df}  = 0,  
\label{eq_conserved_2nd}
\end{align}
for any parameterisation $\sigma$ which parameterise $\gamma$.
\end{theorem}

\begin{proof}
Let the symmetry be $u$, with its local coordinate expression 
$\displaystyle{u = {u^\mu }\frac{\partial }{{\partial {x^\mu }}}}$, 
and the induced vector field $Y$. 
Then from (\ref{variation_2nd_mech}), we have 
\begin{align}
&0 = \int_{{\gamma ^3}} {{L_Y}{\mathcal{K}}}  \hfill \nonumber \\
&   = \int_{{\gamma ^3}} {{u^\mu } \circ \tau _M^{3,0} \cdot 
\left[ {\left\{ {\frac{{{\partial ^2}K}}{{\partial {x^\mu }\partial {y^\rho }}}d{x^\rho } 
+ 2\frac{{{\partial ^2}K}}{{\partial {x^\mu }\partial {z^\rho }}}d{y^\rho } 
- d\left( {\frac{{\partial K}}{{\partial {y^\mu }}}} \right)} \right\} \circ \tau _M^{3,2} 
+ d\left( {{D_3}\frac{{\partial K}}{{\partial {z^\mu }}}} \right)} \right]}  \hfill \nonumber \\
&  \quad \quad  + \int_{{\gamma ^3}} {d\left( {{u^\mu } \circ \tau _M^{3,0} \cdot 
\frac{{\partial K}}{{\partial {y^\mu }}} \circ \tau _M^{3,2} + 2D{u^\mu } \circ \tau _M^{3,1} 
\cdot \frac{{\partial K}}{{\partial {z^\mu }}} \circ \tau _M^{3,2} 
- {D_3}\left( {{u^\mu } \circ \tau _M^{2,0} \cdot \frac{{\partial K}}{{\partial {z^\mu }}}} \right)} \right)}  \hfill \nonumber \\
&   = \int_{{\gamma ^3}} {d\left( {{u^\mu } \circ \tau _M^{3,0} 
\cdot \frac{{\partial K}}{{\partial {y^\mu }}} \circ \tau _M^{3,2} 
+ 2D{u^\mu } \circ \tau _M^{3,1} \cdot \frac{{\partial K}}{{\partial {z^\mu }}} \circ \tau _M^{3,2} 
- {D_3}\left( {{u^\mu } \circ \tau _M^{2,0} \cdot \frac{{\partial K}}{{\partial {z^\mu }}}} \right)} \right)} . \hfill \nonumber \\ 
\end{align}
The third equality comes from the fact we consider along the extremal $\gamma$. 
Therefore we have a function on $T^3 M$,
\begin{align}
f = {u^\mu } \circ \tau _M^{3,0} \cdot \frac{{\partial K}}{{\partial {y^\mu }}} \circ \tau _M^{3,2} + 2D{u^\mu } \circ \tau _M^{3,1} \cdot \frac{{\partial K}}{{\partial {z^\mu }}} \circ \tau _M^{3,2} - {D_3}\left( {{u^\mu } \circ \tau _M^{2,0} \cdot \frac{{\partial K}}{{\partial {z^\mu }}}} \right)
\end{align}
such that satisfies the condition. 
\end{proof}

We call the relation (\ref{eq_conserved_2nd}), the {\it conservation law}. 
We can express the conservation law (\ref{eq_conserved_2nd}) by 
taking arbitrary parameterisation for this $\gamma$. 

\begin{defn} Noether current \\
The quantity $f$ is called the {\it Noether current associated with $u$}.  
\end{defn}

\begin{remark} \label{convL-2ndFins-rel}
As in the case of first order mechanics (Remark \ref{convL-Fins-rel}), here we will show that 
given a second order ``conventional'' Lagrange function on a certain coordinate chart on $J^2 Y$, 
we can construct its homogeneous counterpart. In general, this cannot be done globally. ( It is possible only when $Y=\mathbb{R}\times Q$. )
Let $(U,\psi )$, $\psi =(t,{{q}^{i}})$, $i=1,\cdots ,n$ 
be the adapted chart on a $(n+1)$-dimensional manifold $Y$, 
and the induced chart on $\mathbb{R}$ be $(\pi(U), t)$. 
We denote the induced chart on ${{J}^{2}}Y$ 
by $({{({{\pi }^{2,0}})}^{-1}}(U),{{\psi }^{2}}),{{\psi }^{2}}
=(t,{{q}^{i}},{{\dot{q}}^{i}},{{\ddot{q}}^{i}})$.
Take the induced chart on ${{T}^{2}}Y$ as $(V,{{\tilde{\psi }}^{2}})$, 
$V={{(\tau_{Y}^{2,0})}^{-1}}(U)$, 
${{\tilde{\psi }}^{2}}=({{x}^{0}},{{x}^{i}},{{y}^{0}},{{y}^{i}},{{z}^{0}},{{z}^{i}})$, 
$i=1, \cdots , n$, such that ${{y}^{0}} \ne 0$. 
It is always possible to choose such coordinates for a single chart. 
(In order to avoid confusion we use different symbols, 
but clearly ${x^0} = t \circ {\tau _Y},{x^i} = {q^i} \circ {\tau _Y}$.) 
Now consider a map $\rho :V\hookrightarrow {{J}^{2}}Y,\ $, $\rho (V)={{({{\pi }^{2,0}})}^{-1}}(U)$, 
which in coordinates are defined by 
\begin{align}
t\circ \rho ={{x}^{0}},\ {{q}^{i}}\circ \rho ={{x}^{i}},
\ {{\dot{q}}^{i}}\circ \rho =\frac{{{y}^{i}}}{{{y}^{0}}},
\ {{\ddot{q}}^{i}}\circ \rho =\frac{{{z}^{i}}{{y}^{0}}-{{z}^{0}}{{y}^{i}}}{{{({{y}^{0}})}^{3}}}.
\end{align}
Let $K$ be a function on $V$, defined by, 
\begin{align}
K/{y^0} = {\rho ^*}{\mathcal{L}}.
\end{align}
In coordinates, 
 \begin{align}
 K({x^\mu },{y^\mu },{z^\mu }) = {\mathcal{L}}(t \circ \rho ,{q^i} \circ \rho ,
 {\dot q^i} \circ \rho ,{\ddot q^i} \circ \rho ){y^0} 
 = {\mathcal{L}}({x^0},{x^i},\frac{{{y^i}}}{{{y^0}}},
 \frac{{{z^i}{y^0} - {z^0}{y^i}}}{{{{({y^0})}^3}}})\,{y^0} \label{convL-2ndFinsL}
\end{align}
for $\mu =0,\cdots ,n$, $i=1,\cdots ,n$. 
Then on $V$, $K$ satisfies the homogeneity function, and we can obtain the local expression of 
Finsler-Kawaguchi form 
$\displaystyle{{\mathcal{K}} = \frac{{\partial K}}{{\partial {y^\mu }}}d{x^\mu } 
+ 2\frac{{\partial K}}{{\partial {z^\mu }}}d{y^\mu }}$ on $V$. 
This expression can be used for the reparameterisation invariant action provided that the arc segment 
where the integration is carried out is covered by this single chart. 
\end{remark}

\begin{remark} 
Similarly as in the first order case, locally (on a single chart) we can also construct 
a different Finsler-Kawaguchi function and local Finsler-Kawaguchi form from ${\mathcal{L}}$ 
by choosing an appropriate map $\rho$, which on this single chart can be used 
for reparameterisation invariant action. 
\end{remark} 

\newpage
\section{First order field theory} \label{sec_1st_field}
Here we present the Lagrange formulations for the first order field theory, 
in terms of Kawaguchi geometry. 

By the term {\it first order field theory}, we 
mean that the total space we are considering is 
the first order $k$-multivector bundle ${{\Lambda }^{k}}TM$ with $k$-patch in its base space. 
The total space is sometimes also called the {\it ambient space},
and its coordinate functions represents the physical fields as well as the 
spacetime coordinates. 
In this sense, the coordinate functions of $M$ are regarded as unified variables, 
and the $k$-dimensional submanifold of $M$ represents the actual spacetime. 
 
	The basic structure we consider in this section is introduced in  
Chapter 2 and 4 (Section \ref {sec_1st_k_Kawaguchi}), 
namely the $n$-dimensional $k$-areal Kawaguchi manifold $(M,\mathcal{K})$, 
the $k$-multivector bundle $({{\Lambda }^{k}}TM,{{\Lambda }^{k}}{{\tau }_{M}},M)$, 
and a $k$-curve ($k$-patch) $\Sigma $ on $M$, parameterised by $\sigma$. 
 
We take the Kawaguchi form $\mathcal{K}$ as the Lagrangian, 
and the action will be defined by considering the integration 
over the lift of the parameterisable $k$-curve ($k$-patch). 
The Euler-Lagrange equations are derived by taking the variation of the action with respect 
to the flow on $M$ that deforms the $k$-patch $\Sigma$, and fixed on the boundary. 
We can show that the action 
and consequently the Euler-Lagrange equations are 
independent with respect to the parameterisation belonging to the same equivalent class.

\subsection{Action} \label{subsec_1st_field_Action}
Suppose we have a dynamical system (differential equations expressing motions) 
where the configurations of the $k$-dimensional spacetime (or any extended object of dimension $k$) 
is expressed as a smooth $k$-patch $\Sigma $ of a parameterisable $k$-area, 
such that $\Sigma =\sigma (P)\subset M$, where $P$ is a closed rectangle 
$P=[t_{i}^{1},t_{f}^{1}] \times [t_{i}^{2},t_{f}^{2}] \times \cdots \times [t_{i}^{k},t_{f}^{k}] 
\subset {{\mathbb{R}}^{k}}$.

When we can express this system by first order $k$-areal Kawaguchi geometry, 
namely the pair $(M, \mathcal{K})$ where $\mathcal{K}$ is a first order Kawaguchi $k$-form, 
we refer to this dynamical system as {\it first order fields},  
and conversely call the pair $(M,\mathcal{K})$ a {\it dynamical system}, 
and $\mathcal{K}$ the {\it Lagrangian} of first order fields. 

The action of first order fields is defined as follows.

\begin{defn} Action of first order fields \\
Let  $(M, {\mathcal{K}})$ be a $n$-dimensional $k$-areal Kawaguchi manifold.
	Consider a $k$-multivector bundle $({{\Lambda }^{k}}TM,{{\Lambda }^{k}}{{\tau }_{M}},M)$ and let 
$(U, \varphi)$, $\varphi =({{x}^{\mu }})$ be a chart on $M$, 
and $(V,\psi)$, 
$\psi=({{x}^{\mu }}, {{y}^{{{\mu }_{1}} \cdots {{\mu }_{k}}}})$ 
the induced chart on ${{\Lambda }^{k}}TM$. 

The local expression of the Kawaguchi form $\mathcal{K} \in {{\Omega }^{k}}({{\Lambda }^{k}}TM)$ 
is given by 
\begin{align}
{\mathcal{K}=\frac{1}{k!} 
\frac{\partial K}{\partial {{y}^{{{\mu }_{1}} 
\cdots {{\mu }_{k}}}}}d{{x}^{{{\mu }_{1}}}} \wedge 
\cdots \wedge d{{x}^{{{\mu }_{k}}}}},
\end{align}
where $K$ is the first order $k$-areal Kawaguchi function.  

Let $\Sigma $ be a $k$-patch on $M$, and $\sigma $ its parameterisation,
$\sigma (P)=\Sigma \subset M$ with $P=[t_{i}^{1},t_{f}^{1}]\times [t_{i}^{2},t_{f}^{2}]\times \cdots \times [t_{i}^{k},t_{f}^{k}]\subset {{\mathbb{R}}^{k}}$, 
and $\hat{\sigma}$ the lift of $\sigma$, 
defined in chapter 4 (Definition \ref{def_liftedparameterisation_k}).
We call the functional ${S}^{\mathcal{K}}(\Sigma)$ defined by 
\begin{align}
{S}^{\mathcal{K}}(\Sigma )=\int_{\hat{\Sigma }}{\mathcal{K}}
=\int_{\hat{\sigma }(P)}{\frac{1}{k!}\frac{\partial K}{\partial {{y}^{{{\mu }_{1}}
\cdots {{\mu }_{k}}}}}d{{x}^{{{\mu }_{1}}\cdots {{\mu }_{k}}}}}  \label{KawaguchiAction}
\end{align}
the {\it action of first order field theory associated with $\mathcal{K}$.} 
\end{defn}

As we have seen in Section \ref{subsec_paraminv_k_kawaguchi}, 
Lemma \ref{lem_repinv_kawaguchi_1st}, 
Kawaguchi area is invariant with respect to the reparameterisation, 
therefore the action is also invariant.  

\subsection{Extremal and equations of motion}  \label{subsec_extremal_1st_field}
	Now we will derive the Euler-Lagrange equations by considering the extremal of the action. 
Again, we only consider global flows in this text. 
Nevertheless, with some details added, the formulation can be set up similarly with local flows. 
Consider a ${C^\infty }$-flow, $\alpha :\mathbb{R} \times M \to M$, and its associated 
$1$-parameter group of transformations ${\{ {\alpha _s}\} _{s \in \mathbb{R}}}$ . 
The $1$-parameter group ${\alpha _s}:M \to M$ induces a multi-tangent 
$1$-parameter group 
${{\Lambda }^{k}}T{{\alpha }_{s}}:{{\Lambda }^{k}}TM \to {{\Lambda }^{k}}TM$ 
generated by the tangent mapping of multivectors. 
This will also deform the $k$-area ($k$-patch) $\Sigma $ to ${\Sigma }'={{\alpha }_{s}}(\Sigma )$, 
and since this is a smooth deformation, it again becomes a parameterisable area. 
By the reparameterisation independence, 
we can always choose the parameterisation of this deformed ${\Sigma }'$ by a new 
${\sigma }':P\to M$, ${\sigma }'(P)={\Sigma }'$, 
so that it has the same parameter space as $\Sigma $.  
The variation of the action will be expressed by the small deformations made to the action by 
${\alpha _s}$. 

\begin{figure}
  \centering
  \includegraphics[width=7cm]{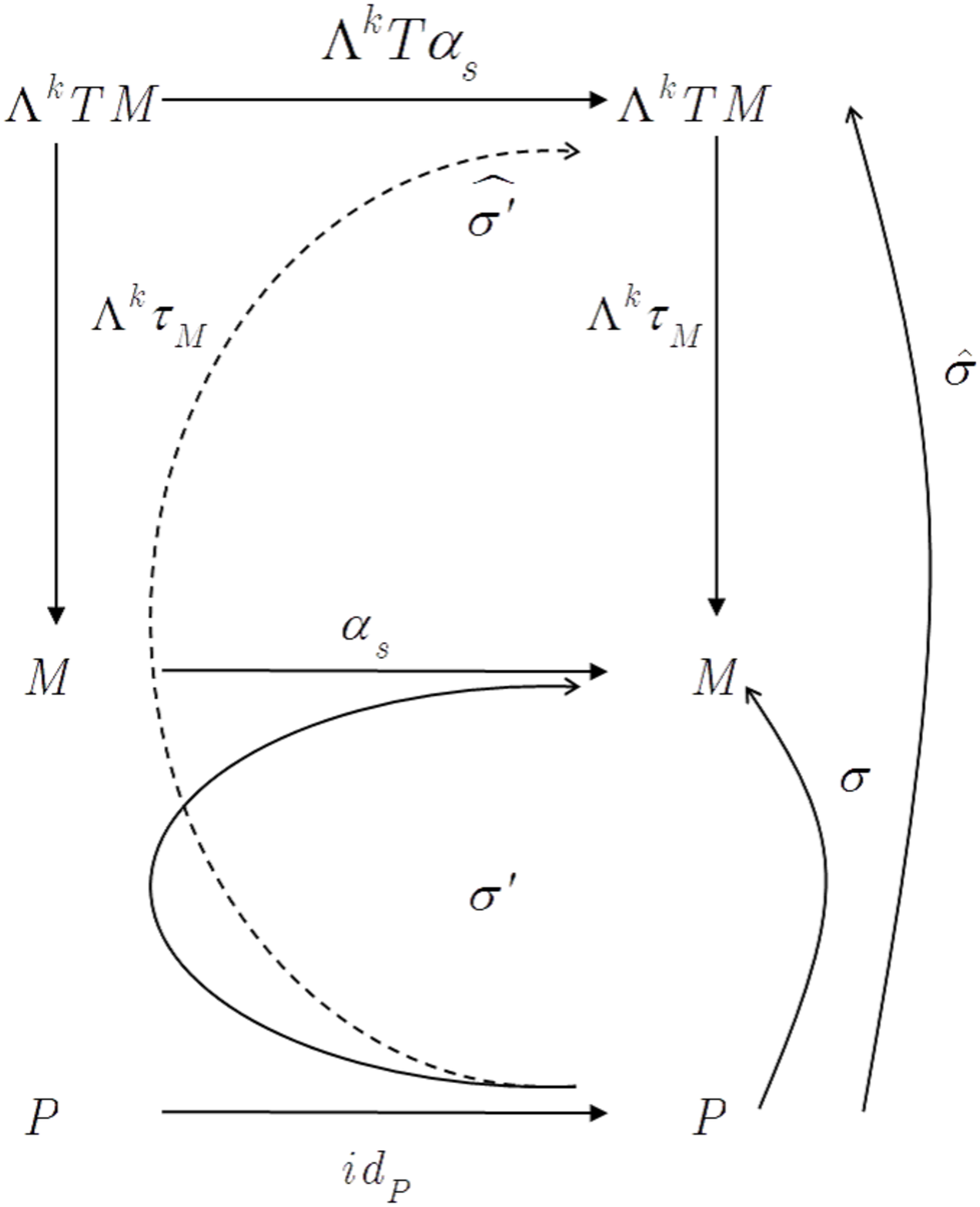}
  \caption{First order fields} \label{fig_1st_field}
\end{figure}

\begin{defn} Variation of the action
Let $\xi$ be a vector field on $M$ which generates the $1$-parameter group $\alpha_s$, 
i.e., $\displaystyle{\xi ={{\left. \frac{d{{\alpha }_{s}}}{ds} \right|}_{s=0}}}$. 
We call the functional 
\begin{align}
 {{\delta }_{\xi }}{{S}^{\mathcal{K}}}(\Sigma ) &:=\underset{s\to 0}{\mathop{\lim }}\, 
\frac{1}{s}\left\{ {{S}^{\mathcal{K}}}({{\alpha }_{s}}(\Sigma ))-{{S}^{\mathcal{K}}}(\Sigma ) \right\} \nonumber \\ 
& =\underset{s\to 0}{\mathop{\lim }}\,\frac{1}{s}\left\{ \int_{\widehat{{{\alpha }_{s}}(\sigma )}(P)}
{\mathcal{K}}-\int_{\hat{\sigma }(P)}{\mathcal{K}} \right\} 
\end{align}
the {\it variation of the action ${{S}^{\mathcal{K}}}(\Sigma )$ 
with respect to the flow $\alpha$, associated to $\mathcal{K}$ }.
\end{defn}

Since the lift of this modified parameterisation is given by, 
$\widehat{{{\sigma }'}}=\widehat{{{\alpha }_{s}}\circ \sigma }={{\Lambda }^{k}}T{{\alpha }_{s}}
\circ \hat{\sigma }={{\Lambda }^{k}}T{{\alpha }_{s}}\circ \hat{\sigma }\circ i{{d}_{P}}^{-1}$, 
We get, 
\begin{align}
 {{\delta }_{\xi }}{{S}^{\mathcal{K}}}(\Sigma ) 
  &=\underset{s\to 0}{\mathop{\lim }}\,\frac{1}{s}\left[ \int_{{{\Lambda }^{k}}T{{\alpha }_{s}} 
 \circ \hat{\sigma }(P)}{\mathcal{K}}-\int_{\hat{\sigma }(P)}{\mathcal{K}} \right]
 =\underset{s\to 0}{\mathop{\lim }}\,\frac{1}{s}
 \left[ \int_{\hat{\sigma }(P)}{{{(T{{\alpha }_{s}})}^{*}}\mathcal{K}}
 - \int_{\hat{\sigma }(P)}{\mathcal{K}} \right] \nonumber \\ 
 & =\int_{\hat{\sigma }(P)}{{{L}_{X}}\mathcal{K}} \nonumber \\ 
 & =\int_{{\hat{\Sigma }}}{{{L}_{X}}\mathcal{K}}  
\end{align}
where $X$ is a vector field on ${{\Lambda }^{k}}TM$ that generates the multi-tangent 
$1$-parameter group ${{\Lambda }^{k}}T{{\alpha }_{s}}$, i.e., 
$\displaystyle{X={{\left. \frac{d({{\Lambda }^{k}}T{{\alpha }_{s}})}{ds} \right|}_{s=0}}}$, 
and $L_X$ is a Lie derivative with respect to $X$. 

Now we will calculate the vector field $X$ and the equation of motion in local coordinates. 
As usual, let $(U,\varphi )$, 
$\varphi =({{x}^{\mu }})$ be a chart on $M$, and the induced chart of ${{\Lambda }^{k}}TM$ 
$(V,\psi )$, $V={{\tau }_{M}}^{-1}(U)$, $\psi =({{x}^{\mu }},{{y}^{{{\mu }_{1}}\cdots {{\mu }_{k}}}})$. 
Let $\xi $ be generator of the $1$-parameter group 
${{\alpha }_{s}}$, $\displaystyle{\xi =\frac{d{{\alpha }_{s}}}{ds}}$, 
and its local coordinate expression 
$\displaystyle{\xi ={{\xi }^{\mu }}\frac{\partial }{\partial {{x}^{\mu }}}}$, 
where ${{\xi }^{\mu }}\in {{C}^{\infty }}(M)$. 
The multi-tangent map ${{\Lambda }^{k}}{{T}_{p}}{{\alpha }_{s}}$ at $p\in M$ 
sends the vector 
$v:={{v}_{p}}\in {{\Lambda }^{k}}{{T}_{p}}M$ to ${{\Lambda }^{k}}{{T}_{{{\alpha }_{t}}(p)}}M$ by 
\begin{align}
{{\Lambda }^{k}}{{T}_{p}}{{\alpha }_{s}}(v)=\frac{1}{k!}{{\left. 
\frac{\partial {{x}^{{{\mu }_{1}}}}{{\alpha }_{s}}{{\varphi }^{-1}}}{\partial {{x}^{{{\nu }_{1}}}}} 
\right|}_{\varphi (p)}}{{\left. \cdots \frac{\partial {{x}^{{{\mu }_{k}}}}{{\alpha }_{s}}
{{\varphi }^{-1}}}{\partial {{x}^{{{\nu }_{k}}}}} \right|}_{\varphi (p)}}{{v}^{{{\nu }_{1}} 
\cdots {{\nu }_{k}}}}{{\left( \frac{\partial }{\partial {{x}^{{{\mu }_{1}}}}}\wedge 
\cdots \wedge \frac{\partial }{\partial {{x}^{{{\mu }_{k}}}}} \right)}_{{{\alpha }_{s}}(p)}},
\end{align} 
and since $({{\Lambda }^{k}}T{{\alpha }_{s}},{{\alpha }_{s}})$ is a bundle morphism and from the 
definition of canonical coordinates of a tangent vector, we have 
\begin{align}
&{{x}^{\mu }}\circ {{\Lambda }^{k}}{{T}_{p}}{{\alpha }_{s}}(v)={{x}^{\mu }}
\circ {{\alpha }_{s}}\circ {{\Lambda }^{k}}{{\tau }_{M}}(v), \nonumber  \\ 
&{{y}^{{{\mu }_{1}}\cdots {{\mu }_{k}}}}\circ {{\Lambda }^{k}}{{T}_{p}}{{\alpha }_{s}}(v)
={{\left. \frac{\partial {{x}^{{{\mu }_{1}}}}{{\alpha }_{s}}{{\varphi }^{-1}}}
{\partial {{x}^{{{\nu }_{1}}}}} \right|}_{\varphi ({{\Lambda }^{k}}{{\tau }_{M}}(v))}} 
\cdots {{\left. \frac{\partial {{x}^{{{\mu }_{k}}}}{{\alpha }_{s}}{{\varphi }^{-1}}} 
{\partial {{x}^{{{\nu }_{k}}}}} \right|}_{\varphi ({{\Lambda }^{k}}{{\tau }_{M}}(v))}}{{y}^{{{\nu }_{1}} 
\cdots {{\nu }_{k}}}}(v).
\end{align} 
The induced vector field $X$ by the $1$-parameter group ${{\Lambda }^{k}}T{{\alpha }_{s}}$ 
at a point $q\in {{\Lambda }^{k}}TM$ has a local expression, 
\begin{align}
 & {{X}_{q}}={{\left. \frac{d({{x}^{\mu }}\circ {{\Lambda }^{k}}T{{\alpha }_{s}})}{ds} \right|}_{s=0}}
 {{\left( \frac{\partial }{\partial {{x}^{\mu }}} \right)}_{q}}
 +\frac{1}{k!}{{\left. \frac{d({{y}^{{{\mu }_{1}}\cdots {{\mu }_{k}}}} 
 \circ {{\Lambda }^{k}}T{{\alpha }_{s}})}{ds} \right|}_{s=0}}{{\left( \frac{\partial }
 {\partial {{y}^{{{\mu }_{1}}\cdots {{\mu }_{k}}}}} \right)}_{q}} \nonumber \\ 
 & \quad ={{\left. \frac{d}{ds}({{x}^{\mu }}\circ {{\alpha }_{s}} 
 \circ {{\Lambda }^{k}}{{\tau }_{M}}) \right|}_{s=0}} 
 {{\left( \frac{\partial }{\partial {{x}^{\mu }}} \right)}_{q}} \nonumber  \\
 & \quad \quad +\frac{1}{k!}{{y}^{{{\nu }_{1}}\cdots {{\nu }_{k}}}}(q)\frac{d}{ds}{{\left( {{\left. 
 \frac{\partial {{x}^{{{\mu }_{1}}}}{{\alpha }_{s}}{{\varphi }^{-1}}}{\partial {{x}^{{{\nu }_{1}}}}} 
 \right|}_{\varphi ({{\Lambda }^{k}}{{\tau }_{M}}(q))}} 
 \cdots {{\left. \frac{\partial {{x}^{{{\mu }_{k}}}}{{\alpha }_{s}}
 {{\varphi }^{-1}}}{\partial {{x}^{{{\nu }_{k}}}}} \right|}_{\varphi ({{\Lambda }^{k}}{{\tau }_{M}}(q))}} 
 \right)}_{s=0}}{{\left( \frac{\partial }{\partial {{y}^{{{\mu }_{1}}\cdots {{\mu }_{k}}}}} \right)}_{q}} \nonumber \\ 
 & \quad =\left( {{\xi }^{\mu }}\circ {{\Lambda }^{k}}{{\tau }_{M}} \right)(q) 
 {{\left( \frac{\partial }{\partial {{x}^{\mu }}} \right)}_{q}}+\frac{1}{(k-1)!} 
 \left( \frac{\partial {{\xi }^{{{\mu }_{1}}}}}{\partial {{x}^{\nu }}}\circ {{\Lambda }^{k}}{{\tau }_{M}} 
 \cdot {{y}^{\nu {{\mu }_{2}}\cdots {{\mu }_{k}}}} \right)(q){{\left( \frac{\partial }
 {\partial {{y}^{{{\mu }_{1}}\cdots {{\mu }_{k}}}}} \right)}_{q}}, 
\end{align}
Therefore, 
\begin{align}
X={{\xi }^{\mu }}\circ {{\Lambda }^{k}}{{\tau }_{M}}
\left( \frac{\partial }{\partial {{x}^{\mu }}} \right)
+\frac{1}{(k-1)!}\frac{\partial {{\xi }^{{{\mu }_{1}}}}}{\partial {{x}^{\nu }}} 
\circ {{\Lambda }^{k}}{{\tau }_{M}}\cdot {{y}^{\nu {{\mu }_{2}}\cdots {{\mu }_{k}}}}
\left( \frac{\partial }{\partial {{y}^{{{\mu }_{1}}\cdots {{\mu }_{k}}}}} \right).   \label{pro.v.f.onTM}
\end{align}
The Lie derivative ${{L}_{X}}\mathcal{K}$ in coordinate expression becomes, 
\begin{align}
 {{L}_{X}}\mathcal{K} &={{L}_{X}}\left( \frac{1}{k!}\frac{\partial K}{\partial {{y}^{{{\rho }_{1}}
 \cdots {{\rho }_{k}}}}}d{{x}^{{{\rho }_{1}}\cdots {{\rho }_{k}}}} \right) \nonumber  \\
 & =X\left( \frac{1}{k!}\frac{\partial K}{\partial {{y}^{{{\rho }_{1}}\cdots {{\rho }_{k}}}}} 
 \right)d{{x}^{{{\rho }_{1}}\cdots {{\rho }_{k}}}}+\frac{1}{(k-1)!}
 \frac{\partial K}{\partial {{y}^{{{\rho }_{1}}\cdots {{\rho }_{k}}}}}d{{L}_{X}}{{x}^{{{\rho }_{1}}}}
 \wedge d{{x}^{{{\rho }_{2}}\cdots {{\rho }_{k}}}} \nonumber \\ 
 &=\frac{1}{k!}\left\{ {{\xi }^{\mu }}\circ {{\Lambda }^{k}}{{\tau }_{M}}
 \left( \frac{{{\partial }^{2}}K}{\partial {{x}^{\mu }}\partial {{y}^{{{\rho }_{1}}
 \cdots {{\rho }_{k}}}}} \right)+\frac{\partial {{\xi }^{{{\mu }_{1}}}}}{\partial {{x}^{\nu }}}
 \circ {{\Lambda }^{k}}{{\tau }_{M}}\cdot {{y}^{\nu {{\mu }_{2}}\cdots {{\mu }_{k}}}}
 \left( \frac{{{\partial }^{2}}K}{\partial {{y}^{{{\mu }_{1}}\cdots {{\mu }_{k}}}}
 \partial {{y}^{{{\rho }_{1}}\cdots {{\rho }_{k}}}}} \right) \right\}d{{x}^{{{\rho }_{1}}
 \cdots {{\rho }_{k}}}} \nonumber \\ 
 & \quad +\frac{1}{(k-1)!}\frac{\partial K}{\partial {{y}^{{{\rho }_{1}}\cdots {{\rho }_{k}}}}}
 d\left( {{\xi }^{{{\rho }_{1}}}}\circ {{\Lambda }^{k}}{{\tau }_{M}} \right)\wedge d{{x}^{{{\rho }_{2}}
 \cdots {{\rho }_{k}}}} \nonumber \\ 
 & =\frac{1}{k!}{{\xi }^{\mu }}\circ {{\Lambda }^{k}}{{\tau }_{M}}
 \left\{ \frac{{{\partial }^{2}}K}{\partial {{x}^{\mu }}\partial {{y}^{{{\rho }_{1}}
 \cdots {{\rho }_{k}}}}}d{{x}^{{{\rho }_{1}}}}-kd\left( \frac{\partial K}{\partial {{y}^{\mu {{\rho }_{2}}
 \cdots {{\rho }_{k}}}}} \right) \right\}\wedge d{{x}^{{{\rho }_{2}}\cdots {{\rho }_{k}}}} \nonumber  \\
 &\quad +\frac{\partial {{\xi }^{{{\mu }_{1}}}}}{\partial {{x}^{\nu }}}\circ {{\Lambda }^{k}}{{\tau }_{M}}
 \cdot {{y}^{\nu {{\mu }_{2}}\cdots {{\mu }_{k}}}}\left( \frac{{{\partial }^{2}}K}{\partial {{y}^{{{\mu }_{1}}
 \cdots {{\mu }_{k}}}}\partial {{y}^{{{\rho }_{1}}\cdots {{\rho }_{k}}}}} \right)d{{x}^{{{\rho }_{1}}
 \cdots {{\rho }_{k}}}} \nonumber \\
 &\quad  +\frac{1}{(k-1)!}d\left( \frac{\partial K}{\partial {{y}^{{{\rho }_{1}}\cdots {{\rho }_{k}}}}}
 \cdot {{\xi }^{{{\rho }_{1}}}}\circ {{\Lambda }^{k}}{{\tau }_{M}} \right)\wedge d{{x}^{{{\rho }_{2}}
 \cdots {{\rho }_{k}}}}. \label{Lie_1st_field}
\end{align}
The result of (\ref{Lie_1st_field}) is called the {\it infinitesimal first variation formula} for the 
Kawaguchi $k$-form $\mathcal{K}$.

The variation of action becomes, 
\begin{align}
 & {{\delta }_{\xi }}{{S}^{\mathcal{K}}}(\Sigma )=\int_{\hat{\sigma }(P)}{{{L}_{X}}\mathcal{K}}
 =\int_{P}{{{{\hat{\sigma }}}^{*}}{{L}_{X}}\mathcal{K}} \nonumber \\ 
 &=\int_{P}{{{{\hat{\sigma }}}^{*}}\left[ \frac{1}{k!}{{\xi }^{\mu }} 
 \circ {{\Lambda }^{k}}{{\tau }_{M}}\left( \frac{{{\partial }^{2}}K}{\partial {{x}^{\mu }} 
 \partial {{y}^{{{\rho }_{1}}\cdots {{\rho }_{k}}}}}d{{x}^{{{\rho }_{1}}}}
 -kd \left( \frac{\partial K}{\partial {{y}^{\mu {{\rho }_{2}}\cdots {{\rho }_{k}}}}} \right) \right) 
 \wedge d{{x}^{{{\rho }_{2}}\cdots {{\rho }_{k}}}} \right.} \nonumber \\ 
 & \quad +\frac{\partial {{\xi }^{{{\mu }_{1}}}}}{\partial {{x}^{\nu }}}
 \circ {{\Lambda }^{k}}{{\tau }_{M}}\cdot {{y}^{\nu {{\mu }_{2}}\cdots {{\mu }_{k}}}} 
 \left( \frac{{{\partial }^{2}}K}{\partial {{y}^{{{\mu }_{1}}\cdots {{\mu }_{k}}}}\partial {{y}^{{{\rho }_{1}}
 \cdots {{\rho }_{k}}}}} \right)d{{x}^{{{\rho }_{1}}\cdots {{\rho }_{k}}}} \nonumber \\
 &\left. \quad +\frac{1}{(k-1)!}
 d\left( \frac{\partial K}{\partial {{y}^{{{\rho }_{1}}\cdots {{\rho }_{k}}}}}
 \cdot {{\xi }^{{{\rho }_{1}}}}\circ {{\Lambda }^{k}}{{\tau }_{M}} \right)\wedge d{{x}^{{{\rho }_{2}}
 \cdots {{\rho }_{k}}}} \right] \nonumber \\ 
 &  =\int_{{\hat{\Sigma }}}{\left[ \frac{1}{k!}{{\xi }^{\mu }}\circ {{\Lambda }^{k}}{{\tau }_{M}}
 \left( \frac{{{\partial }^{2}}K}{\partial {{x}^{\mu }}\partial {{y}^{{{\rho }_{1}}
 \cdots {{\rho }_{k}}}}}d{{x}^{{{\rho }_{1}}}}-kd\left( \frac{\partial K}{\partial {{y}^{\mu {{\rho }_{2}}
 \cdots {{\rho }_{k}}}}} \right) \right) \right.} \nonumber \\ 
 &  \quad \left. +\frac{1}{(k-1)!}d\left( \frac{\partial K}{\partial {{y}^{{{\rho }_{1}} 
 \cdots {{\rho }_{k}}}}}\cdot {{\xi }^{{{\rho }_{1}}}}\circ {{\Lambda }^{k}}{{\tau }_{M}} \right) \right]
 \wedge d{{x}^{{{\rho }_{2}}\cdots {{\rho }_{k}}}}, \label{variation_1st_field} 
\end{align}
which is called the {\it integral first variation formula}. 
We used the homogeneity condition: 
\begin{align}
\left( \left( \frac{{{\partial }^{2}}K}{\partial {{y}^{{{\mu }_{1}}
\cdots {{\mu }_{k}}}}\partial {{y}^{{{\rho }_{1}}\cdots {{\rho }_{k}}}}} \right)\cdot {{y}^{{{\rho }_{1}}
\cdots {{\rho }_{k}}}} \right)\circ \hat{\sigma }=0. \label{cond_hom3}
\end{align}
(\ref{cond_hom3}) is obtained by taking the derivative of (\ref{1st_k_Euler}) 
with respect to ${{y}^{{{\mu }_{1}}\cdots {{\mu }_{k}}}}$, and then taking the pull back. 

Now we can proceed to find the equations of motion to this system. 
We will begin with the definition of an extremal.

\break \begin{defn} Extremal of an action  \label{def_extremal_k} 
\begin{enumerate}
\item We say that a $k$-area $\Sigma$ is {\it stable} with respect to the flow $\alpha$, 
when it satisfies 
\begin{align}
{\delta _\xi }{S^ {\mathcal{K}}} = 0, \label{c_stable_k}
\end{align}
where $\xi$ is the generator of $\alpha$. 
\item We say that a $k$-area $\Sigma$ is an {\it extremal} of the action $S^{\mathcal{K}}$, when it satisfies 
(\ref{c_stable_k}) for any $\alpha$ such that its associated $1$-parameter group 
$\alpha_s$ satisfies ${{\alpha }_{s}}(\partial \Sigma )=\partial \Sigma $, 
$\forall s\in \mathbb{R}$, where $\partial \Sigma $ is the boundary of $\Sigma$.
\end{enumerate}
\end{defn}

Now we can obtain the following theorem. 
\begin{theorem} Extremals  \label{thm_1st_var} \\
Let $\Sigma$ be an $k$-patch. 
The following statements are equivalent. 
\begin{enumerate}
\item $\Sigma$ is an extremal. 
\item The equation 
\begin{align}
 & {\mathcal{EL}^{K}}_{\mu } \circ \hat{\sigma }=0, \nonumber \\ 
 & {\mathcal{EL}^{K}}_{\mu }:=\frac{1}{k!}
 \left( \frac{{{\partial }^{2}}K}{\partial {{x}^{\mu }}\partial {{y}^{{{\rho }_{1}}
 \cdots {{\rho }_{k}}}}}d{{x}^{{{\rho }_{1}}}}-kd\left( \frac{\partial K}{\partial {{y}^{\mu {{\rho }_{2}}
 \cdots {{\rho }_{k}}}}} \right) \right)\wedge d{{x}^{{{\rho }_{2}}\cdots {{\rho }_{k}}}}, \label{1st_k_EL_pullback}
\end{align}
holds for arbitrary parameterisation $\sigma$. 
\end{enumerate}
\end{theorem}
The proof is given similarly as in the case of mechanics. (see Section \ref{subsec_extremal})

The equations (\ref{1st_k_EL_pullback}) are called the {\it Euler-Lagrange equations} or 
{\it equations of motion} of the Lagrangian of first order fields, $\mathcal{K}$.  

\begin{defn} Symmetry of the dynamical system \\ 
Let $u$ be a vector field over $M$, and $Y$ an induced vector field by $u$ over ${{\Lambda }^{k}}TM$.  
We say that {\it $\mathcal{K}$ is invariant with respect to $u$}, if 
\begin{align}
{{L}_{Y}}\mathcal{K}=0, 
\end{align} 
and $u$ is called a {\it symmetry} of the dynamical system $(M, \mathcal{K})$. 
We also say that $u$ generates the invariant transformations on the Kawaguchi manifold $(M, \mathcal{K})$. 
\end{defn}

Now we will have the following conservation law. 
\begin{theorem}  Noether  (first order field)\\
Suppose we are given a symmetry of $(M, \mathcal{K})$. 
Then there exists a $(k-1)$-form $f$ on ${{\Lambda }^{k}}TM$ 
which along the extremal $\gamma$ of $S^\mathcal{K}$ satisfies, 
\begin{align}
\int_{{\hat{\gamma }}}{df}=0, \label{eq_conserved_1st_field}
\end{align}
for any parameterisation $\sigma$ which parameterise $\gamma$.
\end{theorem}
\begin{proof}
Let the symmetry be $u$, with its local coordinate expression 
$\displaystyle{u={{u}^{\mu }}\frac{\partial }{\partial {{x}^{\mu }}}}$, 
and the induced vector field $Y$. 
Then from (\ref{variation_1st_field}), we have 
\begin{align}
 0&=\int_{{\hat{\gamma }}}{{{L}_{Y}}\mathcal{K}} \nonumber \\ 
 & =\int_{{\hat{\gamma }}}{\left[ \frac{1}{k!}{{u}^{\mu }} 
 \circ {{\Lambda }^{k}}{{\tau }_{M}}\left( \frac{{{\partial }^{2}}K}
 {\partial {{x}^{\mu }}\partial {{y}^{{{\rho }_{1}}\cdots {{\rho }_{k}}}}}d{{x}^{{{\rho }_{1}}}}
 - kd\left( \frac{\partial K}{\partial {{y}^{\mu {{\rho }_{2}} \cdots {{\rho }_{k}}}}} \right) \right) \right.} \nonumber \\ 
 &  \quad \left. +\frac{1}{(k-1)!}d\left( \frac{\partial K}{\partial {{y}^{{{\rho }_{1}} 
 \cdots {{\rho }_{k}}}}}\cdot {{u}^{{{\rho }_{1}}}}\circ {{\Lambda }^{k}}{{\tau }_{M}} \right) \right] 
 \wedge d{{x}^{{{\rho }_{2}}\cdots {{\rho }_{k}}}} \nonumber \\ 
 & =\int_{{\hat{\gamma }}}{d\left( \frac{1}{(k-1)!}
 \frac{\partial K}{\partial {{y}^{{{\rho }_{1}}\cdots {{\rho }_{k}}}}}\cdot {{u}^{{{\rho }_{1}}}} 
 \circ {{\Lambda }^{k}}{{\tau }_{M}}\, d{{x}^{{{\rho }_{2}}\cdots {{\rho }_{k}}}} \right)} 
\end{align}
The third equality comes from the fact we consider along the extremal $\gamma$. 
Therefore we have a $(k-1)$-form on ${{\Lambda }^{k}}TM$, 
\begin{align}
f=\frac{1}{(k-1)!}\frac{\partial K}{\partial {{y}^{{{\rho }_{1}}\cdots {{\rho }_{k}}}}}
\cdot {{u}^{{{\rho }_{1}}}}\circ {{\Lambda }^{k}}{{\tau }_{M}}\,d{{x}^{{{\rho }_{2}}\cdots {{\rho }_{k}}}}
\end{align}
such that satisfies the condition. 
\end{proof}

We call the relation (\ref{eq_conserved_1st_field}), the {\it conservation law}. 

\begin{defn} Noether current \\
The quantity $f$ is called the {\it Noether current of first order field theory, 
associated with $u$}.  
\end{defn}
By the coordinate transformation 
\begin{align}
{{x}^{\mu }}\to {{\tilde{x}}^{\mu }}={{\tilde{x}}^{\mu }}({{x}^{\nu }}),\ {{y}^{{{\mu }_{1}}\cdots {{\mu }_{k}}}}\to {{\tilde{y}}^{{{\mu }_{1}}\cdots {{\mu }_{k}}}}=\frac{\partial {{{\tilde{x}}}^{{{\mu }_{1}}}}}{\partial {{x}^{{{\nu }_{1}}}}}\cdots \frac{\partial {{{\tilde{x}}}^{{{\mu }_{k}}}}}{\partial {{x}^{{{\nu }_{k}}}}}{{y}^{{{\nu }_{1}}\cdots {{\nu }_{k}}}},
\end{align} 
the $k$-form ${{\mathcal{EL}}^{K}}_{\mu }$ in (\ref{1st_k_EL_pullback}) transforms as 
\begin{align}
&\frac{1}{k!}\left( \frac{{{\partial }^{2}}K}{\partial {{{\tilde{x}}}^{\mu }}
\partial {{{\tilde{y}}}^{{{\rho }_{1}}\cdots {{\rho }_{k}}}}}d{{{\tilde{x}}}^{\rho }}
-kd\left( \frac{\partial K}{\partial {{{\tilde{y}}}^{\mu {{\rho }_{2}} \cdots {{\rho }_{k}}}}} \right) 
\right)\wedge d{{\tilde{x}}^{{{\rho }_{2}}\cdots {{\rho }_{k}}}} \nonumber  \\
& \quad \quad \quad =\frac{1}{k!} \left( \frac{\partial {{x}^{\nu }}}{\partial {{{\tilde{x}}}^{\mu }}} \right)
\left( \frac{{{\partial }^{2}}K}{\partial {{x}^{\nu }}\partial {{y}^{{{\rho }_{1}}
\cdots {{\rho }_{k}}}}}d{{x}^{{{\rho }_{1}}}}-kd\left( \frac{\partial K}
{\partial {{y}^{\nu {{\rho }_{2}}\cdots {{\rho }_{k}}}}} \right) \right) 
\wedge d{{x}^{{{\rho }_{2}}\cdots {{\rho }_{k}}}}.
\end{align}
This observation leads us to define a new coordinate invariant form.

\begin{lemma} Euler-Lagrange form \\*
There exist a global $(k+1)$ form on $\Lambda^k TM$, which in local coordinates are expressed by
\begin{align}
\mathcal{E}{{\mathcal{L}}^{K}}=d{{x}^{\mu }}\wedge \mathcal{E}{{\mathcal{L}}^{K}}_{\mu }
= \frac{1}{k!}\left\{ \frac{{{\partial }^{2}}K}{\partial {{x}^{\mu }}\partial {{y}^{{{\rho }_{1}} 
\cdots {{\rho }_{k}}}}}d{{x}^{\mu }}+d\left( \frac{\partial K}{\partial {{y}^{{{\rho }_{1}} 
\cdots {{\rho }_{k}}}}} \right) \right\}\wedge d{{x}^{{{\rho }_{1}}\cdots {{\rho }_{k}}}}. \label{EL_form_1st_k}
\end{align}
\end{lemma} 
From the previous coordinate transformations, this is obviously coordinate independent. 

There is a direct relation between the exterior derivative of $\mathcal{K}$ and $\mathcal{EL}$,   
\begin{align}
d\mathcal{K}=\text{ }\mathcal{E}{{\mathcal{L}}^{K}} 
-\frac{1}{k!}\frac{{{\partial }^{2}}K}{\partial {{x}^{\mu }}\partial {{y}^{{{\rho }_{1}} 
\cdots {{\rho }_{k}}}}}d{{x}^{\mu }}\wedge d{{x}^{{{\rho }_{1}}\cdots {{\rho }_{k}}}}. \label{K-EL_rel_1st_field}
\end{align}
It can be also checked easily that this is also a coordinate invariant relation. 

\begin{remark} \label{convL-Kawaguchi-rel}
For the case for $k=1$ (Finsler), we have shown in Remark \ref{convL-Fins-rel}, 
that once we are given a conventional Lagrangian depending on a parameter, 
we can construct a Finsler function from that and reformulate the theory in parameter invariant style. 
Here we will construct the Kawaguchi function in a similar way. 
As in the case of $k=1$, in general, this cannot be done globally. 
Consider a bundle $(Y,\pi ,{{\mathbb{R}}^{k}})$ and let $(U,\psi )$, 
$\psi =({{t}^{1}},\cdots {{t}^{k}},{{q}^{k+1}},\cdots ,{{q}^{k+m}}),$ 
be the adapted chart on a $(m+k)$-dimensional manifold $Y$, 
and the induced chart on $\mathbb{R}^k$ be $(\pi(U),t^1,\dots,t^k)$. 
We denote the induced chart on ${{J}^{1}}Y$ by 
$({{({{\pi }^{1,0}})}^{-1}}(U),{{\psi }^{1}}),{{\psi }^{1}}=({{t}^{1}},
\cdots ,{{t}^{k}},{{q}^{k+1}},\cdots ,{{q}^{k+m}},{{\dot{q}}^{k+1}},\cdots ,{{\dot{q}}^{k+m}})$. 
Take the induced chart on ${{\Lambda }^{k}}TY$ as $(V,{{\tilde{\psi }}^{1}})$, 
$V={{({{\tau }_{Y}})}^{-1}}(U)$, ${{\tilde{\psi }}^{1}}=({{x}^{\mu }},{{y}^{{{\mu }_{1}}\cdots {{\mu }_{k}}}})$, 
$\mu =1,\cdots ,n,$ with $\ n=k+m$, such that ${{y}^{12\cdots k}}\ne 0$. 
It is always possible to choose such coordinates for a single chart. 
(In order to avoid confusion we use different symbols, 
but clearly ${{x}^{a}}={{t}^{a}},\ {{x}^{i}}={{q}^{i}}$, $a=1,\cdots ,k,\ i=k+1,\cdots ,n$, on $U$.)
Now consider a map $\rho :V \hookrightarrow {{J}^{1}}Y$, $\rho (V)={{({{\pi }^{1,0}})}^{-1}}(U)$, 
which in coordinates are defined by 
\begin{align}
{{t}^{a}}\circ \rho ={{x}^{a}},{{q}^{i}}\circ \rho = {{x}^{i}},q_{a}^{i}\circ \rho = 
\frac{{{y}^{1\cdots \overset{i}{\mathop{\overset{\vee }{\mathop{a}}\,}}\, 
\cdots k}}}{{{y}^{1\cdots k}}},
\end{align}
 for $a=1,\cdots ,k,\ i=k+1,\cdots ,n,$
and \[{{y}^{1\cdots \overset{i}{\mathop{\overset{\vee }{\mathop{a}}\,}}\,\cdots k}}\] 
means to replace the $a$ th index by $i$. 
Then we define the Kawaguchi function $K$ by a function on $V$, 
\begin{align}
K/{{y}^{12\cdots k}}={{\rho }^{*}}{\mathcal{K}}.
\end{align} 
In coordinates, 
\begin{align}
K({{x}^{\mu }},{{y}^{{{\mu }_{1}}\cdots {{\mu }_{k}}}})=\mathrm{}({{t}^{a}}\circ \rho ,{{q}^{i}}
\circ \rho ,q_{a}^{i}\circ \rho ){{y}^{12\cdots k}}={\mathcal{L}}({{x}^{a}},{{x}^{i}},
\frac{{{y}^{1\cdots \overset{i}{\mathop{\overset{\vee }{\mathop{a}}\,}}\, 
\cdots k}}}{{{y}^{12\cdots k}}}){{y}^{12\cdots k}},  \label{convL-FinKawaL}
\end{align}
for $\mu ,{{\mu }_{1}},\cdots ,{{\mu }_{k}}=1,\cdots ,n$, $a=1,\cdots ,k,$ $i=k+1,\cdots ,n$. 
Then on $V$, $K$ satisfies the homogeneity function. 
We now have 
\begin{align}
{\mathcal{K}}=\frac{1}{k!}\frac{\partial K}{\partial {{y}^{{{\mu }_{1}} 
\cdots {{\mu }_{k}}}}}d{{x}^{{{\mu }_{1}}}}\wedge \cdots \wedge d{{x}^{{{\mu }_{k}}}}
\end{align} 
on $V$. 
In this way, for a local coordinate chart, we can construct a local Kawaguchi $k$-form from a Lagrangian, 
which also can be used as an reparameterisation invariant action provided that the $k$-patch 
where the integration is carried out is covered by this single chart. 
\end{remark}

\begin{ex} De Broglie field (Schrodinger field) \\
Here we will give a rather trivial, but instructive example for the case of 
$M = {\mathbb{R}^4}$, with Kawaguchi structure corresponding to the De Broglie field 
on $2$-dimensional spacetime ${\mathbb{R}^2}$. 
The conventional Lagrangian function of the De Broglie field is given by
\begin{align}
{\mathcal{L}} = \frac{i}{2}\left( {\bar \psi \, \cdot {\partial _t}\psi  
- {\partial _t}\bar \psi  \cdot \psi } \right) 
- \frac{1}{{2m}}{\partial _x}\bar \psi  \cdot \,{\partial _x}\psi  
+ e\bar \psi \, \cdot \varphi \, \cdot \psi ,  \label{DeBroglieL}
\end{align}
on ${{J}^{1}}Y={{J}^{1}}{{\mathbb{R}}^{4}}$, 
where ${{J}^{1}}Y$ is the prolongation of the total space regarding the bundle 
$(Y, p{{r}_{1}},{{\mathbb{R}}^{2}})$, 
$Y={{\mathbb{R}}^{2}}\times {{\mathbb{R}}^{2}}$. 
We take for the global fibre bundle coordinates; 
$(t,x)$ for ${\mathbb{R}^2}$, $(t,x,\psi ,\bar \psi )$ for $Y$, 
and $(t,x,\psi ,\bar{\psi },{{\partial }_{t}}\psi ,{{\partial }_{t}}\bar{\psi },
{{\partial }_{x}}\psi ,{{\partial }_{x}}\bar{\psi })$ for ${{J}^{1}}Y$. 
$t,\,\,x$ denotes the $2$-dimensional spacetime, 
$\psi ,\,\,\bar{\psi }$ the fields,  
$\varphi =\varphi (x)$ the external field, and $m,\,\,e$ are constants. 
In the orthodox physics notation, 
the pull-back of ${\mathcal{L}}$ to ${\mathbb{R}^2}$ is also called the Lagrangian, namely 
\begin{align}
&L: = \user2{\mathcal{L}} \circ {J^1}\gamma  
= \frac{i}{2}\left( {\bar \psi \frac{{\partial \psi }}{{\partial t}} 
- \frac{{\partial \bar \psi }}{{\partial t}}\psi } \right) 
- \frac{1}{{2m}}\frac{{\partial \bar \psi }}{{\partial x}}\frac{{\partial \psi }}{{\partial x}} 
+ e\bar \psi \,\varphi \,\psi , 
\end{align}
with $\psi :=\psi \circ \gamma =\psi (t,x)$, 
$\bar{\psi }=\bar{\psi }\circ \gamma =\bar{\psi }(t,x)$, 
where $\gamma $ is a section of the bundle $(Y,p{{r}_{1}},{{\mathbb{R}}^{2}})$, 
and ${{J}^{1}}\gamma $ its prolongation. 

We will try to construct the Kawaguchi manifold that corresponds to such model. 
Let $M={{\mathbb{R}}^{4}}$, and the parameter space $P={{\mathbb{R}}^{2}}$. 
In this case, we have global charts on $M$ and $P$. 
Let the canonical coordinates on $P$ and $M$ be $({{t}^{0}},{{t}^{1}}) = (t,x)$ and 
$({{x}^{1}},{{x}^{2}},{{x}^{3}},{{x}^{4}})$ respectively, 
$\sigma :P \to M$ be a parameterisation, and $\hat{\sigma }:P\to {{\Lambda }^{2}}TM$ its lift. 

The Kawaguchi form is a $2$-form on ${{\Lambda }^{2}}TM$. 
Let $({{x}^{\mu }},{{y}^{{{\nu }_{1}}{{\nu }_{2}}}})$ be the induced global chart on 
${{\Lambda }^{2}}TM={{\Lambda }^{2}}T{{\mathbb{R}}^{4}}$ with $\mu ,{{\nu }_{1}},{{\nu }_{2}}=1,2,3,4$. 
By the formula (\ref{convL-FinKawaL}), we have 
\begin{align}
K=\frac{i}{2}\left( {{x}^{4}}{{y}^{32}}-{{x}^{3}}{{y}^{42}} \right)
-\frac{1}{2m}\frac{{{y}^{14}}{{y}^{13}}}{{{y}^{12}}}
+e\varphi ({{x}^{2}}){{x}^{3}}{{x}^{4}}{{y}^{12}} .
\end{align}
We see that this $K$ satisfies the homogeneity condition. 
From this Kawaguchi function, we can construct the Kawaguchi $2$-form, 
\begin{align}
 {\mathcal{K}} = \left( {\frac{1}{{2m}}\frac{{{y^{14}}{y^{13}}}}{{{{({y^{12}})}^2}}} 
 + e\varphi ({x^2}){x^3}{x^4}} \right)d{x^{12}} 
 - \frac{1}{{2m}}\left( {\frac{{{y^{14}}}}{{{y^{12}}}}d{x^{13}} 
 + \frac{{{y^{13}}}}{{{y^{12}}}}d{x^{14}}} \right) 
 - \frac{i}{2}\left( {{x^4}d{x^{23}} + {x^3}d{x^{42}}} \right). \label{ex_debroglie_f2}
 \end{align}
This is the reparameterisation invariant Lagrangian of the De Broglie field theory. 

The Euler-Lagrange equations obtained from ${\mathcal{K}}$ 
can be calculated by the formula (\ref {1st_k_EL_pullback}), as  
\begin{align}
&   {\mathcal{E}}{{\mathcal{L}}_1} \circ \hat \sigma  = d\left( {\frac{{\partial K}}{{\partial {y^{12}}}}d{x^2} + \frac{{\partial K}}{{\partial {y^{13}}}}d{x^3} + \frac{{\partial K}}{{\partial {y^{14}}}}d{x^4}} \right) \circ \hat \sigma  = 0, \hfill\nonumber  \\
&   {\mathcal{E}}{{\mathcal{L}}_2} \circ \hat \sigma  = d\left( {\frac{{\partial K}}{{\partial {y^{21}}}}d{x^1} + \frac{{\partial K}}{{\partial {y^{23}}}}d{x^3} + \frac{{\partial K}}{{\partial {y^{24}}}}d{x^4}} \right) \circ \hat \sigma  = 0, \hfill \nonumber \\
&   {\mathcal{E}}{{\mathcal{L}}_3} \circ \hat \sigma  = \left( {e\varphi {x^4}d{x^{12}} - id{x^{42}} - \frac{1}{{2m}}\left( { - \frac{{{y^{14}}}}{{{{({y^{12}})}^2}}}d{y^{12}} + \frac{1}{{{y^{12}}}}d{y^{14}}} \right) \wedge d{x^1}} \right) \circ \hat \sigma  = 0, \hfill \nonumber \\
&   {\mathcal{E}}{{\mathcal{L}}_4} \circ \hat \sigma  = \left( {e\varphi {x^3}d{x^{12}} - id{x^{23}} - \frac{1}{{2m}}\left( { - \frac{{{y^{13}}}}{{{{({y^{12}})}^2}}}d{y^{12}} + \frac{1}{{{y^{12}}}}d{y^{13}}} \right) \wedge d{x^1}} \right) \circ \hat \sigma  = 0, \hfill 
\label{ex_debroglie_f3}
\end{align}
which is true for any parameterisation $\sigma $. 
To compare these equations with the conventional De Broglie field equations; 
in other name the Schrodinger equations, choose the parameterisation, 
which in coordinates are given by 
\begin{align}
({{x}^{1}} \circ \sigma ,{{x}^{2}}\circ \sigma ,{{x}^{3}}\circ \sigma ,{{x}^{4}}\circ \sigma )
=(t,x,\psi ,\bar{\psi }).
\end{align} 
We get,
\begin{align}
&  {{\hat \sigma }^*}{\mathcal{EL}_1} \equiv 0, \hfill \nonumber \\
&  {{\hat \sigma }^*}{\mathcal{EL}_2} \equiv 0, \hfill \nonumber \\
&  {{\hat \sigma }^*}{\mathcal{EL}_3} =  - \left( {i{\partial _t}\bar \psi  - \frac{1}{{2m}}{\partial _x}{\partial _x}\bar \psi  - e\varphi \bar \psi } \right) dt \wedge dx = 0, \hfill \nonumber \\
&  {{\hat \sigma }^*}{\mathcal{EL}_4} = \left( {i{\partial _t}\psi  + \frac{1}{{2m}}{\partial _x}{\partial _x}\psi  + e\varphi \psi } \right) dt \wedge dx = 0, \hfill  
\end{align}
which are indeed the well-known Schrodinger equations. 
The first two equations becomes identity, when the latter two are taken into account, 
meaning these equations are indeed dependent. 
\end{ex}

\begin{remark}
The expressions such as (\ref{ex_debroglie_f2}), (\ref{ex_debroglie_f3}) are reparameterisation invariant, 
and there are other possibilities to choose different parameterisations such that 
their  pulled back expressions on the parameter space would not look like the conventional expressions. 
To consider their meaning and applications would be an interesting theme for future research. 
\end{remark}

%% file: thesis2012_chap6.tex
\chapter{Discussion}

In this thesis, we have introduced the foundations needed for the calculus of variation in Finsler and Kawaguchi geometry. 
For Kawaguchi geometry, we especially constructed the second order $1$-dimensional parameter case, and first order $k$-dimensional case. 
For the second order $k$-dimensional case, only local version was presented. 
We have used a less restricted definition for both Finsler manifold and Kawaguchi manifold compared to the standard definition, 
which is considered more applicable to the problems of physics. 
We constructed a global form for the Kawaguchi geometry, 
which has a similar property as the Finsler-Hilbert form in the Finsler geometry case, 
in the sense that they define a reparameterisation invariant $k$-dimensional area on the subset of the base manifold $M$. 
Lagrange formulation was introduced on these structures in a natural way, 
and we obtained the reparameterisation invariant Euler-Lagrange expression. 
We had compared the results with examples such as Newtonian mechanics and De Broglie field, 
and confirmed that with a special choice of parameterisation, the results will reduce to the conventional expression of these theories. 
Throughout the discussion, we only used basic methods in differential geometry, and took the most straightforward path to considering Lagrangian formulation. 

There are many issues in the thesis that remains for further discussions and research. 
The main reason of the difficulty in the case of higher order $k$-dimensional case comes from the fact that our main pillar; the homogeneity conditions in the simplest expression is not coordinate independent. 
However, we believe this difficulty can be solved soon, and global Kawaguchi form could be constructed for this case too. 
Nevertheless, for many concrete problems for physics, our local formalism should also be applicable. 
In this thesis, we only considered the case where the subset $\Sigma$ of $M$ is diffeomorphic to the closed $k$-rectangle in $\mathbb{R}^k$. 
For a more general case, the action of $\Sigma$ associated to the Kawaguchi (Finsler) form is well-defined provided that there exists an inclusion map $\iota: P \to M$, 
where $\Sigma=\iota(P)$, such that $P$ is an oriented compact manifold. 
In the case where $P$ has no boundary, we should have to consider an extension of variational principle, such as Cartan's principle, which is also an interesting problem. 
Since Finsler and Kawaguchi geometry is less restrictive than Riemannian geometry, 
we expect it should embrace wider area of physics where it cannot be expressed by Riemannian geometry. 
For instance, system that is irreversible with time, or shows hysteresis, may be a good non-trivial example to be modelled by our approach. 
However, these problems are for the moment left for the future, 
and would be presented another time.

%% file: thesis2012_chap_app.tex
\chapter*{Appendix}
\phantomsection
\addcontentsline{toc}{chapter}{Appendix}

Bellow we will show some detailed calculations that were omitted in the text. \\
\\
The variation of the action of second order mechanics is calculated as follows.

\begin{align*}
&  {\delta _\xi }{S^ {\mathcal{K}}}(C) = \int_{{\sigma ^2}(I)} {{L_X}{\mathcal{K}}}  = \int_I {({\sigma ^2})^*{L_X}{\mathcal{K}}}  \hfill \\
&  \quad  = \int_I {{{({\sigma ^2})}^*}\left[ {{\xi ^\mu } \circ \tau _M^{2,0}\left\{ {\frac{{{\partial ^2}K}}{{\partial {x^\mu }\partial {y^\rho }}}d{x^\rho } + 2\frac{{{\partial ^2}K}}{{\partial {x^\mu }\partial {z^\rho }}}d{y^\rho } - d\left( {\frac{{\partial K}}{{\partial {y^\mu }}}} \right)} \right\}} \right.\quad \quad }  \hfill \\
&  \quad \quad  + D{\xi ^\mu } \circ \tau _M^{2,1}\left( {\frac{{{\partial ^2}K}}{{\partial {y^\mu }\partial {y^\rho }}}d{x^\rho } + 2\frac{{{\partial ^2}K}}{{\partial {y^\mu }\partial {z^\rho }}}d{y^\rho } - 2d\left( {\frac{{\partial K}}{{\partial {z^\mu }}}} \right)} \right) \hfill \\
&  \quad \quad  + {D_2}(D{\xi ^\mu })\left( {\frac{{{\partial ^2}K}}{{\partial {z^\mu }\partial {y^\rho }}}d{x^\rho } + 2\frac{{{\partial ^2}K}}{{\partial {z^\mu }\partial {z^\rho }}}d{y^\rho }} \right) \hfill \\
&  \left. {\quad \quad  + d\left( {{\xi ^\mu } \circ \tau _M^{2,0} \cdot \frac{{\partial K}}{{\partial {y^\mu }}} + 2D{\xi ^\mu } \circ \tau _M^{2,1} \cdot \frac{{\partial K}}{{\partial {z^\mu }}}} \right)} \right] \hfill \\ 
&  \quad  = \int_I {({\xi ^\mu } \circ \tau _M^{2,0} \circ {\sigma ^2})\left\{ {\frac{{{\partial ^2}K}}{{\partial {x^\mu }\partial {y^\rho }}}d{x^\rho } + 2\frac{{{\partial ^2}K}}{{\partial {x^\mu }\partial {z^\rho }}}d{y^\rho } - d\left( {\frac{{\partial K}}{{\partial {y^\mu }}}} \right)} \right\} \circ {\sigma ^2}\quad \quad }  \hfill \\
&  \quad \quad  + \frac{{d({\xi ^\mu } \circ \sigma )}}{{dt}}\left( {\left( {\frac{{{\partial ^2}K}}{{\partial {y^\mu }\partial {y^\rho }}}{y^\rho } + 2\frac{{{\partial ^2}K}}{{\partial {y^\mu }\partial {z^\rho }}}{z^\rho }} \right) \circ {\sigma ^2} - 2\frac{d}{{dt}}\left( {\frac{{\partial K}}{{\partial {z^\mu }}} \circ {\sigma ^2}} \right)} \right)dt \hfill \\
&  \quad \quad  + \frac{{{d^2}({\xi ^\mu } \circ \sigma )}}{{d{t^2}}}\left( {\left( {\frac{{{\partial ^2}K}}{{\partial {z^\mu }\partial {y^\rho }}}{y^\rho } + 2\frac{{{\partial ^2}K}}{{\partial {z^\mu }\partial {z^\rho }}}{z^\rho }} \right) \circ {\sigma ^2}} \right)dt \hfill \\
&  \left. {\quad \quad  + \frac{d}{{dt}}\left( {\left( {{\xi ^\mu } \circ \tau _M^{2,0} \cdot \frac{{\partial K}}{{\partial {y^\mu }}} + 2D{\xi ^\mu } \circ \tau _M^{2,1} \cdot \frac{{\partial K}}{{\partial {z^\mu }}}} \right) \circ {\sigma ^2}} \right)dt} \right] \hfill \\
&  \quad  = \int_I {({\xi ^\mu } \circ \tau _M^{2,0} \circ {\sigma ^2})\left\{ {\frac{{{\partial ^2}K}}{{\partial {x^\mu }\partial {y^\rho }}}d{x^\rho } + 2\frac{{{\partial ^2}K}}{{\partial {x^\mu }\partial {z^\rho }}}d{y^\rho } - d\left( {\frac{{\partial K}}{{\partial {y^\mu }}}} \right)} \right\} \circ {\sigma ^2}}  \hfill \\
&  \quad \quad  + \frac{{d({\xi ^\mu } \circ \sigma )}}{{dt}}\left( { - 2\frac{d}{{dt}}\left( {\frac{{\partial K}}{{\partial {z^\mu }}} \circ {\sigma ^2}} \right)} \right)dt \hfill \\
&  \quad \quad  + \frac{{{d^2}({\xi ^\mu } \circ \sigma )}}{{d{t^2}}}\left( { - \frac{{\partial K}}{{\partial {z^\mu }}} \circ {\sigma ^2}} \right)dt \hfill \\
&  \left. {\quad \quad  + \frac{d}{{dt}}\left( {\left( {{\xi ^\mu } \circ \tau _M^{2,0} \cdot \frac{{\partial K}}{{\partial {y^\mu }}} + 2D{\xi ^\mu } \circ \tau _M^{2,1} \cdot \frac{{\partial K}}{{\partial {z^\mu }}}} \right) \circ {\sigma ^2}} \right)dt} \right] \hfill \\
&  \quad  = \int_I {({\xi ^\mu } \circ \tau _M^{2,0} \circ {\sigma ^2})\left\{ {\frac{{{\partial ^2}K}}{{\partial {x^\mu }\partial {y^\rho }}}d{x^\rho } + 2\frac{{{\partial ^2}K}}{{\partial {x^\mu }\partial {z^\rho }}}d{y^\rho } - d\left( {\frac{{\partial K}}{{\partial {y^\mu }}}} \right)} \right\} \circ {\sigma ^2}}  \hfill \\
&  \quad \quad  + ({\xi ^\mu } \circ \sigma )\left( {\frac{{{d^2}}}{{d{t^2}}}\left( {\frac{{\partial K}}{{\partial {z^\mu }}} \circ {\sigma ^2}} \right)} \right)dt \hfill \\
&  \quad \quad  - \frac{{{d^2}}}{{d{t^2}}}\left( {\left( {{\xi ^\mu } \circ \tau _M^{2,0} \cdot \frac{{\partial K}}{{\partial {z^\mu }}}} \right) \circ {\sigma ^2}} \right)dt \hfill \\
&  \left. {\quad \quad  + \frac{d}{{dt}}\left( {\left( {{\xi ^\mu } \circ \tau _M^{2,0} \cdot \frac{{\partial K}}{{\partial {y^\mu }}} + 2D{\xi ^\mu } \circ \tau _M^{2,1} \cdot \frac{{\partial K}}{{\partial {z^\mu }}}} \right) \circ {\sigma ^2}} \right)dt} \right] \hfill \\ 
&  \quad  = \int_I {\left[ {({\xi ^\mu } \circ \tau _M^{2,0} \circ {\sigma ^2})\left[ {\left\{ {\frac{{{\partial ^2}K}}{{\partial {x^\mu }\partial {y^\rho }}}d{x^\rho } + 2\frac{{{\partial ^2}K}}{{\partial {x^\mu }\partial {z^\rho }}}d{y^\rho } - d\left( {\frac{{\partial K}}{{\partial {y^\mu }}}} \right)} \right\} \circ {\sigma ^2} + \frac{{{d^2}}}{{d{t^2}}}\left( {\frac{{\partial K}}{{\partial {z^\mu }}} \circ {\sigma ^2}} \right)dt} \right]} \right.}  \hfill \\
&  \left. {\quad \quad  + \frac{d}{{dt}}\left( {\left( {{\xi ^\mu } \circ \tau _M^{2,0} \cdot \frac{{\partial K}}{{\partial {y^\mu }}} + 2D{\xi ^\mu } \circ \tau _M^{2,1} \cdot \frac{{\partial K}}{{\partial {z^\mu }}}} \right) \circ {\sigma ^2} - \frac{d}{{dt}}\left( {\left( {{\xi ^\mu } \circ \tau _M^{2,0} \cdot \frac{{\partial K}}{{\partial {z^\mu }}}} \right) \circ {\sigma ^2}} \right)} \right)dt} \right] \hfill \\
&  \quad  = \int_I {({\xi ^\mu } \circ \tau _M^{2,0} \circ {\sigma ^2})\left[ {\left\{ {\frac{{{\partial ^2}K}}{{\partial {x^\mu }\partial {y^\rho }}}d{x^\rho } + 2\frac{{{\partial ^2}K}}{{\partial {x^\mu }\partial {z^\rho }}}d{y^\rho } - d\left( {\frac{{\partial K}}{{\partial {y^\mu }}}} \right)} \right\} \circ {\sigma ^2} + d\left( {\left( {{D_3}\frac{{\partial K}}{{\partial {z^\mu }}}} \right) \circ {\sigma ^3}} \right)} \right]}  \hfill \\
&  \left. {\quad \quad  + \frac{d}{{dt}}\left( {\left( {{\xi ^\mu } \circ \tau _M^{2,0} \cdot \frac{{\partial K}}{{\partial {y^\mu }}} + 2D{\xi ^\mu } \circ \tau _M^{2,1} \cdot \frac{{\partial K}}{{\partial {z^\mu }}}} \right) \circ \tau _M^{3,2} \circ {\sigma ^3} - \left( {{D_3}\left( {{\xi ^\mu } \circ \tau _M^{2,0} \cdot \frac{{\partial K}}{{\partial {z^\mu }}}} \right)} \right) \circ {\sigma ^3}} \right)dt} \right] \hfill \\
&  \quad  = \int_{{C^3}} {{\xi ^\mu } \circ \tau _M^{3,0}\left[ {\left\{ {\frac{{{\partial ^2}K}}{{\partial {x^\mu }\partial {y^\rho }}}d{x^\rho } - d\left( {\frac{{\partial K}}{{\partial {y^\mu }}}} \right) + 2\frac{{{\partial ^2}K}}{{\partial {x^\mu }\partial {z^\rho }}}d{y^\rho }} \right\} \circ \tau _M^{3,2} + d\left( {{D_3}\frac{{\partial K}}{{\partial {z^\mu }}}} \right)} \right]}  \hfill \\
&  \quad \quad  + \int_{{C^3}} {d\left( {{\xi ^\mu } \circ \tau _M^{3,0} \cdot \frac{{\partial K}}{{\partial {y^\mu }}} \circ \tau _M^{3,2} + 2D{\xi ^\mu } \circ \tau _M^{3,1} \cdot \frac{{\partial K}}{{\partial {z^\mu }}} \circ \tau _M^{3,2} - {D_3}\left( {{\xi ^\mu } \circ \tau _M^{2,0} \cdot \frac{{\partial K}}{{\partial {z^\mu }}}} \right)} \right)}  \hfill \\ 
\end{align*}

which gives (\ref{variation_2nd_mech}). 

%% file: thesis2012_main_hepth.bbl
\begin{thebibliography}{18}
\expandafter\ifx\csname natexlab\endcsname\relax\def\natexlab#1{#1}\fi
\providecommand{\url}[1]{\texttt{#1}}
\providecommand{\href}[2]{#2}
\providecommand{\path}[1]{#1}
\providecommand{\DOIprefix}{doi:}
\providecommand{\ArXivprefix}{arXiv:}
\providecommand{\URLprefix}{URL: }
\providecommand{\Pubmedprefix}{pmid:}
\providecommand{\doi}[1]{\href{http://dx.doi.org/#1}{\path{#1}}}
\providecommand{\Pubmed}[1]{\href{pmid:#1}{\path{#1}}}
\providecommand{\bibinfo}[2]{#2}
\ifx\xfnm\relax \def\xfnm[#1]{\unskip,\space#1}\fi
\bibitem[{Chern et~al.(2000)Chern, Chen and Lam}]{ChernChenLam}
\bibinfo{author}{S.S. Chern}, \bibinfo{author}{W.H. Chen},
  \bibinfo{author}{K.S. Lam}, \bibinfo{title}{Lectures on Differential
  Geometry}, \bibinfo{publisher}{World Scientific}, \bibinfo{year}{2000}.
\bibitem[{Chern and Shen(2005)}]{Chern2}
\bibinfo{author}{S.S. Chern}, \bibinfo{author}{Z.~Shen},
  \bibinfo{title}{RIEMANN-{F}INSLER GEOMETRY}, \bibinfo{publisher}{World
  Scientific}, \bibinfo{year}{2005}.
\bibitem[{Grigore(1992)}]{Grigore1}
\bibinfo{author}{D.R. Grigore}, \bibinfo{title}{A geometric lagrangian
  formalism for extended objects}, \bibinfo{journal}{Proc. Conf. Opava, Aug
  24-28, (1992)}  (\bibinfo{year}{1992}) \bibinfo{pages}{439--448}.
\bibitem[{Kawaguchi(1976)}]{AK5}
\bibinfo{author}{A.~Kawaguchi}, \bibinfo{title}{On the concepts and theories of
  higher order spaces}, \bibinfo{journal}{Periodica Mathematica Hungarica}
  \bibinfo{volume}{7} (\bibinfo{year}{1976}) \bibinfo{pages}{291--299}.
\bibitem[{Krupka(shed)}]{Krupka1}
\bibinfo{author}{D.~Krupka}, \bibinfo{title}{Advanced Analysis on Manifolds},
  \bibinfo{publisher}{Elsevier}, \bibinfo{year}{to be published}.
\bibitem[{Krupka and Krupka(2010)}]{Krupka-MKrupka}
\bibinfo{author}{D.~Krupka}, \bibinfo{author}{M.~Krupka},
  \bibinfo{title}{Higher order {G}rassmann fibrations and the calculus of
  variations}, \bibinfo{journal}{Balkan J. Geom. Appl.} \bibinfo{volume}{15}
  (\bibinfo{year}{2010}) \bibinfo{pages}{68--79}.
\bibitem[{Matsumoto(1975)}]{matsumoto2}
\bibinfo{author}{M.~Matsumoto}, \bibinfo{title}{Œv—Ê"÷•ªŠô‰½Šw (Metrical
  differential geometry)}, \bibinfo{publisher}{Shoukabou},
  \bibinfo{year}{1975}.
\bibitem[{Ootsuka(2012)}]{Oo1}
\bibinfo{author}{T.~Ootsuka}, \bibinfo{title}{New covariant {L}agrange
  formulation for field theories}, \bibinfo{journal}{arXiv:1206.6040v1}
  (\bibinfo{year}{2012}).
\bibitem[{Ootsuka and Tanaka(2010)}]{OT2}
\bibinfo{author}{T.~Ootsuka}, \bibinfo{author}{E.~Tanaka},
  \bibinfo{title}{{F}insler geometrical path integral}, \bibinfo{journal}{Phys.
  Lett. A} \bibinfo{volume}{374} (\bibinfo{year}{2010})
  \bibinfo{pages}{1917--1921}.
\bibitem[{Saunders(1989)}]{Saunders2}
\bibinfo{author}{D.~Saunders}, \bibinfo{title}{The Geometry of Jet Bundles},
  \bibinfo{publisher}{Cambridge University Press}, \bibinfo{year}{1989}.
\bibitem[{Saunders(2010)}]{Saunders1}
\bibinfo{author}{D.~Saunders}, \bibinfo{title}{Some geometric aspects of the
  calculus of variations in several independent variables},
  \bibinfo{journal}{Communications in Mathematics} \bibinfo{volume}{18}
  (\bibinfo{year}{2010}) \bibinfo{pages}{3--19}.
\bibitem[{Spivak(1999)}]{Spivak1}
\bibinfo{author}{M.~Spivak}, \bibinfo{title}{A Comprehensive Introduction to
  DIFFERENTIAL GEOMETRY}, \bibinfo{publisher}{Publish or Perish, Inc.},
  \bibinfo{address}{HOUSTON, Texas}, \bibinfo{year}{1999}.
\bibitem[{Tamassy(2007)}]{Tamassy3}
\bibinfo{author}{L.~Tamassy}, \bibinfo{title}{Finsler geometry in the tangent
  bundle}, \bibinfo{journal}{2007 {F}insler Geometry, Sapporo2005 - In Memory
  of Makoto Matsumoto, Advanced Studies in Pure Mathematics}
  \bibinfo{volume}{48} (\bibinfo{year}{2007}) \bibinfo{pages}{163--194}.
\bibitem[{Tanaka(2013)}]{Ta2}
\bibinfo{author}{E.~Tanaka}, \bibinfo{title}{General relativity by {K}awaguchi
  geometry}, \bibinfo{journal}{Proceedings for THE TIME MACHINE FACTORY 2012}.
  \bibinfo{note}{To appear in EPJ Web of Conferences}.
\bibitem[{Tanaka and Krupka(2013)}]{Ta3}
\bibinfo{author}{E.~Tanaka}, \bibinfo{author}{D.~Krupka}, \bibinfo{title}{On
  the structure of {F}insler and areal spaces}, \bibinfo{journal}{Proceedings
  for AGMP-8 Brno}. \bibinfo{note}{To appear in Miskolc Mathematical Notes}.
\bibitem[{Tanaka et~al.(2012)Tanaka, Ootsuka and Yahagi}]{OTY3}
\bibinfo{author}{E.~Tanaka}, \bibinfo{author}{T.~Ootsuka},
  \bibinfo{author}{R.~Yahagi}, \bibinfo{title}{{L}agrange formulation of
  {E}instein's general relativity using {K}awaguchi geometry},
  \bibinfo{journal}{Soryuushiron Kenkyu} \bibinfo{volume}{13}
  (\bibinfo{year}{2012}).
\bibitem[{Urban(2011)}]{Urban1}
\bibinfo{author}{Z.~Urban}, \bibinfo{title}{VARIATIONAL SEQUENCES IN MECHANICS
  ON {G}RASSMANN FIBRATIONS}, Ph.D. thesis, UNIVERSITY OF OSTRAVA,
  \bibinfo{year}{2011}.
\bibitem[{Urban and Krupka(2013)}]{Krupka-Urban}
\bibinfo{author}{Z.~Urban}, \bibinfo{author}{D.~Krupka}, \bibinfo{title}{The
  {Z}ermelo conditions and higher order homogeneous functions},
  \bibinfo{journal}{Publ. Math. Debrecen} \bibinfo{volume}{82}
  (\bibinfo{year}{2013}) \bibinfo{pages}{59--76}.
  \DOIprefix\doi{10.5486/PMD.2013.5265}.

\end{thebibliography}
